%% file: document.tex
\include{settings}	

\ifcompileiliad
	\documentclass[10pt]{book}
\else
 	\documentclass[12pt]{book}
\fi

\ifcompileiliad
	\usepackage[papersize={122mm,160mm},
				hmargin=0.2cm,
				vmargin=0.2cm,
				tmargin=1.0cm,
				bmargin=0.4cm
				]
				{geometry}
\else
	\usepackage[paper = a4paper,         
				twoside,                 
				bindingoffset = 2mm,     
				hmargin = 25mm,          
				vmargin = 25mm,          
				dvipdfm]
				{geometry}
\fi

\usepackage[T1]{fontenc}	
\usepackage{mathptmx}       
\usepackage[scaled]{helvet} 
\usepackage{sectsty}        
\allsectionsfont{\usefont{OT1}{phv}{bc}{n}\selectfont} 

\usepackage{graphicx}       
\usepackage[utf8]{inputenc} 
\usepackage{amsmath}
\usepackage{amsfonts}
\usepackage[mathcal]{euscript} 
\usepackage{booktabs}       
\usepackage{setspace}       
\usepackage{forloop}		

\usepackage{hyperref}
\usepackage{color}
\usepackage{enumerate}
\usepackage{amssymb} 
\usepackage{amsthm}
\usepackage{afterpage}
\usepackage{pdflscape}
\usepackage{makeidx}
\makeindex
\usepackage{mathrsfs}
\usepackage{epstopdf}
\usepackage{pdflscape}
\usepackage{caption} 
\usepackage{subcaption}
\usepackage{rotating}
\usepackage{float}
\usepackage{adjustbox}
\usepackage{array}
\usepackage{booktabs}
\usepackage{multirow}

\usepackage{fancyhdr}
\fancypagestyle{plain}{
	\fancyhead{} 
}
\include{macros}

\title{
{\fontsize{28}{30}\usefont{OT1}{phv}{bc}{n}\selectfont
General Slit Stochastic L\"owner Evolution and Conformal Field
Theory} \author{
	\textbf{Alexey Tochin}\vspace{3cm}\\
		\includegraphics[width=74mm]{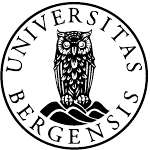}\vspace{1em}\\
		Dissertation for the degree of Philosophiae Doctor (PhD)\vspace{3.5em}\\
		Department of Mathematics \\
		University of Bergen
	}
	\date{September 2015}
}

\begin{document}

\ifDownscaledFinalDoc
	\fontsize{\TextSize}{\BaseLineSkip}
	\selectfont
\fi

\ifDraft
	\doublespacing
\fi

\maketitle
\frontmatter
\include{acknowledgements}
\tableofcontents

\mainmatter

%
%

\include{introduction}    	    

\include{Loewner_chain}

\include{SHSF}

\include{NumericalSimulation}

\include{BB_approach}

\include{CFT}

\include{Coupling}

\include{Classical_SLE}

%
%
\appendix
\include{Stochastic_calculus_stuff}

\backmatter

\listoffigures

\printindex

%
%

%
%
\bibliographystyle{alpha}

\bibliography{library}
	
\end{document}

%% file: settings.tex
\newif\ifcompileiliad 
	\compileiliadfalse 

\newcommand{\TextSize}{13}
\newcommand{\BaseLineSkip}{15}
\newif\ifDownscaledFinalDoc
	\DownscaledFinalDoctrue		

\newif\ifDraft
	\Draftfalse		

%% file: macros.tex
\newtheorem{theorem}{Theorem}[chapter]
\newtheorem{proposition}{Proposition}[chapter]
\newtheorem{corollary}{Corollary}[chapter]
\newtheorem{lemma}{Lemma}[chapter]
\newtheorem{remark}{Remark}[chapter]

\newtheorem{definition}{Definition}[chapter]
\newtheorem{conjecture}{Conjecture}[chapter]
\newtheorem{example}{Example}[chapter]

\newcounter{ct}

\newcommand{\includepaperpages}[2]
{
	\forloop{ct}{0}{\value{ct} < #2}
	{
		\begin{figure}[ht]
			\includegraphics[width=.99\textwidth]{#1\the\value{ct}.ps}
		\end{figure}
		\clearpage
	}
}

\renewcommand{\emph}[1]{\textbf{#1}}
\renewcommand{\Re}{\,\mathrm{Re}\,}
\renewcommand{\Im}{\,\mathrm{Im}\,}
\renewcommand{\th}{\,\mathrm{th}\,}
\newcommand{\sh}{\,\mathrm{sh}\,}
\newcommand{\ch}{\,\mathrm{ch}\,}
\renewcommand{\th}{\,\mathrm{th}\,}
\newcommand{\cth}{\,\mathrm{cth}\,}

\newcommand{\tg}{\,\mathrm{tg}\,}
\newcommand{\ctg}{\,\mathrm{ctg}\,}
\newcommand{\arcth}{\,\mathrm{arcth}\,}
\newcommand{\arctg}{\,\mathrm{arctg}\,}

\DeclareMathOperator{\id}{id}

\DeclareMathOperator{\de}{\partial}
\DeclareMathOperator{\dS}{d^{\text{S}}}
\DeclareMathOperator{\dI}{d^{\text{It\^{o}}}}
\DeclareMathOperator{\Lc}{\mathcal{L}}
\DeclareMathOperator{\Ac}{\mathcal{A}}
\DeclareMathOperator{\Dc}{\mathcal{D}}

\DeclareMathOperator{\D}{\mathbb{D}}
\DeclareMathOperator{\C}{\mathbb{C}}
\DeclareMathOperator{\HH}{\mathbb{H}}

\DeclareMathOperator{\SSS}{\mathbb{S}}
\DeclareMathOperator{\LL}{\mathbb{L}}
\DeclareMathOperator{\Hc}{\mathcal{H}}
\DeclareMathOperator{\K}{\mathcal{K}}

\DeclareMathOperator{\map}{\rightarrow}
\DeclareMathOperator{\then}{~\Rightarrow~}

\newcommand{\Ev}[1]{\mathbb{E}\left[{#1}\right]}
\newcommand{\Evv}[2]{\mathbb{E}_{#1}\left[{#2}\right]}
\DeclareMathOperator{\supp}{\mathrm{supp }}
\newcommand{\law}[1]{\mathrm{Law}[#1]}

%% file: acknowledgements.tex
\chapter{Acknowledgements}

I would like to express my sincere gratitude to everyone who helped me along
the way to this thesis. Throughout this journey, I was fortunate to receive
excellent guidance and generous support.

First and foremost, I am deeply grateful to my
supervisor, \emph{Alexander Vasil'ev}. My study of mathematics at the
University of Bergen would not have been possible without him. I thank
Alexander for his generous help and support, and for always being ready to read
my early drafts, full of misprints and technical mistakes.

I am also grateful to my co-author, \emph{Georgy Ivanov}, for our numerous
lively discussions. I would like to thank \emph{Bruno Carneiro da Cunha} and
\emph{Tiago Anselmo} for their contributions and for the valuable discussions
we had during my stay in Brazil.

Special thanks go to my colleagues and friends who read the text at the
final stage and corrected numerous grammatical and stylistic errors. This
painstaking work was done by \emph{Anastasia Frolova}, \emph{Valentin
Krasontovitsch}, \emph{Daulet Moldabayev}, and, once again, \emph{Alexander Vasil'ev}.

Finally, I would like to thank my colleagues at the Department of Mathematics, including
\emph{Irina Markina}, \emph{Bjørn Ian Dundas}, \emph{Henrik Kalisch},
\emph{Mauricio Godoy Molina}, \emph{Erlend Grong}, \emph{Christian Autenried},
\emph{Mirjam Solberg}, \emph{Sergey Alyaev}, \emph{Viktor Kiselev},
\emph{Anna Varzina}, \emph{Anna Kvashchuk}, \emph{Lina Astrakova}, and many
others.

%% file: introduction.tex
\chapter*{Introduction}
\addcontentsline{toc}{chapter}{Introduction}

The thesis is presented in the style of a monograph and is dedicated to a
generalization of the L\"owner equation in its stochastic form, known as the
Schramm--L\"owner evolution (SLE), and to its coupling with the Gaussian free field,
with the ultimate aim of constructing a boundary conformal field theory
with one free scalar
\footnote{More precisely, pre-pre-Schwarzian; see Section 
\ref{Section: Pre-pre-Schwarzian}.}
bosonic field. 
The exposition aims to be a
systematic, and hopefully concise, account and generalization of known
elements of the theory of L\"owner evolution. This topic is closely related
to lattice models of statistical physics in their scaling limit. However, we
deal only with continuous models, focusing, in particular, on the domain Markov
property, the Gaussian free field, etc. Numerical simulations are presented in
Chapter 
\ref{Chapter: Numerical Simulation}.

The main results are proofs of the basic properties of  
($\delta,\sigma$)-SLE and general necessary and sufficient conditions for its 
coupling with the Gaussian free field. We also introduce a framework that allows one to treat all L\"owner
equations and known types of coupling as different manifestations of the
same phenomenon.

The thesis is divided into six chapters, each of which contains a short
introductory part outlining the general ideas and listing the main results. Chapter
\ref{Chapter: Slit Loewner equation and its stochastic version}
(deterministic and stochastic L\"owner equations) and Chapter 
\ref{Chapter: CFT} 
(Gaussian free
field and conformal field theory) are self-contained and can be studied
independently. The coupling with the Gaussian free field (and hence with CFT) is considered in
Chapter 
\ref{Chapter: Coupling}, 
where the results of Chapters 
\ref{Chapter: Slit Loewner equation and its stochastic version}
and 
\ref{Chapter: CFT}  
come together.
Another aspect of the same connection to conformal field theory, related to the
representation theory of the Virasoro algebra, is discussed
in Chapter
\ref{Chapter: Representation theory approach}.
Important special cases of the general slit stochastic L\"owner evolution and
its possible couplings with CFT are collected in Chapter 
\ref{Chapter: Classical cases}.


%% file: Loewner_chain.tex
\chapter{Slit L\"owner equation and its stochastic version}
\label{Chapter: Slit Loewner equation and its stochastic version}

This chapter mostly repeats 
\cite{Ivanov2014}. Here, we use a different approach; in particular, we focus on
working in a generic domain rather than the unit disk. We also present some
additional results such as Theorem
\ref{Theorem: Absolute continuity of SLE}
and discuss the \emph{domain Markov property} in detail.  
   
The L\"owner theory was introduced in 1923 by Karl L\"owner (Charles Loewner)
\cite{Lowner}, and was later developed by Kufarev 
\cite{Kufareff1943} 
and Pommerenke
\cite{Pommerenke.1975}.
The L\"owner differential equation became one of the most powerful tools for
solving extremal problems in the theory of univalent functions. In the modern
period, L\"owner theory has again attracted a lot of interest due to the
discovery of the Stochastic (Schramm)-L\"owner Evolution (SLE), a stochastic
process that has made it possible to describe analytically the scaling limits of
several 2-dimensional lattice models in statistical physics, see 
\cite{Lawler2001b}
and
\cite{Schramm2000}. 
SLE theory focuses on describing probability measures on families of curves
which possess the property of \emph{conformal invariance} and the \emph{domain
Markov property}. The second property, in its turn, is related to the diffusion
form of SLE. So far, the following types of SLE have been studied:
the \emph{chordal} SLE
\cite{Lawler2001b,Rohde2005,Schramm2000},
the \emph{dipolar} SLE
\cite{Bauer2004a},
the \emph{radial} SLE 
\cite{Lawler2008,Lawler2001b},
SLE$(\kappa,\rho)$
\cite{Dubedat2005},
and multiple SLE
\cite{Dubedat2004}.

In this thesis we address the following questions. What are other possible
diffusion equations with holomorphic coefficients that generate random families
of curves? How similar are the properties of these families of curves to the
properties of known SLE curves? 

We start by introducing notation and basic concepts. In the second section,
we introduce 
\emph{$(\delta,\sigma)$-L\"owner chains} 
(or \emph{slit L\"owner chains}) 
and their stochastic versions are given in the third section.
This construction includes the 
\emph{chordal}, 
\emph{dipolar}, 
and 
\emph{radial}
L\"owner equations as special cases. 
SLE$(\kappa,\rho)$ and multiple SLE are beyond the scope of this thesis because
they cannot be formulated as a single diffusion equation (rather than a system). In
particular, the SLE$(\kappa,\rho)$ measure on curves does not possess the domain Markov
property defined in terms of slits only because the conditional random law
depends not only on the fixed part of the curve but also on some additional parameters.

We try to keep the presentation as self-contained as possible, avoiding,
however, a lengthy introduction and referring instead to, e.g.,  
\cite{Lawler2008}. Stochastic calculus, including It\^{o} and Stratonovich differentiation,
is assumed to be familiar to the reader. Some additional facts are given
in Appendix \ref{Appendix: Some relations from stochastic calculus}.

\section{Preliminaries}
\label{Section: SLE preliminaries}

Each version of the L\"owner equation, and, more generally, of
holomorphic stochastic flows, is usually associated with a certain canonical
domain $D\subset \mathbb C$ in the complex plane, e.g., the upper half-plane
(in the case of the chordal SLE) or the unit disk (in the case of the radial
SLE), etc. Focusing on conformally invariant properties, we avoid this specific
choice and map the canonical domains onto one another if necessary. For example,
the number of fixed points of the flow or the algebraic properties of vector
fields that define the flow are presented in this invariant way. The invariance
is achieved by considering a general hyperbolic simply connected domain, or a
chart, from the very beginning. It is also natural to go further and work with
a simply connected hyperbolic Riemann surface 
$\Dc$ 
\index{$\Dc$}
(with a boundary $\de\Dc$). Below,
\index{$\de \Dc$}
$\Dc$ is understood as a generic domain with a well-defined boundary as well as
a Riemann surface.
A simply connected hyperbolic Riemann surface with a
boundary is denoted by 
$\bar\Dc=\Dc\cup\de\Dc$,
\index{$\bar \Dc$} 
where $\de\Dc$ is the boundary
and $\Dc$ is the open interior.
We will mostly use global chart maps 
$\psi\colon\Dc\map D^{\psi}\subset \mathbb{C}$ from 
\index{chart map $\psi$}
$\Dc$ to a domain of the complex plane, 
writing $\psi$ for a chart 
$(D,\psi)$ for simplicity.
For any other chart
$\tilde\psi\colon\Dc\map D^{\tilde\psi}\subset \mathbb{C}$
\index{$D^{\psi}$}
the transition map is denoted by
\index{crossing map $\tau$}
\begin{equation}
	\tau^{\psi,\tilde\psi}:=\psi \circ \tilde\psi^{-1}\colon 
	 D^{\tilde\psi}\map D^{\psi}.
\end{equation}
For example, we will often use the half-plane chart map 
$\psi_{\HH}\colon\Dc\map\HH$,
\index{half-plane chart}
\begin{equation}
	\HH:=\{z\in\mathbb{C}\colon \Im(z)>0\}.
\end{equation}
\index{$\HH$}
Another example is the unit disk chart map
\index{unit disk chart}
$\psi_{\D}\colon\Dc\map\D$, where
\begin{equation}
	\D:=\{z\in\mathbb{C}\colon |z|<1\}.
\end{equation} 
\index{$\D$}
These charts are related by the transition map 
\begin{equation}
 \tau^{\HH,\D}(z)
 :=\psi^{\HH} \circ \left(\psi^{\D}\right)^{-1} =
 i \frac{1-z}{1+z}
 \colon\D \map\HH.
 \label{Formula: tau_H D}
\end{equation}

Thus, the point $z=1$ in the unit disk chart corresponds to the origin in the
half-plane chart and the point $z=-1$ corresponds to infinity.
We will also use a non-global multivalued chart 
$\psi^{\mathbb{L}}:\Dc\map\HH$ in 
Section \ref{Section: Radial case}.

Consider now a holomorphic vector field $v$ on $\Dc$, that is, a
holomorphic section of the complexified tangent bundle.
\index{holomorphic vector field}
We can also define it
as a map 
$\psi\mapsto v^{\psi}$ 
from the set of all possible global charts 
$\psi:\Dc\map D^{\psi}$ 
to the set of holomorphic functions 
$v^{\psi}:D^{\psi}\map\C$
defined on $D^{\psi}:=\psi(\Dc)$.
For the vector fields, the following coordinate change holds. Any chart map
$\tilde \psi:\Dc\map D^{\tilde \psi}$
induces the transition
\begin{equation}
	v^{\tilde \psi}(\tilde z) =
	v^{\tau^{-1} \circ \psi}(\tilde z) =
	\frac{1}  { \tau'(\tilde z)} 
	v^{\psi}(\tau(\tilde z)),\quad
	\tau :=\psi\circ\tilde\psi^{-1} \colon D^{\tilde \psi}\map D^{\psi},
	\label{Formula: tilde v = 1/dtau v(tau)}
\end{equation}
If $v$ is defined for one chart, then it is automatically defined on all other
charts.
A vector field can also be called a $(-1,0)$-differential.

Consider now a conformal map
$F\colon\Dc\map \tilde\Dc$ 
between two hyperbolic simply connected Riemann surfaces
(or generic domains)
$\Dc$ and $\tilde \Dc$, 
and let 
$\psi:\Dc\map D^{\psi}\subset \C$,
and let
$\tilde \psi:\tilde \Dc\map \tilde D^{\tilde\psi}\subset \C$ 
be chart maps. Define
\begin{equation}
	F^{\tilde\psi,\psi}:= \tilde \psi \circ F \circ \psi^{-1}
	\colon {D}^{\psi} \map \tilde D^{\tilde \psi},
\end{equation} 
consequently,
\begin{equation}
	\left( F^{\tilde\psi,\psi} \right)^{-1} =	
	\left( F^{-1} \right)^{\psi,\tilde \psi}= \psi \circ F^{-1} \circ \tilde \psi^{-1}
	\colon \tilde D^{\tilde\psi} \map D^{\psi}.
\end{equation}
The pushforward 
$F_*:v^{\psi} \mapsto \tilde v^{\tilde \psi}$ 
is defined by the rule
\begin{equation}\begin{split}
	v^{\psi}(z) \map \tilde v^{\tilde \psi}(\tilde z) 
	=&	(F_* v)^{\tilde	\psi}(\tilde z)
	:= v^{\tilde \psi \circ F}(\tilde z) = 
	v^{\tilde \psi \circ F \circ \psi^{-1} \circ \psi}(\tilde z) =
	v^{F^{\tilde \psi,\psi} \circ \psi}(\tilde z) 
	=\\=&
	\frac{1}{\left(\left(F^{\tilde\psi,\psi}\right)^{-1}\right)'(\tilde z)} 
	v^{\psi}\left(\left(F^{\tilde\psi,\psi}\right)^{-1}(\tilde z)\right),\quad 
	\tilde z\in \tilde D^{\tilde\psi},
\end{split}\end{equation} 
because 
$\left(F^{\tilde\psi,\psi}\right)^{-1}$ 
plays the same role as $\tau$ in 
\eqref{Formula: tilde v = 1/dtau v(tau)}.

Now let $F$ be an endomorphism of $\Dc$, namely,
$F(\Dc)=\tilde \Dc=\Dc\setminus \mathcal{K}$ 
for some compact subset 
$\mathcal{K}$,
$\tilde \psi \equiv \psi|_{\Dc \setminus \mathcal{K}}$,
and
$F^{\psi}:=F^{\psi,\psi}$.
Then,
\begin{equation}\begin{split}
	&F_* v^{\psi} (z) = v^{\psi \circ F} (z) =
	\frac{1}{  \left(\left(F^{\psi}{}\right)^{-1}\right) {}'(z) } v^{\psi}
	\left( 		\left(\left(F^{\psi}\right)^{-1}\right)(z)\right),\\
	&z\in \psi(\tilde \Dc) = \psi(\Dc \setminus \mathcal{K}) \subset D^{\psi},
	\label{Formula: G v(z) = 1/G' v(G(z))}
\end{split}\end{equation}
is a vector field defined in 
$\Dc\setminus \mathcal{K}$.
For the inverse map 
$F^{-1}\colon \Dc\setminus\K\map \Dc$ 
the vector field 
$F^{-1}_*v^{\psi}$ 
is defined in the entire 
$\Dc$ 
but the values of
$v$ 
in 
$\K$ 
are not taken into account.
It is also easy to see that
\begin{equation}
 \tilde F_* F_* = \left(F \circ \tilde F \right)_*. 
\end{equation}


The pushforward $F_*$
\index{pushforward}
can also be understood as a bundle map between the corresponding tangent bundles
induced by $F$.
We stick to the approach above because it can be extended to more general
transformations with respect to the change of charts such as the
\emph{pre-pre-Schwarzian},
\index{pre-pre-Schwarzian}
see Section \ref{Section: Pre-pre-Schwarzian}.
We will use the notation 
$X^{\HH}:=X^{\psi_{\HH}}$ 
if 
$X$ 
is a vector field, conformal map or pre-pre-Schwarzian (defined below) in 
$\Dc$ 
as well as
$X^{\D}:=X^{\psi_{\D}}$ 
for the unit disk chart, and similarly, for other standard charts.

It will also be convenient to use a basis of holomorphic vector fields
given by holomorphic functions as
\begin{equation}
	\ell_n^{\HH}(z) := z^{n+1} ,\quad z\in\HH
	\label{Formula: ell_n^H = ...}
\end{equation}
in the half-plane chart.
In the unit disk chart they take the form
\begin{equation}
 	\ell_n^{\D}(z) = -\frac{i^{n}}{2}(1-z)^{1+n}(1+z)^{1-n} 
	\label{Formula: ell_n^D = ...}
\end{equation}
according to
\eqref{Formula: tau_H D}
and
\eqref{Formula: tilde v = 1/dtau v(tau)}.
We remark that $\ell_n$ are holomorphic in $\Dc$, tangent to the boundary
except at two points, which are the origin and infinity in the half-plane chart
or $z=\pm 1$ in the unit disk chart. At these points, $\ell_n$ has critical
points of order $1+n$ and $1-n$, respectively, which are zeros for positive
order and poles for negative order. 

A holomorphic  vector field $\sigma$ in $\Dc$ is called 
\emph{complete} 
\index{complete vector field}
\index{$\sigma$}
if the solution 
$H_t[\sigma]^{\psi}(z)$ 
of the initial value problem 
\begin{equation}
 \dot H_t[\sigma]^{\psi}(z) = \sigma^{\psi}(H_t[\sigma]^{\psi}(z)),\quad 
 H_0[\sigma]^{\psi}(z) = z,\quad 
 z\in D^{\psi}
 \label{Formula: d H[sigma](z) = sigma( H[sigma] ) (z)}
\end{equation}
is defined for
$t\in(-\infty,\infty)$ 
as a conformal automorphism 
$H_t[\sigma]^{\psi}:D^{\psi}\map D^{\psi}$. 
Here and below, we denote the partial derivative with respect to $t$ (called
\emph{time} henceforth)
\index{time} 
as
$\dot H_t:=\frac{\de}{\de t} H_t$. 
It is straightforward to see that the differential equation has the same form
in any chart $\psi$. That is why it is reasonable to drop the index $\psi$
and the argument $z$ as 
\begin{equation}
 \dot H_t[\sigma] = \sigma \circ H_t[\sigma] ,\quad 
 H_0[\sigma] = \id.
 \label{Formula: d H = sigma H ds}
\end{equation}
We will use the notation `$\circ$' in what follows. The advantage is the explicit
independence of the choice of a chart.

The collection $\{H_t[\sigma]\}_{t\in\mathbb{R}}$ forms a one-parameter
group:
\begin{equation}
	H_t[\sigma]\circ H_s[\sigma]=H_{t+s}[\sigma],\quad
	H_t^{-1}[\sigma] = H_{-t}[\sigma],\quad t,s\in(-\infty,+\infty).
\end{equation}
Besides, for any complete vector field $\sigma$ we have 
\begin{equation}
 H_t[\sigma]_* \sigma = \sigma,\quad t\in(-\infty,+\infty).
 \label{Formula: H[v] v = v}
\end{equation}

It is possible to show, see \cite{Shoikhet2001}, that a complete vector field
$\sigma$ is a linear combination of $\ell_{-1}$, $\ell_{0}$, and $\ell_{1}$
with real coefficients
\begin{equation}
	\sigma = \sigma_{-1} \ell_{-1} + \sigma_{0} \ell_{0} + \sigma_{1} \ell_{1},\quad
	\sigma_{-1}, \sigma_0, \sigma_1 \in \mathbb{R}.
	\label{Formula: v = v l_-1 + v l_0 + v l_1}
\end{equation}
In the half-plane chart this relation reads 
\begin{equation}
	\sigma^{\HH}(z) = \sigma_{-1} + \sigma_0 z + \sigma_1 z^2,\quad z\in \HH,\quad
	\sigma_{-1}, \sigma_0, \sigma_1 \in \mathbb{R}
 \label{Formula: v = v l_-1 + v l_0 + v l_1 in H}
\end{equation}
due to 
\eqref{Formula: ell_n^H = ...}.
It is also true that a holomorphic vector field is complete if and only if it
is holomorphic and tangent to the boundary.

Define the parameter
\begin{equation}
	\Delta_2 := \sigma_0^2 - \sigma_1 \sigma_{-1}.
\end{equation}
It is invariant with respect to the choice of basis $\ell_n$, or equivalently,
M\"obius automorphisms of $\Dc$. We distinguish three cases:
\begin{enumerate}
	\item $\Delta_2=0$. We call this case \emph{parabolic}. The vector field
	$\sigma$ has one zero of order 2 at the boundary.
	\index{parabolic $\sigma$}
	\item $\Delta_2>0$. We call this case \emph{hyperbolic}. The vector field
	$\sigma$ has two zeros of order 1 at the boundary.
	\index{hyperbolic $\sigma$}
	\item $\Delta_2<0$. We call this case \emph{elliptic}. The vector
	field $\sigma$ has one zero inside $\Dc$.
	\index{elliptic $\sigma$}
\end{enumerate}

We illustrate each type in Fig.
\ref{Figure: Complete vector fields}.
\begin{figure}[h]
\centering
  \begin{subfigure}[t]{0.33\textwidth}
	\centering
	\captionsetup{justification=centering}
  \includegraphics[width=5cm,keepaspectratio=true]
   	{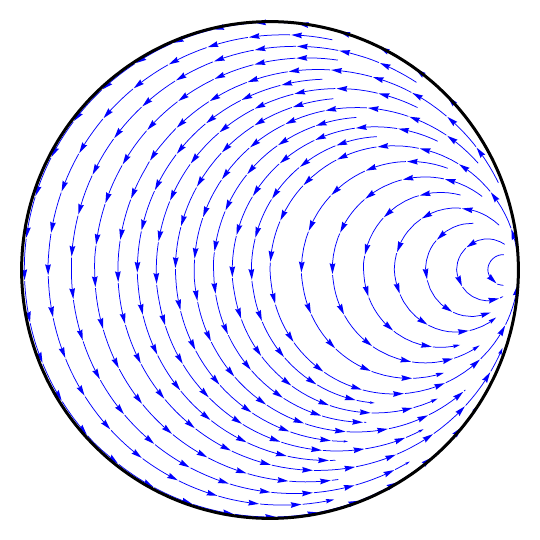}
    \caption{ $\ell_{1}$\\
    Parabolic vector field,\\ 
    {\centering$\Delta_2=0$}.}
  \end{subfigure}%
  ~  
  \begin{subfigure}[t]{0.33\textwidth}
	\centering
	\captionsetup{justification=centering}
  \includegraphics[width=5cm,keepaspectratio=true]
   	{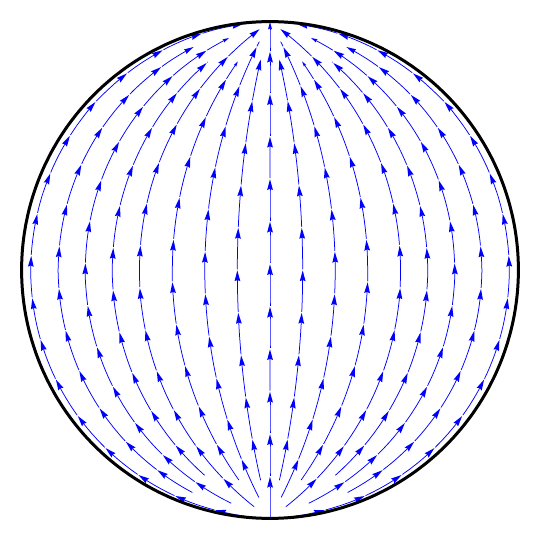} 
  	\caption{
  	$\ell_{-1}-\ell_{1}$\\
   	Hyperbolic vector field,\\  
  	$\Delta_2>0$.}  	
	\end{subfigure}%
	~
	\begin{subfigure}[t]{0.33\textwidth}
	\centering
	\captionsetup{justification=centering}
	\includegraphics[width=5cm,keepaspectratio=true]
		{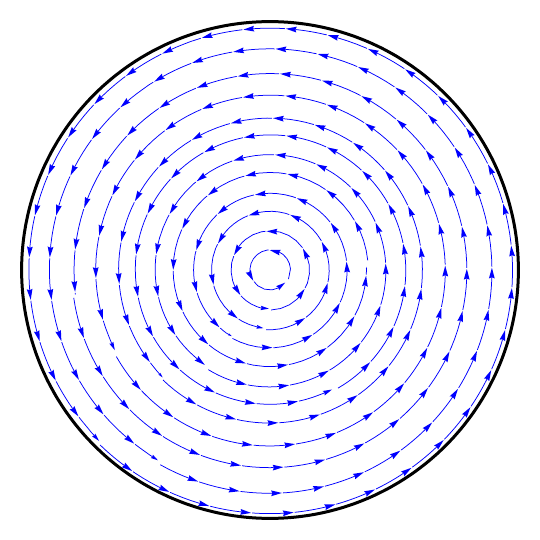}
		\caption{$\ell_{-1}+\ell_{1}$\\
		Elliptic vector field,\\ 
    $\Delta_2<0$.}
	\end{subfigure}%
 	\caption{Examples of the three types of complete vector fields in the unit disk
 	chart.
 	\label{Figure: Complete vector fields}}
\end{figure}
We note that any complete vector field can be reduced to one of those in
Fig.
\ref{Figure: Complete vector fields}
up to a constant factor by a M\"obius transformation. We remark that $H_t[\sigma]$ is
periodic with respect to $t$ if and only if $\sigma$ is elliptic.

A holomorphic vector field $\delta$ is called \emph{semicomplete}
\index{semicomplete vector field}
if the initial value problem
\eqref{Formula: d H = sigma H ds}
with $\delta$ in place of $\sigma$
has a solution 
$H_t[\delta]$, 
which is a conformal map 
$H_t[\delta]:\Dc\map\Dc\setminus K_t$ 
for all 
$t\in[0,+\infty)$ 
and for some family 
$\{\K_t\}_{t\geq 0}$
of subsets
$\K_t\subset \Dc$. 
The equation 
\eqref{Formula: d H = sigma H ds}
also has a solution for $t\in(-\infty,0]$ 
but $\{H_t[\delta]\}_{t\leq 0}$ 
is the family of inverse endomorphisms
$H_t[\delta]:\Dc\setminus \K_t\map\Dc$.

\emph{Antisemicomplete} vector fields can be defined simply as the negatives of
semicomplete fields.
\index{antisemicomplete vector field}
\index{$\delta$}
Thus, $-\delta$ is an antisemicomplete vector field if and only if $\delta$
is a semicomplete field, and vice versa.
Equivalently, we can define an antisemicomplete vector field by assuming that
$\{H_t[\delta]\}_{t\leq 0}$
is a collection of endomorphisms or, which is the same, 
$\{H_t[\delta]\}_{t\geq 0}$ 
is a collection of inverse endomorphisms.
A complete field is, in particular, semicomplete and antisemicomplete at the
same time.

The collection $\{H_t[\delta]\}_{t\geq 0}$ for semicomplete $\delta$ is a
one-parameter semigroup of endomorphisms with respect to composition.
Analogously, for an antisemicomplete $\delta$ the collection
$\{H_t[\delta]\}_{t\geq 0}$ 
is also a one-parameter semigroup of inverse endomorphisms with respect to
composition.

\begin{proposition} (Berkson-Porta representation) \cite{Berkson}\\
A vector field $\delta$ is semicomplete if and only if 
\begin{equation}
	\delta^{\D}(z) = (z-\tau)(\bar\tau z - 1)p(z),
	\quad \tau\in\bar\D, \quad \Re p(z) \geq 0, \quad z\in\bar \D,
\end{equation}	 
in the unit disk chart for some holomorphic $p:\bar\D\map\C$.
\label{Proposition: Berkson-Porta representation}
\end{proposition}

The space of all semicomplete fields is much larger than the space of
complete fields; it is infinite-dimensional. We restrict ourselves to the
case of fields which are holomorphic in $\Dc$ and tangent to the boundary
except at one point $a\in\de\Dc$, which is called the \emph{source}
\index{source}
below.

\begin{proposition}
A semicomplete vector field $\delta$ is holomorphic in $\Dc$ and tangent to
the boundary except at one point $a\in\de\Dc$ if and only if it is of the form
\begin{equation}
 \delta =  
 \delta_{-2} \ell_{-2} + \delta_{-1} \ell_{-1} + \delta_0 \ell_{0} 
 + \delta_1 \ell_{1},\quad
	\delta_{-1}, \delta_0, \delta_1 \in \mathbb{R},\quad \delta_{-2}\leq 0.
 \label{Formula: delta with 1 simple pole}
\end{equation}
\end{proposition}

\begin{proof}
Without loss of generality assume that $\tau=1$ and $p(1)=0$ under the
conditions of Proposition \ref{Proposition: Berkson-Porta representation}, which
can be achieved by adding an appropriate complete field.

Consider now a chart 
\begin{equation}
	\psi^{\tilde \HH}:=-\frac{1}{\psi^{\HH}} :\Dc\map \HH,
\end{equation}
which is similar to 
$\psi^{\HH}$ 
defined above, but it maps the source point $a$ to infinity 
$\psi^{\tilde \HH}(a)=\infty$.
The function 
$\delta^{\tilde \HH}(z)$ 
can be expressed as
\begin{equation}
	\delta^{\tilde \HH}(z)
	= 2 i\, z^2 p \left(\frac{i-z}{i+z}\right),\quad z\in\HH.
\end{equation}
On the other hand, 
$\delta^{\tilde \HH}(z)$ 
is an entire function after the
Schwarz reflection to the lower half-plane. 
The Taylor series of 
$\delta^{\tilde \HH}(z)$ 
and of 
$i\,p \left(\frac{i-z}{i+z}\right)$
about $z=0$ have real coefficients because 
$\delta$ 
is tangent to the boundary. The conditions 
$\Re \left(p \left(\frac{i-z}{i+z}\right)\right) \geq 0$ 
and $p(1)=0$ leave
a unique possibility:
\begin{equation}
	i\, p \left(\frac{i-z}{i+z}\right) = c\,z,\quad c>0.
\end{equation} 
We conclude that
\begin{equation}
	\delta^{\tilde \HH}(z) = 2 c\, z^3,\quad c>0.
\end{equation}
After the coordinate change $z \mapsto 1/z$ to the standard half-plane chart and
adding a generic complete part, we obtain 
\eqref{Formula: delta with 1 simple pole}.
\end{proof}

In the half-plane chart, we have 
\begin{equation}
 \delta^{\HH}(z) = 
 \frac{\delta_{-2}}{z} + \delta_{-1} + \delta_0 z + \delta_1 z^2
 ,\quad z\in \HH,\quad
 \delta_{-1}, \delta_0, \delta_1 \in \mathbb{R},\quad \delta_{-2}\leq 0.
 \label{Formula: delta with 1 simple pole in H}
\end{equation}
Thus, $\delta$ has a simple pole at $a$ with a non-positive residue.
The sum of the last three terms is just a complete field.

We note that if a vector field $v$ has a zero at $z_0$, then the parameter 
$v^{\psi}{}'(z_0)$ does not depend on the choice of chart. We call a zero of a
vector field \emph{attracting}
\index{attracting zero of vector field} 
if $\Re v^{\psi}{}'(z_0)> 0$.  

Analogously to the case of complete fields, define the parameter
\begin{equation}
	\Delta_3 := 
	18 \delta_{-1} \delta_{0} \delta_{1} 
	- 4 \delta_0^3 
	+ \frac{\delta_{-1}^2 \delta_0^2}{\delta_{-2}}
	- 4 \frac{\delta_{-1}^3 \delta_1}{\delta_{-2}} 
	- 27 \delta_{-2} \delta_{1}^2   
\end{equation}
It is likewise invariant with respect to the choice of basis $\ell_n$, or
equivalently, M\"obius automorphisms of $\Dc$ preserving the pole position. We
distinguish 3 cases:
\begin{enumerate}
	\item $\Delta_3=0$. We call this case \emph{parabolic}. The vector field
	$\delta$ has either one zero of order 2 at the boundary and one
	zero of order 1 at another boundary point, or one zero of order 3 at the
	boundary.
	\index{parabolic $\sigma$}
	\item $\Delta_3>0$. We call this case \emph{hyperbolic}. The vector field
	$\delta$ has three distinct zeros of order 1 at the boundary. The middle
	one is attracting.
	\index{hyperbolic $\sigma$}
	\item $\Delta_3<0$. We call this case \emph{elliptic}. The vector
	field $\delta$ has one attracting zero of order 1 inside $\Dc$ and one zero
	of order 1 at the boundary.
	\index{hyperbolic $\sigma$}
\end{enumerate}
We illustrate each type in Fig.
\ref{Figure: Semicomplete vector fields}.
\begin{figure}[h]
\centering
  \begin{subfigure}[t]{0.33\textwidth}
	\centering
	\captionsetup{justification=centering}
  \includegraphics[width=5cm,keepaspectratio=true]
   	{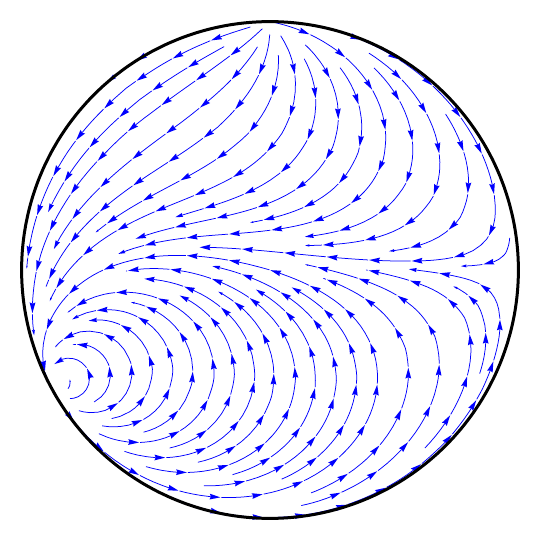}
    \caption{Parabolic case,\\ 
    {\centering$\Delta_3=0$}.}
  \end{subfigure}%
  ~  
  \begin{subfigure}[t]{0.33\textwidth}
	\centering
	\captionsetup{justification=centering}
  \includegraphics[width=5cm,keepaspectratio=true]
   	{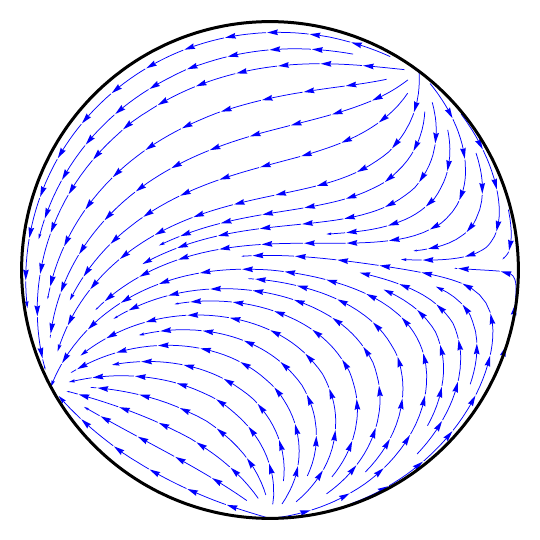} 
  	\caption{Hyperbolic case,\\  
  	$\Delta_3>0$.}  	
	\end{subfigure}%
	~
	\begin{subfigure}[t]{0.33\textwidth}
	\centering
	\captionsetup{justification=centering}
	\includegraphics[width=5cm,keepaspectratio=true]
		{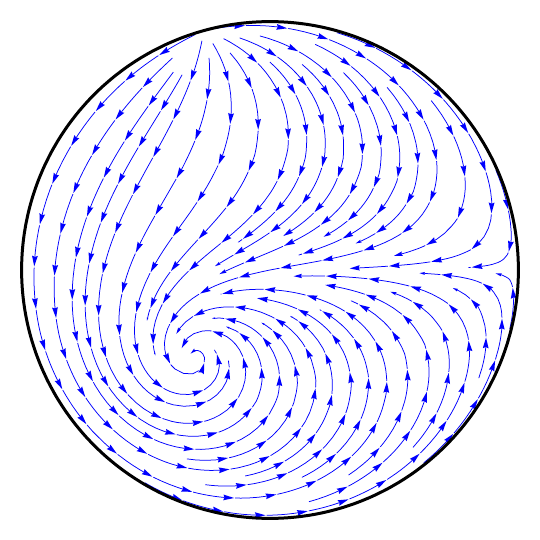}
		\caption{Elliptic case,\\ 
    $\Delta_3<0$.}
	\end{subfigure}%
 	\caption{Examples of the three types of vector fields $\delta$ defined by 
 	\eqref{Formula:	delta with 1 simple pole} 
 	in the unit disk chart. Each field has a simple pole at $\psi^{\D}(a)=1$. The
 	attracting zeros of order one correspond to singular points with divergent
 	arrows.
 	\label{Figure: Semicomplete vector fields}}
\end{figure}

The solution $H_t[\delta]$ of 
\eqref{Formula: d H = sigma H ds}
with $\delta$ from 
\eqref{Formula: delta with 1 simple pole}
is a conformal endomorphism 
$H_t[\delta]\colon \Dc\map\Dc\setminus\gamma_t$.
Here, 
$\{\gamma_t\}_{t\in[0,+\infty)}$
is a family of curves such that
$\gamma_t\subset\gamma_s$,
$t<s$ 
and
$\gamma_0=\emptyset$. 
The curves start 
from $a$, lie along a flow line of 
$\delta$ 
and tend to the attracting (or degenerate) zero when 
$t\map+\infty$.
We call such a family $\{\gamma_t\}_{t\in[0,+\infty)}$
a \emph{growing curve}.
\index{growing curve}

Let $v_t$ and $\tilde v_t$ be two holomorphic vector fields depending continuously on time
such that the following differential equations have continuously
differentiable solutions $F_t$ and $\tilde F_t$ in some time interval
\begin{equation}\begin{split}
	&\dot F_t = v_t \circ F_t,\\
	&\dot {\tilde F}_t = \tilde v_t \circ \tilde F_t.
\end{split}\end{equation}
Then, we can conclude that
\begin{equation}
	\frac{\de}{\de t}(F_t \circ {{\tilde F}_t}) = 
	\left(
		v_t + {F_t}_* \tilde v_t
	\right) \circ F_t \circ \tilde F_t
	\label{Formula: d G G = (v + G G v) G G}
\end{equation}
\begin{equation}
	\dot F_t^{-1}  = 
	- \left( {F_t^{-1}}_* v_t \right) \circ F_t^{-1}. 
	\label{Formula: dot F^-1 = -v F^-1 v_t F^-1}
\end{equation}
in the same interval and in the region of $\Dc$ where $F_t$ and $\tilde F_t$ are
defined.
The latter relation can be reformulated in a fixed chart $\psi$ as
\begin{equation}
	\dot {F_t^{-1}}^{\psi}(z)  = 
	- \left(\left( {F_t^{-1}}\right)^{\psi}\right)'(z)~  v^{\psi}_t(z). 
\end{equation}

\section{$(\delta,\sigma)$-L\"owner chain (slit holomorphic L\"owner chain)}
\label{Section: Slit holomorphic Loewner chain}

\subsection{Definition and basic properties}
\label{Section: Definition and basic properties}

Consider the autonomous initial value problem 
\begin{equation}
 \dot G_t = \delta \circ G_t + \dot u_t\, \sigma \circ G_t ,\quad 
 G_0 = \id,\quad t\in[0,+\infty),
 \label{Formula: d G = delta G dt + sigma G du}
\end{equation}
for a conformal map $G_t$ 
defined for a complete vector field $\sigma$, for a semicomplete or
antisemicomplete vector field $\delta$, and for a continuously differentiable
function $u_t:[0,+\infty)\map \mathbb{R}$. 

We now drop the requirement of
differentiability of $u_t$ using the following method.
Define first a conformal map 
\begin{equation}
	g_t := H_{u_t}[\sigma]^{-1} \circ G_t,
	\label{Formula: g_t = H[sigma] circ G_t}
\end{equation}
where $H[\sigma]$ is defined by 
\eqref{Formula: d H = sigma H ds}.
Thanks to 
\eqref{Formula: d G G = (v + G G v) G G}
the map $g_t$ satisfies 
\begin{equation}
 \dot g_t = (H_{u_t}[\sigma]^{-1}_* \delta) \circ g_t,\quad g_0=\id.
 \label{Formula: d g = h delta g dt}
\end{equation}

The converse is also true: 
\eqref{Formula: d G = delta G dt + sigma G du} 
can be obtained from 
\eqref{Formula: d g = h delta g dt}. 
However, 
\eqref{Formula: d g = h delta g dt} 
is defined for a more general set of driving functions, not necessarily
continuously differentiable. Thus, we can define 
\begin{equation}
	G_t:=H_{u_t}[\sigma]\circ g_t
	\label{Formula: G = H circ g}
\end{equation}
for non-differentiable driving functions $u_t$. Working with $g_t$ is
technically less convenient than with $G_t$; however, it has the advantage of
allowing non-differentiable driving functions. This motivates
the following definition.

\begin{definition}
Let $\delta$ and $\sigma$ be holomorphic vector fields of the form
\begin{equation}\begin{split}
	\delta =&
	\delta_{-2} \ell_{-2} + \delta_{-1} \ell_{-1} + \delta_0 \ell_{0} 
	+ \delta_1 \ell_{1},\quad
	\delta_{-2},\delta_{-1}, \delta_0, \delta_1 \in \mathbb{R},\quad
	\delta_{-2}\neq 0\\
	\sigma =&
	\sigma_{-1} \ell_{-1} + \sigma_0 \ell_{0} + \sigma_1 \ell_{1},\quad
	\sigma_{-1}, \sigma_0, \sigma_1 \in \mathbb{R},\quad \sigma_{-1}\neq 0,\\
	\label{Formula: delta and sigma for Lowner chain}
\end{split}\end{equation}
and let $u_t$ be a continuous function 
$u:[0,\infty)\map \mathbb{R}$.  
Then the initial value problem
\eqref{Formula: d g = h delta g dt}
is called the 
\emph{$(\delta,\sigma)$-L\"owner equation},
\index{$(\delta,\sigma)$-L\"owner equation} 
or equivalently,
the \emph{slit holomorphic L\"owner equation}, 
\index{slit holomorphic L\"owner equation}
and its solution $\{g_t\}_{t\in[0,+\infty)}$ given by 
\begin{equation}
	\{u_t\}_{t\in[0,+\infty)} \mapsto \{g_t\}_{t\in[0,+\infty)}
	\label{Formula: u_t -> g_t}
\end{equation}
or $\{G_t\}_{t\in[0,+\infty)}$ given by
\begin{equation}
	\{u_t\}_{t\in[0,+\infty)} \mapsto
	\{G_t:=H_{u_t}[\sigma]\circ g_t\}_{t\in[0,+\infty)}
	\label{Formula: u_t -> G_t}
\end{equation}
is called the \emph{$(\delta,\sigma)$-L\"owner chain},
or equivalently,
the \emph{slit holomorphic L\"owner chain}. 
The chain is called \emph{forward} 
\index{forward chain}
if $\delta$ is antisemicomplete
($\delta_{-2}> 0$), and it is called \emph{reverse} 
\index{reverse chain}
if $\delta$ is semicomplete ($\delta_{-2}<0$).
\end{definition} 

We do not consider the degenerate cases when $\delta_{-2}=0$ or $\sigma_{-1}=0$,
because most of the propositions below do not hold in these cases. If 
$\delta_{-2}=0$,
then we have just M\"obius automorphisms. The case 
$\sigma_{-1}=0$
is considered in Section
\ref{Section: Degenerate case}.
Henceforth, we always assume that $\delta$ and $\sigma$ are of the form 
\eqref{Formula: delta and sigma for Lowner chain}.
\index{$\delta$}
\index{$\sigma$}
\index{forward case}
\index{reverse case}

The well-known chordal, radial, and dipolar equations are special cases
of
\eqref{Formula: d g = h delta g dt}. 
We summarize them in Table
\ref{Table: Calssical cases}.
\begin{table}[h]
\begin{center}
\begin{tabular}{|c|c|c|}
 \hline 
 Equation type & $\delta$ & $\sigma$ \\ 
 \hline 
 Chordal & $2\ell_{-2}$ & $-\ell_{-1}$ \\ 
 \hline 
 Dipolar & $\frac12\ell_{-2}-\frac12\ell_0$ 
 	& $-\frac12\ell_{-1}+\frac12\ell_1$ \\
 \hline 
 Radial & $\frac12\ell_{-2}+\frac12\ell_0$ 
 	& $-\frac12\ell_{-1}-\frac12\ell_1$ \\ 
 \hline 
 ABP, see \cite{Ivanov2012a}  & $\frac12\ell_{-2}$ 
 	& $-\frac12\ell_{-1}-\frac12\ell_1$ \\ 
 \hline 
\end{tabular} 
\end{center}
\caption{Known cases of forward $(\delta,\sigma)$-L\"owner chains.}
\label{Table: Calssical cases}
\end{table}

We call the first three cases in the table \emph{classical} 
\index{classical L\"owner equations}
and discuss each of them in the corresponding sections of Chapter 
\ref{Chapter: Classical cases}.
Each of the three classical cases is a combination of $\delta$ and $\sigma$ of
the same type with identical positions of zeros. This leads to simple
restrictions on the corresponding maps $G_t$ and $g_t$ that we call
\emph{normalization}.
\index{normalization}

We apply the general theory of L\"owner chains to solve the initial value
problem \eqref{Formula: d g = h delta g dt}.
According to the terminology from 
\cite{Contreras2012} 
the vector field
$(H_{u_t}[\sigma]_*^{-1} \delta)^{\D}$ 
in the unit disk chart is Herglotz in the reverse case and minus Herglotz in
the forward case. We conclude that the solution 
$\{g_t\}_{t\in[0,+\infty)}$ 
exists, is unique, and is given by the conformal maps 
$g_t\colon \Dc\map\Dc\setminus\K_t$ 
in the reverse case and by the conformal maps 
$g_t\colon \Dc\setminus\K_t\map\Dc$
in the forward case for some collection 
$\{\K_t\}_{t\in[0,+\infty)}$ 
of subsets
$\K_t\subset\mathbb{C}$. 
Moreover, in the forward case, the collection 
$\{\K_t\}_{t\in[0,+\infty)}$  
is strictly increasing: $\K_s \subset \K_t$,
$s<t$. In Fig. 
\ref{Figure: G map example},
we show how a typical map $G_t$ acts in the unit disk chart. Table
\ref{Table: Sin driving function examples}
shows the limit $t\map+\infty$ of the hull $\K_t$ (which is a curve in these
cases) for some fixed choice of the driving function and for various choices of
$\delta$ and $\sigma$. Some exact solutions of the chordal L\"owner equation
are considered in 
\cite{Kager2004}.

\begin{center}
\begin{figure}[h]
\centering
	\includegraphics[keepaspectratio=true]
    	{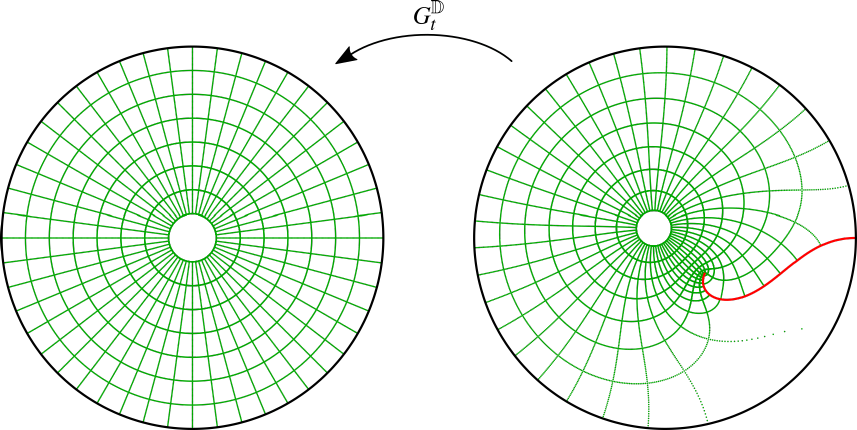}
  \caption{
  This is how a typical conformal map $G_t$ acts in the unit disk chart
  $\D$ for some choice of $\delta$, $\sigma$, the driving function $u$, and
  $t$ in the forward case.
  The red line is the hull $\K_t$, which is a simple curve (slit) in this case.
  In the reverse case, the map $G_t$ acts in the opposite direction, see
  Proposition
  \ref{Proposition: Forward - Inverse LE connection}.
  \label{Figure: G map example}}
\end{figure}
\end{center}

\newcolumntype{R}[2]{%
    >{\adjustbox{angle=#1,lap=\width-(#2)}\bgroup}%
    l%
    <{\egroup}}
\newcommand*\rot{\multicolumn{1}{R{90}{1em}}}

\begin{table}[h]
\begin{center}
\begin{tabular}{ c c c c }
  & $\sigma=\ell_{-1}$ & $\sigma=\ell_{-1}-\ell_{1}$ & 
  $\sigma=\ell_{-1}+\ell_{1}$\\ 
	\rot{$~~~~~~~~~~~~~~\delta=2\ell_{-2}$} & 
 	\includegraphics[width=4.5cm,keepaspectratio=true]
    	{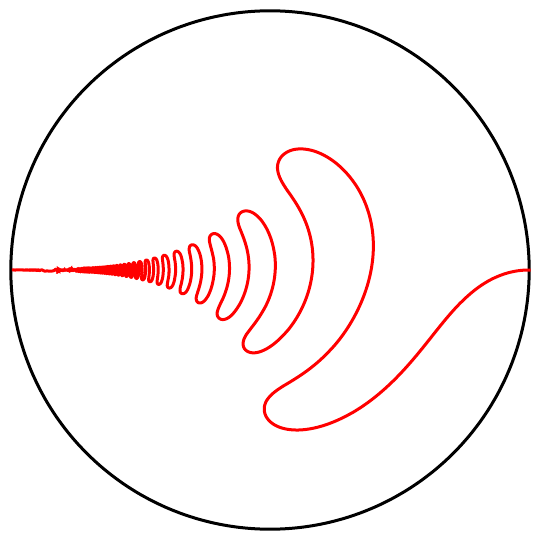} & 
 	\includegraphics[width=4.5cm,keepaspectratio=true]
    	{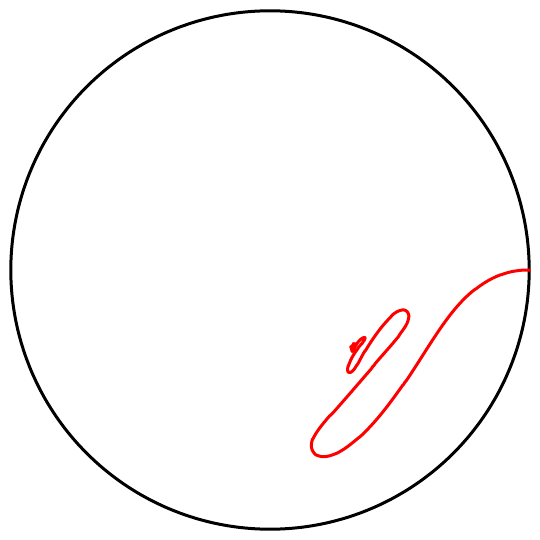} & 
 	\includegraphics[width=4.5cm,keepaspectratio=true]
    	{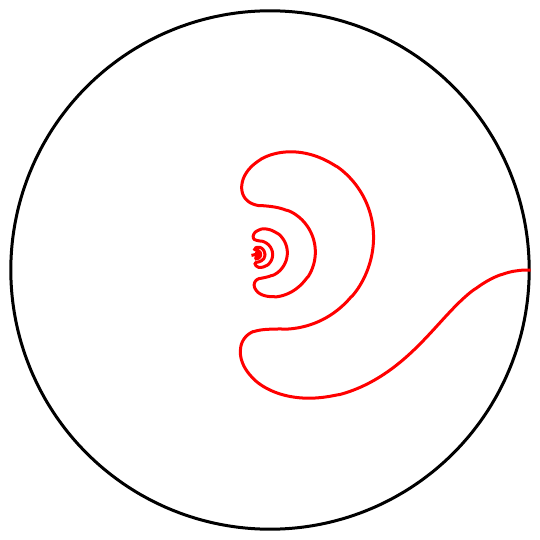} \\ 
	\rot{$~~~~~~~~~\delta=2\ell_{-2}-2\ell_0$} & 
 	\includegraphics[width=4.5cm,keepaspectratio=true]
    	{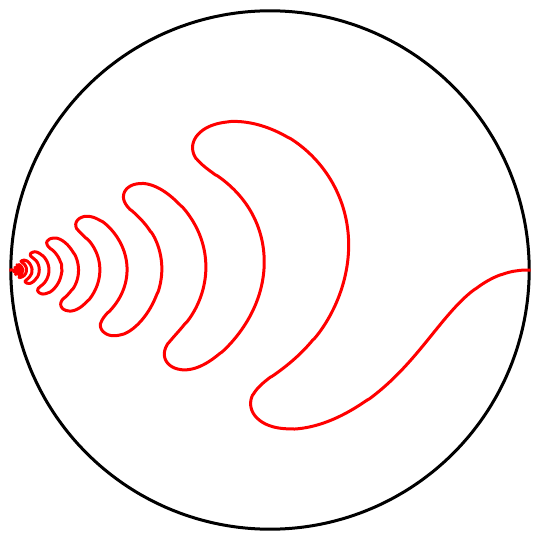} & 
 	\includegraphics[width=4.5cm,keepaspectratio=true]
    	{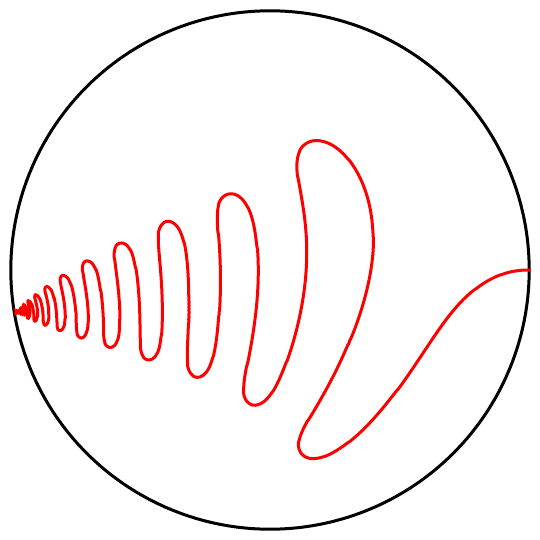} & 
 	\includegraphics[width=4.5cm,keepaspectratio=true]
    	{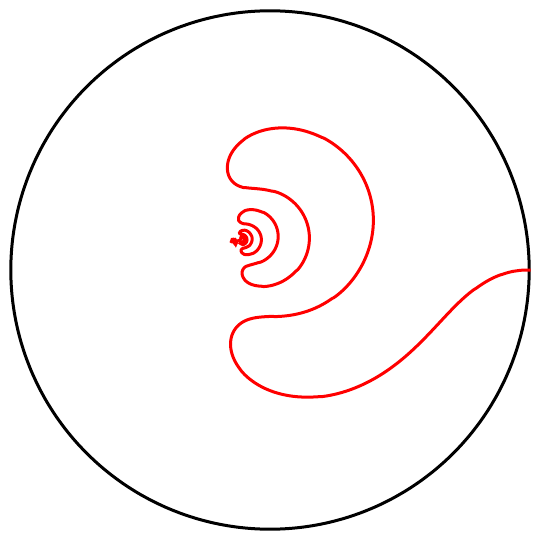} \\ 
	\rot{$~~~~~~~~~\delta=2\ell_{-2}+2\ell_0$} & 
 	\includegraphics[width=4.5cm,keepaspectratio=true]
    	{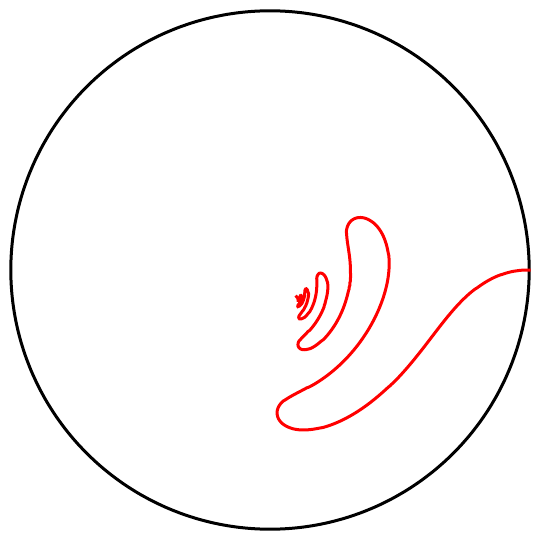} & 
 	\includegraphics[width=4.5cm,keepaspectratio=true]
    	{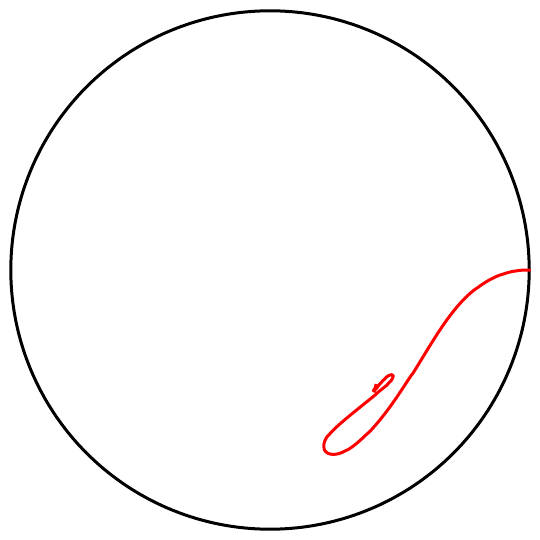} & 
 	\includegraphics[width=4.5cm,keepaspectratio=true]
    	{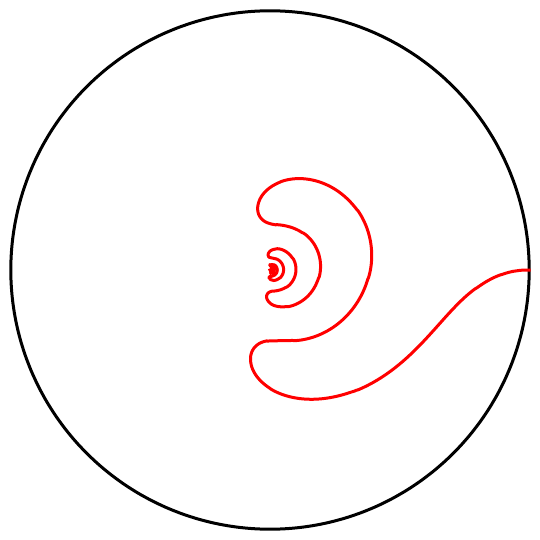} \\ 
\end{tabular}
\end{center}
\caption{Slits of ($\delta,\sigma$)-L\"owner chains ($t\map+\infty$) in the
unit disk chart driven by $u_t=\sin 20\, t$.
Each row
corresponds to a fixed $\delta$, and each column to a fixed $\sigma$.
Thus, the three cases on the diagonal are classical.}
\label{Table: Sin driving function examples}
\end{table}

Let
$G_{t,s}$ be the solution of 
\eqref{Formula: d g = h delta g dt}
and
\eqref{Formula: G = H circ g}
for $t\in[0,+\infty)$ parametrized by 
$s\in[0,+\infty)$
with the initial condition $G_{s,s}=\id$. Hence, 
$\{G_{t,s}\}_{t\in[s,+\infty)}$ 
is a $(\delta,\sigma)$-L\"owner chain with the
driving function
\begin{equation}
	u_{t,s}:=u_{t} - u_s,\quad t\in[s,+\infty).
	\label{Formula: u_t,s = u_t - u_s}
\end{equation}
In particular, 
$G_{t,0}=G_t$, $t\in[0,+\infty)$.
We see that 
$G_{t+s,s}$ is also defined for negative values of $t$ 
($t\in[-s,0]$, $s\geq 0$)
and
\begin{equation}
	G_{t,s}=G_t \circ G_s^{-1},\quad t,s\geq 0
\end{equation}
Besides, 
\begin{equation}
	G_{t,s} = G_{t,r} \circ G_{r,s},\quad t,s,r \geq 0.
	\label{Formula: G_t,s = G_t,r circ G_r,s}
\end{equation}
In particular,
\begin{equation}
	G_{t} = G_{t,s} \circ G_s,\quad s,t\geq 0,
	\label{Formula: G_t = G_t,s circ G_s}
\end{equation}
and
\begin{equation}
	G_{t,s} = G^{-1}_{s,t},\quad s,t\geq 0.
\end{equation}
Analogous relations are valid for the differentiable version of the chain
$g_{s,t}$
\begin{equation}
	g_{t,s} = g_{t,r} \circ g_{r,s},\quad
	g_{t,s} = g^{-1}_{s,t},\quad
	g_{t,s}:=
	g_t\circ g_s^{-1},\quad t,s,r \geq 0.
\end{equation}
The collection 
$\{ g_{t,s}^{\D} \}_{0\leq s\leq t<+\infty}$ 
is known as an `evolution family'.

The inverse maps 
$G_t^{-1}$ 
and 
$g_t^{-1}$ 
satisfy the partial differential equations:
\begin{equation}
	\dot G_t^{-1} = - G_t^{-1}{}_* \left( \delta 
	+ \sigma \dot u_t \right)\circ G_t^{-1}
\end{equation}
for continuously differentiable $u_t$, and
\begin{equation}
	\dot g_t^{-1} = 
	-\left( g_t^{-1}{}_*  H_{u_t}[\sigma]_* \delta \right)\circ g_t^{-1}
\end{equation}
for continuous $u_t$.

\begin{proposition}
Fix $T\geq 0$.
Let $\{G_t\}_{t\in[0,T]}$ be a forward $(\delta,\sigma)$-L\"owner chain 
driven by 
$\{u_t\}_{t\in[0,T]}$,
and let
$\{\tilde G_t\}_{t\in[0,T]}$ 
be a reverse 
$(-\delta,\sigma)$-L\"owner chain
driven by 
\begin{equation}
	\tilde u_t:=u_{T-t}-u_T,\quad t\in[0,T].
\end{equation}
Then 
\begin{equation}
	\tilde G_T = G_T^{-1},
	\label{Formula: tilde G = G^-1}
\end{equation}  
and
\begin{equation}
	\tilde g_T = H_{u_T}[\sigma] \circ g^{-1}_{T} \circ H_{u_T}[\sigma]^{-1}.
	\label{Formula: tilde g = H g^-1 H}
\end{equation}  
\label{Proposition: Forward - Inverse LE connection} 
\end{proposition}
\begin{proof}
Set 
\begin{equation}
	\hat g_t: = H_{u_T}[\sigma] \circ g_{T-t,T} \circ H_{u_T}[\sigma]^{-1}
	, \quad t\in[0,T],
\end{equation}
and observe that 
$\hat g_T=\tilde g_T$ 
because they are solutions of the same initial value problem. Indeed,
\begin{equation}\begin{split}
  \frac{\de}{\de t} \hat g_{t} 
	=& \frac{\de}{\de t} 
	\left(  
		H_{u_T}[\sigma] \circ g_{T-t} \circ g^{-1}_T \circ H_{u_T}[\sigma]^{-1}
	\right) 
	=\\=&
	-\left( H_{u_T}[\sigma]_* H_{u_{T-t}}[\sigma]^{-1}_* \delta \right) 
		\circ H_{u_T}[\sigma] \circ g_{T-t} \circ g^{-1}_T \circ H_{u_T}[\sigma]^{-1}
	=\\=&
  \left( H_{u_{T-t}-u_T}[\sigma]^{-1}_* (-\delta) \right) \circ	\hat g_t
  , \quad t\in[0,T].
	\label{Formula: 1 from Proposition: Forward - Inverse LE connection}
\end{split}\end{equation}
and $\hat g_0 = \id$.
Putting $t=T$, we obtain
\eqref{Formula: tilde g = H g^-1 H}
and 
\eqref{Formula: tilde G = G^-1}.
\end{proof}

Given a pair of $\delta$ and $\sigma$, we denote by
$\mathcal{G}[\delta,\sigma]\subset\mathcal{G}$ 
the collection of all endomorphisms and inverse endomorphisms $G_t$ of $\Dc$ 
from all possible $(\delta,\sigma)$-L\"owner chains for all possible continuous
driving functions $u$ and $t\in[0,+\infty)$.

Proposition 
\ref{Proposition: Forward - Inverse LE connection},
in particular, states that the collection
$\mathcal{G}[\delta,\sigma]$ coincides with $\mathcal{G}[-\delta,\sigma]$ up to
inversion. We will mostly consider forward chains below. The advantage of the
forward case is that the hulls are growing: $\K_s\subset \K_t$, $s<t$.

The equation
\eqref{Formula: d g = h delta g dt}
together with
\eqref{Formula: G = H circ g}
induce a surjective map 
$C^0[0,T]\map\mathcal{G}[\delta,\sigma]$
($T\in[0,+\infty)$)
acting as
$\{u_t\}_{t\in[0,T]}\mapsto G_T\in\mathcal{G}[\delta,\sigma]$,
where $C^0[0,T]$ is the set of all continuous functions
$u\colon[0,T]\map\mathbb{R}$ on the interval $[0,T]$ such that $u_0=0$.

\subsection{General properties of ($\delta,\sigma$)-L\"owner chains}
\label{Section: General properties of Loewner chain}

Relations
\eqref{Formula: u_t,s = u_t - u_s}
-
\eqref{Formula: 1 from Proposition: Forward - Inverse LE connection}
are valid for arbitrary semicomplete or antisemicomplete vector fields
$\delta$, not necessarily of the form 
\eqref{Formula: delta and sigma for Lowner chain}. 
However, the theorem below essentially uses the fact that $\delta$ has only one simple pole
and is tangent to the boundary at all other points. The result of the
theorem motivates the term `slit L\"owner chain'.

\begin{theorem}
\label{Theorem: The master theorem}
Let $\{g_t\}_{t\in[0,+\infty)}$ be a forward $(\delta,\sigma)$-L\"owner chain  
with a driving function 
$\{u_t\}_{t\in[0,+\infty)}$ 
and with the hulls 
$\{\K_t\}_{t\in[0,+\infty)}$.
Let also $\tilde \delta$ and $\tilde \sigma$ be as in 
\eqref{Formula: delta and sigma for Lowner chain}
with $\tilde \delta_{-2}>0$.
Then the following statements hold:
\begin{enumerate} [1.]
\item 
There exists a 
$(\tilde \delta,\tilde \sigma)$-L\"owner chain 
$\{\tilde g_t\}_{t\in[0,\tilde T)}$ 
for some maximal 
$\tilde T\in(0,+\infty]$ 
driven by a function 
$\{\tilde u_{\tilde t}\}_{\tilde t\in[0,\tilde T)}$
such that its hulls are
\begin{equation}
	\tilde \K_{\tilde t} 
	= \K_{\lambda^{-1}_{\tilde t}},\quad \tilde t \in[0,\tilde T)
	\label{Formula: tilde K = K_lambda}
\end{equation}
for a continuous time reparametrization 
$\lambda:[0,T)\map[0,\tilde T)$
defined for some $T\in(0,\infty]$.
\item
If the $(\tilde \delta,\tilde \sigma)$-L\"owner chain, as in item 1, exists
until $\tilde T\in[0,+\infty)$, then it is unique.
\item
The time reparametrization $\lambda_t$ is continuously differentiable.
\item
The L\"owner chains 
$\{g_t\}_{t\in[0,\infty)}$ 
and 
$\{\tilde g_t\}_{t\in[0,\tilde T)}$ 
are related by a family of M\"obius automorphisms
\begin{equation}
	M_t := {\tilde g}_{\lambda_t} \circ g^{-1}_t,\quad t\in[0,T)
	\label{Formula: M = tilde g g^-1}
\end{equation}
that satisfy
\eqref{Formula: dot M_t = ( dot l H delta - M H delta ) circ M_t}
for $t\in[0,T)$.
\item
The functions $\lambda_t$ and $\tilde u_{\tilde t}$ are defined by 
\eqref{Formula: lambda_t = ...}
and by
\eqref{Formula: tilde u_tilde t = ...},
respectively.
\item
Let $b\in\Dc$ (or $b\in\de\Dc$) be a point such that $b\not \in \bar \K_t$ for
$t\in[0,+\infty)$ (where $\bar \K\subset \bar \Dc$ is the closure of $\K$ in
$\bar \Dc$). Assume that $\tilde\delta$ and $\tilde\sigma$ are radial (see
Table \ref{Table: Calssical cases} and Section \ref{Section: Radial case}) 
with the common zero at $b\in\Dc$ or chordal 
(see Section \ref{Section: chordal case}) 
with the common zero at
$b\in\de\Dc$. Then $T=+\infty$.
\end{enumerate}

\end{theorem}

\begin{proof}
Assume first that there exists a 
$(\tilde\delta,\tilde\sigma)$-L\"owner chain 
$\{\tilde g_{\tilde t}\}_{\tilde t\in[0,\tilde T)}$ 
($\tilde T\in(0,+\infty]$),
a collection of hulls 
$\{ \tilde \K_{\tilde t} \}_{t\in[0,\tilde T)}$
constructed as a reparametrization of
$\{\K_t\}_{t\in[0,T)}$ 
with some 
$T\in(0,+\infty]$,
and a continuous strictly increasing 
$\lambda:[0,T)\map[0,\tilde T)$.
Such a function $\lambda$ is invertible and differentiable a.e. in $[0,T)$.
Then define $M_t$ by the relation
\eqref{Formula: M = tilde g g^-1}
and conclude that it is also differentiable with respect to $t$ a.e. due to the
differentiability of $g_t$ and $\tilde g_{\tilde t}$.
Moreover, according to 
\eqref{Formula: d G G = (v + G G v) G G}
and
\eqref{Formula: g_t = H[sigma] circ G_t}
we have
\begin{equation}
	\dot M_t = 
	\left( 
		\dot \lambda_{t}~ H_{\tilde u_{\lambda_t}}[\tilde \sigma]^{-1}_*~ 
		\tilde \delta 
		- M_{t}{}_*~ H_{u_t}[\sigma]^{-1}_*~ \delta 
	\right) \circ M_t,\quad M_0=\id.
	\label{Formula: dot M_t = ( dot l H delta - M H delta ) circ M_t}
\end{equation}
a.e. in $[0,T)$. From the fact that $M_t$ is a M\"obius
automorphism, it follows that the expression in parentheses above is a
complete vector field for a.a. $t$.

We now study the conditions for completeness. Define the vector field
\begin{equation}\begin{split}
	m(M, x, y, c):= 
	c~ H_{y}[\tilde \sigma]^{-1}_*~ \tilde \delta 
	- M_*~ H_{x}[\sigma]^{-1}_*~ \delta,\quad
	x,y\in\mathbb{R},\quad c>0,\quad M:\Dc\map\Dc.
\end{split}\end{equation}
It is complete if and only if the field 
\begin{equation}\begin{split}
	H_{y}[\tilde \sigma]_* m(M, x, y, c)=
	c~ \tilde \delta 
	- \hat M_* \delta
	\label{Formula: H m = delta - c M delta}
\end{split}\end{equation}
is complete, where
\begin{equation}
	\hat M:= H_{y}[\tilde \sigma] \circ M \circ H_{x}[\sigma]^{-1}.
	\label{Formula: hat M = H M H^-1}
\end{equation}
In general, it contains the sum of two poles at $a\in\Dc$ with residues of opposite signs
because of the structure of $\tilde \delta$ and $\delta$. Let us obtain
the necessary and sufficient conditions on $M,~x,~y$,
and $c$ for the poles to cancel. The positions of the poles must coincide, and the absolute values of
the residues must be equal (for example, it is clear in the half-plane chart).
These conditions are
\begin{equation}\begin{split}
 \hat M (a) = a,
 \label{Formula: hat M(a) = a}
\end{split}\end{equation}
and
\begin{equation}\begin{split}
	&c~ \mathrm{res}_a \tilde \delta = 
	\mathrm{res}_a 
	\left( \tilde M_* \delta \right) \quad \Leftrightarrow \\
	&c~ \tilde \delta_{-2} = 
	\delta_{-2} \left( \left( \tilde M^{\psi} \right)'(a) \right)^2
	\quad \Leftrightarrow \\
	&c = \frac{\delta_{-2}}{\tilde \delta_{-2}} 
		\left( \left( \tilde M^{\psi} \right)'(a) \right)^2.
	\label{Formula: c = delta/delta M'^2}
\end{split}\end{equation}
Remark that 
$\left( \tilde M^{\psi} \right)'(a)>0$ 
does not depend on the choice of the chart $\psi$, and neither does the ratio 
$\frac{\delta_{-2}}{\tilde \delta_{-2}}$. 

The first condition
\eqref{Formula: hat M(a) = a}
holds uniquely as an implicit solution to $y=y(x,M)$ 
(see \eqref{Formula: hat M = H M H^-1} ) 
due to the condition 
$\tilde \sigma(a) \neq 0$, 
at least for sufficiently small $|x|$, $|y|$, and $M$
close enough to the identity map.
In the case of an elliptic $\tilde \sigma$, the function $y(x,M)$ is unique up
to the transformation $y\map y+k d$, $k\in \mathbb{Z}$, where $d>0$
is the minimal number such that $H_d[\tilde \sigma]=\id$. 

Assume now 
\begin{equation}
	M = M_{ t}, \quad 
	x = u_{t}, \quad
	y = \tilde u_{\lambda_{\tilde t}}, \quad
	c = \dot \lambda_{ t}, 	
\end{equation}
which leads to
\begin{equation}
	\hat M = \hat M_{ t} = 
	H_{\tilde u_{\lambda_t}} [\tilde \sigma] \circ 
	M_{t} \circ H_{u_{t}}[\sigma]^{-1}
	\label{Formula: hat M_t = H_u M_t H_u}
\end{equation}
and
\begin{equation}\begin{split}
	&
	m(M_{t}, u_{t}, \tilde u_{\lambda_{t}}, \dot \lambda_{t})= 
	\dot \lambda_{t}~ H_{\tilde u_{\lambda_t}}[\tilde \sigma]^{-1}_*~ \tilde \delta 
	-  M_{t}{}_*~ H_{u_t}[\sigma]^{-1}_*~ \delta,\\	
\end{split}\end{equation}
which is exactly the vector field from 
\eqref{Formula: dot M_t = ( dot l H delta - M H delta ) circ M_t}.
The condition 
\eqref{Formula: c = delta/delta M'^2}
gives
\begin{equation}
	\dot \lambda_{t} = 
	\frac{\delta_{-2}}{\tilde \delta_{-2}} 
		\left( \left( \hat M_{t} \right)'(a) \right)^2
	\label{Formula: dot lambda = M'^2}
\end{equation}
for a.a $t\in[0,T)$. On the other hand, the continuity of $M_t'(a)$ 
follows from
\eqref{Formula: M = tilde g g^-1}. 
This and 
\eqref{Formula: dot lambda = M'^2}
in turn imply that $\lambda_t$ is continuously differentiable in
$[0,T)$, and that
\eqref{Formula: dot lambda = M'^2}
holds for all $t\in[0,T)$, and not only a.e.
Due to the property $M_t'(a)>0$ of the M\"obius maps, we also conclude
that the inverse function $\lambda^{-1}$ is continuously differentiable in
$[0,\tilde T)$.

Conditions 
\eqref{Formula: hat M(a) = a}
and
\eqref{Formula: hat M = H M H^-1}
give
\begin{equation}
	\tilde u_{\lambda_t} = y(M_{t},u_{t}),
	\quad	t\in[0,T), 
	\label{Formula: tilde u = f(M,u)}
\end{equation}
Here we decrease the value of $T$ according to the domain of definition of
$y(M_{t},u_{t})$.
For an elliptic $\tilde \sigma$ we impose the continuity condition and 
$\tilde u_0=0$. This removes the above-mentioned ambiguity.

Thus, we have found $\lambda_{t}$ as an integral functional of $\{M_t\}_{t\in[0,T)}$ and 
$\{u_t\}_{t\in[0,T)}$: 
\begin{equation}\begin{split}
	\lambda_{t} =& 
	\frac{ \delta_{-2}}{\tilde \delta_{-2}} 
	\int\limits_0^{t} 
	\left( 
		\left( 
			\hat M_{\tau} 
		\right)'(a) 
	\right)^2	
	d\tau
	=\\=&
	\frac{\delta_{-2}}{\tilde \delta_{-2}} 
	\int\limits_0^{t} 
	\left( 
		\left( 
			H_{\tilde u_{\tau}}[\tilde \sigma] \circ M_{\tau} \circ H_{u_{\lambda_{\tau}}}[\sigma]^{-1} 
		\right)'(a) 
	\right)^2	
	d\tau
	=\\=&
	\frac{\delta_{-2}}{\tilde \delta_{-2}} 
	\int\limits_0^{t} 
	\left( 
		\left( 
			H_{ y(M_{\tau}, u_{\tau}) }[\tilde \sigma] \circ M_{\tau} \circ H_{u_\tau}[\sigma]^{-1} 
		\right)'(a) 
	\right)^2	
	d\tau,\\
	&t\in[0,T),
	\label{Formula: lambda_t = ...}
\end{split}\end{equation}
where we have used 
\eqref{Formula: hat M_t = H_u M_t H_u}
and
\eqref{Formula: tilde u = f(M,u)}.

The driving function $\tilde u_{\tilde t}$ can now be expressed as
\begin{equation}
	\tilde u_{\tilde t} = y(M_{\lambda^{-1}_{\tilde t} },u_{\lambda^{-1}_{\tilde t} }),\quad
	\tilde t\in[0, \lambda^{-1}_{T}).
	\label{Formula: tilde u_tilde t = ...}
\end{equation}

Substituting now 
\eqref{Formula: hat M_t = H_u M_t H_u},
\eqref{Formula: tilde u = f(M,u)}
and
\eqref{Formula: dot lambda = M'^2}
in 
\eqref{Formula: dot M_t = ( dot l H delta - M H delta ) circ M_t}
gives the initial value problem
\begin{equation}\begin{split}
	\dot M_t =&
	\left( 
		\frac{\delta_{-2}}{\tilde \delta_{-2}}
		\left( \left( H_{y(M_t,u_t)} [\tilde \sigma] \circ 
		M_{t} \circ H_{u_{t}}[\sigma]^{-1} \right)'(a) \right)^2
		H_{y(M_t,u_t)}[\tilde \sigma]^{-1}_*~ \tilde \delta 
	\right. 
	-\\-&
	M_{t}{}_*~ H_{u_t}[\sigma]^{-1}_*~ \delta \bigg)
	\circ M_t,\\
	M_0=&\id.
	\label{Formula: dot M_t = F(M_t,u_t) circ M_t}
\end{split}\end{equation}
which is defined by the driving function $\{u_t\}_{t\in[0,T)}$
only. By construction, it is a partial differential equation for $M$ in
$[0,T)\times \Dc$ because it contains the derivatives of $M_t^{\psi}(z)$ with
respect to $z$. However, the family of M\"obius automorphisms is three-parametric.
For example, in the half-plane chart, we have
\begin{equation}
	M_t^{\HH}(z) = \frac{a_t z + b_t}{c_t z + d_t},\quad z\in\HH,\quad a_t d_t - b_t c_t = 1
\end{equation}
(explicit relations in the unit disk chart are
presented in \cite{Ivanov2014}).
Thus, 
\eqref{Formula: dot M_t = F(M_t,u_t) circ M_t}
is actually a system of three ordinary differential equations with continuous
coefficients. Its solution always exists and is unique at least until some
$T_1>0$.
Moreover, if the chain $\{\tilde g_t\}_{t\in[0,\tilde T)}$ exists, it is unique
due to the uniqueness of the solution of
\eqref{Formula: dot M_t = F(M_t,u_t) circ M_t}
and due to the arguments above about the differentiability of $\lambda_t$ and
$M_t$.

In order to prove item 6, we construct the second chain 
$\{\tilde g_{\tilde t}\}_{\tilde t\in[0,+\infty)}$.
Let us define 
$\{\hat M_t\}_{t\in[0,\infty)}$ 
by the condition
\eqref{Formula: hat M(a) = a} 
and 
\begin{equation}
	\hat M_t \circ G_t (b) = b.
	\label{Formula: hat M G(b) = b}
\end{equation}
If $b\in\de\Dc$, then we can choose a chart $\psi$ such that the boundary of 
$D^{\psi}:=\psi(\Dc)$ 
is a straight line segment in a neighborhood of $\psi(b)$ and we require in
addition
\begin{equation}
	(\hat M_t \circ G_t)' (b) = 1.
\end{equation}
The derivative 
$(G_t^{\psi})'(\psi(b))$ 
is defined because 
$b\not\in\bar{\K_t}$, $t\in[0,+\infty)$.
Thus, the function
$\hat M_t$ 
is continuous in 
$t\in[0,+\infty)$.
Define the time reparametrization 
$\{\lambda_t\}_{t\in[0,+\infty)}$ 
by the first line of 
\eqref{Formula: lambda_t = ...}, which is possible for $t\in[0,+\infty)$ due to the continuity of 
$\left(\hat M_{t}\right)'(a)$.
On the other hand, the properties of the radial and chordal L\"owner equations ensure that 
$\tilde g_{\tilde t}$ and $\tilde G_{\tilde t}$   
are defined for any collection of hulls such that 
$b\not\in\K_t$ 
with some driving function $\tilde u$, and 
$\tilde G_{\tilde t} = \hat M_{\lambda^{-1}_{\tilde t}} \circ G_{\lambda^{-1}_{\tilde t}}$, 
$t\in[0,+\infty)$.
We use it now to define $M_t$ for $t\in[0,+\infty)$. 
\end{proof}

We remark that if the original chain $\{g_t\}_{t\in[0,+\infty)}$ is defined
only until some $T'\in(0,+\infty)$, then the modification of the theorem is
straightforward. The same is true for the reverse L\"owner chains.

With the aid of this theorem, we can extend the properties of the classical
L\"owner equation to general $(\delta,\sigma)$-L\"owner chains. Following 
\cite{Lawler2008}, 
we give the following definition:
\begin{definition}
We say that a collection of hulls 
$\{\K_t\}_{t\in[0,\infty)}$ ($\K_0=\emptyset$)
is \emph{continuously increasing}
\index{continuously increasing collection of hulls}
if
\begin{enumerate}[1.]
\item The set
\begin{equation}
	U_t:=\bigcap\limits_{\varepsilon>0} \K_{t+\varepsilon,t},\quad t\in[0,+\infty)
\end{equation} 
is a point on the boundary $\de\Dc$, where
\begin{equation}
	\K_{t,s}:=\Dc \setminus g_{t,s}(\Dc);
\end{equation}
\item
The function $U:[0,+\infty)\map \de\Dc$ is continuous;
\item
There exists a continuously differentiable time reparametrization 
$$\lambda:[0,+\infty)\map[0,+\infty)$$
such that the family of hulls
$\{\K_{\lambda_t}\}_{t\in[0,+\infty)}$ 
is induced by some radial or chordal L\"owner chain. 
\end{enumerate}
\end{definition}

\begin{proposition} 
The hulls induced by the radial or chordal L\"owner equations are continuously
increasing.
\label{Proposition: Radial or chordal chain existence for con incr hulls}
\end{proposition}

\begin{proof}
This follows from the results of \cite[Chapter 4]{Lawler2008} and from the
formulation in the unit disk chart ($\psi^{\D}(b)=0$) for the radial case and
in the half-plane chart 
($\psi^{\HH}(b)=\infty,~\mathrm{res}_a\tilde\delta^{\HH}=2$)
for the chordal case.
\end{proof}

\begin{corollary}
\label{Corollary: Con incr hulls and L chains}
Any continuously increasing collection of hulls
$\{\K_{t}\}_{t\in[0,+\infty)}$ 
can be induced by any $(\delta,\sigma)$-L\"owner equation 
up to a continuously differentiable
time reparametrization 
at least until some maximal time $T\in(0,+\infty]$.
The converse is also true: a collection of hulls induced by
a $(\delta,\sigma)$-L\"owner equation is continuously increasing for
$t\in[0,+\infty)$.
\end{corollary}

\begin{proof}
It follows from Proposition 
\ref{Proposition: Radial or chordal chain existence for con incr hulls} 
and Theorem \ref{Theorem: The master theorem}.

%
%
\end{proof}

We call the point
\begin{equation}
	z_{t}:=g_t(U_t)
\end{equation}
the \emph{tip} of the hull $\K_t$.
\index{tip}
Due to 
\eqref{Formula: hat M(a) = a}, 
$\hat M = \tilde G_{\lambda_t} \circ G_t^{-1}$,
and the property of forward chordal and radial chains of mapping the tip to the
source point 
(see \cite{Lawler2008}),
we conclude that the same is true in the general case of
($\delta,\sigma$)-L\"owner chains
\begin{equation}\begin{split}
	&G_t(z_{t}) = a,\quad t\in[0,+\infty),\quad \text{in the forward case},\\
	&G_t(a) = z_{t},\quad t\in[0,+\infty),\quad \text{in the reverse case}.
\end{split}\end{equation}
Besides, for the map $g_t$ we have
\begin{equation}\begin{split}
	&g_t(z_{t}) = H_{u_t}^{-1}[\sigma](a),\quad t\in[0,+\infty),\quad \text{in the
	forward case},\\
	&g_t \circ H_{u_t}^{-1}[\sigma](a) = z_{t},\quad t\in[0,+\infty),\quad \text{in
	the reverse case}.
\end{split}\end{equation}

We have studied how to construct a $(\delta,\sigma)$-L\"owner chain for a given
family of hulls. Consider now a fixed hull $\K\subset\Dc$. In general, 
for given $\delta$, $\sigma$, and $\K\subset\mathbb{C}$ a
$(\delta,\sigma)$-L\"owner chain such that $\K=\K_T$ for some
($T\in[0,\infty)$) may not exist or may not be unique. For example, the radial
chain does not exist if $b\in\K$. However, if $\K$ is a simple curve, we can
state the following.

We will consider only simple curves that start from the source point
$a\in\de\Dc$ 
and that are parametrized in an open interval. 
Namely, let $\gamma\colon (0,T)\map \Dc$ (for some
$T\in(0,\infty]$)
be an endomorphism such
that $\lim\limits_{t\map +0}\gamma(t)=a$ 
and the limit 
$\lim\limits_{t\map -T}\gamma(t)$ 
may not exist.
Let 
$\mathscr{G}_{\Dc,a}$ 
\index{$\mathscr{G}_{\Dc,a}$}
be the space of all such curves 
up to a continuous reparametrization 
$\lambda:[0,T)\map[0,\tilde T)$ ($\tilde T\in(0,\infty]$).

\begin{theorem}
\label{Theorem: Curve -> chain}
For any curve from $\mathscr{G}_{\Dc,a}$ there exists a parametrization 
$\gamma:(0,T)\map\Dc$, $T\in(0,+\infty]$,	
such that 
$\{\gamma_{(0,t]}\}_{t\in[0,T)}$ 
is a continuously increasing family of hulls. In particular, for any pair 
of $\delta$ and $\sigma$, any simple curve contains a subcurve starting from
$a\in\de\Dc$ that can be induced by a $(\delta,\sigma)$-L\"owner chain for the given
$\delta$ and $\sigma$, and such a chain is uniquely defined.
\end{theorem}

\begin{proof}
This follows from the results in \cite[Chapter 3, Chapter 4]{Lawler2008}, in
particular, from Remark 4.4.
\end{proof}

\begin{figure}[h]
\centering
	\includegraphics[keepaspectratio=true]
		{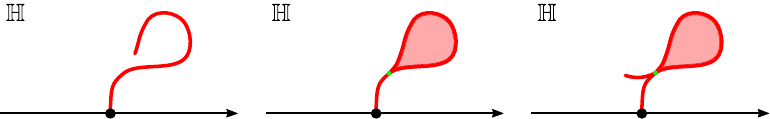}
\caption{This figure illustrates how a slit curve $\gamma_t$ (the red line)
generates a hull $\K_t$ (the red line and the pink interior) which is not a
curve. The time $t$ increases from left to right. The slit touches itself
at the green point and swallows up the pink connected component of
$\HH\setminus\gamma_t$.
\label{Figure: Closed set generation}}
\end{figure}

We remark that if the slit is not a simple curve but a curve parametrized as
$\gamma(t)=G^{-1}(a)$ for the forward chain $\{G_t\}_{t\in[0,+\infty)}$, then
the hull $\K_t$ consists of the curve and all subsets of $\Dc$ swallowed by the
curve except the one where the map $G_t$ is defined, see Fig. 
\ref{Figure: Closed set generation}.
In the classical cases, this is the connected component of $\Dc\setminus\gamma_t$
that contains the fixed points (see Sections
\ref{Section: Chordal Loewner equation},
\ref{Section: Dipolar Loewner equation},
and
\ref{Section: Radial Loewner equation}).
In the general case of ($\delta,\sigma$)-L\"owner chains, it is an interesting
problem to determine which of the components are swallowed.

It is important, but more difficult, to establish analogous results for curves
that touch themselves or the boundary $\de\Dc$. We do not discuss this
problem here. However, if there is a simple normalization for the
$(\delta,\sigma)$-L\"owner chain maps such as for the classical cases,
see 
\eqref{Formula: chordal normalization for G in H},
\eqref{Formula: dipolar normalization for G in S},
and
\eqref{Formula: radila normalization for G in D},	
then the existence is guaranteed by the Riemann mapping theorem and the
uniqueness can be shown by Corollary 
\ref{Corollary: Con incr hulls and L chains}.

It is also an interesting, yet unsolved, problem to specify the collection
of hulls induced by $(\delta,\sigma)$-L\"owner chains for arbitrary 
$\delta$ and $\sigma$ as we did for the classical cases
discussed in Chapter \ref{Chapter: Classical cases}. 
It would also be important to prove the following conjecture.

\begin{conjecture}
\label{Conjecture: uniquence of G and T given K}
For any pair of $\delta$ and $\sigma$ and for any hull $\K\subset\Dc$ there
exists at most one $T\in[0,\infty)$ and at most one map 
$G\in\mathcal{G}[\delta,\sigma]$
such that
$\K=\Dc\setminus G(\Dc)$ for the forward case,  
$\K=\Dc\setminus G^{-1}(\Dc)$ for the reverse case,
and
$G=G_T$ for some, possibly non-unique, $(\delta,\sigma)$-L\"owner chain
$\{G_t\}_{t\in[0,T]}$.
\end{conjecture}

In other words, the chain 
$\{G_t\}_{t\in[0,T]}$
may not be unique, but the value 
$T\in[0,\infty)$
is unique 
(which is an interesting conformal invariant characteristic of the set
$\K$);
moreover, the final map $G_T$ is also unique. The proof of the
conjecture is straightforward for the classical cases thanks to the simple
normalization conditions. The parameter $T$ is simply related to the conformal
radius in the radial case and to the half-plane capacity in the chordal case.
However, for a general choice of $\delta$ and $\sigma$ the problem is still
unsolved.

\subsection{Equivalence and normalization of $(\delta,\sigma)$-L\"owner chains}
\label{Section: Equivalence and normalization of slit L chains}

A slit L\"owner chain is determined by a triple $(\delta,\sigma, u_t)$. This
correspondence, however, is not one-to-one. It may happen that different
combinations of $\delta$, $\sigma$ and $u_t$ produce the same L\"owner chain
$\{g_t\}_{t\geq 0}$. It may  also happen that the resulting chains can be
transformed into one another by means of a simple transformation, for instance,
by a linear time reparametrization.

In this section, we define precisely what we mean by a `simple', or
\emph{elementary transformation}
\index{elementary transformations of L\"owner chain}
of a triple $(\delta,\sigma, u_t)$. If two slit L\"owner chains are determined
by triples that can be transformed into one another by means of elementary
transformations, we call such chains \emph{essentially equivalent} 
\index{essentially equivalent L\"owner chains}.

In particular, we show that we can always find a representative in the
equivalence class of triples so that the conditions
\begin{equation} 
	\delta_{-2} = \pm 2,
	\label{Formula: delta_-2 = pm2}
\end{equation}
and
\begin{equation}
	\sigma_{-1} = -1
	\label{Formula: sigma_-1 = -1}
\end{equation}
are satisfied. 

If the vector fields $\delta$ and $\sigma$ are of the form 
\eqref{Formula: delta_-2 = pm2} 
and 
\eqref{Formula: sigma_-1 = -1}, 
then we say that a general slit L\"owner chain driven by $\delta$ and $\sigma$ is a 
\emph{normalized slit L\"owner chain}
\index{normalized slit L\"owner chain}.

Below we list the transformations that we regard as elementary. We use 
the notation 
$\mathscr{V}_c$, $\mathscr{T}_c$, $\mathscr{D}_c$, $\mathscr{R}_c$, and
$\mathscr{S}_c$ with the index $c$ for these transformations, parametrized by
$c\in\mathbb{R}$ and the same letters without the index for the corresponding
one-parameter groups of automorphisms of the triples $(\delta,\sigma, u_t)$.
We summarize the properties of the transformations described below in Table
\ref{Table: elementary transforms}
together with a special combination of them 
\begin{equation}
	\mathscr{P}_c
	:=\mathscr{V}_{e^c} \circ \mathscr{T}_{e^{2c}} \circ \mathscr{S}_c
\end{equation}
that will be used in Section \ref{Section: Equivalence and normalization SLE}.

\subsubsection{Dilation of the driving function $\mathscr{V}$.} 
\index{$\mathscr{V}$}
Let 
$c\in \mathbb{R}\setminus\{0\}$. 
Then we define the transformation $\mathscr{V}_c$ by the formula
\begin{equation}
 \mathscr{V}_c : (\delta, \sigma, u_t) \mapsto 
 	(\delta, c \, \sigma, c^{-1} \,u_t).
\end{equation} 

The triple 
$ (\delta,\,\tilde{\sigma}, \, \tilde{u}_t)= (\delta,c\sigma, c^{-1} u_t)$ 
produces the same L\"owner chain as the triple 
$(\delta,\sigma, u_t)$. 
In order to see this, note that
the flow 
$\{\tilde{H}_t[\tilde \sigma]\}_{t\in \mathbb{R}}$ 
differs from the flow of 
$\{H_t[\sigma]\}_{t \in \mathbb{R}}$ 
only by time reparametrization, 
$H_t[\tilde \sigma] = H_{c t}[\sigma],$ 
so that
\begin{equation}
	H_{\tilde{u}_t}[\tilde \sigma]^{-1}_* \delta 
 	= {H_{u_t}[\sigma]}^{-1}_*\, \delta.
\end{equation}

\subsubsection{Time dilation $\mathscr{T}$} 
\index{$\mathscr{T}$}
For a constant $c>0$ we can consider the time
reparametrized chain $\{\tilde{g}_t\}_{t\geq 0},$ 
where 
$\tilde{g}_t = g_{ct},$
and note that 
$\{\tilde{g}_t\}_{t\geq 0}$ 
is generated by the triple 
$(c\,\delta,\sigma, u_{ct})$. 
Indeed,
\begin{equation}
 \frac{\partial}{\partial t} \tilde{g}_{t} 
 = \left(H_{u_{ct}}[\sigma]^{-1}_*\, (c\, \delta) \right) \circ \tilde{g}_t
\end{equation}

This motivates us to define the transformation $\mathscr{T}_c$
\begin{equation}
 \mathscr{T}_c : (\delta, \sigma, u_t) \mapsto  
 (c\,\delta, \sigma, u_{c\,t}), \quad c>0.
\end{equation}

The transformations 
$\mathscr{V}_c$ 
and 
$\mathscr{T}_c$ 
are sufficient for imposing the normalization conditions 
\eqref{Formula: delta_-2 = pm2} 
and 
\eqref{Formula: sigma_-1 = -1}.
Indeed, due to the fact that 
$\sigma_{-1} \neq 0,$ 
we can apply 
$\mathscr{V}_{1/\sigma_{-1}}$ 
and the transform 
$(\delta,\sigma, u_t) \mapsto (\delta,\frac{\sigma}{\sigma_{-1}}, 
\sigma_{-1} \, u_t),$ 
so that the vector field
\begin{equation}
 \frac{\sigma}{\sigma_{-1}} = \ell_{-1} 
 + \frac{\sigma_0}{\sigma_{-1}} \,\ell_0 
 + \frac{\sigma_1}{\sigma_{-1}}  \,\ell_1
\end{equation}
has the coefficient $1$ in front of $\ell_{-1}$. 

Then, since 
$\delta_{-2}>0,$ 
we apply 
$\mathscr{T}_{2/\delta_{-2}}$, 
so that the triple 
$(\delta,\frac{\sigma}{\sigma_{-1}},  \sigma_{-1} \, u_t)$ 
is further transformed to 
\begin{equation}
 (\tilde{\delta}, \,\tilde{\sigma}, \,\tilde{u}_t) 
 =\left(\frac{2 \delta}{\delta_{-2}}, \,
 \frac{\sigma}{ \sigma_{-1}}, \,\sigma_{-1}\, 
 u_{\frac{2}{\delta_{-2}}t}\right)
\end{equation}
and the vector field 
$\tilde{\delta}$ 
has the coefficient $2$ in front of 
$\ell_{-2}$.

\subsubsection{Drift $\mathscr{D}$}
\index{$\mathscr{D}$} 
We add a linear drift to the driving 
 function $u_t$ and simultaneously modify
the field $\delta$:
\begin{equation}
 \mathscr{D}_c: (\delta, \sigma, u_t) \mapsto 
  (\delta - c \,\sigma, \sigma, u_t +  c\,t), \quad c\in \mathbb{R}.
\end{equation}
The new slit L\"owner chain 
$\{\tilde{g}_t\}_{t\geq 0}$ 
can be expressed as 
\begin{equation}
 \tilde g_{t} = H_{-ct}[\sigma] \circ g_t.
\end{equation}
Indeed, due to 
\eqref{Formula: d G G = (v + G G v) G G}
and
\eqref{Formula: H[v] v = v},
\begin{equation}\begin{split}
 \dot {\tilde g}_{t} 
 =& \frac{\de}{\de t} H_{-ct}[\sigma] \circ g_t 
 = \left(-c \sigma + H_{-ct}[\sigma]_* H_{u_t}[\sigma]_*^{-1} \delta \right) 
 \circ H_{-ct}[\sigma] \circ g_t
 =\\=&
 H_{u_t+ct}[\sigma]_*^{-1} \left(\delta - c \sigma \right) \circ \tilde g_t.
\end{split}\end{equation}
For a differentiable $u_t$ this is straightforward from 
\eqref{Formula: d G = delta G dt + sigma G du}.

We can apply $\mathscr{D}_c$ with $c= \frac{\frac32\,\delta_{-2} \sigma_0 - \delta_{-1}}{\sigma_{-1}}$, 
so that the coefficients in the normalized fields 
$\tilde{\delta}$ 
and 
$\tilde{\sigma}$  
satisfy the following normalization condition:
\begin{equation}
 2\,\frac{\delta_{-1}}{\delta_{-2}} - 3 \frac{\sigma_0}{\sigma_{-1}} = 0.  
 \label{Formula: drift normalization}
\end{equation}
This particular choice is motivated by the stochastic equation that we consider
below. In particular, this restriction appears in Theorem
\ref{Theorem: Absolute continuity of SLE},
Section
\ref{Section: Classification and normalization from algebraic point of view},
the end of Section
\ref{Section: Loewner equation and representation theory},
and Theorem
\ref{Theorem: ppS -> simple cases}.

\subsubsection{Parabolic M\"{o}bius transformation $\mathscr{R}$}
\label{Formula: R - transform}
\index{$\mathscr{R}$}
This elementary transformation is related to a holomorphic automorphism 
$H_c[\ell_1]\colon \Dc\map\Dc$, $c\in\mathbb{R}$ of $\Dc$ induced by the vector
field $\ell_1^{\HH}(z)=z^2$, which has a double zero at the point $a$.
The map ${H_c[\ell_1]}_*$ preserves the form
\eqref{Formula: delta and sigma for Lowner chain}
of $\delta$ and $\sigma$.
In the half-plane chart we have
\begin{equation}
	H_c[\ell_1]^{\HH} = \frac{z}{1-c z},\quad c\in\mathbb{R}.
\end{equation}

The transformation $\mathscr{R}_c$ is defined by
\begin{equation}
 \mathscr{R}_c: (\delta, \sigma, u_t) \mapsto 
 ({H_c[\ell_1]}_* \delta,~{H_c[\ell_1]}_* \sigma,~ u_t), \quad c\in \mathbb{R}.
\end{equation}
The new chain $\{\tilde g_t\}_{t\geq0}$ can be related to the original chain by
$\tilde g_t = H_c[\ell_1]\circ g_t \circ H_c[\ell_1]^{-1}$
with the aid of 
\eqref{Formula: d G G = (v + G G v) G G}.
This transformation changes the hull as $K_t\map \tilde K_t = H_c[\ell_1](K_t)$.

We present how the components 
$\delta_n$, 
$\sigma_n$ 
of the fields 
$\delta$ 
and 
$\sigma$ 
are transformed under the action of 
$H_c[\ell_1]_*$
in Table 
\ref{Table: elementary transforms}.

Application of $\mathscr{R}_c$ preserves the normalization conditions 
\eqref{Formula: delta_-2 = pm2}, 
\eqref{Formula: sigma_-1 = -1} 
and 
\eqref{Formula: drift normalization}.

\subsubsection{Hyperbolic M\"{o}bius transformation $\mathscr{S}$ (space dilation)}
\label{Formula: S - transform}
\index{$\mathscr{S}$} 
The transformation $\mathscr{S}_c$ ($c\in\mathbb{R}$)
is analogous to 
$\mathscr{R}_c,$ 
except for the fact that we use the holomorphic automorphisms  
$H_c[\ell_0]$ 
generated by $\ell_0$ in this case.	The vector field 
$\ell_0(z)=z$ 
has a simple zero at the source point $a$. 
The map ${H_c[\ell_0]}_*$,
as well as 
${H_c[\ell_1]}_*$,
preserves the form
\eqref{Formula: delta and sigma for Lowner chain}
of $\delta$ and $\sigma$.

In the half-plane chart  
$H_c[\ell_0]$ 
is simply the multiplication by the real constant $e^{c}$, i.e., 
\begin{equation}
 H_c[\ell_0]^{\HH}(z) = e^c z. 
\end{equation}
This motivates the term `space dilation'. 
In the unit disk chart, 
\begin{equation}
 H_c[\ell_0]^{\D}(z) 
 = \tau_{\D,\HH} \circ H_c[\ell_0]^{\HH} \circ \tau_{\HH,\D} 
 =  \frac
  { e^{-\frac{c}{2}} (1+z) - e^{\frac{c}{2}} (1-z) }
  { e^{-\frac{c}{2}} (1+z) + e^{\frac{c}{2}} (1-z) } 
\end{equation}
is a M\"obius automorphism of $\D$ that keeps the points $z=\pm 1$ fixed.

The transformation $\mathscr{S}_c$ is defined by
\begin{equation}
 \mathscr{S}_c: (\delta, \sigma, u_t) \mapsto 
 ({H_c[\ell_0]}_* \delta,~{H_c[\ell_0]}_* \sigma,~ u_t), \quad c\in \mathbb{R}.
\end{equation}
The new chain 
$\{\tilde g_t\}_{t\geq0}$ 
can be related to the original chain by
$\tilde g_t = H_c[\ell_0]\circ g_t \circ H_c[\ell_0]^{-1}$
with the aid of 
\eqref{Formula: d G G = (v + G G v) G G}.
This transformation changes the hull as $K_t\map \tilde K_t = H_c[\ell_0](K_t)$.

We present how the components 
$\delta_n$, 
$\sigma_n$ 
of the fields 
$\delta$ 
and 
$\sigma$ 
are transformed under the action of 
$H_c[\ell_0]_*$
in Table 
\ref{Table: elementary transforms}.

Application of $\mathscr{S}_c$ keeps the normalization conditions 
\eqref{Formula: delta_-2 = pm2}, 
\eqref{Formula: sigma_-1 = -1} 
and 
\eqref{Formula: drift normalization}
unchanged.

\afterpage{
\begin{landscape}
 \vspace*{\fill}
\begin{table}
{%
\newcommand{\mc}[3]{\multicolumn{##1}{##2}{##3}}
\begin{center}
\begin{tabular}{l|l|llllll}\cline{2-8}
 & \mc{7}{c|}{Action of the transformation }\\\cline{2-8} & 
 $\delta\map$ & 
 \mc{1}{l|}{$\delta_n$} & 
 \mc{1}{l|}{$\sigma\mapsto$} & 
 \mc{1}{l|}{$\sigma_n$} & 
 \mc{1}{l|}{$u_t\mapsto$} & 
 \mc{1}{l|}{$g_t\mapsto$} &
 \mc{1}{l|}{$K_t\mapsto$} \\\hline
\mc{1}{|l|}{$\mathscr{V}_c$} & 
 $\delta$ & 
 \mc{1}{l|}
 {$\delta_n\map\delta_n$} & 
 \mc{1}{l|}
 {$c\sigma$} & 
 \mc{1}{l|}
 {$\sigma_n\map c\sigma_n$} & 
 \mc{1}{l|}
 {$c^{-1} u_t$} & 
 \mc{1}{l|}
 {$g_t$} &
 \mc{1}{l|}
 {$K_t$} \\\hline
\mc{1}{|l|} {$\mathscr{T}_c$} & 
 $c\delta$ & 
 \mc{1}{l|}
 {$\delta_n\mapsto c\delta_n$} & 
 \mc{1}{l|}
 {$\sigma$} & 
 \mc{1}{l|}{$\sigma_n\mapsto \sigma_n$} & 
 \mc{1}{l|}{$u_{ct}$} & 
 \mc{1}{l|}{$g_{ct}$} & 
 \mc{1}{l|}{$K_{ct}$} \\\hline
\mc{1}{|l|} {$\mathscr{D}_c$} & 
 $ \delta - c \sigma$ & 
 \mc{1}{l|} {$
 \begin{aligned}
  	\delta_{-2}	& \mapsto \delta_{-2} \\ 
  	\delta_{-1}	& \mapsto \delta_{-1} - c \sigma_{-1} \\
  	\delta_{0}	& \mapsto \delta_{0} - c \sigma_{0} \\
  	\delta_{1}	& \mapsto \delta_{1} - c \sigma_{1}
 \end{aligned}$} & 
 \mc{1}{l|}
 {$\sigma$} & 
 \mc{1}{l|}
 {$\sigma_n\mapsto\sigma_n$} & 
 \mc{1}{l|}
 {$u_t+c t$} & 
 \mc{1}{l|}
 {$H_{-ct}[\sigma] \circ g_t$} & 
 \mc{1}{l|}
 {$K_t$} \\\hline
\mc{1}{|l|}{$\mathscr{R}_c$} &
 ${H_c[\ell_1]}_* \delta$ & 
 \mc{1}{l|}{
 $\begin{aligned} 
  \delta_{-2}	& \mapsto \delta_{-2} \\
  \delta_{-1}	& \mapsto \delta_{-1} -3c\delta_{-2} \\
  \delta_{0}	& \mapsto \delta_{0} -2c\delta_{-1} + 3c^2 \delta_{-2} \\
  \delta_{1}	& \mapsto \delta_{1} -c\delta_{0} + c^2 \delta_{-1} -c^3 \delta_{-2}
 \end{aligned}$
 } & 
 \mc{1}{l|}
  {${H_c[\ell_1]}_* \sigma$ } & 
 \mc{1}{l|}{$\begin{aligned}
  \sigma_{-1}	&\mapsto\sigma_{-1}\\
  \sigma_0		&\mapsto\sigma_{0} -2c\sigma_{-1} \\
  \sigma_1 		&\mapsto \sigma_{1} -c\sigma_{0} +c^2\sigma_{-1} \\
 \end{aligned}$} & 
 \mc{1}{l|}{$ u_t$} & 
 \mc{1}{l|}{$H_c[\ell_1] \circ g_t \circ H_c[\ell_1]^{-1}$} & 
 \mc{1}{l|}{$H_c[\ell_1] (K_t)$}\\\hline
\mc{1}{|l|}{$\mathscr{S}_c$} & 
 $H_c[\ell_0]_* \delta$ & 
 \mc{1}{l|}
 {$\delta_n \mapsto e^{nc} \delta_n$} & 
 \mc{1}{l|}{$ H_c[\ell_0]_* \sigma$} & 
 \mc{1}{l|}{$\sigma_n \mapsto e^{nc} \sigma_n$ } & 
 \mc{1}{l|}{$u_t$} & 
 \mc{1}{l|}{$H_c[\ell_0] \circ g_t \circ H_c[\ell_0]^{-1}$} & 
 \mc{1}{l|}{$H_c[\ell_0] (K_t) $} \\\hline
\mc{1}{|l|}
 {$\mathscr{P}_c$} & 
 $e^{2c} H_c[\ell_0] \delta$ & \mc{1}{l|}{$
 \begin{aligned}
 	\delta_{-2}	& \mapsto \delta_{-2} \\
	\delta_{-1}	& \mapsto  e^{c} \delta_{-1} \\
	\delta_{0}	& \mapsto  e^{2c} \delta_{0} \\
	\delta_{1}	& \mapsto  e^{3c} \delta_{1}
 \end{aligned}$} & 
 \mc{1}{l|}
 {$e^c H_c[\ell_0]_* \sigma$} & 
 \mc{1}{l|}{$
 \begin{aligned}
	\sigma_{-1}&\mapsto   \sigma_{-1} \\
	\sigma_{0}&\mapsto  e^{c} \sigma_{0} \\
	\sigma_{1} &\mapsto  e^{2c} \sigma_{1}
 \end{aligned}$} & 
 \mc{1}{l|}
 {$ e^{-c} u_{e^{2c} t}$} & 
 \mc{1}{l|}{$H_c[\ell_0] \circ g_{e^{2c}t} \circ H_c[\ell_0]^{-1}$} & 
 \mc{1}{l|}{$H_c[\ell_0] \left(K_{e^{2c}t}\right) $}\\\hline
\end{tabular}
\end{center}
} 
\caption{Elementary transformations of slit Löwner chains. The last transformation, $\mathscr{P}_c$, is simply the composition 
$\mathscr{P}_c:=\mathscr{V}_{e^c} \circ \mathscr{T}_{e^{2c}} \circ \mathscr{S}_c$, 
see the text for the motivation.}
\label{Table: elementary transforms}
\end{table}
\vspace*{\fill}
\end{landscape}
}

\subsubsection{The structure of the family of 
	essentially different $(\delta,\sigma)$-L\"{o}wner chains} 
\label{Section: The str of the family of essentially diff slit L chains}

The collection of the slit L\"{o}wner chains given by 
$\delta$ and $\sigma$ as in 
\eqref{Formula: delta and sigma for Lowner chain}
can be considered as a seven-parameter space
$\mathbb{R}^{\times 5} \times (\mathbb{R}\setminus \{0\})^2$, 
where one factor $(\mathbb{R}\setminus \{0\})$ corresponds to $\delta_{-2}\neq 0$
and the other factor 
$\mathbb{R}\setminus \{0\}$
corresponds to 
$\sigma_{-1}\neq 0$.
The group $\mathscr{G}$ generated by the five elementary transformations introduced above
defines equivalence classes of \emph{essentially equivalent} L\"owner chains. 
\index{essentially equivalent L\"owner chains}
Thus, the
family of essentially different forward chains can be represented by the factor
space 
$\mathbb{R}^{\times 5} \times \mathbb{R}^ 
	+ \times (\mathbb{R}\setminus \{0\})/ \mathscr{G}
	$.
The classical chordal, dipolar and radial chains are represented by three
single points in this two-dimensional space.

To specify this factorized collection of L\"owner chains we have already
introduced three normalization conditions \eqref{Formula: delta_-2 = pm2},
\eqref{Formula: sigma_-1 = -1}
and
\eqref{Formula: drift normalization} 
above.
Using the last two transformations, $\mathscr{R}$ and $\mathscr{S}$, two more
parameters can be fixed, but there is no need to introduce such restrictions
in this text.

%% file: SHSF.tex
\section{Slit holomorphic stochastic flow or ($\delta,\sigma$)-SLE}
\label{Section: SLE}

This section is dedicated to the central object of this monograph, called the 
slit holomorphic stochastic flow or $(\delta,\sigma)$-SLE. The author prefers
the second name to highlight that it is a version of the Schramm-L\"owner
evolution.

\subsection{Definition and basic properties}
\label{Section: SLE defenition and basic properties}


Let 
($\Omega,\mathcal{F},P$)
be a probability space, and let $A$ and $B$ be two random variables taking
values in a measurable space $S$. We denote by $\law{A}$ the induced measure on
$S$. We say that the laws of $A$ and $B$
are \emph{identical}
(equivalently, $A$ and $B$ \emph{agree in law}):
\index{identical random laws (agreement in law)}
\begin{equation}
	\law{A} = \law{B}
	\label{Formula: Law A = Law B}
\end{equation}
if the following expectations are equal 
\begin{equation}
	\Ev{f(A)} = \Ev{f(B)}
\end{equation}
for any bounded measurable function $f:S\map\mathbb{R}$.
More generally, an equality of conditional random laws
\begin{equation}
	\law{A~|~\mathcal{F}_1} = \law{B~|~\mathcal{F}_2}
\end{equation}
is, by definition, the identity
\begin{equation}
	\Ev{f(A)~|~\mathcal{F}_1}(\omega) = 
	\Ev{f(B)~|~\mathcal{F}_2}(\omega)\quad
	\text{a.s.}
\end{equation}
for any bounded $\mathcal{F}$-measurable function $f:S\map\mathbb{R}$.
One can understand a conditional law as a law parametrized by $\omega$.

 
Let $\{B_t\}_{t\in[0,+\infty)}$ be the standard Brownian motion naturally adapted
to a filtration $\{\mathcal{F}_t\}_{t\in[0,+\infty)}$ 
of a filtered probability space
\begin{equation}
	(\Omega^B,\{\mathcal{F}_t^B\}_{t\in[0,+\infty)},P^B).
	\label{Formula: (Omega,F_t,P)}
\end{equation}
Define
\begin{equation}
	\mathcal{F}_T := \{ B\in\mathcal{F}^B \colon B \cap \{t\leq T\}\in
	\mathcal{F}_t^B,~ \forall t\in[0,+\infty)\}.
\end{equation}
The following statement is known as 
the \emph{strong Markov property of the Brownian motion}
\index{strong Markov property of the Brownian motion}
\begin{equation}
	\law {\{B_{s+T} - B_T\}_{s\in[0,+\infty)}} = 
	\law {\{B_s\}_{s\in[0,\infty)}}
	\label{Formula: Strong markov propery of B}
\end{equation}
and 
$\{B_{s+T} - B_T\}_{s\in[0,\infty)}$
is independent of 
$\mathcal{F}_T$
for any stopping time $T<+\infty$ a.s.
We consider 
$\{B_s\}_{s\in[0,\infty)}$ 
as a random variable taking values in the
space $C^0_{[0,+\infty)}$ of the continuous functions $B:[0,+\infty)\map
\mathbb{R}$ such that $B_0=0$.

One of the consequences of the strong Markov property is
\begin{equation} 
	\law {B_s~|~\mathcal{F}_T} = 
	\law {B_s~|~B_T},\quad s\in[0,\infty)
	\label{Formula: law B_s | F_T = law B_s | B_T}
\end{equation}
for any stopping time $T<+\infty$ a.s.
In particular, we have
\begin{equation} 
	\law {B_s~|~\mathcal{F}_t^B} = 
	\law {B_s~|~B_t},\quad t,s\in[0,\infty),
	\label{Formula: law B_s | F_t = law B_s | B_t}
\end{equation}
which is known as
the \emph{simple Markov property} of $\{B_t\}_{t\in[0,+\infty)}$.
\index{simple Markov property of a stochastic process}
Here and below, by a conditional expectation (or conditional law) with
respect to a random variable $X$, we mean the expectation (or the law) with
respect to the sigma-algebra $\mathcal{F}_X$ generated by $X$.

We are now ready to discuss a stochastic version of the $(\delta,\sigma)$-L\"owner equation. 
We simply equip the driving function $u_t$ with the Brownian measure
\begin{equation}
 u_t:=B_t,\quad t\geq 0.
 \label{Formula: u_t = B_t}
\end{equation}
Equivalently, we formulate the following definition.

\begin{definition}
Let $\delta$ and $\sigma$ be as in
\eqref{Formula: delta and sigma for Lowner chain},
and let $\{B_t\}_{t\geq 0}$ be the standard Brownian motion. 
Then the solution to the stochastic differential equation in the Stratonovich
form
\begin{equation}
	\dS G_t = \delta \circ G_t dt + \sigma \circ G_t \dS B_t,\quad G_0=\id,
	\label{Formula: Slit hol stoch flow Strat}
\end{equation}
is called a forward (reverse)
\emph{slit holomorphic stochastic flow}
\index{slit holomorphic stochastic flow} 
or forward (reverse)
\emph{($\delta,\sigma$)-SLE}.
\index{($\delta,\sigma$)-SLE}
\end{definition}

The stochastic differential equation
\eqref{Formula: Slit hol stoch flow Strat}
is a convenient shorthand for the Stratonovich integral equation
\begin{equation}
	\int\limits_{0}^{t} G_{\tau} \dS B_{\tau}  
	=\int\limits_{0}^{t} \delta \circ G_{\tau} d\tau 
	+ \int\limits_{0}^{t} \sigma \circ G_{\tau} \dS B_{\tau}.
	\label{Formula: Slit hol stoch flow Strat integral form}
\end{equation}
We use the notation `$\dS$' ($\dI$) to indicate that the integrals are taken in the
Stratonovich (It\^o) sense.
The relation between these two integrals is given in, e.g.,
\cite[Section 4.3.6]{Gardiner1982} and in Appendix
\ref{Appendix: Some relations from stochastic calculus}.

The Stratonovich form is more convenient in our setup than the more
frequently used It\^{o} form. The equation
\eqref{Formula: Slit hol stoch flow Strat}
in the It\^{o} form is
\begin{equation}\begin{split}
 & \dI G_t^{\psi}(z) 
 =
 \left(\delta^{\psi} + \frac12\sigma^{\psi} {\sigma^{\psi}}' \right) 
  \circ G_t^{\psi}(z) dt 
 + \sigma^{\psi} \circ G_t^{\psi}(z) \dI B_t,\\
 &G^{\psi}_0(z)=z,
 \label{Formula: Slit hol stoch flow Ito}
\end{split}\end{equation}
choosing some chart $\psi$. A disadvantage of the It\^o form is that the
expression in parentheses in
\eqref{Formula: Slit hol stoch flow Ito}
transforms from chart to chart in a complicated
manner, whereas the functions $\delta$ and $\sigma$ in the
Stratonovich form 
\eqref{Formula: Slit hol stoch flow Strat}
transform as vector fields.

An equivalent definition of ($\delta,\sigma$)-SLE can be given in terms of 
$\{g_t\}_{t\in[0,+\infty)}$ and equation 
\eqref{Formula: d g = h delta g dt}
if we equip the set of driving functions 
$\{u_t\}_{t\in[0,+\infty)}$
with the Brownian measure by setting
\begin{equation}
	u_t:=B_t,\quad t\in[0,+\infty).
\end{equation}

The existence and uniqueness of the solution to 
\eqref{Formula: Slit hol stoch flow Strat}
can be directly obtained from the existence and uniqueness of the solution to 
\eqref{Formula: d g = h delta g dt} 
for each sample of 
$u_t = B_t$. 
An alternative approach is to use a general theory of stochastic differential
equations and flows, see, for instance,
\cite{Protter2004a}. The main difficulty with this approach is that the function
$\delta^{\psi}(z)$ does not satisfy the Lipschitz condition. 

Since there is at most one driving function for a given chain, we have the
following one-to-one correspondence
\begin{equation}
	\{B_t\}_{t\in[0,\infty)} \longleftrightarrow \{G_t\}_{t\in[0,\infty)}
	\label{Formula: B_t <-> G_t}
\end{equation}
for each fixed pair $(\delta,\sigma)$.
Thus, the $(\delta,\sigma)$-L\"owner chain $\{G_t\}_{t\in[0,+\infty)}$ is
a random variable on the same probability space 
\eqref{Formula: (Omega,F_t,P)}
as the Brownian motion $B_t$.
Moreover, the chain
$\{G_s\}_{s\in[0,t]}$ is an $\mathcal{F}^B_t$-measurable random variable
defined by the one-to-one correspondence
\begin{equation}
	\{B_s\}_{s\in[0,t]} \longleftrightarrow \{G_s\}_{s\in[0,t]}.
\end{equation}
Using the map
\begin{equation}
	\{G_s\}_{s\in[0,t]} \mapsto G_t 
\end{equation}
we can define a stochastic process 
$\{G_t\}_{t\in[0,+\infty)}$
taking values in 
$\mathcal{G}[\delta,\sigma]$
adapted to the filtration 
$\{\mathcal{F}^B_t\}_{t\in[0,+\infty)}$.

%
%
%

From the strong Markov property of the Brownian motion 
\eqref{Formula: Strong markov propery of B}
and
\eqref{Formula: B_t <-> G_t}
we conclude that
\begin{equation}
	\law{\{G_{t+T} \circ G_{T}^{-1}\}_{t\in[0,+\infty)}} = 
	\law{\{G_t\}_{t\in[0,\infty)}}
\end{equation}
with the process on the left-hand side independent of $\mathcal{F}_T$
for any stopping time $T<+\infty$ a.s.
This property of the random law of the L\"owner chains 
$\{G_t\}_{t\in[0,+\infty)}$
can be called the
\emph{strong Markov property of $(\delta,\sigma)$-SLE}.
\index{strong Markov property of $(\delta,\sigma)$-SLE}
From the properties of 
$\{B_t\}_{t\in[0,+\infty)}$ 
and from the relation
\eqref{Formula: G_t = G_t,s circ G_s}
it follows that the process 
$\{G_t\}_{t\in[0,+\infty)}$ 
also possesses
the simple Markov property:
\begin{equation}
	\mathrm{Law}[G_{t+s} ~|~\mathcal{F}_t] = \mathrm{Law}[G_{t+s}~|~G_t],	\quad
	s,t\geq 0.
\end{equation}
Equivalently,
\begin{equation}
	\Ev{ f(G_{t+s}) ~|~\mathcal{F}_t}(\omega) =
	\Ev{f(G_{t+s})~|~\mathcal{F}_{G_t}}(\omega), \quad \text{a.s.,} \quad s,t\geq 0
\end{equation}
for any function 
$f:\mathcal{G}[\delta,\sigma]\map\mathbb{R}$
that is bounded and measurable with respect to the sigma-algebra on
$\mathcal{G}[\delta,\sigma]$ generated by $\{B_{\tau}\}_{\tau\in[0,t+s]}$. We
denote by $\mathcal{F}_{G_t}$ the sigma-algebra on $\Omega_B$ induced by the random
variable $G_t$.
Moreover, 
the random conformal maps 
$G_{t+s,s}:=G_{t+s}\circ G_s^{-1}$ 
and 
$G_s$ 
are independent for $t,s\geq 0$.
More generally, for any finite collection of times 
$t_1>t_2>t_3>\dotso t_n\geq 0$
the random variables
$G_{t_1,t_2}$, $G_{t_2,t_3}$ , ... , $G_{t_{n-1},t_n}$ are independent thanks to
the independence of 
$B_{t_1}-B_{t_2}$, $B_{t_2}-B_{t_3}$, ... , $B_{t_{n-1}}-B_{t_n}$
and 
\eqref{Formula: G_t,s = G_t,r circ G_r,s}.
If we take into account that $\mathcal{G}[\delta,\sigma]$ is a semigroup with
respect to the composition, the last property can be interpreted as the
`independence of semigroup increments'.

Thanks to the map
\begin{equation}
	G_t \mapsto \K_t=\Dc\setminus G_t(\Dc)
\end{equation}
there is an $\mathcal{F}_t^B$-process on the subsets of $\Dc$.
The strong Markov property of random maps $G_t$ can be extended to
the random law on the subsets $\K_t$ if Conjecture 
\ref{Conjecture: uniquence of G and T given K}
is true:
\begin{equation}
	\law{\{G_T(\K_{t+T}\setminus \K_T) \}_{t\in[0,+\infty)}} = 
	\law{\{\K_t\}_{t\in[0,+\infty)}}.
\end{equation}
\index{strong Markov property on hulls}
According to the results of Section 
\ref{Section: Slit holomorphic Loewner chain} 
this is true 
at least for the classical cases (due to the normalization
conditions), and 
when $\K_t$ is a simple curve a.s. (due to Theorem
\ref{Theorem: Curve -> chain}). 
We study when $\K_t$ is a simple curve
in Section 
\ref{Section: Relations between essentially different SLEs}.

\begin{figure}[hp]
\centering
  \begin{subfigure}[t]{0.33\textwidth}
		\centering
    \includegraphics[width=4.5cm,keepaspectratio=true]
    	{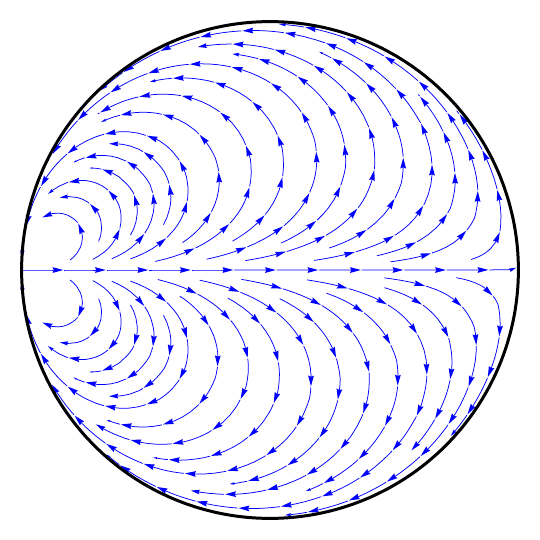}
  \end{subfigure}%
  ~  
  \begin{subfigure}[t]{0.33\textwidth}
		\centering
    \includegraphics[width=4.5cm,keepaspectratio=true]
    	{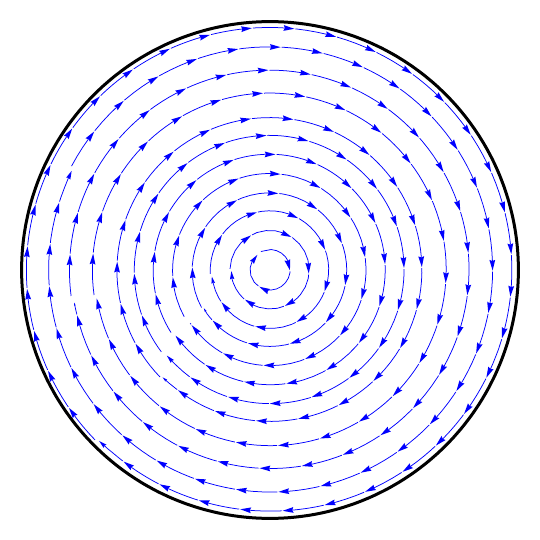} 
	\end{subfigure}%
	~
	\begin{subfigure}[t]{0.33\textwidth}
	\centering
		\includegraphics[width=4.5cm,keepaspectratio=true]
			{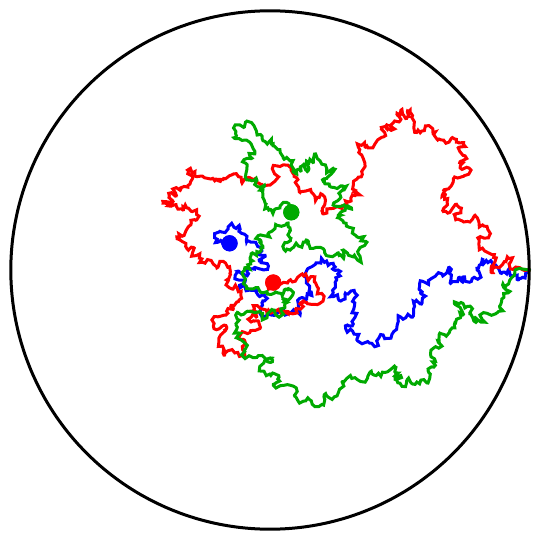}
	\end{subfigure}
	~
	
  \begin{subfigure}[t]{0.33\textwidth}
		\centering
    \includegraphics[width=4.6cm,keepaspectratio=true]
    	{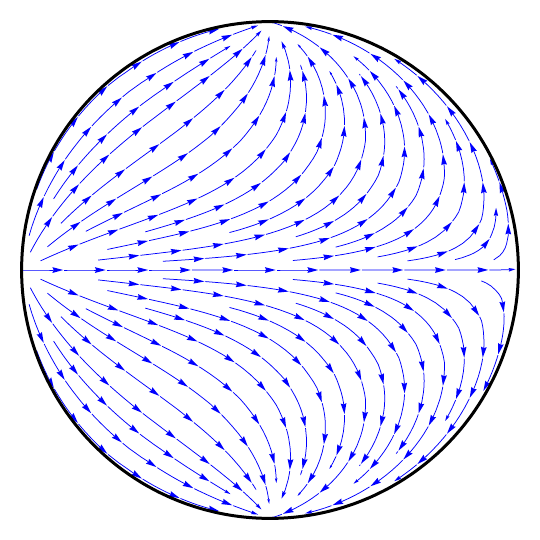}
  \end{subfigure}%
  ~  
  \begin{subfigure}[t]{0.33\textwidth}
		\centering
    \includegraphics[width=4.5cm,keepaspectratio=true]
    	{Pictures/Vector_fileds_v=-_-1_-_1_.pdf} 
	\end{subfigure}%
	~
	\begin{subfigure}[t]{0.33\textwidth}
	\centering
		\includegraphics[width=4.5cm,keepaspectratio=true]
			{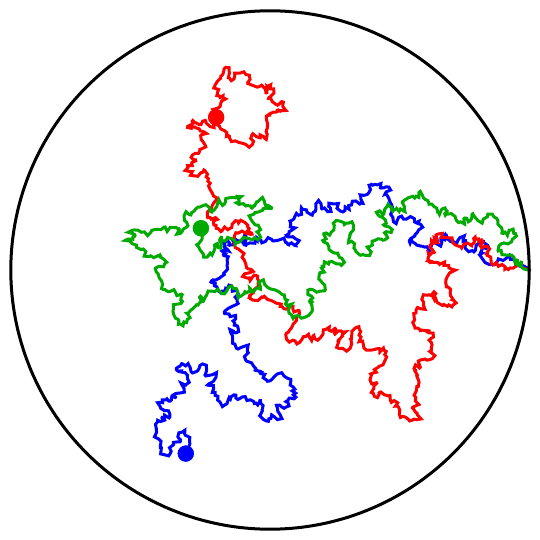}
	\end{subfigure}
	~

  \begin{subfigure}[t]{0.33\textwidth}
		\centering
    \includegraphics[width=4.5cm,keepaspectratio=true]
    	{Pictures/Vector_fileds_v=_-2_.pdf}
  \end{subfigure}%
  ~  
  \begin{subfigure}[t]{0.33\textwidth}
		\centering
    \includegraphics[width=4.5cm,keepaspectratio=true]
    	{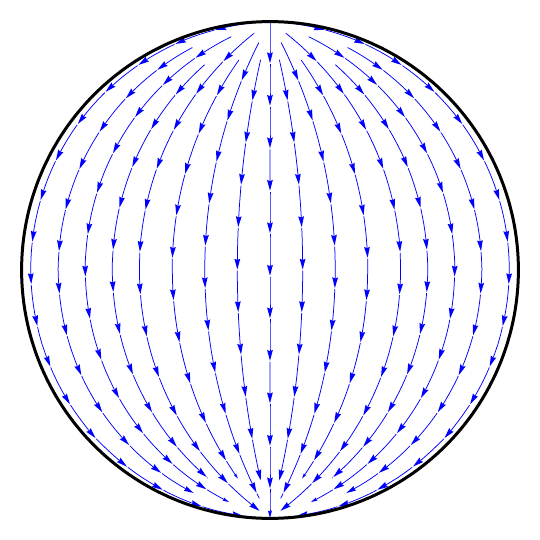} 
	\end{subfigure}%
	~
	\begin{subfigure}[t]{0.33\textwidth}
	\centering
		\includegraphics[width=4.5cm,keepaspectratio=true]
			{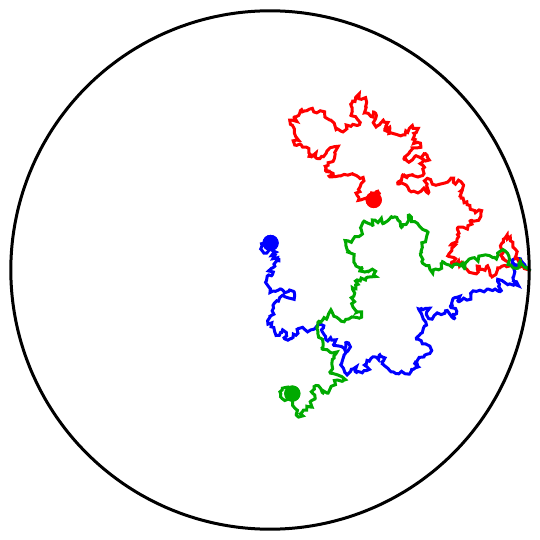}
	\end{subfigure}
	~
	
  \begin{subfigure}[t]{0.33\textwidth} 
		\centering
    \includegraphics[width=4.5cm,keepaspectratio=true]
    	{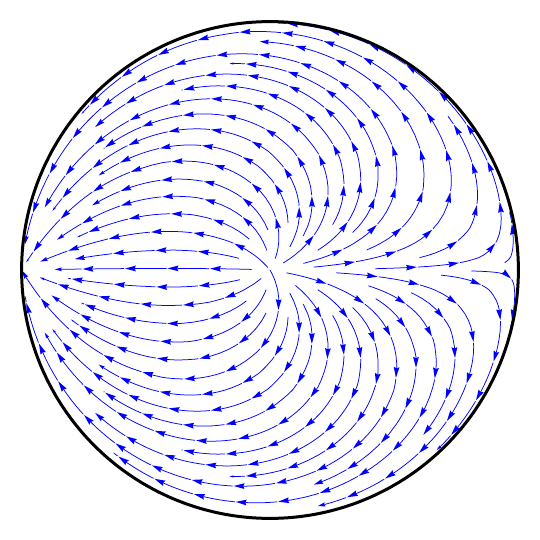}
  \end{subfigure}%
  ~  
  \begin{subfigure}[t]{0.33\textwidth}
		\centering
    \includegraphics[width=4.5cm,keepaspectratio=true]
    	{Pictures/Vector_fileds_v=-_-1_+_1_.pdf} 
	\end{subfigure}%
	~
	\begin{subfigure}[t]{0.33\textwidth}
	\centering
		\includegraphics[width=4.5cm,keepaspectratio=true]
			{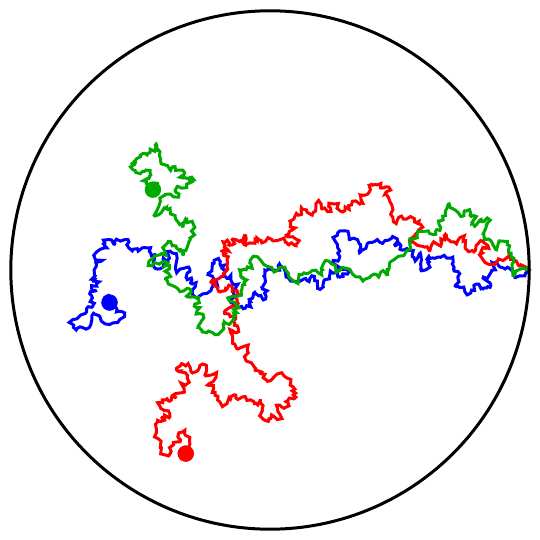}
	\end{subfigure}
 	\caption{Samples of the ($\delta,\sigma$)-SLE slits (in the right column) for
 	different choices of $\delta$ (the left column) and $\sigma$ (the
 	central column). We use the normalization from Section 
 	\ref{Section: Equivalence and normalization SLE} and $\kappa=2$. Each
 	of the four rows corresponds to some choice of $\delta$ and $\sigma$. 
 	We pick these four combinations just for illustration; they do not have any
 	special properties. Some other combinations are
 	studied in Chapter 
 	\ref{Chapter: Classical cases}.
 	\label{Figure: Some SLE Examples}}
\end{figure}

 \subsection{Equivalence and normalization of $(\delta,\sigma)$-SLE}
 \label{Section: Equivalence and normalization SLE}

We now use the results of Section
\ref{Section: Equivalence and normalization of slit L chains}
to study how the elementary transformations change the random law on the
driving function $u_t=B_t$. We remark that the transforms 
$\mathscr{V}$, $\mathscr{T}$, and $\mathscr{D}$ change the random law as
follows:
\begin{equation}
	\mathscr{V}_c:B_t\mapsto c^{-1} B_t,
\end{equation}
\begin{equation}
	\mathscr{T}_c:B_t\mapsto B_{ct}= c^{\frac12} \tilde B_t,
\end{equation}
\begin{equation}
	\mathscr{D}_c:B_t\mapsto B_t+ct,
\end{equation}
where $\tilde B_t$ is another sample of the standard Brownian motion.
The transforms 
$\mathscr{R}$ 
and
$\mathscr{S}$ do not change the law of $B_t$.

Using the transform 
\begin{equation}
	\mathscr{T}_{c^2} \circ \mathscr{V}_{c}
	\label{Formula: T_c circ V_s^1/2}
\end{equation}
with
$c=2/\delta_{-2}$, we can impose the normalization condition 
\eqref{Formula: delta_-2 = pm2}
keeping the random law of the driving function unchanged. The equation 
\eqref{Formula: Slit hol stoch flow Strat} takes the form 
\begin{equation}
  \dS G_t = 
  \underline\delta \circ G_t dt 
  + \underline\sigma \circ G_t \sqrt{\kappa} \dS B_t ,\quad G_0=\id,
  \label{Formula: dS G = delta_N dt + kappa^1/2 sigma_N dB} 
\end{equation}
where we have defined the parameter
\begin{equation}
	\kappa :=\frac{2 \sigma_{-1}^2}{\delta_{-2}}
	\label{Formula: kappa := ...}
\end{equation}
and denoted the normalized $\delta$ and $\sigma$ by an underline:  
\begin{equation}\begin{split}
	\underline\delta =& 
	\pm (2\ell_{-2} + \underline\delta_{-1} \ell_{-1} 
	+ \underline\delta_{0} \ell_{0} 
	+ \underline\delta_{1} \ell_{1}) \\
	\underline\sigma =& 
	-\ell_{-1} + \underline\sigma_{0} \ell_{0} +	\underline\sigma_{1}\ell_{1}
	\label{Formula: delta_N = ..., sigma_N = ... without drift}
\end{split}\end{equation}
This definition of $\kappa$ agrees with its analog from the classical
chordal, dipolar, and radial cases. It will be clear below that the parameter
$\kappa$ defined this way is responsible for the local geometric properties of
the hull $\K_t$ in the general case as well. We will often assume below that $\delta$
and $\sigma$ are normalized and drop the underline.

The transform 
$\mathscr{R}$
keeps the form of
\eqref{Formula: dS G = delta_N dt + kappa^1/2 sigma_N dB}
unchanged because it has no influence on the driving function $B_t$
and on the normalizations
\eqref{Formula: delta_-2 = pm2} 
and
\eqref{Formula: sigma_-1 = -1}.
The transform
$\mathscr{P}$ introduced at the end of 
Section \ref{Section: Equivalence and normalization of slit L chains}
acts on the driving function as
\begin{equation}
	\mathscr{P}_c: B_t \mapsto e^{-c} B_{e^{2c}t} =: \tilde B_t  
\end{equation}
and does not change the random law of the Brownian motion. On the other hand,
$\mathscr{P}$ does not change the normalization conditions 
\eqref{Formula: delta_-2 = pm2} 
and
\eqref{Formula: sigma_-1 = -1}
introduced above.

In addition, we can use the transform $\mathscr{D}$ to impose
\eqref{Formula: drift normalization}. After that, we have 
\begin{equation}\begin{split}
  \dS G_t = 
  \underline\delta \circ G_t dt 
  + \underline\sigma \circ G_t \left( \sqrt{\kappa} \dS B_t 
  + \nu dt\right),\quad 
  G_0=\id,
  \label{Formula: dS G = delta_N dt + sigma_N (kappa^1/2 dB + v dt)}
\end{split}\end{equation}
where the parameter $\nu$ is chosen such that the additional (drift)
normalization condition 
\begin{equation}\begin{split}
	\underline\delta_{-1} + 3 \underline\sigma_0 = 0
	\label{Formula: delta_-1 - 3 sigma_0 = 0}
\end{split}\end{equation}
is satisfied.

The parameter $\nu$ can be defined by
\begin{equation}
 	\nu:= 
 	\frac{\sqrt{2}	\delta_{-1}}{\sqrt{\delta_{-2}}}
 	-\frac{ 3 \sqrt{\delta_{-2}} \sigma_0}{\sqrt{2} \sigma_{-1}}. 
	\label{Formula: nu := ...}
\end{equation}
The transform $\mathscr{P}$ preserves $\kappa$, but changes $\nu$ as
\begin{equation}
	\mathscr{P}_c \colon \nu\mapsto e^{c} \nu.
\end{equation}
Thus, there are three nonequivalent cases: $\nu=0$, $\nu>0$, and $\nu<0$. Each
case is invariant with respect to $\mathscr{P}$. We call the case $\nu=0$
\emph{driftless}.
\index{driftless $(\delta,\sigma)$-SLE}
The last two cases are the mirror images of each other.

The relations 
\eqref{Formula: delta_-1 - 3 sigma_0 = 0}
and
\eqref{Formula: nu := ...}
will be motivated below 
(in Theorems 
\ref{Theorem: Absolute continuity of SLE} 
and 
\ref{Theorem: ppS -> simple cases}).
For now we just remark that 
\eqref{Formula: delta_-1 - 3 sigma_0 = 0}
is invariant with respect to $\mathscr{R}$ and $\mathscr{P}$.
The right-hand side of 
\eqref{Formula: nu := ...}
is invariant with respect to 
$\mathscr{S}$  and $\mathscr{R}$.
In particular, $\nu$ does not depend on the choice of the basis
vectors $\ell_n$ if one of the poles is placed at the source
point $a\in\de\Dc$.

From the above, we conclude that the family of \emph{essentially different}
\index{essentially different ($\delta,\sigma$)-SLE}
($\delta,\sigma$)-SLEs is a two-parameter family ($7-2-1-2=2$),
and each ($\delta,\sigma$)-SLE is in
addition parametrized by 
$\kappa>0$ 
and 
$\nu\in\mathbb{R}$. 
The classical chordal, radial and dipolar equations are just elements of this
two-parameter set. We emphasize that
both $\kappa$ and $\nu$ are parameters related to the stochastic
version of the L\"owner equation only. In the deterministic case, they can
always be absorbed by the driving function after the transforms 
$\mathscr{V}$
and
$\mathscr{D}$.

\subsection{Relations between essentially different $(\delta,\sigma)$-SLEs}
\label{Section: Relations between essentially different SLEs}

In Section \ref{Section: General properties of Loewner chain},
we showed that
any $(\delta,\sigma)$-L\"owner chain $\{G_t\}_{t\in[0,+\infty)}$
after the transform 
\begin{equation}
	\tilde G_{\tilde t} = \hat M_{\tilde t} \circ
	G_{\lambda_{\tilde t}},\quad \tilde t\in[0,\tilde T)
	\label{Formula: G_t = M_t circ G_l_t}
\end{equation}
is a 
$(\tilde \delta, \tilde \sigma)$-L\"owner chain for a given pair 
of $\tilde \delta$ and $\tilde \sigma$
and some properly chosen $\{\hat M_{\tilde t}\}_{\tilde t\in[0,\tilde T)}$ and
$\{\lambda_{\tilde t}\}_{\tilde t\in[0,\tilde T)}$. 
It will be convenient in this setup to parametrize the map $M$ by $\tilde t$
rather than $t$ and to define the time reparametrization in the inverse way
$\lambda\colon[0,\tilde T)\map[0, T)$.
In this subsection, we consider a stochastic analog of this result.
Namely, we ask how the random law on the driving functions changes under this
transform. We establish that, in the general case, the initial Brownian random law
on the driving functions $\{u_t\}_{t\in[0,+\infty)}$ 
of
$\{G_t\}_{t\in[0,+\infty)}$
and the transformed random law on
$\{\tilde u_{\tilde t}\}_{\tilde t\in[0,+\infty)}$
of
$\{\tilde G_{\tilde t}\}_{\tilde t\in[0,+\infty)}$
are absolutely continuous with respect to each other.
This has far-reaching consequences for the local properties of $\K_t$.
Moreover, in some cases, the law of  
$\{\tilde u_{\tilde t}\}_{\tilde t\in[0,\tilde T)}$
is also Brownian. A special case of this for the chordal and radial
SLEs was considered before in 
\cite{Schramm2006}.

\begin{theorem}
The random laws on 
$\{\mathcal{K}_t\}_{t\in[0,\infty)}$ 
and 
$\{\tilde {\mathcal{K}}_{\tilde t}\}_{\tilde t\in[0,\infty)}$ 
generated by a forward
$(\delta,\sigma)$-SLE 
and a forward 
$(\tilde \delta,\tilde \sigma)$-SLE,
respectively, with the same $\kappa$, are locally equivalent until some
stopping time $T>0$ modulo some random strictly increasing time change 
$\lambda\colon[0,\tilde T)\map[0, T)$. 
Moreover, if $\kappa=6$ and if both SLEs are driftless, then the random laws 
of 
$\K_{\lambda_{\tilde t}}$ 
and
$\tilde \K_{\tilde t}$
are identical for $\tilde t\in[0,\tilde T)$.

\label{Theorem: Absolute continuity of SLE}
\end{theorem}

`Locally' means that the equivalence is valid until each stopping time
$T_n>0$ from some collection
$\{T_n\}_{n=1,2,3,\dotso}$
such that $T_n \geq T_m$ if $n>m$ and $T_n\map T$
a.s. when $n\map \infty$.

\begin{proof}
We use the results of the proof of Theorem 
\ref{Theorem: The master theorem}.
Assume $u_t$ and $\tilde u_{\tilde t}$ are some continuous It\^o processes.
Then the stochastic version of formula
\eqref{Formula: dot M_t = ( dot l H delta - M H delta ) circ M_t}
is
\begin{equation}
	\dS \hat M_{\tilde t}  \circ \hat M_{\tilde t}^{-1} = 
	\left( \tilde \delta d \tilde t + \tilde \sigma \dS \tilde u_{\tilde t} \right)
	- \dot \lambda_{\tilde t}~ \hat M_{\tilde t}{}_* 
		\left( \delta d \tilde t + \sigma \dS u_{\lambda_{\tilde t}} \right)
	\label{Formula: dS M circ M^-1 = ...}
\end{equation}
It will be convenient to consider this relation in the half-plane chart.
The expression for $\hat M_t^{\HH}\colon\HH\map\HH$ is
\begin{equation}
	\hat M_{\tilde t}^{\HH}(z) 
	= \frac{\alpha_{\tilde t} z}{1+\beta_{\tilde t} z},\quad z\in\HH, 
	\label{Formula: hat M^H(z) = az/(1+bz)}
\end{equation} 
for some real-valued stochastic processes 
$\{\alpha_{\tilde t}\}_{\tilde t\in[0,\tilde T)}$ 
and
$\{\beta_{\tilde t}\}_{\tilde t\in[0,\tilde T)}$. 
Since $G_0=\tilde G_0=\id$, we impose the initial conditions
\begin{equation}
	\alpha_{0}=1,\quad \beta_0=0.
	\label{Formula: alpha_0=1, beta_0=0}
\end{equation}
For 
$\left(\hat M_{\tilde t}{}_* v \right)^{\HH}(z)$ 
we have
\begin{equation}
	\left(\hat M_{\tilde t}{}_* v \right)^{\HH}(z) = 
	\frac{(\alpha_{\tilde t}-\beta_{\tilde t}~ z)^2}{\alpha_{\tilde t}} 
		v^{\HH} \left( \frac{z}{\alpha_{\tilde t} - \beta_{\tilde t}~z} \right),\quad 
		z\in\HH.
\end{equation}

Expand both sides of the identity 
\eqref{Formula: dS M circ M^-1 = ...}
into Laurent series with respect to $z$ and impose the normalization
\eqref{Formula: delta_N = ..., sigma_N = ... without drift}.
The coefficients of $z^{-1}$ give
\begin{equation}
	0 = 2 - 2~ \alpha_{\tilde t}^2~ \dot \lambda_{\tilde t}
	\quad \Leftrightarrow \quad
	\alpha_{\tilde t} = \frac{1}{\sqrt{ \dot \lambda_{\tilde t} }},
	\label{Formula: a = 1/sqrt lambda}
\end{equation}
Recall that $\dot \lambda_{\tilde t}>0$, $\tilde t\in[0,\tilde T)$ a.s.

The coefficients of $z^0$ give
\begin{equation}
	0 =
	\tilde \delta_{-1}d\tilde t 
	- \dS \tilde u_{\tilde t} 
	+ 6~ \dot\lambda_{\tilde t}~ \alpha_{\tilde t}~ \beta_{\tilde t}~ d \tilde t
	- \dot \lambda_{\tilde t} \alpha_{\tilde t} \delta_{-1} d \tilde t
	+ \dot\lambda_{\tilde t}~ \alpha_{\tilde t}~ \dS u_{\lambda_{\tilde
	t}},
	\label{Formula: 0 = ... 1} 
\end{equation}
and we conclude that
\begin{equation}
	\dS \tilde u_{\tilde t} = 
	\left(
		6~ \dot\lambda_{\tilde t}^{\frac12}~ \beta_{\tilde t} 
		+ \tilde \delta_{-1}
		- \dot\lambda_{\tilde t}^{\frac12}~ \delta_{-1}
	\right) d \tilde t
	+
	\dot\lambda_{\tilde t}^{\frac12}~ \dS u_{\lambda_{\tilde t}}.
	\label{Formula: d tilde u = ... 1}
\end{equation}
If we assume that the driving function $u_{t}$ is a Brownian
motion, $u_{t} = \sqrt{\kappa }B_{t}$, the last formula
is an analogue of the relation from Proposition 4.2 in
\cite{Lawler2001a}.
We apply the same argument and item 6 from Theorem
\ref{Theorem: The master theorem}
to obtain the first statement of the theorem.

In order to obtain the second statement, we consider the coefficient of $z^1$
in the identity
\eqref{Formula: dS M circ M^-1 = ...}
\begin{equation}\begin{split}
	&\frac{\dS \alpha_{\tilde t}}{\alpha_{\tilde t}} =
	\tilde \delta_0~ d\tilde t 
	+  \tilde \sigma_0~\dS \tilde u_{\tilde t} 
	- 6~ \beta_{\tilde t}^2~ \dot\lambda_{\tilde t}~ d\tilde t
	-\\-&
	2~ \beta_{\tilde t}~ \dot\lambda_{\tilde t} 
		\dS u_{\lambda_{\tilde t}} 
	+ 2~ \dot\lambda_{\tilde t}~ \beta_{\tilde t}~ \delta_{-1}~ d\tilde t
	- \dot\lambda_{\tilde t}~ \delta_0~ d\tilde t
	- \dot\lambda_{\tilde t}~ \sigma_0~ \dS u_{\lambda_{\tilde t}}.
\end{split}\end{equation} 
We use this formula and
\eqref{Formula: a = 1/sqrt lambda}
to obtain 
$\dS \dot \lambda_{\tilde t}$ 
\begin{equation}
	\dS \dot \lambda_{\tilde t} 
	=-2 \dot \lambda_{\tilde t} \frac{\dS \alpha_{\tilde t}}{\alpha_{\tilde t}}
	= (\dotso) d\tilde t  
	-2 \dot \lambda_{\tilde t} \left(
		-2 \dot \lambda_{\tilde t}^{\frac12} \beta_{\tilde t}
		+ \tilde \sigma_0
		- \dot \lambda_{\tilde t}^{\frac12} \sigma_0
	\right) \dS \tilde u_{\tilde t}.
\end{equation}
Assume now that $\tilde u_{\tilde t}$ is a Brownian motion, 
$\tilde u_{\tilde t}=\sqrt{\tilde \kappa}\tilde B_{\tilde t}$, 
$\tilde t\in[0,\tilde T]$.
We can apply
\eqref{Formula: d lambda = a dt + b dB}
and
\eqref{Folrmula: dB = lambda dB - 1/4 b/lambda dt}
to conclude
\begin{equation}
	\sqrt{\kappa} \dI \tilde B_{\tilde t} = 
	\left(
		(6-\kappa) \dot \lambda_{\tilde t}^{\frac12} \beta_{\tilde t}
		+ \left( \tilde \delta_{-1} + \frac{\kappa}{2} \tilde \sigma_0 \right)
		- \dot \lambda_{\tilde t}^{\frac12} \left(\delta_{-1} + \frac{\kappa}{2}
		\sigma_0\right) \right) d \tilde t 
	+\dot\lambda_{\tilde t}^{\frac12}~ \dI u_{\lambda_{\tilde t}}
	\label{Formula: dI u = ...}
\end{equation}
Analogously to 
\cite{Schramm2006}
the SLEs are locally equivalent if and only if 
$u_t=\sqrt{ \kappa}\tilde B_t$ is a Brownian motion with 
$\kappa=\tilde \kappa$.

If $\kappa=6$ and 
$\tilde \delta_{-1} + 3\tilde \sigma_{0}=\delta_{-1} + 3\sigma_{0}=0$,
then the random law on 
$\{u_t\}_{t\in[0,T]}$
is also Brownian (multiplied by $\sqrt{\kappa}$).
\end{proof}

With the aid of this theorem, we can extend various properties of the
well-studied chordal SLE to the general case of $(\delta,\sigma)$-SLE. In
particular, we can prove the following.

\begin{corollary}
\label{Corollary: local SLE properties}
For $\delta$ and $\sigma$ as in
\eqref{Formula: delta_N = ..., sigma_N = ... without drift}
the random family of hulls 
$\{\mathcal{K}_t\}_{t\in[0,\infty)}$ possesses the following properties:
\begin{enumerate}[1.]
\item
$\{\mathcal{K}_t\}_{t\in[0,\infty)}$ is a curve-generated set a.s. 
Namely, there exists a parametrized 
random curve 
$\gamma:[0,\infty)\map\bar\Dc$ (not necessarily simple), 
started from the source point
$\lim\limits_{t\map +0}\gamma_t=a$, 
such that $G_t^{-1}(\Dc)\subset \Dc$ is a connected component of 
$\Dc\setminus \gamma_{[0,t]}$ for the forward case.
\item
\begin{itemize}
\item $0<\kappa\leq 4$: $\gamma$ is a simple curve in the interior of $\Dc$;
\item $4<\kappa< 8$: $\gamma$ touches itself and the boundary $\de\Dc$;
\item $8\leq \kappa$: $\gamma$ is a space-filling curve.
\end{itemize}
See Fig. 
\ref{Figure: kappa dependence picture} 
for a qualitative illustration
and Figs. 
\ref{Figure: SLE4_demo},
\ref{Figure: SLE7_demo}
for numerical simulations with $\kappa=4$ and $\kappa=7$, respectively.
\item The Hausdorff dimension of $\gamma$ is equal to $\min\{2,1+\kappa/8\}$.
\end{enumerate}
\label{Corollary: SLE analitic properties}
\end{corollary}

\begin{figure}[H]
\centering
  \begin{subfigure}[t]{0.33\textwidth}
		\centering
    \includegraphics[width=5cm,keepaspectratio=true]
    	{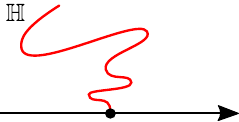}
    \caption{Simple curve, $\kappa\leq4$}
  \end{subfigure}%
  ~  
  \begin{subfigure}[t]{0.33\textwidth}
		\centering
    \includegraphics[width=5cm,keepaspectratio=true]
    	{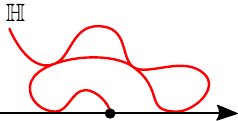} 
  	\caption{Self-touching curve, $4<\kappa<8$}
	\end{subfigure}%
	~
	\begin{subfigure}[t]{0.33\textwidth}
	\centering
		\includegraphics[width=5cm,keepaspectratio=true]
			{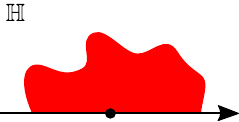}
		\caption{Space-filling curve, $8\leq \kappa$.}
	\end{subfigure}
 	\caption{Schematic illustration of the qualitative behavior of the SLE slit
 	for different values of $\kappa$ in the half-plane chart.
 	\label{Figure: kappa dependence picture}}
\end{figure}

\begin{proof}
Let 
$\{\tilde G_{\tilde t}\}_{\tilde t\in[0,+\infty)}$ 
be the chordal SLE,
see Section 
\ref{Section: Chordal Loewner equation},
and let 
$b\in\de\Dc$ 
be the common zero of 
$\tilde \delta$ 
and 
$\tilde \sigma$.
Let also $T_1$ be the first stopping time from the collection and let 
$\lambda$
be the time reparametrization from Theorem
\ref{Theorem: Absolute continuity of SLE}.
Then the law of 
$\{\K_t\}_{t\in[0,T_1]}$ 
stopped at $T_1>0$ 
is equivalent to the chordal law 
$\{\tilde \K_{\tilde t}\}_{\tilde t\in[0,\lambda_{T_1}]}$ 
stopped at 
$\lambda_{T_1}>0$. 
In particular, 
$\{\mathcal{K}_t\}_{t\in[0,T_1]}$ 
is curve-generated a.s.

In order to extend this property to $t\in[0,\infty)$, observe that the law
of $\{G_{t+T_1} \circ G_{T_1}^{-1}\}_{t\in[0,+\infty)}$
coincides with the law of 
$\{G_t\}_{t\in[0,+\infty)}$, 
which is the strong Markov
property. Thus, applying Theorem
\ref{Theorem: Absolute continuity of SLE} 
again to the pair consisting of the 
$(\delta,\sigma)$-SLE
$G_{t+T_1} \circ G_{T_1}^{-1}$ 
and to the chordal SLE with the
same zero point $b$ shows that 
$\mathcal{K}_t$ 
is curve-generated until 
$T_1+ T_1'$, 
where $T_1'$ has the same law as $T_1$.
Continuing by induction, we conclude that the hull is generated by a curve until a
countable sum of independent identically distributed positive stopping times. Such a
sum diverges a.s.

The remaining two items can be proved similarly.
\end{proof}

Thanks to Corollary \ref{Corollary: local SLE properties},  
the hull $\K_t$ is a simple curve 
a.s. if $\kappa\leq 4$.
Moreover, due to Theorem \ref{Theorem: Curve -> chain}, 
there is at most one sample of the Brownian motion for any curve. 
On the other hand, the sigma-algebra and the measure on curves can be defined by
methods not related to conformal maps. For example, we can consider a scaling
limit of a lattice path. In particular, the self-avoiding walk and the
loop-erased random walk are constructed this way.
Another approach is to define the curve as an interface in various
two-dimensional models from statistical physics such as the Ising model or
percolation.

In recent decades, the relation between measures constructed in these two
different ways has been established; see, for example,
\cite{Lawler2001b}
and
\cite{Schramm2000},
for classical SLEs. We restrict ourselves to remarking that 
Theorem \ref{Theorem: Absolute continuity of SLE} might be a key tool in this
direction.

We established above that the random law of the Brownian motion is preserved
under the transition from one pair of $\delta$ and $\sigma$ to another if
$\kappa=6$ and $\nu=0$ for both pairs. Consider now 
two essentially different $(\delta,\sigma)$-SLEs
such that the random law is kept unchanged under the transition from one to the other for
any $\kappa>0$. This question is motivated by Theorems
\ref{Theorem: ppS -> simple cases}
and
\ref{Theorem: ppS -> simple cases 2}.
The case when one of the SLEs is chordal is the subject of the
following theorem.


\begin{theorem}
Let 
$\{G_t\}_{t\in[0,+\infty)}$ 
be a forward $(\delta,\sigma)$-SLE, and let its
transform by
\eqref{Formula: G_t = M_t circ G_l_t}
be a forward chordal SLE
$\{{\tilde G}_{\tilde t}\}_{\tilde t\in[0,+\infty)}$ 
for some 
$\{\lambda_{\tilde t}\}_{\tilde t\in[0,+\infty)}$,
and
$\{\hat M_{\tilde t}\}_{\tilde t\in[0,\tilde T)}$.
Then, 
the $(\delta,\sigma)$-SLE 
$\{ G_{ t}\}_{t\in[0,+\infty)}$ 
is one of the following:
\begin{enumerate}[1.]
  \item Chordal;
  \item The case described in Section \ref{Section: Chordal case with fix time
  change};
  \item The case described in Section \ref{Section: Case with one fixed point}.
\end{enumerate}
\label{Theorem: cordal -> fix time change and fixed boundary point}
\end{theorem}

\begin{proof}
We use the proof of Theorem 
\ref{Theorem: Absolute continuity of SLE}.
The term in the parentheses of 
\eqref{Formula: dI u = ...}
is identically zero if and only if $\beta_{\tilde t}$ is a constant because 
$\dot \lambda_{\tilde t}\neq 0$. Due to 
\eqref{Formula: alpha_0=1, beta_0=0}
\begin{equation}	
	\beta_{\tilde t} \equiv 0.
\end{equation}
Hence, 
$\hat M^{\HH}(z)=\alpha_{\tilde t}\,z = e^{-c_{\tilde t}} z$ 
for some real-valued function $c_{\tilde t}$, 
and
\begin{equation}
 \tilde G_{\tilde t}^{\HH}(z) = e^{-c_{\tilde t}} G_{\lambda_{\tilde t}}^{\HH}(z).
 \label{Formula: tilde G = e^c G}
\end{equation}
The Stratonovich differential of $\tilde G^{\HH}_{\tilde t}(z)$ is 
\begin{equation}\begin{split}
 \dS \tilde G^{\HH}_{\tilde t}(z) =& 
 (\dS e^{-c_{\tilde t}}) G^{\HH}_{\lambda_{\tilde t} }(z) +
 e^{-c_{\tilde t}} \dS G^{\HH}_{\lambda_{\tilde t} }(z)  
 =\\=&
 (\dS e^{-c_{\tilde t}}) G^{\HH}_{\lambda_{\tilde t} }(z) +
 e^{-c_{\tilde t}} \dot \lambda_{\tilde t}
 \left( 
  \frac{2}{G^{\HH}_{\lambda_{\tilde t} }(z)}d\tilde t - \sqrt{\kappa} \dS B_{\lambda_{\tilde t}}
 \right)
\end{split}\end{equation}
Due to 
\eqref{Formula: a = 1/sqrt lambda},
we have 
\begin{equation}
 e^{-c_{\tilde t}} = \dot \lambda^{-\frac12}_{\tilde t},
\end{equation}
and consequently,
\begin{equation}\begin{split}
 \dS e^{-c_{\tilde t}} =
 - \frac12 e^{-3 c_{\tilde t}} x_{\tilde t} d \tilde t
 - \frac12 e^{-3 c_{\tilde t}} y_{\tilde t} \dS \tilde B_{\tilde t},
\end{split}\end{equation}
where we used 
(\ref{Formula: d lambda = a dt + b dB}).
Thus, we conclude that
\begin{equation}\begin{split}
 \dS \tilde G^{\HH}_{\tilde t}(z) =&
 \left( 
  - \frac12 e^{-3 c_{\tilde t}} x_{\tilde t} d \tilde t
  - \frac12 e^{-3 c_{\tilde t}} y_{\tilde t} \dS \tilde B_{\tilde t}
 \right) 
 e^{ c_{\tilde t}} \tilde G^{\HH}_{\tilde t}(z) 
 +\\+&
 \frac{2}{\tilde G^{\HH}_{\tilde t }(z)}d\tilde t 
 - \sqrt{\kappa} \dS \tilde B_{\tilde t} 
 + \frac14 \sqrt{\kappa} e^{-2c_{\tilde t}} y_{\tilde t} d \tilde t.
 \label{Formula: 3}   
\end{split}\end{equation}

In order to have time-independent coefficients, we assume that 
$x_{\tilde t}$
and 
$y_{\tilde t}$ 
are proportional to $e^{2 c_{\tilde t}}$. Hence, define
$\xi\in \mathbb{R}$ by
\begin{equation}
 x_{\tilde t} = - 4 \xi e^{2c_{\tilde t}}.
\end{equation}
Without loss of generality, we can assume that $y_{\tilde t}$ has one of three
possible forms:
\begin{enumerate}[1.]
\item $y_{\tilde t}=0$,
\item $y_{\tilde t}=4\sqrt{\kappa} e^{2c_{\tilde t}} $,
\item $y_{\tilde t}=-4\sqrt{\kappa} e^{2c_{\tilde t}} $,
\end{enumerate}
because all other choices can be reduced to these three with the transform
$\mathscr{S}$.
The first case is considered in Section 
\ref{Section: Chordal case with fix time change}.
The other two cases are discussed in Section
\ref{Section: Case with one fixed point}.
\end{proof}

\subsection{Domain Markov property and conformal invariance of random laws on
planar curves}
\label{Section: Domain Markov property and conformal invariance of random laws on
planar curves}

In this section, we consider the property of random laws on planar curves that
are defined as above, not necessarily as the complement of images of
conformal maps.

Let $\mathscr{G}_{D,a}$ be the set of all simple open curves $\gamma$ in a
domain $D\subset\C$. All curves start from
some point $a\in\de D$ on a smooth piece of the boundary. 
Let $(\mathscr{G}_{D,a},\mathcal{F}_{D,a},P_{D,a})$ 
be a probability space.

Let 
$\tau:\tilde D \map D$ be a conformal map that can be continuously extended to
the boundary near the point $a$ and let
$\tilde a:=\tau^{-1}(a)\in\de \tilde D$. 
Let 
$({\mathscr{G}}_{\tilde D,\tilde a},
\tilde{\mathcal{F}}_{\tilde D,\tilde a},\tilde P_{\tilde D, \tilde a})$
be the probability space induced by the map $\tau$. We will use the notation  
\begin{equation}
 \left(\tau_* P\right)(\tilde B) 
 := P ( \tau(\tilde B) ),\quad 
 \tilde B\in \mathcal{\tilde F}.
 \label{Formula: push forward of measure}
\end{equation}
For example, $\tau_* P_{D,a} = \tilde P_{\tilde D,\tilde a}$. 


On the other hand, a measure on $\mathscr{G}_{\tilde D,\tilde a}$
can be defined independently, regardless of the
conformal mapping. For example, self-avoiding and loop-erased random
walks are usually defined on a lattice grid embedded in an arbitrary domain. Then
we can consider a limiting measure as the mesh size tends to zero.
If such laws are related by
\eqref{Formula: push forward of measure},
then the limiting measure is called
\emph{conformally invariant}. 
\index{conformal invariance of a random law on curves}
It arises in the scaling limit, and only some special
random laws on lattice paths, such as those mentioned above, possess it. We emphasize that, if the
measure $P_{D,a}$ is constructed with the aid of $(\delta,\sigma)$-SLE as in the
present text, the conformal invariance follows directly from the very
definition formulated in terms of $\Dc$.

In Section 
\ref{Section: Relations between essentially different SLEs},
we define an $\mathcal{F}_t$-random law on hulls $\K_t$ for 
$t\in[0,+\infty)$
and for any pair of $\delta$ and $\sigma$. In particular, for $\kappa\leq 4$
and for any given domain $D$,
this gives a random law on planar curves as above 
with $\mathcal{F}_{D,a}$ induced by $\mathcal{F}_t^B$ and 
$P_{D,a}$ given by $P^B$. 

\begin{remark}
It is an interesting problem to study the relation between the sigma-algebra 
$\mathcal{F}_{D,a}$ as above and the sigma-algebra induced by the Hausdorff
distance on $\mathscr{G}_{D,a}$.
\end{remark}

Let us now discuss another (and independent) property of a family of measures
on planar curves called the domain Markov property. It is related to the strong
Markov property of $G_t$ considered above, but it is formulated for a general, not
necessarily conformally invariant, random law on curves.

Let 
$\mathscr{G}_{D,a}[\gamma~|~\tilde \gamma]$ 
be a conditional law on curves 
$\gamma\in\mathscr{G}_{D,a}$.
The condition is such that 
$\gamma_{\sim}\subset\gamma$ 
is a fixed subcurve also starting from the point
$a$ and ending at 
$a_{\sim}\in\gamma$. 
Assume that a random law on 
$\mathscr{G}_{D\setminus\gamma_{\sim},a_{\sim}}$ 
is given for some curve 
$\gamma_{\sim}$.
The 
\emph{domain Markov property} 
\index{domain Markov property}
is usually defined as
\begin{equation}
	\mathrm{Law}_{D,a}[\gamma~|~\gamma_{\sim}]=
	\mathrm{Law}_{D\setminus\gamma_{\sim},a_{\sim}}[\gamma\setminus\gamma_{\sim}].
	\label{Formula: Domain Markov property (heuristic)}	
\end{equation}
%
%
%
%
%
%
%
%
The conditional law on the left-hand side can be understood heuristically or
defined as below. 

Consider first the case of discrete measures, for example, self-avoiding or
loop-erased random walks defined on a finite lattice grid.
The conditional law of their continuous versions (scaling limits) can be
understood as the limit law. For more details, we refer, for example, to
\cite{Schramm2000}. 
However, if the law is defined on a continuous set of curves (and the
conditional expectation is taken with respect to an event of probability zero)
from the very beginning, as we did using $(\delta,\sigma)$-SLE,
then we need the following constructions for a mathematically rigorous
definition.

Consider a quotient space 
$(\mathscr{G}_{D,a})_{\sim}$ 
of 
$\mathscr{G}_{D,a}$ 
with the following property. For each equivalence class there exists a maximal
curve 
$\gamma_{\sim}\in\mathscr{G}_{D,a}$ 
which is strictly contained in each curve from this class and all curves
with this subcurve are in this class.
Thus, each element of 
$(\mathscr{G}_{D,a})_{\sim}$
is represented by $\gamma_{\sim}\in\mathscr{G}_{D,a}$. We use the letter 
$\gamma_{\sim}$ 
for both the equivalence class and the common part of curves it represents. We
denote by 
$a_{\sim}$ 
the end of the curve 
$\gamma_{\sim}$ 
which differs from 
$a$. 
Let 
$\mathcal{F}_{\sim}$ 
be the sigma-subalgebra of 
$\mathcal{F}$ 
generated by this factor space. Namely, each subset from 
$\mathcal{F}_{\sim}$
is a preimage of the quotient map of some set of 
$\mathcal{F}$.
Remark that 
$\gamma_{\sim}$ 
is an 
$\mathcal{F}_{\sim}$-random 
variable. 


Let 
$\mathcal{F}_{D\setminus \gamma_{\sim},a_{\sim}}$ 
be the sigma-algebra on 
$\mathscr{G}_{D\setminus \gamma_{\sim},a_{\sim}}$
consisting of all sets from $\mathcal{F}_{D,a}$ that are subsets of the
equivalence class given by $\gamma_{\sim}$.
Assume now that for a.a. $\gamma_{\sim}$  a 
$\mathcal{F}_{D\setminus \gamma_{\sim},a_{\sim}}$- 
random law on 
$\mathscr{G}_{D\setminus \gamma_{\sim},a_{\sim}}$
is given.
Hence, for any bounded $\mathcal{F}_{D,a}$-measurable function 
$f:\mathscr{G}_{D,a}\map\mathbb{R}$,
the function 
$f(\cdot \cup \gamma_{\sim}):\mathscr{G}_{D\setminus
\gamma_{\sim},a_{\sim}}\map\mathbb{R}$ 
is bounded and
$\mathcal{F}_{D\setminus \gamma_{\sim},a_{\sim}}$-measurable. We are now ready
to give meaning to 
\eqref{Formula: Domain Markov property (heuristic)}:
\begin{equation}
	\Evv{D,a}{f(\gamma)|\mathcal{F}_{\sim}}(\omega) =
	\Evv{D\setminus \gamma_{\sim}(\omega),a_{\sim}(\omega)}{f(\gamma \cup
	\gamma_{\sim}(\omega))}\quad a.s. 
	\label{Formula: Domain Markov property (strict)}
\end{equation}
for any bounded $\mathcal{F}_{D,a}$-measurable function $f$.

Thus, to formulate the domain Markov property, we have to specify 
not only random laws for each $\gamma_{\sim}$ on 
$D\setminus \gamma_{\sim}$, 
but also the collection of curves 
$\gamma_{\sim}$. 
For example, in the case of curves on a finite lattice grid, we can consider the
set of all curves of fixed length. In the case of $(\delta,\sigma)$-SLE-generated
curves, we can pick any stopping time $0<T<+\infty$ and consider the
set of curves $\gamma_{\sim}$ generated by the random map $G_T$. 
We use Theorem
\ref{Theorem: Curve -> chain}
and denote by $G[\gamma_{\sim}]$ the unique map which corresponds to the curve
$\gamma_{\sim}$.
We can now formulate what we mean by the
\emph{domain Markov property of $(\delta,\sigma)$-SLE}, see also Fig. 
\ref{Figure: domain Markov Property illustration}.
\index{domain Markov property of $(\delta,\sigma)$-SLE}

\begin{proposition}
Let $\{G_t\}_{t\in[0,+\infty)}$ be a 
$(\delta,\sigma)$-SLE with $\kappa\leq 4$.
Choose some chart $\psi\colon \Dc\map D^{\psi}$.
For a positive stopping time $T<+\infty$ a.s.,
let $\mathcal{F}_{\sim}:=\mathcal{F}_T^B$,
and let us define 
\begin{equation}
	P_{D\setminus\gamma_{\sim},a_{\sim}}(\tilde B) 
	:= \left( {G[\gamma_{\sim}]^{\psi}}_* P_{D,a} \right) (\tilde B),\quad
	\tilde B \in \mathcal{F}_{\sim}.
\end{equation}
The domain Markov property is satisfied with respect to $\mathcal{F}_{\sim}$.
\end{proposition}

\begin{proof}
This follows from the strong Markov property of $(\delta,\sigma)$-SLE and from
the construction above.
\end{proof}

\begin{center}
\begin{figure}[H]
\centering
	\includegraphics[keepaspectratio=true]
    	{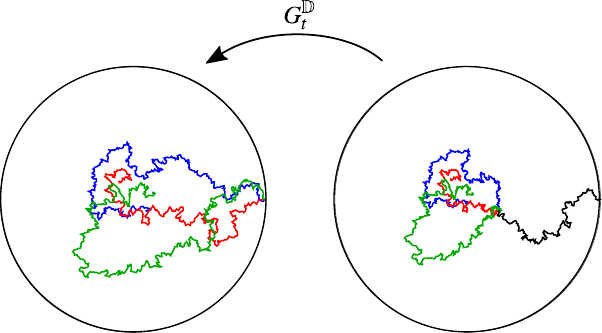}
  \caption{This is an illustration of the domain Markov property of
  ($\delta,\sigma$)-SLE.
  The three colored lines on the left-hand side are samples of the
  ($\delta,\sigma$)-SLE slit $\gamma$. 
  The corresponding colored lines on the right-hand side are the preimages with
  respect to an independently sampled forward ($\delta,\sigma$)-SLE map $G_t$,
  which has the slit denoted by the black line $\gamma_{\sim}$. The domain
  Markov property states that the union of the black line and any of the colored lines has the
  same law as the ($\delta,\sigma$)-SLE slit.
  For this figure, we simulated the slit for
  $\delta=2\ell_{-2}$ and $\sigma=-\sqrt{2}(\ell_{-1}+\ell_{1})$ in the unit
  disk chart.
  \label{Figure: domain Markov Property illustration}}
\end{figure}
\end{center}

The domain Markov property can also be considered for continuously growing hulls
that are not curves. This can be done for the classical cases, but for other
cases we would need to prove Conjecture 
\ref{Conjecture: uniquence of G and T given K}.

\section*{Conclusions and perspectives}

\begin{enumerate}

\item 
It is important to prove Conjecture
\ref{Conjecture: uniquence of G and T given K}
and to determine which normalization conditions, analogous to those in the classical cases, are satisfied in the general case of
($\delta,\sigma$)-SLE.

\item
We considered only holomorphic vector fields and conformal maps. This
restriction can be relaxed as follows. Let $\delta$ and $\sigma$ be tangent to
the boundary (except at the point $a$ for $\delta$) and holomorphic only on some
neighborhood of the boundary of $\Dc\setminus\K_t$. On the rest of $\Dc$
the vector fields can be smooth and such that the initial value problem
\eqref{Formula: d g = h delta g dt}
has a unique solution $\{G_t\}_{t\in[0,+\infty)}$ that is a family of
endomorphisms.
Such a map $G_t$ is holomorphic only on some neighborhood of the boundary and
smooth inside. On the other hand, the local behavior of the hull is the same
as for holomorphic $\delta$ and $\sigma$.
This provides a more general random law on curves with the SLE local
behavior.

\item The generalization of the SLE to multiply connected Riemann surfaces
(surfaces of higher genus with boundaries) is not straightforward because the number of complete
and semicomplete vector fields tangent to the boundaries is severely restricted
by the Riemann-Roch theorem. The approach from the previous item is a possible
way to avoid this difficulty.

\end{enumerate}

%% file: NumericalSimulation.tex
\chapter{Numerical Simulation}
\label{Chapter: Numerical Simulation}

\section*{Introduction}

This chapter is dedicated to methods for the numerical solution of forward
$(\delta,\sigma)$-L\"owner chains and simulation of the forward
$(\delta,\sigma)$-SLE.
A general method for the approximation of conformal maps was introduced in
\cite{Marshall2007}. Its application to the chordal L\"owner equation is
considered in
\cite{Kennedy2009}
together with some technical details.

Numerical simulation of the $(\delta,\sigma)$-SLE was
developed and implemented independently by the author using Wolfram Mathematica
software. All figures with $(\delta,\sigma)$-SLE samples in this monograph,
except Fig. \ref{Figure: bad kennedy figure}, were made with this program. Here
we describe the method used, together with some technical
difficulties that motivate the structure of the
program.
The correctness (convergence) in the deterministic case is proved in
\cite{Marshall2007}.
The analogous statement for the stochastic case is more difficult, see 
\cite{HuyV.Tran2014}
for details. We do not consider this
problem in the present text.

We organize this chapter as follows. In the first section, we study how to
approximate $(\delta,\sigma)$-L\"owner chains with a given driving
function. If the driving function is a sample of the Brownian motion, there
are additional technical difficulties that are considered in the second
section. The third part of the chapter is a conclusions and perspectives
section. We discuss how to extend these methods to non-continuous stochastic
processes such as the stable L\'evy process.


\section{Approximation of $(\delta,\sigma)$-L\"owner chains}
\label{Section: Simulation of Loewner chain}

\subsection{The zipper method}
\label{Section: The zipper method}

Let 
$G_t$
be a conformal map from a forward L\"owner chain 
$\{G_s\}_{s\in[0,\,t]}$ 
with the
driving function 
$\{u_s\}_{s \in[0,\,t]}$, $t\in[0,+\infty)$. 
Our purpose is to obtain a
reasonable approximation map $\bar G_t$
($\bar G_t \approx G_t$). 
Namely, for a given chart $\psi$ and a large integer
$N$, 
we define a conformal map
$\bar G_t^N$ 
and claim 
\begin{equation}
	\left(\bar G_t^N\right)^{\psi}(z) \xrightarrow{N\map\infty} G_t^{\psi}(z)\quad
	\text{uniformly on }\psi(\Dc).
	\label{Formula: G^N -> G}
\end{equation}

To this end, consider a finite partition of the interval
\index{partition of time}
$[0,\,t]$
given by a collection
$\{t_0,\,t_1,\,t_2,\dotso t_N\}$ 
such that
\begin{equation}
	0=t_0<t_1<t_2<\dotso<t_{N-1}<t_N=t<\infty.
	\label{Formula: 0=t_0<t_1<...<t_N}
\end{equation}
For example, one can use the uniform partition
\begin{equation}
	t_n:=\frac{n}{N} t,\quad n=1,2,\dotso N. 
	\label{Formula: uniform partition of time}
\end{equation}
We use the results of 
Section 
\ref{Section: Definition and basic properties}
to split $G_t$ into the composition
\begin{equation}\begin{split}
	&G_t = 
	G_{t_N,\,t_{N-1}} \circ G_{t_{N-1},\,t_{N-2}} \circ \dotso \circ
	G_{t_3,\,t_2} \circ G_{t_2,\,t_1} \circ G_{t_1,\,t_0}.
\end{split}\end{equation}
Each function 
$G_{t_{n},\,t_{n-1}}$ ($n=1,2,\dotso,N$)
is a L\"owner map 
$\tilde G_{t_n-t_{n-1}}$ 
obtained with the driving function 
$\{\tilde u_s := u_{s+t_{n-1}}-u_{t_{n-1}} \}_{s \in[0,\,t_n-t_{n-1}]}$. 
Let 
$\bar G_n$ ($n=1,2,\dotso,N$) 
be some approximation of 
$G_{t_n,\,t_{n-1}}$.
We call $\bar G_n$ the \emph{step map}.
\index{step map}
We choose the partition
\eqref{Formula: 0=t_0<t_1<...<t_N}
such that 
\begin{equation}
	\max\limits_{n=1,2,\dotso,N}
	\left(t_n-t_{n-1}\right) \xrightarrow{N\map \infty} 0,
	\label{Formula: max (t_n-t_n-1) -> 0}
\end{equation}
and define the approximation $\bar G_t^N$ by
\begin{equation}
 \bar G_t^N := \bar G_N \circ \bar G_{N-1} \circ \dotso \circ \bar G_2 \circ
 \bar G_1,\quad N=1,2,\dotso.
 \label{Formula: G^N = G_n circ G_n-1 ... G_1}
\end{equation}
 
Our main purpose is to approximate the curve
$\gamma_t$ 
that generates the hull 
\begin{equation}
	\K_t=\Dc \setminus G_t^{-1}(\Dc).
\end{equation}
We assume first that 
$\K_t=\gamma_t$ is a simple curve.
Then the curve
\begin{equation}
	{{\bar \gamma}}^{\,N}_t := \Dc \setminus \left(\bar G^N_t\right)^{-1}(\Dc)
\end{equation}
is such that
\begin{equation}
	\psi({{\bar \gamma}}^{\,N}_t) \xrightarrow{N\map\infty} \psi(\gamma_t)
\end{equation}
due to
\eqref{Formula: G^N -> G}. 

We note that 
${{\bar \gamma}}^{\,N}_t$ 
consists of arcs such that the first one is 
$\Dc\setminus \bar G_1^{-1}(\Dc)$,
the second one
is the image of 
$\Dc\setminus \bar G_2^{-1}(\Dc)$ 
with respect to $\bar G_1^{-1}$,
the third one is the image of 
$\Dc\setminus \bar G_3^{-1}(\Dc)$ 
with respect to 
$\bar G_1^{-1}\circ \bar G_2^{-1}$,
and so on.

Consider the 
\emph{joint points}
\index{joint points} 
$\bar \gamma_{n,a}^N\in {\bar \gamma}^{\,N}_t$ 
defined by 
\begin{equation}
 \bar \gamma_{n,a}^N:=
 \bar G_1^{-1} \circ \bar G_2^{-1} \circ \dotso \circ \bar G_n^{-1}(a),
 \quad n=1,2,\dotso N,
 \label{Formula: gamma^N_n = G_1 G_2 ... G_n}
\end{equation}
where
$a\in\de \Dc$ 
is the source point. In the limit 
$N\map\infty$, 
the distance between them tends to zero: 
\begin{equation}
	\max\limits_{n=1,2,\dotso,N} 
	\left| \psi(\bar \gamma_{n,a}^N) - \psi(\bar \gamma_{n-1,a}^N) \right|
	\xrightarrow{N\map \infty} 0,
\end{equation}
and the arcs tend to
straight line segments due to the continuity of the driving function 
$\{u_s\}_{s\in[0,\,t]}$.
Thus, to approximate the curve $\gamma_t$ in the chart $\psi$, we can simply find
all joint points $\bar\gamma_{n,a}^N$ 
with large enough $N$ and a proper partition 
\eqref{Formula: 0=t_0<t_1<...<t_N}.
The joint points 
$\{\bar \gamma_{n,a}^{N}\}_{n=1,2,\dotso N}$
can then be connected by straight
line segments in the chart $\psi$.

To find each point 
$\bar \gamma_{n,a}^N$, 
we need to make $n$ iterations (calculate
the map 
$\bar G_i$ 
for different values of $i$, $n$ times). Hence, the computation time for 
$\bar \gamma^{\,N}_t$ 
grows as $O(N^2)$. A faster method of
approximation is considered in 
\cite{Kennedy2007}. 
We do not implement it here.

Below, we discuss how to choose the step map and
the partition 
\eqref{Formula: 0=t_0<t_1<...<t_N}.

\subsection{Choice of the step map $\bar G_n$}
\label{Formula: Choice of bar G_n}

One possible choice of the approximation function
$\bar G_n$ 
for the chordal case is made, e.g., in 
\cite{Kennedy2009}. 
The map 
$\bar G_n$ 
is chosen such that the slit 
$\bar \gamma:=\Dc\setminus \bar G_n^{-1}(\Dc)$ 
is a straight line segment in the half-plane chart from $a$ to some point inside
the domain. See 
\cite{Kennedy2009} 
for exact formulas.

An alternative is given by 
the solution of the equation
\begin{equation}
	\dot {G}_s 
	= \delta \circ {G}_s 
	+ \frac{u_{t_n}-u_{t_{n-1}}} {t_n-t_{n-1}} \sigma \circ {G}_s
	,\quad {G}_0=\id,\quad s \in[0,\,t_n-t_{n-1}], 
	\label{Formula: dot G_t = delta + (u-u)/(t-t)sigma}
\end{equation}
which is 
$H_s\left[\delta+\frac{u_{t_n}-u_{t_{n-1}}} {t_n-t_{n-1}} \sigma\right]$.
With such a step map the full approximation 
$\bar G_t^N$ 
is a L\"owner chain with the driving function 
$\{\bar u_s^N\}_{s\in[0,\,t]}$ 
obtained as a piecewise linear continuous approximation of 
$\{u_s\}_{s\in[0,\,t]}$
such that 
$\bar u^N_{t_n} = u_{t_n},~n=0,1,2,\dotso N$ at each joint point of 
$\bar u_t^N$.

The third alternative considered here corresponds to the approximation of 
$\{u_s\}_{s\in[0,\,t]}$
by a piecewise constant function 
$\{\bar u_s^N\}_{s\in[0,\,t]}$ 
such that
\begin{equation}
	\bar u^N_{s} = u_{t_n},\quad s \in [t_{n},\,t_{n-1}),\quad n=1,2,\dotso N.
	\label{Formula: bar u = u_n}
\end{equation}
Thus, 
$\{\bar u_s^N\}_{s\in[0,\,t]}$
is not continuous,
and the step function is a composition 
\begin{equation}
	\bar G_n = H_{u_{t_n}-u_{t_{n-1}}}[\sigma] \circ H_{t_n-t_{n-1}}[\delta],\quad
	n=1,2,\dotso N.
	\label{Formula: tilde G_n = H[sigma] H[delta]}
\end{equation}

The term
$H_{t_n-t_{n-1}}[\delta]$ 
corresponds to the interval  
$[t_{n-1},\,t_{n})$, 
where 
$\bar u_t^N$ 
is constant. The curve 
$\bar \gamma :=\Dc\setminus H_{s}[\delta]^{-1}(\Dc)$ 
is a flow line of 
$\delta$ 
that starts at $a$ and tends to the attracting (or degenerate) zero of
$\delta$ when $s\map +\infty$.
Meanwhile, 
$H_{t_n-t_{n-1}}[\delta]$
maps the tip of 
$\gamma$ to $a$ 
and both sides of the slit $\bar \gamma$ to the boundary $\de\Dc$ of $\Dc$.

The second term 
$H_{u_{t_n}-u_{t_{n-1}}}[\sigma]$
in 
\eqref{Formula: dot G_t = delta + (u-u)/(t-t)sigma}
is a M\"obius automorphism that moves the point $a$ along the boundary of
$\Dc$. 
The length of the shift is defined by the jump of  
$\bar u_s^N$
at 
$s=t_n$, 
which is
$u_{t_n}-u_{t_{n-1}}$.
Hence, the inverse step map 
$\bar G_n^{-1}$
produces a slit $\bar \gamma$ which starts at the source $a$, but  
maps $a$ to the boundary of 
$\Dc\setminus\bar\gamma$,
which is either a side of $\bar \gamma$ or, if the jump is
big enough, the boundary $\de\Dc$. The next piece of $\bar \gamma^{\,N}$
starts from this point.


Thus, in the chain 
\eqref{Formula: gamma^N_n = G_1 G_2 ... G_n},
each subsequent part of the slit $\bar {\gamma_t}^{N}$ grows not from the tip of the
previous one, but from a shifted point on a side of the slit. Thus, 
${\bar {\gamma}}_t^{\,N}$ 
is not a curve, but a tree with the property that in the limit 
\eqref{Formula: max (t_n-t_n-1) -> 0}
the source of each arc of 
${\bar {\gamma}}^{\,N}$ 
tends to the tip of the previous arc, see Fig.
\ref{Figure: different approxiamtions}. 
This is a consequence of the continuity
of the driving function $\{u_s\}_{s\in[0,\,t]}$,
see Fig.
\ref{Figure: different approxiamtions}.
\begin{center}
\begin{figure}[h]
\centering
	\includegraphics[keepaspectratio=true]
    	{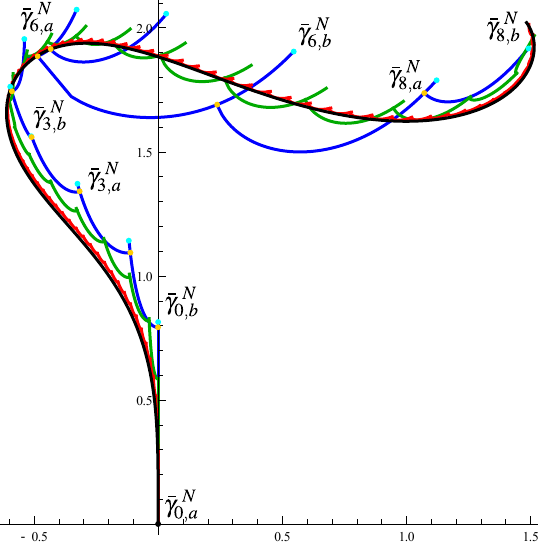}
  \caption{
  We show the slit (the blue line) of the approximation map $\bar G_t^N$.
  The points $\bar\gamma_{i,a}^N$ are shown in yellow and the points 
  $\bar\gamma_{i,b}^N$ in blue.
  We also demonstrate the convergence of $\bar \gamma_t^N$ to $\gamma_t$
  when $N\map+\infty$.
  This figure shows how the slit 
  $\bar\gamma_t^N$ 
  of $\bar G_t^N$ looks
  for different values of $N$. For illustration, we consider 
  the chordal case, half-plane chart, the driving function
  $u_s:=-4s^2 (s-1)(s-2)$, $t=1.8$, and use the partition 
  \eqref{Formula: uniform partition of time}.  
  The black line is the exact slit
  $\gamma_t$, the blue, green, and red lines are the slits $\bar \gamma_t^N$
  for $N=10$, $N=20$, and $N=80$, respectively.
  \label{Figure: different approxiamtions}}
\end{figure}
\end{center}

The choice 
\eqref{Formula: tilde G_n = H[sigma] H[delta]}
has the disadvantage that the map 
$H_{s}[\delta]$ 
cannot be expressed in terms of elementary functions. This increases the time
of numerical calculations. To handle this, we combine the second
\eqref{Formula: dot G_t = delta + (u-u)/(t-t)sigma}
and the third 
\eqref{Formula: tilde G_n = H[sigma] H[delta]}
approaches as follows. 
Assume the normalization 
\eqref{Formula: delta_-2 = pm2}
and consider the semicomplete vector field 
\begin{equation}
	\delta + \frac{u_{t_n}-u_{t_{n-1}}} {t_n-t_{n-1}} \sigma 
  \label{Formula: delta + (u-u)/(t-t)sigma}
\end{equation}
from 
\eqref{Formula: dot G_t = delta + (u-u)/(t-t)sigma}.
Now we apply the classification from Section
\ref{Section: SLE preliminaries}
and 
Fig.~\ref{Figure: Semicomplete vector fields}
to define a vector field 
$\bar\delta_n$.
\begin{enumerate}
  \item
  If 
  \eqref{Formula: delta + (u-u)/(t-t)sigma}
 	is parabolic, then 
  $\bar\delta_n$
  is chordal 
  (see Section 
 	\ref{Section: Chordal Loewner equation})
  with a triple zero at the same point where 
  \eqref{Formula: delta + (u-u)/(t-t)sigma} has the double zero.
  \item 
  If 
  \eqref{Formula: delta + (u-u)/(t-t)sigma}
  is hyperbolic, let $b_1,b_2\in\de\Dc$ be the 
  non-attracting zeros of 
  \eqref{Formula: delta + (u-u)/(t-t)sigma}. 
  Assume $\bar\delta_n$ is dipolar 
  (see Section \ref{Section: Dipolar Loewner equation})
  with non-attracting zeros at $b_1$ and $b_2$.
  \item
  If 
  \eqref{Formula: delta + (u-u)/(t-t)sigma}
  is elliptic, then 
  $\bar\delta_n$ is
  radial 
	(see Section \ref{Section: Radial Loewner equation})
  with its zero at the position of the attracting zero of 
  \eqref{Formula: delta + (u-u)/(t-t)sigma}.
\end{enumerate}
We also define a complete vector field $\bar\sigma_n$ by
\begin{equation}
	\bar\sigma_n := \delta + \frac{u_{t_n}-u_{t_{n-1}}} {t_n-t_{n-1}} \sigma -
	\bar\delta_n
\end{equation}
and the step function by
\begin{equation}
	\bar G_n = H_{t_n-t_{n-1}}[\bar\sigma_n] \circ
	H_{t_n-t_{n-1}}[\bar\delta_n],\quad n=1,2,\dotso N.
	\label{Formula: bar G_n = H[hat sigma] H[hat delta]}
\end{equation}
Thus, the slit 
$\bar \gamma$ 
obtained with $\bar G_n$ tends to the attracting
zero of 
\ref{Formula: delta + (u-u)/(t-t)sigma}
in the chordal and radial cases.
The map $\bar G_n$ is given by the solution of the driftless chordal, dipolar, or
radial equation with a piecewise constant driving function.
In fact, this method is a
version of
\eqref{Formula: tilde G_n = H[sigma] H[delta]}
after some drift transform $\mathscr{D}_c$ of the triple
\begin{equation}
	\left(\delta,\sigma,
	\left\{\frac{u_{t_n}-u_{t_{n-1}}} {t_n-t_{n-1}}s\right\}
	_{s=[0,\,t_n-t_{n-1})}\right)
\end{equation}
of 
\eqref{Formula: dot G_t = delta + (u-u)/(t-t)sigma}.

\subsection{Choice of the partition}
\label{Section: Choice of the partition}

The simplest choice of partition 
\eqref{Formula: 0=t_0<t_1<...<t_N}
is the uniform one
\eqref{Formula: uniform partition of time}.
A critical disadvantage of this approach is discussed in
\cite{Kennedy2009}.
The difficulty is that the distance between the joint points
$\bar\gamma_n^{\,N}$ is far from homogeneous in this case. 
Define the distance 
\index{$d(\cdot,\cdot)$}
\begin{equation}
	d(z_2,z_1):=|\psi(z_2)-\psi(z_1)|,\quad z_1,z_2\in\Dc
\end{equation}
for a given chart $\psi$.
For a fixed driving function $\{u_s\}_{s\in[0,s]}$ and a fixed integer $N$, the
distance 
$d(\bar \gamma_n^N,\bar \gamma_{n-1}^N)$ 
may vary from numerically very large to very small,
see Fig. 
\ref{Figure: bad kennedy figure}.
Thus, the curve $\gamma$ is approximated with too high accuracy in some regions
and with too low accuracy in other regions. As an illustration, we
present a simulation of a chordal L\"owner chain with a driving
function obtained as a sample of Brownian motion (a sample of chordal SLE); see
Fig.
\ref{Figure: bad kennedy figure}.
\begin{figure}
\centering
\includegraphics[width=10cm,keepaspectratio]{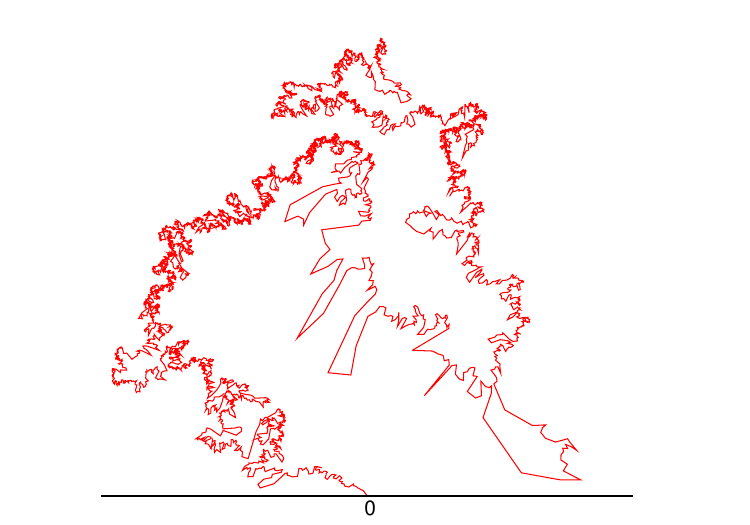}
\caption{An approximation of the slit of a chordal L\"owner chain with the
driving function given by a sample of Brownian motion (SLE slit) in the
half-plane chart.
The partition of the time interval is uniform. 
The parameters are chosen to be $\kappa=6$ and $N=10000$. (\cite{Kennedy2009})}
\label{Figure: bad kennedy figure} 
\end{figure}
The thin straight lines correspond to regions of time where
$d(\bar\gamma_{n,a}^N,\bar\gamma_{n-1,a}^N)$ 
is much bigger than the resolution of the picture. On the other hand, in other
regions, where the curve is bold, the distance 
$d(\bar\gamma_{n,a}^N,\bar\gamma_{n-1,a}^N)$ 
is excessively small.

To avoid this situation, it is enough to choose each subsequent $t_n$ not
uniformly as above, but such that
$d_{\text{min}}<d(\bar\gamma_n^N,\gamma_{n-1,a}^N)<d_{max}$ 
for some parameters 
$d_{\text{min}},d_{\text{max}}>0$ 
that correspond to the desired degree of resolution. Thus, the choice of
partition strongly depends on the driving function $u_t$. We remark that such a
partition also depends on the choice of the chart $\psi$, as does the Euclidean
distance $d(\cdot,\cdot)$.

If the driving function 
$\{u_s\}_{s\in[0,\,t]}$ 
is given analytically for all
values of 
$s\in[0,\,t]$,
then the implementation of this method is straightforward. Once we have found
$\{\bar \gamma^{\,N}_i\}_{i=1,2,\dotso n}$ 
for some $n<N$, we can simply take some
$t_{n+1}>t_n$, 
calculate 
$\bar\gamma_{n+1,a}^N$, 
and increase or decrease the value of 
$t_{n+1}$
depending on the calculated value of 
$d(\bar\gamma_{n+1,a}^N,\bar\gamma_{n,a}^N)$.
Only after obtaining an optimal 
$\bar\gamma_{n+1,a}^N$ 
do we add this 
$\bar\gamma_{n+1,a}^N$
to the collection
$\{\bar {\gamma}^{\,N}_i\}_{i=1,2,\dotso n}$. One can also use the distance
\begin{equation}
	d(\bar\gamma_{n+1,a}^N,\bar\gamma_{n,b}^N)+
	d(\bar\gamma_{n,b}^N,\bar\gamma_{n,a}^N)
\end{equation}
instead of 
$d(\bar\gamma_{n+1,a}^N,\bar\gamma_{n,a}^N)$.

\section{Simulation of $(\delta,\sigma)$-SLE}
\label{Section: Simulation of SLE}

Above, we considered a method of numerical simulation for
$(\delta,\sigma)$-L\"owner chains. The problem becomes more complicated
when the driving function is a sample of a random process such as the Brownian
motion $\{B_t\}_{t\in[0,+\infty)}$. This is due to the fact that we have to
sample the process during the simulation. We prefer to avoid sampling $B_t$
beforehand, because we do not know the partition in advance (the required number
of floating-point operations and amount of memory are enormous if we simply take the
smallest possible mesh $t_n-t_{n-1}$ and sample $B_{t_n}$ uniformly at all
points $t_n$); see also the caption to Fig.
\ref{Figure: histograms}.

On the other hand, if we have sampled $B_t$ at
$t=t_n$ and have concluded that $\bar \gamma_{n,a}$ is too close to or too
far away from $\bar \gamma_{n-1,a}$, we cannot ignore the value of $B_{t_n}$ in
the future sampling of $B_t$ for other values of $t$ (say, $\tilde t_{n}$) as we
did above for a non-random driving function, since $B_{\tilde t_{n}}$
and $B_{t_n}$ are not independent random variables.
This motivates the following method (algorithm) for the ($\delta,\sigma$)-SLE
simulation. The scheme
(see Fig. \ref{Figure: the routine R})
presented below is motivated by the discussion above and
the experience of the author. We do not present any proof of correctness. 

We consider a routine 
$R(\bar \gamma_{x,a},\{t_x,B_x\},\{t_y,B_y\})$ that is supposed to obtain all points 
$\bar\gamma_n$ that correspond to a given time interval $[t_x,\,t_y]$. It also
samples all necessary values of $B_t$ inside $[t_x,\,t_y]$.
The values of $B_t$ for $t=t_x$ and $t=t_y$ are assumed to be given as $B_x$
and $B_y$. All points 
$\{\bar \gamma_{i,a}\in\Dc\}_{i=1,2,\dotso, n}$ 
for the interval 
$[0,\,t_x]$ 
are assumed to be obtained during the previous steps. The last joint point 
$\bar \gamma_{n,a}\in\Dc$, which corresponds
to $t=t_x$ and is denoted by 
$\bar \gamma_{x,a}$, 
is the first argument of the routine
$R$.
The algorithm is recursive, as the routine $R$ calls itself to
obtain $\bar\gamma_{n,a}$ on subintervals of $[t_x,\,t_y]$.  

At the very beginning of the simulation, we pick some $T>0$, 
sample $B_T$, and set $\bar\gamma_{0,a}:=a$. Then we call the routine $R$ for
the first time with the arguments $R(\bar\gamma_{0,a},(0,0),(T,B_T))$, which
initiates the simulation.

The points $\bar\gamma_{n,a}$, $n=1,2,\dotso$, are obtained sequentially in
different execution instances of $R$.
There is also a global Boolean parameter `additional point' that is equal to
`false' by default and is equal to `true' if the last sampled point
$\bar\gamma_{n+1,a}$ is too close to $\bar\gamma_{n,a}$, see details below.

\begin{sidewaystable}[hp]
\begin{center}
\includegraphics{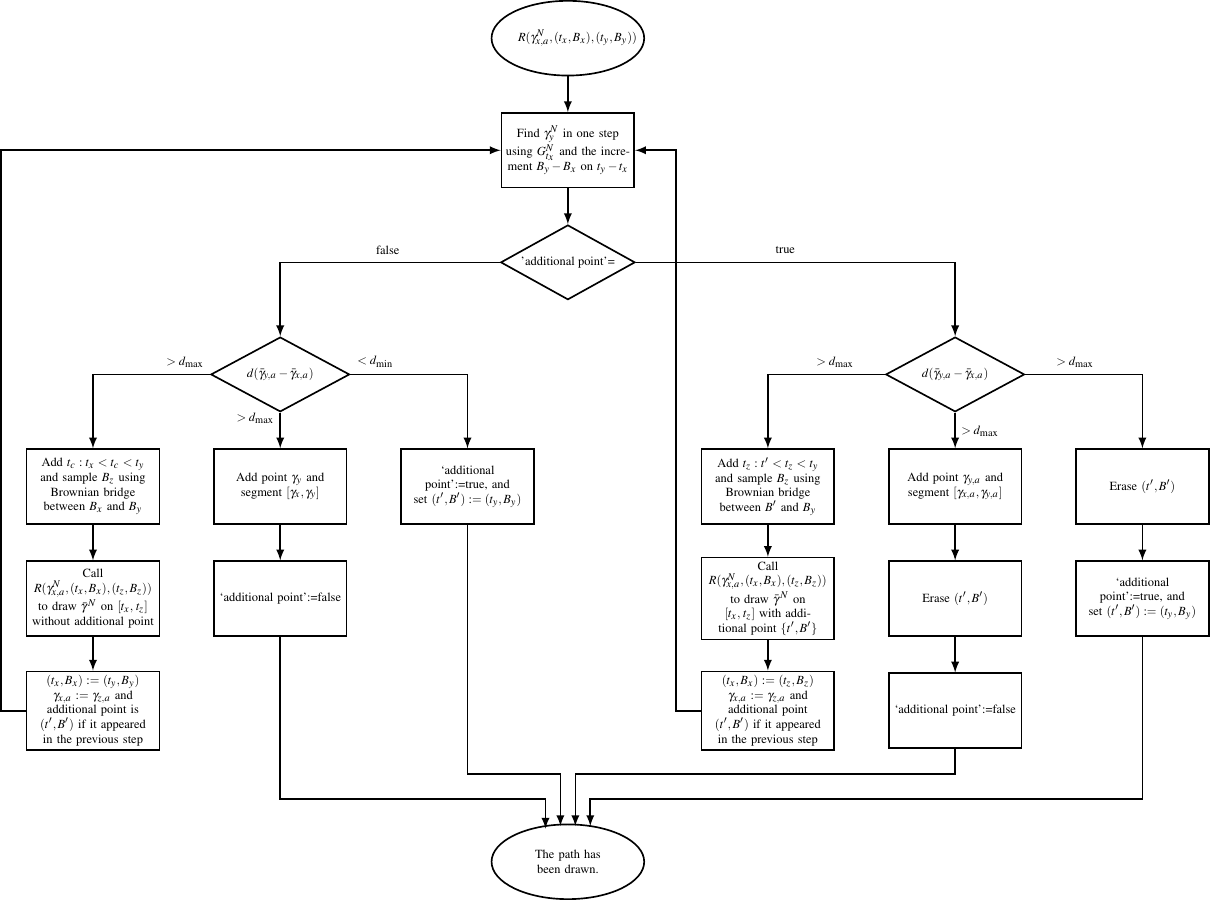}
\end{center}
\caption{Block scheme for the routine $R$.}
\label{Figure: the routine R}
\end{sidewaystable}
The design of the routine $R$ is presented in Fig. 
\ref{Figure: the routine R}
as a block scheme. First of all,  
the routine calculates the point $\bar\gamma_{y,a}$ by applying the
step map $\bar G_{n+1}$ with the given time interval $t_y-t_x$ and jump
$B_y-B_x$
(the top rectangle in Fig. 
\ref{Figure: the routine R}).
The method is described in Section \ref{Formula: Choice of bar G_n}.
The next decision (the diamond below the rectangle) will be discussed below. For
now, assume that there is no additional point and go left. Afterwards, there are three
alternatives related to the value of the distance 
$d(\bar\gamma_{y,a},\bar\gamma_{x,a})$.
\begin{enumerate}[1.]
	\item $d(\bar\gamma_{y,a},\bar\gamma_{x,a})$ 
	is too big. In this case, we pick
	some point $t_z\in(t_x,\,t_y)$, sample $B_z:=B_{t_z}$, 
	and call the routine $R$ for the interval $[t_x,\,t_z]$. Afterwards, we return
	to the beginning but with a changed left boundary of the time interval 
	$t_x:=t_z$ 
	(or smaller, see below for details) 
	and new 
	$\bar \gamma_{x,a}$, which corresponds to the new value of $t_x$.
	\item $d(\bar\gamma_{y,a},\bar\gamma_{x,a})$ is neither too big nor too small.
	This means that the point $\bar\gamma_{y,a}$ obtained before can be added to
	the collection $\{\bar\gamma_{i,a}\}_{i=0,1,2,\dotso n}$ of the joint points.
	Hence, the task of the routine is complete, and we finish it with the attribute
	`additional point'=`false'.
	\item $d(\bar\gamma_{y,a},\bar\gamma_{x,a})$ is too small. This means that the
	point $\bar\gamma_{y,a}$ is too close to $\bar\gamma_{x,a}$, and it is not
	reasonable to add it to the collection of the joint points because this would
	needlessly increase the number of iterations in 
	\eqref{Formula: gamma^N_n = G_1 G_2 ... G_n}. We finish the routine but with
	the attribute `additional	point'=`true'. We also register the pair 
	$(t',B'):=(t_y,B_y)$ for future needs.
\end{enumerate}

This completes the description of the left half of the block scheme
that corresponds to the option `additional point'=`false'. The alternative
happens when the last considered point 
$\bar\gamma_{n+1,a}$ 
is too close to the
last point added to the collection 
$\{\bar\gamma_i\}_{i=0,1,2,\dotso n}$ 
as in item 3 above. So if, after the
calculation of 
$\bar\gamma_{y,a}$ 
in the upper rectangle, the global attribute  
`additional	point'=`true', we go right from the diamond. In addition, the pair 
$(t',B')$ 
is given and it is known that it corresponds to a point too close to
$\bar\gamma_{x,a}$. The right half of the
algorithm is analogous to the left half. We have the same three alternatives
after the calculation of the distance 
$d(\bar\gamma_{y,a},\bar\gamma_{x,a})$.
\begin{enumerate}[1.]
	\item $d(\bar\gamma_{y,a},\bar\gamma_{x,a})$ is too big. In this case, we pick
	some point $t_z\in(t',\,t_y)$, sample $B_z:=B_{t_z}$, 
	and call the routine $R$ for the interval $[t_x,\,t_z]$
	with the option `additional	point'=`true'. Afterwards, we return to
	the beginning, but the value of $t_x$ can be increased and the point
	$\gamma_{x,a}$ can be changed to another one	as in item 1 above.
	The global attribute `additional	point' depends on the results of the called
	routine for the interval $[t_x,\,t_z]$. Recall that it equals
	`false' if the point 
	$\bar\gamma_{z,a}$ for $(t_z,B_z)$ is added to the collection, and it  
	equals `true' if $\bar\gamma_{z,a}$ for $(t_z,B_z)$ is too close to the last
	point $\bar\gamma_n$ added to the collection. In the second case, the pair
	$\{t_z,B_z\}$ is registered as $(t',B')$ and 
	$\{t_x,B_x\}$ corresponds to the last point $\bar\gamma_{n,a}$ added to the
	collection.
	\item $d(\bar\gamma_{y,a},\bar\gamma_{x,a})$ is neither too big nor too
	small.
	This means that the point $\bar\gamma_{y,a}$ obtained before can be added to
	the collection $\{\bar\gamma_{i,a}\}_{i=0,1,2,\dotso n}$ of the joint points. 
	The pair 
	$(t',B')$ can be erased, as it has no influence on values of $B_t$ if 
	$t>t_y>t'$. 
	The	task of the routine is complete after that, and we finish it with
	the attribute `additional point'=`false'.
	\item $d(\bar\gamma_{y,a},\bar\gamma_{x,a})$ is too small. This means that the
	point $\bar\gamma_{y,a}$ is also too close to $\bar\gamma_{x,a}$. Just as in
	the previous alternative, we no longer need $(t',B')$, and we replace it by
	$(t',B'):=(t_y,B_y)$. 
	We finish the routine with the attribute `additional
	point'=`true'.
\end{enumerate}

This is how the program works. The routine $R$ calls itself for
smaller and smaller intervals of time until the distance between the points
$\bar\gamma_{n,a}$ is small enough. The choice of the point $t_z$ in the first
alternatives of both halves of the block scheme can be prescribed by the
time interval between the two previously added points. 

The simulation can be stopped when the first call of the routine for 
$[0,\,T]$ finishes. If the time intervals between the last added points are of
the same order of magnitude as $T$, the last point may be interpreted
as the limit point $t\map+\infty$ of the slit $\gamma_t$. The value $T=10^8$
was used in all simulations. This choice is motivated by some technical
characteristics of numerical calculations in Wolfram Mathematica.

The limit point of the slit can be observed in all simulations in Chapter
\ref{Chapter: Classical cases}. 
In some situations, however, such as a chordal SLE
in the half-plane chart, the limit points were not observed (this is in
accordance with the analytic theory, which states that the limit point is
infinity). In this case, the simulation can be stopped when the number $n$ of
the joint points reaches some maximal value $N$.

Here we present two high-resolution examples of samples of the chordal
SLE slit in the half-plane chart; see Figs.
\ref{Figure: SLE4_demo} 
and 
\ref{Figure: SLE7_demo}. 
In most other figures with slit samples in this monograph, the number of points
$N$ is several hundred, and 
$d_{\text{max}}=0.02$, $d_{\text{min}}=0.01$.
Figure 
\ref{Figure: histograms}
demonstrates the importance of a dynamical choice of the
partition of time.  
 
\begin{figure}[hp]
\centering
\includegraphics[height=\textheight,
		keepaspectratio]{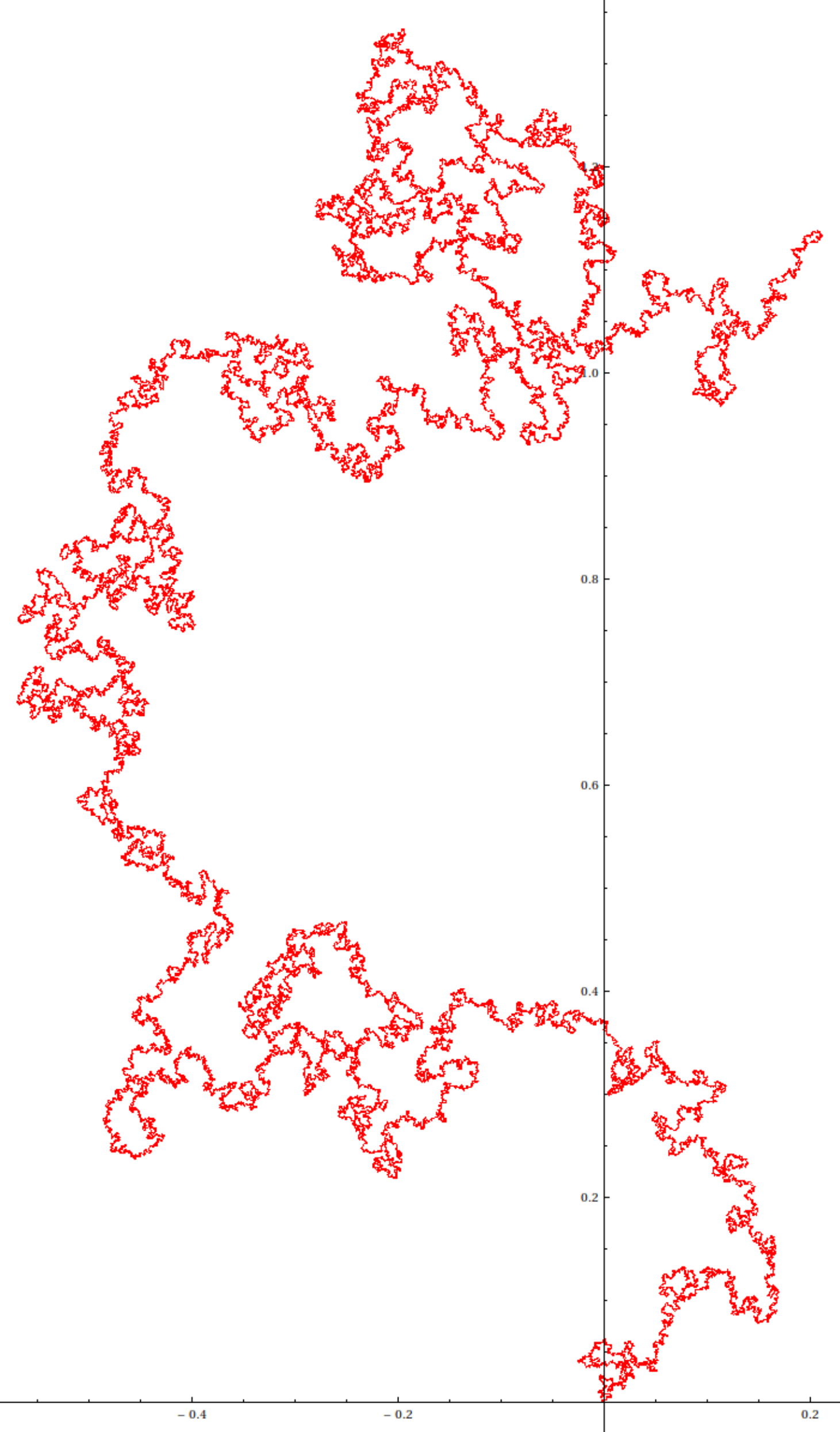}
\caption{Simulation of a chordal SLE slit in the half-plane chart for $\kappa=4$ 
($10^5$ points, 
$d_{\text{max}}=10^{-3}$, and $d_{\text{min}}=0.5 \cdot 10^{-3}$).} 
\label{Figure: SLE4_demo}
\end{figure}

\begin{figure}[hp]
\centering
\includegraphics[height=19cm,
		keepaspectratio]{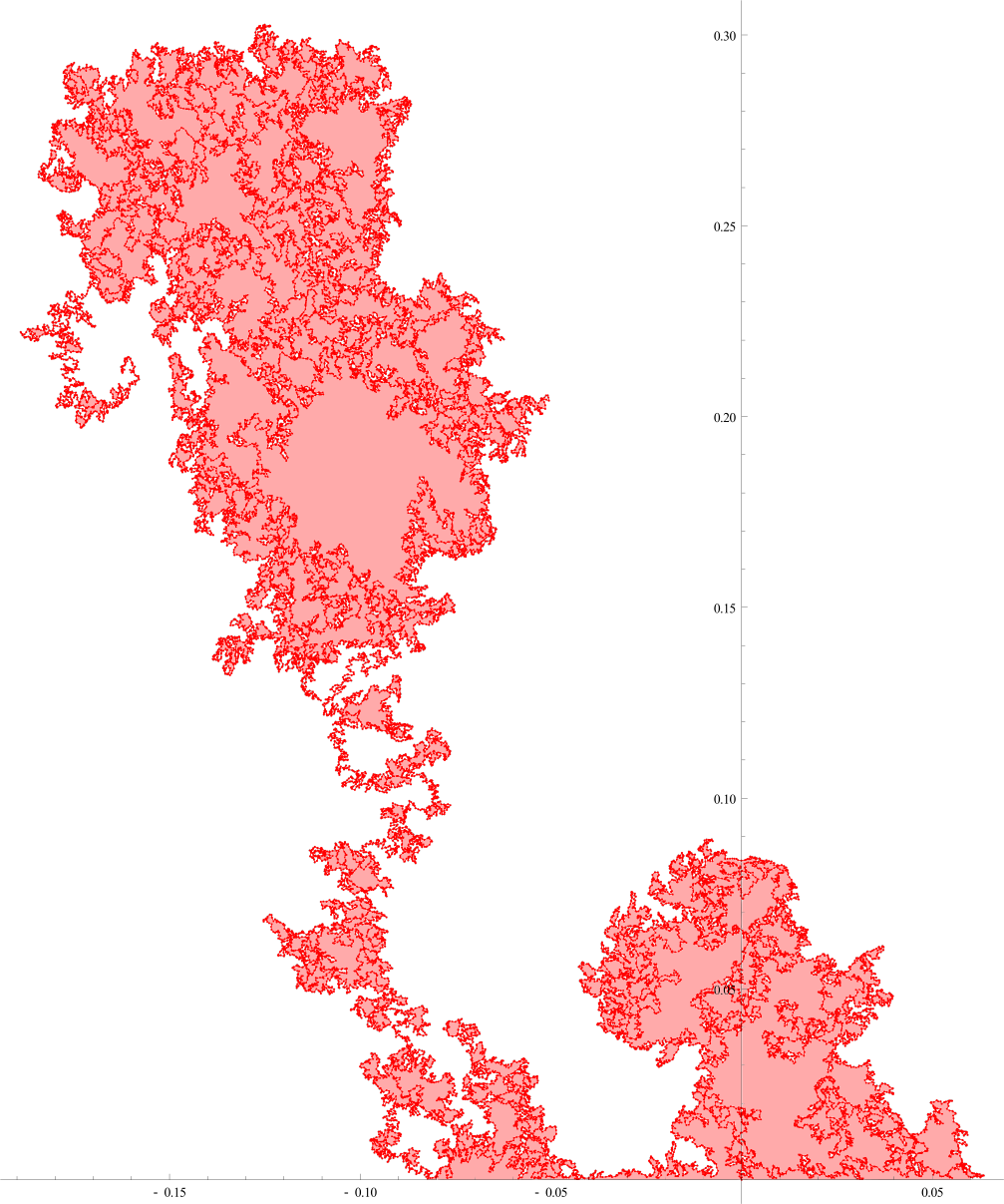}
\caption{Simulation of the chordal SLE slit for
$\kappa=7$ in the half-plane chart ($38793$ points), $d_{\text{max}}=10^{-3}$,
and $d_{\text{min}}=0.5 \cdot 10^{-3}$. The method gives only an
approximation of the curve that generates the hull $\K_t$. 
This approximate curve ${\bar \gamma}^{\,N}_t$ never touches itself, but passes
close to itself. The regions that are numerically almost enclosed by the curve were filled with
pink by hand in a graphics editor.}
\label{Figure: SLE7_demo}
\end{figure}


\begin{figure}[h]
\centering
  \begin{subfigure}[t]{0.5\textwidth}
		\centering
    \includegraphics[width=5cm,keepaspectratio=true]
    	{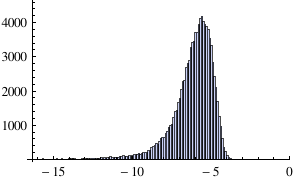}
    \caption{Fig. \ref{Figure: SLE4_demo}}
  \end{subfigure}%
  ~  
  \begin{subfigure}[t]{0.5\textwidth}
		\centering
    \includegraphics[width=5cm,keepaspectratio=true]
    	{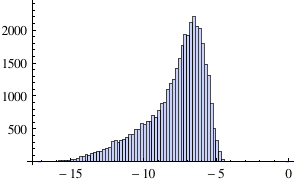} 
  	\caption{Fig. \ref{Figure: SLE7_demo}}
	\end{subfigure}%
\caption{The histograms illustrate the distribution of the length of the
intervals $t_n-t_{n-1}$, $n=1,2,\dotso,N$ in the partitions of time that are
used for the simulations in Figs.~%
\ref{Figure: SLE4_demo} 
and 
\ref{Figure: SLE7_demo}. 
We use a
logarithmic scale on the horizontal axis. The left histogram
demonstrates that the length of the time intervals varies from $10^{-12}$ to
$10^{-4}$, which corresponds to a difference of eight orders of magnitude
between the typical smallest and largest time intervals. The right
histogram demonstrates the same variation, but from $10^{-15}$ to $10^{-5}$.
\label{Figure: histograms}}
\end{figure}

\section{Simulation of $(\delta,\sigma)$-L\"owner equation driven by a
stable L\'evy process}
\label{Section: Levy simulation}

Substituting a stochastic process for the driving function 
is a convenient tool for studying the structure of $(\delta,\sigma)$-L\"owner
chains. Fig. 
\ref{Figure: Some SLE Examples}
shows that a $(\delta,\sigma)$-L\"owner 
chain driven by a sample of the Brownian motion 
(($\delta,\sigma$)-SLE) 
does not allow one to distinguish one choice of $\delta$ and $\sigma$ from
another. In the general case, the curves have identical local behavior and tend
to a random point inside the disk.

To make such figures more representative, one may consider
a more general stochastic process than the Brownian motion. One of the most
natural choices is the stable L\'evy process. This version of the stochastic 
chordal L\"owner equation is considered in 
\cite{Rushkin2006,Oikonomou2008,Chen2009}.
The point is that using a discontinuous but piecewise continuous
driving function produces a tree slit, not just a single curve. Each of the
intervals of continuity of the driving function corresponds to a curve that is a
branch of the tree. Each branch tends to a different random point.
Thus, for each choice of $\delta$ and $\sigma$, we expect to have a unique
characteristic pattern of a random tree.

The problem of a correct numerical simulation of a $(\delta,\sigma)$-L\"owner
equation driven by a stable L\'evy process has not been completely solved in the
present work. The technical difficulty will be described below. However, we
present some intermediate results and a picture that is reasonably close to
correct. The author is not aware of which method was used in the above-cited
papers. The most important step is the choice of the partition of the time
interval.

\subsection{Stable L\'evy process}
\label{Section: Stable Levy process}

In this section, we give the definition and consider some basic properties of
the stable L\'evy process. For more details, we refer to \cite{Applebaum2009}.

\begin{definition}
The 
\emph{stable L\'evy process (symmetric $\alpha$-stable L\'evy process)}
\index{stable L\'evy process}
is an $\mathbb{R}$-valued stochastic process
$\{L^{\alpha}_t\}_{t\in[0,\infty)}$ ($L_0=0$)
with independent increments. The distribution of the increments is defined
by the following characteristic function:
\begin{equation}
	\Ev{e^{i (L^{\alpha}_{t-s}-L_s) \theta}} 
	= e^{-s|\theta|^{\alpha}},\quad \theta\in\mathbb{R},
	\quad t,s\in[0,+\infty).
\end{equation} 
for some fixed $0<\alpha\leq 2$.
\end{definition}

The process 
$\{L^{\alpha}_t\}_{t\in[0,\infty)}$
possesses the following properties: 
\begin{enumerate}[1.]
\item $\{L^{\alpha}_t\}_{t\in[0,\infty)}$ is time homogeneous;
\item $\{L^{\alpha}_t\}_{t\in[0,\infty)}$ is strong Markov;
\item $\{L^{\alpha}_t\}_{t\in[0,\infty)}$ is self-similar (scale covariant):
\begin{equation}
	\mathrm{Law}[ L^{\alpha}_{ct} ] 
	= \mathrm{Law} [ c^{\frac{1}{\alpha}} L_t ],\quad c>0,\quad t\in[0,+\infty);
\end{equation}
\item $\{L^{\alpha}_t\}_{t\in[0,\infty)}$ is stochastically continuous, i.e.,
\begin{equation}
	\lim\limits_{t\map s} \mu_{L^{\alpha}}\left[ |L^{\alpha}_t|>a \right] 
	= 0,\quad a>0,\quad s\in[0,+\infty);
\end{equation}
\item $\{L^{\alpha}_t\}_{t\in[0,\infty)}$ is right-continuous with left limits
a.s.;
\item For $\alpha<2$, $\{L^{\alpha}_t\}_{t\in[0,\infty)}$ is piecewise
continuous, in the sense that on each finite time interval it is represented by
a continuous function with a finite number of jumps of size larger than
$\varepsilon>0$ and a countable number of jumps smaller than $\varepsilon$ a.s.;
\item For $\alpha=2$, $\{L^{\alpha}_t\}_{t\in[0,\infty)}$ is the Brownian motion
$\{B_t\}_{t\in[0,+\infty)}$.
\end{enumerate}

\begin{figure}[h]
\centering
  \begin{subfigure}[t]{0.5\textwidth}
		\centering
    \includegraphics[width=5cm,keepaspectratio=true]
    	{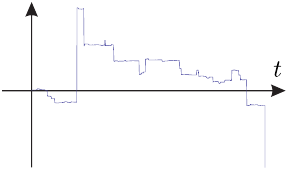}
    \caption{$\alpha=0.5$}
  \end{subfigure}%
  ~  
  \begin{subfigure}[t]{0.5\textwidth}
		\centering
    \includegraphics[width=5cm,keepaspectratio=true]
    	{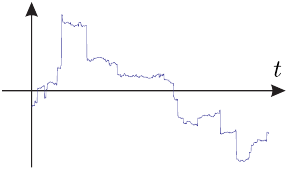} 
  	\caption{$\alpha=1$}
	\end{subfigure}%

	\begin{subfigure}[t]{0.5\textwidth}
	\centering
		\includegraphics[width=5cm,keepaspectratio=true]
			{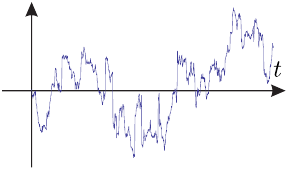}
		\caption{$\alpha=1.5$}
	\end{subfigure}%
	~
	\begin{subfigure}[t]{0.5\textwidth}
	\centering
		\includegraphics[width=5cm,keepaspectratio=true]
			{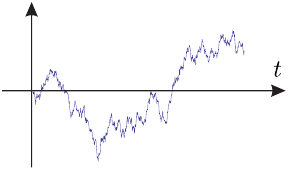}
		\caption{$\alpha=2$}
	\end{subfigure}%
\caption{Samples of the stable L\'evy process
$\{L^{\alpha}_t\}_{t\in[0,+\infty)}$ for different values of $\alpha$.
\label{Figure: Levy samples}}
\end{figure}

\begin{sidewaystable}[h]
\begin{center}
\centering
\includegraphics[width=25cm,
		keepaspectratio]{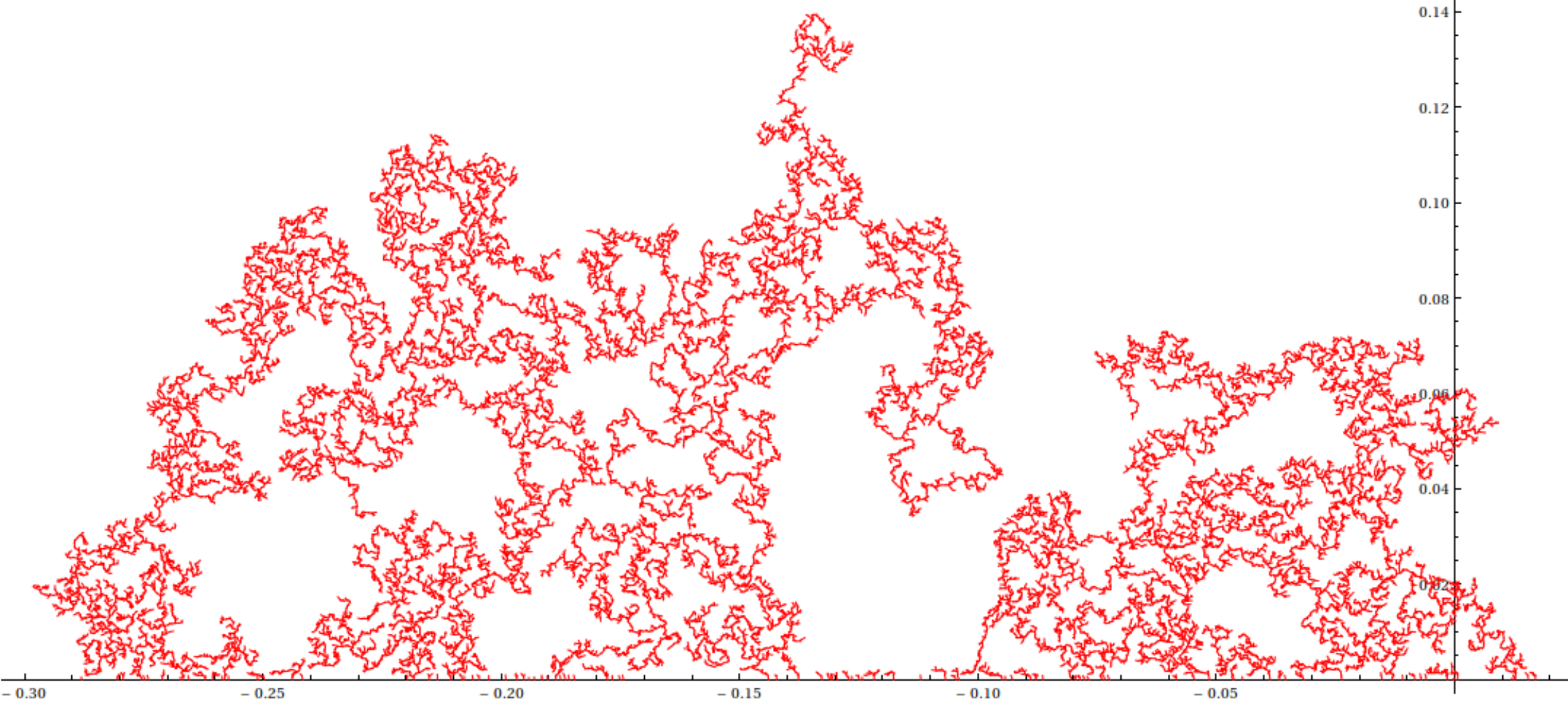}
\caption{Simulation of a chordal SLE slit driven by a stable L\'evy process
($u_t:=L_t^{\alpha}$) with the method described in the text ($3\cdot 10^4$
points).
The parameters are $\alpha=1.9$ and $T=8\cdot10^{-3}$. The slit is a sample
of a random tree.
\label{Figure: SLE_Levy_alpha=1.9_demo}}
\end{center}
\end{sidewaystable}

Figure 
\ref{Figure: Levy samples}
illustrates samples of 
$\{L^{\alpha}_t\}_{t\in[0,+\infty]}$ 
for
different values of $\alpha$. We can see that smaller values of $\alpha$
correspond to larger but rarer jumps and more continuous-looking
behavior between jumps. When $\alpha$ tends to $2$, the number of jumps
increases, but they become smaller. In the limit, we obtain the Brownian motion, which
has no jumps a.s.

We use the following method of numerical sampling of
$\{L^{\alpha}_t\}_{t\in[0,\infty)}$.
Represent $L_t$ as a composition 
\begin{equation}
	L_t=B_{Y_t^{\alpha}},\quad t\in[0,+\infty)
\end{equation}
of the Brownian motion
$\{B_s\}_{s\in[0,+\infty)}$ 
and an independent positive-valued strictly increasing stochastic process
$\{Y^{\alpha}_t\}_{t\in[0,\infty)}$ called the 
\emph{L\'evy subordinator}.
\index{L\'evy subordinator}
The value of $Y^{\alpha}_t$ may be sampled with the aid of the formula
\begin{equation}
	Y^{\alpha}_t = 
	\frac{t^{\frac{2}{\alpha}} \sin\left((x+\frac{\pi}{2})\frac{\alpha}{2}\right)  
		\left( 
			\frac{\cos\left( x - \frac{\pi}{2})\frac{\alpha}{2}  \right) }{y} 
		\right)
		^{\frac{2}{\alpha}-1}	
	}
	{\cos^{\frac{2}{\alpha}}(x)},\quad t\in[0,+\infty),
	\label{Formula: Y_t = ...}
\end{equation} 
where $x$ is a random variable uniformly distributed on the interval 
$[-\pi/2,\pi/2]$ 
and 
$y$ 
is an independent positive random variable with the probability density function 
$p_y=e^{-y}$, $y\geq 0$. 
It can be shown
\footnote{Unfortunately, the author has lost the reference to the book from which the formula 
\eqref{Formula: Y_t = ...}
is taken.}
that a third independent Gaussian random variable with variance 
$Y_t$ 
has the characteristic function 
$e^{-s|\theta|^{\alpha}}$. 
Furthermore, the process 
$\{Y_t\}_{t\in[0,+\infty)}$ 
defined by
\eqref{Formula: Y_t = ...}
is time-homogeneous. 

For the simulation of L\"owner slits, we also need a method to sample 
the L\'evy bridge, i.e., a stable L\'evy process conditioned on 
$L^{\alpha}_T=X$ 
for some fixed 
$X\in\mathbb{R}$ 
and $T>0$, denoted by
$\{L^{\alpha}_{t;T,X}\}_{t\in[0,\,t)}$. 
This can be done by the following method. Let
$\{B_{t;T,X}\}_{t\in[0,\,t]}$
be the Brownian bridge, namely, 
the Brownian motion conditioned on $B_T=X$. Then
\begin{equation}
	L^{\alpha}_{t;T,X} =
	B_{Y^{\alpha}_t;Y^{\alpha}_t+\tilde Y^{\alpha}_{T-t},X},\quad t\in[0,\,t],
\end{equation}
where $Y^{\alpha}_t$ and $\tilde Y^{\alpha}_{T-t}$ are two independently sampled
subordinators as above.

\subsection{Simulation method for L\'evy process}
\label{Section: simulation method for Levy process}

Due to the strong Markov property of 
$\{L_t^{\alpha}\}_{t\in[0,+\infty)}$ 
the stochastic process $G_t$ driven by $u_t=L_t$ is also strong Markov by the
same argument as in Section 
\ref{Section: Domain Markov property and conformal invariance of random laws on
planar curves}.
Thus, we can use the same representation 
\eqref{Formula: G^N = G_n circ G_n-1 ... G_1}
as before. 

Each jump of the driving function corresponds to a new branch of the tree slit.
Hence, instead of connecting the points
$\{\bar\gamma_{n,a}^N\}_{n=1,2,\dotso}$ 
by straight line segments, we apply the following method.
We define a second kind of joint points: 
\begin{equation}\begin{split}
	&\bar \gamma_{n,b}^N:=
	\bar G_1^{-1} \circ \bar G_2^{-1} \circ \dotso \circ \bar G_{n}^{-1} 
	\circ H^{-1}_{t_{n+1}-t_{n}}[\delta] (a),\\
	&n=0,1,2,\dotso N-1
	\label{Formula: gamma^N_n = G_1 G_2 ... G_n Levy}
\end{split}\end{equation}
denoted by the index `$b$'; see also Fig.
\ref{Figure: different approxiamtions}.
Let 
\begin{equation}
	\bar l_n^N:=\{
	\bar G_1^{-1} \circ \bar G_2^{-1} \circ \dotso \circ \bar G_{n}^{-1} 
	\circ H^{-1}_{s}[\delta] (a)\in\Dc,\quad s\in[0,\,t_{n+1}-t_{n}]\}
\end{equation}
be an arc that connects  
$\bar \gamma_{n,a}^N$
and
$\bar \gamma_{n,b}^N$.
We approximate the tree hull 
$\K_t$
by the union 
$\bar \K_t^N:=\bigcup\limits_{n=0,1,\dotso,N} \bar l_n^N$.
If we apply this to a continuous driving function, then the end 
$\bar \gamma_{n,b}^N$ 
of each
arc $\bar l_n^N$ would be close to the beginning point
$\bar \gamma_{n+1,a}^N$
of the next arc
$\bar l_{n+1}^N$. 
Thus, the tree $\bar \K_t^N$ would tend to a simple curve. 
On the other hand, if the driving function has a jump $\Delta u$ at $t=t'$,
$t'\in(t_{n},\,t_{n+1})$, 
the points 
$\bar \gamma_{n+1,a}^N$
and
$\bar \gamma_{n+1,b}^N$
do not tend to each other because the distance between them is defined by the
map $H_{B_{t_{n+1}}-B_{t_n}}[\sigma]$
that tends to 
$H_{\Delta u}[\sigma]$,
when
$t_{n}\map t'-$
and
$t_{n+1}\map t'+$.

We can forcibly extend the method considered in the previous section to this
situation by replacing the distance 
$d(\bar\gamma_{n,a}^N,\bar\gamma_{n-1,a}^N)$
with the distance
$d(\bar\gamma_{n,b}^N,\bar\gamma_{n,a}^N)$.
However, it turns out that $\bar \K_t^N$ is not close to $\K_t$ in this case.
This happens for the following reason. The method of sampling the stable L\'evy
process 
$\{L_t^{\alpha}\}_{t\in[0,+\infty)}$ 
gives the value of 
$L_t^{\alpha}$ 
at a finite number of points, but not the positions of the jumps. We obtain each
visible branch of $\bar \K_t^N$
only because of a relatively large value of 
$L_{t_{n+1}^{\alpha}}-L_{t_{n}^{\alpha}}$
with respect to
$t_{n+1}-t_{n}$.
If the length of the arc  
$\bar l_{n+1}^N$
after the `jump' is neither too big nor too small, the program continues the
simulation
for $t>t_{n+1}$.
However, it may occur, and numerical tests confirm this, that 
the behavior of 
$L_t^{\alpha}$
between
$t_{n+1}$
and
$t_n$
is important. Namely, the jump between 
$t_{n+1}$
and
$t_n$
may consist of two jumps, and a small time interval between them may correspond
to a numerically large subtree of $\K_t$ that is lost. In the case of Brownian
motion, the probability of such a phenomenon was suppressed due to the continuity
of the driving function (absence of jumps). In the case of the L\'evy process, a
numerical test showed that this regularly happens. In other words, the method
of simulation proposed above does not give all the branches, but only some of
them.

Nevertheless, we apply this method anyway. The result for chordal SLE and 
$\alpha=1.9$ is presented in Fig.
\ref{Figure: SLE_Levy_alpha=1.9_demo}.
The application for other types of $(\delta,\sigma)$-SLE in the unit disk chart 
does not give the expected result with the characteristic pattern because too
many branches of the tree are lost. It would be worthwhile to modify the
method discussed in order to solve this problem.

%% file: BB_approach.tex
\chapter{($\delta,\sigma$)-L\"owner equation from the algebraic point of view}
\label{Chapter: Representation theory approach}

\section*{Introduction}

In this chapter, we study the $(\delta,\sigma)$-L\"owner equations from the
point of view of the Lie algebras they are related to.
In the first section, we reconsider the ($\delta,\sigma$)-SLE classification
problem. The second one is dedicated to a generalization of
the results of Roland Friedrich, Michel Bauer, and Denis Bernard
about the SLE/CFT relation from the classical SLEs to
the general case of ($\delta,\sigma$)-SLE. 

\section{Algebraic classification and normalization}
\label{Section: Classification and normalization from algebraic point of view}

In Section
\ref{Section: Equivalence and normalization of slit L chains},
we explained what we mean by essentially different L\"owner chains and
introduced normalization conditions in terms of coefficients 
$\{\delta_i\}_{i=-2,-1,0,1}$ 
and 
$\{\sigma_i\}_{i=-1,0,1}$.
Here, we study the family of essentially different 
$(\delta,\sigma)$-L\"owner equations from a different point of view. We avoid
using the basis
\eqref{Formula: ell_n^H = ...}
and parametrization 
\eqref{Formula: delta and sigma for Lowner chain}.
Instead, we consider purely algebraic properties of vector fields $\delta$
and $\sigma$ and their possible linear combinations and commutators. 
Such objects are invariant with respect to the transforms $\mathscr{R}$ and
$\mathscr{S}$ by construction.
 
Define the commutator of two holomorphic vector fields $v$ and $w$ by
\begin{equation}
	[v,w]^{\psi}(z):= 
	v^{\psi}(z)\left(w^{\psi}\right)'(z) - w^{\psi}(z)\left(v^{\psi}\right)'(z)
	\label{Formula: [v,w] := ...}
\end{equation}
It is straightforward to check that the right-hand side transforms as a
vector field, according to the rule 
\eqref{Formula: tilde v = 1/dtau v(tau)}. Thus, $[v,w]$ is also a
vector field, and the definition above is chart-independent. 
For the basis vectors
\eqref{Formula: ell_n^H = ...}
this commutator takes the form
\begin{equation}
	[\ell_n,\ell_m] = (m-n) \ell_{n+m}.
	\label{Formula: [l_n,l_m] = (m-n) l_n+m}
\end{equation}

Define an infinite-dimensional Lie algebra of holomorphic vector fields on
$\Dc$ that are tangent to the boundary except at the point $a\in\de\Dc$, where a
finite order pole is allowed.
The algebra is a real linear space	
\begin{equation}
	\mathcal{U}_{\leq 1} := \mathrm{span}\{\ell_{n}~|~n=1,0,-1,-2,\dotso\}
	\label{Formula: Witt_<=1 := ...}
\end{equation}
equipped with the Lie operation $[\cdot,\cdot]$.
The algebra $\mathcal{U}_{\leq 1}$ is a subalgebra of the more frequently used Witt algebra
\begin{equation}
	\mathrm{Witt} := \mathrm{span}\{\ell_{n}~|~ n=\dotso,-2,1,0,-1,-2,\dotso\}.
	\label{Formula: Witt := ...}
\end{equation}

Consider the Lie hull 
\begin{equation}
	\mathcal{U}[\delta,\sigma]
	:=\mathrm{span}\{\delta,\sigma,[\delta,\sigma],[[\delta,\sigma],\sigma],
	\dotso \},
\end{equation}
where $\delta$ and $\sigma$ are given by 
\eqref{Formula: delta and sigma for Lowner chain}.
It is a subalgebra of $\mathcal{U}_{\leq 1}$. 
Since we only consider vector fields such that $\delta_{-2}\neq 0$ and
$\sigma_{-1}\neq 0$, the Lie algebra $\mathcal{U}[\delta,\sigma]$ is always
infinite-dimensional. For example, in the chordal $\delta$ and $\sigma$ case,
see Section \ref{Section: Chordal Loewner equation}, we have
\begin{equation}
	\mathcal{U}[2\ell_{-2},\ell_{-1}] 
	= \mathrm{span}\{\ell_{n} ~|~ n=-1,-2,\dotso\}.
\end{equation}

It is remarkable that the transformations
$\mathscr{V}$, $\mathscr{T}$ 
and 
$\mathscr{D}$ 
act on 
$\delta$ 
and 
$\sigma$
linearly and preserve the two-dimensional subspace
$\mathrm{span}\{\delta,\sigma\}$. 
In particular, they leave the Lie hull 
$\mathcal{U}[\delta,\sigma]$
invariant. The transformations 
$\mathscr{R}$ 
and 
$\mathscr{S}$ 
do change
$\mathrm{span}\{\delta,\sigma\}$ 
and  
$\mathcal{U}[\delta,\sigma]$ 
(as a subset of 
$\mathcal{U}_{\leq1}$), 
but preserve the Lie algebra structure. Namely, the transforms
$\mathscr{R}$ 
and 
$\mathscr{S}$ 
induce isomorphisms of subalgebras of 
$\mathcal{U}_{\leq 1}$. 
Thus, in particular, essentially equivalent 
$(\delta,\sigma)$-SLEs 
correspond to isomorphic Lie algebras.

We also mention that the algebra corresponding to the chordal
L\"owner equation is represented by vector fields with a second order zero at
the boundary point; the Lie algebra corresponding to the radial L\"owner
equation consists of vector fields vanishing at an interior point of the disk;
and the Lie algebra corresponding to the dipolar L\"owner equation consists of
vector fields vanishing at two boundary points, see also the corresponding
sections in Chapter 
\ref{Chapter: Classical cases}.
In other words, the Lie algebras
corresponding to these essentially different L\"owner equations are not
isomorphic. It is a more difficult question whether or not the Lie hulls are
isomorphic as algebras for essentially different equations.
The investigation of how the algebraic properties of $\delta$ and $\sigma$
influence the geometric properties of $(\delta,\sigma)$-L\"owner chains seems
very promising. 

A possible way to classify the L\"owner equations is to consider the following
identity, which is the simplest nontrivial commutation relation between
$\delta$ and $\sigma$:
\begin{equation}
	[[\delta,\sigma],\delta] = 
	x_5 [[[\delta,\sigma],\sigma],\sigma] +
	x_4 [[\delta,\sigma],\sigma] +
	x_3 [\delta,\sigma] \\+ 
	x_2 \delta +
	x_1 \sigma + 
	w, 
	\label{Formula: [[delta,sigma],delta] = ...}
\end{equation}	
where $w$ is a holomorphic vector field vanishing at the source point $a$.
The coefficients $x_i$, $i=1,2,3,4,5$ are obviously invariant with respect to
the transformations $\mathscr{R}$ and $\mathscr{S}$.
With the aid of the transforms $\mathscr{T}$ and $\mathscr{D}$, it is always
possible to fix
\begin{equation}
	x_5 = \frac13, \quad x_4=0
	\label{Formula: x_5 = 1/3, x_4 = 0}
\end{equation}
($x_5 = -\frac13$ for the reverse case).
The transform 
\begin{equation}
	\mathscr{T}_c^2 \circ \mathscr{V}_{c},\quad c\in \mathbb{R}\setminus\{0\}
	\label{Formula: T_c^2 circ V_c}
\end{equation}
keeps the above conditions unchanged. On the other parameters, it acts as 
\begin{equation}\begin{split}
	&x_3\mapsto c^2 x_3,\\
	&x_2\mapsto c^3 x_2,\\
	&x_1\mapsto c^4 x_1,\\
	\label{Formula: x_n |-> c^m x_n}
\end{split}\end{equation}
The collection of three real parameters
$\{x_3,x_2,x_1\}$ up to the transform 
\eqref{Formula: x_n |-> c^m x_n},
can be identified with the collection of all essentially different
$(\delta,\sigma)$-L\"owner chains. 
For example, in the normalization 
\ref{Formula: delta_-2 = pm2}
and
\ref{Formula: sigma_-1 = -1},
for the forward dipolar equation $x_3=8/3$
and $x_2=x_1=0$, and for the forward radial case $x_3=-8/3$ and $x_2=x_1=0$. 
The value $8/3$ can be replaced by an arbitrary positive constant due to
\eqref{Formula: x_n |-> c^m x_n}.
The case considered in 
\cite{Ivanov2012a} 
corresponds to $x_3=32/3$, $x_2=0$, and $x_1=-60$ for the normalization
\eqref{Formula: x_5 = 1/3, x_4 = 0}.

The chordal case is special because $x_1=x_2=x_3=0$. All other cases
can be understood as points on a two-dimensional real manifold obtained as a
quotient of $\{x_3,x_2,x_1\}/$ by the transform 
\eqref{Formula: x_n |-> c^m x_n}, 
which is homeomorphic to a disk.
The dipolar, radial, and  
\cite{Ivanov2012a}
cases correspond to different points
on the boundary of this manifold.

Essential differences in the stochastic case can be studied analogously. 
The investigation into how the algebraic properties of $\delta$ and $\sigma$
influence the random law of the $(\delta,\sigma)$-SLE seems very
promising, in particular, their effect on the properties of the highest weight
representation of $\mathcal{U}[\delta,\sigma]$, which plays an important
role in the SLE/CFT connection discussed in 
Section \ref{Section: Coupling between SLE and GFF}.

We use the transform 
\eqref{Formula: T_c circ V_s^1/2}
to fix 
\begin{equation}
	x_5 = \frac{1}{3\kappa},\quad 
	x_4 = \frac{\nu}{\kappa^{\frac12}}
\end{equation}
according to the definitions
\eqref{Formula: kappa := ...}
and
\eqref{Formula: nu := ...}.
The remaining part of the classification of essentially different 
$(\delta,\sigma)$-SLEs is the same as in Section 
\ref{Section: The str of the family of essentially diff slit L chains}.
We only remark that there is one more special property of the chordal case: the
chordal stochastic equation is the only one invariant with respect to
$\mathscr{P}$.

\section{($\delta,\sigma$)-L\"owner equation and representation theory}
\label{Section: Loewner equation and representation theory}

In the theory of the Lie group-algebra correspondence and its
representations (see, e.g., 
\cite{Knapp1996}),
the following construction is considered.
Let $\mathcal{M}$ be a manifold, and let 
$\{v^i\}_{i=1,2,\dotso n}$ 
be a collection of vector fields on it such that the initial value problem
\begin{equation}
 \dot  H_t[v^i] = v^i \circ G_t,\quad G_0=\id 
\end{equation}
induces flows of automorphisms 
$H_t[v^i]:\mathcal{M}\map\mathcal{M}$ 
for
$t\in(-\infty,+\infty)$. 
Consider the Lie algebra 
$\mathcal{U}\{v^1,v^2,\dotso v^n\}$ 
induced by the collection of vector fields 
$\{v^i\}_{i=1,2,\dotso n}$ 
with respect to the commutator $[\cdot,\cdot]$. 
The linear space of this algebra consists of all linear combinations of
$v^1,v^2,\dotso v^n$ and all possible commutators. 
It is a subspace of the space of all vector fields on
$\mathcal{M}$.
Consider also the Lie group generated by all automorphisms $H_t[v^i]$, 
$t\in(-\infty,+\infty)$ with respect to composition. 
The Lie group-algebra theory establishes a connection between this group and
algebra. As a linear space,  
$\mathrm{U}\{v^1,v^2,\dotso v^n\}$
can be finite- or infinite-dimensional.

The Lie algebra $\mathcal{U}[\delta,\sigma]$ and the 
semigroup of endomorphisms or inverse endomorphisms 
$\mathscr{G}[\delta,\sigma]$
on $\Dc$ 
introduced in Chapter
\ref{Chapter: Slit Loewner equation and its stochastic version}
can be considered in the spirit of Lie theory. The author is not aware of any
analog of this theory for Lie semigroups, and we do not attempt to construct such a
theory here.

The Lie algebra 
$\mathcal{U}_{\leq 1}$ 
(see \eqref{Formula: Witt_<=1 := ...}) 
can be associated with the semigroup $\mathcal{G}$ of all conformal
endomorphisms of $\Dc$. We note that any semigroup
$\mathcal{G}[\delta,\sigma]$ is a subsemigroup of $\mathcal{G}$, just as
any algebra 
$\mathcal{U}[\delta,\sigma]$ 
is a subalgebra of 
$\mathcal{U}_{\leq 1}$. 
Just as in the unitary representation theory of Lie groups,
we can expect that for a Lie algebra representation there is a corresponding
group representation connected with the exponential map. 

% ...
%
%
%
% ...

We note that the algebra $\mathcal{U}_{\leq 1}$ is a
subalgebra of the \emph{Virasoro algebra}. On the other hand, a Virasoro algebra
representation is one of the key ingredients of CFT. This is one of the two
approaches to the SLE/CFT relation considered in this monograph.

A peculiar property of the classical cases is that the map $G_t$ has fixed
points. One can consider the Taylor expansion of $G_t^{\psi}(z)$ at these
points, and define a group of germs instead of the semigroup. Michel Bauer and Denis Bernard used
this in their series of papers (see, for example, \cite{Bauer2004b})
and considered an infinite-dimensional group representation, see also 
\cite{Kytola2007}. 
We generalize this construction slightly by considering non-classical
SLEs and semigroups.

Consider the universal enveloping algebra 
(see, e.g., \cite{Knapp1996})
$\widehat{\mathcal{U}}_{\leq 1}$ 
of 
$\mathcal{U}_{\leq 1}$. 
According to the Poincaré–Birkhoff–Witt theorem
it is a span of finite formal ordered
products of the vectors in $\mathcal{U}_{\leq 1}$
\begin{equation}\begin{split}
	&\mathcal{U}_{\leq 1} = 
	\mathrm{span} \{
	\ell_{-N}^{n_{-N}} \dotso \hat\ell_{-2}^{n_{-2}} \hat\ell_{-1}^{n_{-1}}
	\hat\ell_{0}^{n_{0}} \hat\ell_{1}^{n_{1}}
	~|\\&|~ 
	n_1,n_0,n_{-1},n_{-2},\dotso n_{-N+1}= 0,1,2,\dotso
	,\quad n_{-N}=1,2,\dotso
	,\quad N=0,1,2,\dotso\},
	\label{Formula: U_<=1 = ...}
\end{split}\end{equation}
where each 
$\hat \ell_n\in\widehat{\mathcal{U}}_{\leq 1}$ 
corresponds to a basis vector of 
$\mathcal{U}_{\leq 1}$. 
We use the hat 
` $\widehat{~}$ ' 
to distinguish vector fields on $\Dc$ from elements of the universal
enveloping algebra, and similarly for the corresponding algebras.
Following physics terminology, we call elements of the universal
enveloping algebra \emph{operators}
\index{operator}
because they correspond to linear maps on $\widehat{\mathcal{U}}$.
The formal product of operators
$\hat \ell_n$ 
possesses the property
\begin{equation}
	\hat\ell_n \hat\ell_m - \hat\ell_m \hat\ell_n = (m-n) \hat\ell_{n+m} 
\end{equation}
due to 
\eqref{Formula: [l_n,l_m] = (m-n) l_n+m}. 

Our purpose is to construct a representation of the semigroup
$\mathscr{G}[\delta,\sigma]$. To this end, we have to consider infinite linear
combinations in
\eqref{Formula: U_<=1 = ...},
rather than finite ones as above. We denote by 
$\overline{\widehat{\mathcal{U}}}_{\leq 1}$
the linear space of formal infinite linear combinations of the basis vectors
from $\widehat{\mathcal{U}}_{\leq 1}$. However,  
the algebraic structure of 
$\widehat{\mathcal{U}}_{\leq 1}$
cannot be extended to 
$\overline{\widehat{\mathcal{U}}}_{\leq 1}$ 
because infinite divergent sums appear as coefficients of basis vectors after
multiplication. For example, the square of
\begin{equation}
	\hat\ell_{-1}\hat\ell_{1} + \hat\ell_{-1}^2\hat\ell_{1}^2 +
	\hat\ell_{-1}^3\hat\ell_{1}^3 +
	\dotso\in\overline{\widehat{\mathcal{U}}}_{\leq 1}
\end{equation}
is not well-defined because, in the formal series for the square, the
coefficient of $\hat \ell_0$
is an infinite divergent sum.

However, in some cases, like 
$\mathrm{Lie}\{\ell_{-2},\ell_{-1}\}$,
it is possible due to the grading property.
Let 
$\widehat{\mathcal{U}}_{\leq 1}^n$, 
$n\in\mathbb{Z}$ 
be a subset of 
$\widehat{\mathcal{U}}_{\leq 1}$
defined by
\begin{equation}\begin{split}
 	\widehat{\mathcal{U}}_{\leq 1}^n :=& 
	\mathrm{span} \{
	\ell_{-N}^{n_{-N}}\dotso \hat\ell_{-2}^{n_{-2}} \hat\ell_{-1}^{n_{-1}}
	\hat\ell_{0}^{n_{0}} \hat\ell_{1}^{n_{1}}~|\\
	|& n_1,n_0,n_{-1},n_{-2},\dotso,n_{-N} = 0,1,2,\dotso
	,\quad n_{-N}=1,2,\dotso
	\wedge \\ \wedge &
	n_1+n_0+n_{-1}+n_{-2}+\dotso+n_{-N} = n,\quad
	N=0,1,2,\dotso
	\}.
	\label{Formula: hat U^n_<=1 := ...}
\end{split}\end{equation}
Thus,
\begin{equation}
	\widehat{\mathcal{U}}_{\leq 1} 
	= \bigcup\limits_{n\in\mathbb{Z}} \widehat{\mathcal{U}}_{\leq 1}^n
\end{equation}
and
\begin{equation}
	\left[\widehat{\mathcal{U}}_{\leq 1}^n
	,\widehat{\mathcal{U}}_{\leq 1}^m\right] 
	= \widehat{\mathcal{U}}_{\leq 1}^{n+m},\quad n,m\in\mathbb{Z}.
	\label{Formula: [U,U] = U}
\end{equation}

We now remark that the restriction of 
$\widehat{\mathcal{U}}_{\leq 1}^n$ 
to
$\widehat{\mathcal{U}}[\ell_{-2},\ell_{-1}]$
is a finite-dimensional space for 
$n \leq-1$ and a zero-dimensional space for $n\geq 0$. Hence, the product in   
$\overline{\widehat{\mathcal{U}}}[\ell_{-2},\ell_{-1}]$
(formal infinite linear combinations of basis vectors from
$\widehat{\mathcal{U}}[\ell_{-2},\ell_{-1}]$)
is well-defined because each coefficient in the product is given by a 
finite sum.


This case coincides with the chordal case because
$\mathcal{U}[\ell_{-2},\ell_{-1}]=\mathcal{U}[\delta_c,\sigma_c]$,
see Section 
\ref{Section: Chordal Loewner equation}.
The grading property can be generalized for the dipolar and the radial cases as
well. To this end, we can consider the basis of 
$\mathcal{U}[\delta_d,\sigma_d]$ 
for the dipolar case, see Section 
\ref{Section: Dipolar Loewner equation},
defined by
\begin{equation}
	\ell_{d,n} := \ell_n - \ell_{n+2},\quad n=-1,-2,-3,\dotso,
\end{equation}
 and use the decomposition
\begin{equation}\begin{split}
	&\widehat{\mathcal{U}}[\delta_d,\sigma_d] = 
	\bigcup\limits_{n=-1,-2,\dotso} 
	\widehat{\mathcal{U}}^n_d
\end{split}\end{equation}
into finite-dimensional subspaces 
\begin{equation}\begin{split}
 	\widehat{\mathcal{U}}_{d}^n :=& 
	\mathrm{span} \{
	\ell_{d,-N}^{n_{-N}}\dotso \hat \ell_{d,-2}^{n_{-2}}
	\hat \ell_{d,-1}^{n_{-1}} \hat \ell_{d,0}^{n_{0}}
	\hat \ell_{d,1}^{n_{1}}~|\\
	|& n_1,n_0,n_{-1},n_{-2},\dotso,n_{-N} = 0,1,2
	,\dotso \wedge n_1+n_0+n_{-1}+n_{-2}+\dotso+n_{-N} = n,\\
	& N=0,1,2,\dotso
	\}.
	\label{Formula: hat U^n_d := ...}
\end{split}\end{equation}
They satisfy a weaker version of 
\eqref{Formula: [U,U] = U},
namely, 
\begin{equation}\begin{split}
	\left[\widehat{\mathcal{U}}_d^n,\widehat{\mathcal{U}}_d^m \right]\subset
	\bigcup\limits_{i=-1,-2,\dotso,n+m} 
	\widehat{\mathcal{U}}^i_d.
\end{split}\end{equation}
For the radial case, see Section
\ref{Section: Radial Loewner equation},
we can apply an analogous method and obtain the basis
\begin{equation}
	\ell_{r,n} := \ell_n + \ell_{n+2},\quad n=-1,-2,-3,\dotso.
\end{equation}
Let the operators $\hat \delta$ and $\hat \sigma$ be the images of the vector fields
$\delta$ and $\sigma$ in $\widehat{\mathcal{U}}_{\leq 1}$:
\begin{equation}\begin{split}
	&\hat \delta := \delta_{-2} \hat \ell_{-2} + \delta_{-1} \hat \ell_{-1} +
	\delta_{0} \hat \ell_{0} + \delta_{1} \hat \ell_{1},\\ 
	&\hat \sigma := \sigma_{-1} \hat \ell_{-1} +
	\sigma_{0} \hat \ell_{0} + \sigma_{1} \hat \ell_{1} 
	\label{Formula: hat delta = ... , hat sigma = ...}
\end{split}\end{equation}
according to 
\eqref{Formula: delta and sigma for Lowner chain}.
Consider the initial value problem 
\begin{equation}
	\dot {\hat {G}}_t = 
	{\hat {G}}_t \hat \delta + {\hat {G}}_t \hat \sigma \dot u_t
	,\quad \hat G_0 = \hat I, \quad t\in[0,+\infty)
	,\quad \hat G_t\in \overline{\widehat{\mathcal{U}}}[\delta,\sigma]
	\label{Formula: dot hat G_t = hat G hat sigma + hat G hat delta dot u_t}
\end{equation}
for some given $\hat \delta$ and $\hat \sigma$ and
for any continuously differentiable function $u_t$. We denote by $\hat I$ the
identity in $\widehat{\mathcal{U}}$. The vectors $\hat \delta$ and $\hat \sigma$
are placed on the right-hand side of $\hat G_t$ to be in agreement with 
\cite{Bauer2004b}
and the traditional notation of quantum physics.
The grading property ensures the existence and uniqueness of the solution
because the system can be reduced to a countable collection of finite linear
subsystems. For general ($\delta,\sigma$)-SLEs the existence is a more difficult
problem. We assume the following conjecture henceforth.
\begin{conjecture}
For any $\delta$ and $\sigma$ as in 
\eqref{Formula: delta and sigma for Lowner chain}
and any continuously differentiable function 
$\{u_t\}_{t\in[0,+\infty)}$ 
there exists a unique solution of 
\eqref{Formula: dot hat G_t = hat G hat sigma + hat G hat delta dot u_t}.
\label{Conjecture: hat G_t existence}
\end{conjecture}

If we denote the solution with the initial condition 
$\hat G_{0}=\hat I$
by $\hat G_{t}$, we obtain the same properties as for $G_{t}$ from Section 
\ref{Section: Definition and basic properties}
as well as an analogue of Proposition
\ref{Proposition: Forward - Inverse LE connection}. 
Thus, we have obtained a map
\begin{equation}\begin{split}
	\{u_t\}_{t\in[0,+\infty)} \mapsto \{\hat G_t\}_{t\in[0,+\infty)}
	\label{Formula: G -> G}
\end{split}\end{equation}
analogous to 
\eqref{Formula: u_t -> G_t}.
Therefore, the product
$\hat G_t \hat {\tilde G}_{\tilde t}$
of two solutions 
$\hat G_t$ 
and 
$\hat {\tilde G}_{\tilde t}$ of 
\eqref{Formula: dot hat G_t = hat G hat sigma + hat G hat delta dot u_t}
is well-defined, and it is the solution with the driving function
\begin{equation}
	\tilde {\tilde u}_s =
	\begin{cases}
		u_s
		&\mbox{if } s\in[0,t] \\
		\tilde u_{s-t} + u_t
		& \mbox{if } s\in(t,+\infty).
	\end{cases}
\end{equation}

We remark that it may not be true in general (for non-classical cases) that for a
given $G_t$ in $\mathscr{G}[\delta,\sigma]$ there is a unique 
$\hat G_t$ in 
$\overline{\widehat{\mathcal{U}}}[\delta,\sigma]$, 
see the discussion in Section  
\ref{Section: General properties of Loewner chain}.
This is because the map $G_t$ depends only on the hull
$\K_t$, while the vector $\hat G_t$ 
`contains information' about how the hull was generated.
However, if $\K_t=\gamma_t$ is a simple curve, we have the representation 
\begin{equation}
	\mathscr{G}[\delta,\sigma] \map
	\overline{\widehat{\mathcal{U}}} [\delta,\sigma],\quad t\in[0,T].
\end{equation} 
of the semigroup $\mathscr{G}[\delta,\sigma]$.
 

It is straightforward to define the stochastic process 
$\{\hat G_t\}_{t\in[0,+\infty)}$ 
by
\begin{equation}
	\dS {\hat {G}}_t = 
	{\hat {G}}_t \hat \delta dt + {\hat {G}}_t \hat \sigma \dS B_t
	,\quad \hat G_0 = 1, \quad t\in[0,+\infty).
	\label{Formula: dS hat G_t = hat G hat delta dt + hat G hat sigma dS B}
\end{equation}
The existence of the solution is related to Conjecture  
\ref{Conjecture: hat G_t existence}.

\begin{theorem}
The It\^o form for $\hat G_t$ is 
\begin{equation}
	\dI {\hat {G}}_t = 
	{\hat {G}}_t \left( 
		 \hat \delta + \frac12 \hat \sigma^2 
	\right) dt 
	+ {\hat {G}}_t \hat \sigma \dI B_t
	,\quad \hat G_0 = 1, \quad t\in[0,+\infty).
	\label{Formula: dI hat G_t = hat G hat sigma dt + hat G hat delta dS B}
\end{equation}
\end{theorem}

\begin{proof}
We use the fact that the It\^o form of a vector-valued stochastic process
\eqref{Formula: dX = a X dt + b X dB}
is
\begin{equation}
	\dI X^i_t = 
	\left( \alpha^i(X_t) + \frac12 \beta^j(X_t) \frac{\de \beta^i(X_t)}{\de X_t^j}   
	\right) dt
	+ \beta^i(X_t) \dI B_t.
	\label{Formula: dI X_t = a(X)dt + 1/2 b(X)b'(X)dt + b(X)dB}
\end{equation}
Let $\{\hat U_i\}_{i=1,2,\dotso}$ 
be the basis of 
$\hat {\mathcal{U}}[\delta,\sigma]$ 
and let 
$X^i_t$ 
be the corresponding components of $\hat G_t$ such that
\begin{equation}
	\hat G_t = \sum\limits_{i=1,2,\dotso} X^i_t \hat U_i.
\end{equation}
Thus, 
\begin{equation}
	\hat G_t \hat \delta = 
	\sum\limits_{i=1,2,\dotso} \alpha^{i}(X_t) \hat	U_i,\quad 
	\hat G_t \hat \sigma = \sum\limits_{i=1,2,\dotso} \beta^{i}(X_t) \hat U_i,
\end{equation}
and the matrix
$\frac{\de}{\de X_t^j} \beta^i(X_t)$
is represented by right multiplication by $\hat \sigma$.
Thus, we obtain 
\eqref{Formula: dI hat G_t = hat G hat sigma dt + hat G hat delta dS B}
from
\eqref{Formula: dI X_t = a(X)dt + 1/2 b(X)b'(X)dt + b(X)dB}.
\end{proof}

We define an operator in 
$\widehat{\mathcal{U}}[\delta,\sigma]$,
which is an analogue of the diffusion differential operator $\mathcal{A}$
(defined below by 
\eqref{Formula: A = L + 1/2 L^2})
for the differential equation
\eqref{Formula: Slit hol stoch flow Strat},
\begin{equation}
	\hat A := \hat \delta + \frac12 \hat \sigma^2.
	\label{Formula: hat A = hat delta + 1/2 hat sigma^2}
\end{equation}

Now let $\mathcal{V}[\delta,\sigma]$ be a representation space of the algebra 
$\mathcal{U}[\delta,\sigma]$ 
and denote by $|\rangle\in\mathcal{V}[\delta,\sigma]$ 
a vector such that
\begin{equation}
	\hat A|\rangle = 0.
	\label{Formula: A |> = 0}
\end{equation}
Then 
\begin{equation}
	\dI \hat G_t |\rangle = \hat G_t \hat \sigma |\rangle \dI B_t,
\end{equation}
which means that $\hat G_t |\rangle$ is a 
$\overline{\widehat{\mathcal{U}}}[\delta,\sigma]$-valued
local martingale. 
Now we consider a possible way of defining such a space.

Let $\mathrm{Vir}$ be the Virasoro algebra, which is the only nontrivial
central extension of the Witt algebra 
\eqref{Formula: Witt := ...}. We use the basis
$\{ L_n,I\}_{n\in\mathbb{Z}}$, 
in which the Lie brackets are
\begin{equation}
	[ L_n, L_m] = (m-n) L_{n+m} + \frac{1}{12} (n^3-n) \delta_{n+m}
	 I,\quad [ L_n, I]=0
	,\quad n,m\in \mathbb{Z},
	\label{Formula: vir algebra com rules}
\end{equation}
where $\delta_{n+m}$ is the Kronecker delta ($\delta_{0}=1$ and
$\delta_{n}=0,~n\neq 0$).
We use a non-standard choice of the sign of the basis
elements $L_n$ to be compatible with 
\eqref{Formula: ell_n^H = ...}.

It is clear that $\mathcal{U}_{\leq 1}\subset\mathrm{Vir}$; hence, 
any representation of $\mathrm{Vir}$ is also a representation of 
$\mathcal{U}_{\leq 1}$.
The Virasoro algebra has a well-studied
representation theory. In particular, there is an infinite-dimensional highest
weight representation $\mathcal{V}_{h,c}$, $h,c\in\mathbb{R}$, such that
\begin{equation}\begin{split}
	&\hat L_n |\rangle = 0,\quad n=1,2,\dotso\\
	&\hat L_0 |\rangle = -h|\rangle,\\
	&\hat I |\rangle = c |\rangle.
	\label{Formula: vir highest weight rules}
\end{split}\end{equation}
Thus, the collection
\begin{equation}\begin{split}
	&\hat L_{-N}^{n_{-N}} \dotso \hat L_{-3}^{n_{-3}} \hat L_{-2}^{n_{-2}} \hat
	L_{-1}^{n_{-1}}|\rangle,\\
	&n_{-1},n_{-2},n_{-3},\dotso,n_{-N+1} = 0,1,2,\dotso
	,\quad n_{-N} = 1,2,\dotso
	,\quad N=0,1,2,\dotso. 
\end{split}\end{equation}
is a basis of $\mathcal{V}_{h,c}$ and 
$\mathcal{V}_{h,c}\approx \mathcal{U}[\ell_{-2},\ell_{-1}]$ 
as linear spaces.  

We denote the vectors from the space dual to $\mathcal{V}_{h,c}$ by
$\langle \alpha|$. In this notation, $\langle|$ is a vector from the dual
basis that corresponds to $|\rangle$, namely,
\begin{equation}\begin{split}
	&\langle||\rangle = 1,\\
	&\langle| \hat L_{-N}^{n_{-N}} \dotso \hat L_{-3}^{n_{-3}} \hat L_{-2}^{n_{-2}}
	\hat L_{-1}^{n_{-1}}|\rangle = 0,\\ 
	&n_{-1},n_{-2},n_{-3},\dotso,n_{-N+1} = 0,1,2,\dotso
	,\quad n_{-N} = 1,2,\dotso
	,\quad N=1,2,\dotso. 
\end{split}\end{equation}

The parameters $h$ and $c$ can be chosen such that 
\begin{equation}\begin{split}
	& \hat L_n\left( 2 \hat L_{-2} \pm \frac{\kappa}{2} \hat L_{-1}^2\right)
	|\rangle = 0 ,\quad n=1,2,\dotso,
	\label{Formula: L_n(2 L_-2 + k/2 L_-1^2)> = 0}
\end{split}\end{equation}
for given $\kappa>0$. The necessary and sufficient conditions are 
\begin{equation}\begin{split}
	h = -\frac{ \pm \kappa - 6}{ \pm 2 \kappa},\quad
	c = h (\pm 3 \kappa - 8).
	\label{Formula; h = h(k)}
\end{split}\end{equation}
This can be obtained by elementary calculations using 
\eqref{Formula: vir algebra com rules}
and
\eqref{Formula: vir highest weight rules}, 
see 
\cite{Bauer2003a} 
for details.
The condition 
\eqref{Formula: L_n(2 L_-2 + k/2 L_-1^2)> = 0} 
implies that the quotient space 
\begin{equation}
	\left.\mathcal{V}_{h,c} \right|_
	{\left( 2 L_{-2} \pm \frac{\kappa}{2}	L_{-1}^2\right) |\rangle \sim 0}
	\label{Formula: The quotient rep space}
\end{equation}
is a representation of $\mathrm{Vir}$ with the property 
\eqref{Formula: A |> = 0},
if we assume the upper sign in the `$\pm$' pairs for the forward case
and the lower one for the reverse case.

If, in addition, 
$\hat L_{-1} |\rangle =0$,
then the representation is trivial: $\hat L_n|\rangle =0$,
$n\in\mathbb{Z}$. This is why we assume
\begin{equation}
	\hat L_{-1} |\rangle \neq 0.
\end{equation}
It can be shown 
(see \cite{Kac1987}) 
that, in this case, the space 
\eqref{Formula: The quotient rep space}
is a non-trivial irreducible infinite-dimensional
representation of $\mathrm{Vir}$ and, consequently, of $\mathcal{U}_{\leq 1}$.

Thus, to obtain the space of local martingales, we extend the algebra
$\mathcal{U}[\delta,\sigma]$ 
to the algebra 
$\mathrm{Vir}$ 
and consider the vector space $\mathcal{V}_{h,c}$. 
From 
\eqref{Formula: hat delta = ... , hat sigma = ...}
and the rules 
\eqref{Formula: vir algebra com rules}
and
\eqref{Formula: vir highest weight rules}
we conclude that 
\begin{equation}\begin{split}
	\hat A |\rangle =& 
	\left[
		\delta_{-2} \hat L_{-2} + \frac{\sigma_{-1}^2}{2} \hat L_{-1}^2 +
 		\left(
 			\delta_{-1} - \sigma_{-1}\sigma_{0} \left( h + \frac12 \right)
	 	\right) \hat L_{-1} 
	\right.
	+\\+&
	\left.
		h	\left(
			- \delta_0 + \sigma_{-1} \sigma_{1} + \frac12 h \sigma_0^2
	 	\right)
	\right] |\rangle
\end{split}\end{equation}
Now we use the normalization from Section
\ref{Section: Equivalence and normalization SLE},
\eqref{Formula; h = h(k)},
and
\eqref{Formula: nu := ...}
to obtain
\begin{equation}\begin{split}
	&\hat A |\rangle 
	=\\=& 
	\left(
		\pm 2 \hat L_{-2} + \frac{\kappa}{2} \hat L_{-1}^2 
		\pm \left( \delta_{-1} + 3 \sigma_0 \right) \hat L_{-1} +
		\frac{(\kappa \mp 6 ) \left(\pm 4 \underline{\delta}_0+(\mp 6 + \kappa)
		\underline{\sigma}_0^2 + 4 \kappa \underline{\sigma}_1\right)}
			{8 \kappa } \right) |\rangle
	=\\=&
	\pm \left(
		2 \hat L_{-2} + \frac{\pm \kappa}{2} \hat L_{-1}^2 +
		\nu \hat L_{-1} +
		\frac{(\pm \kappa - 6 ) \left(4 \underline{\delta}_0+(\pm \kappa - 6)
		\underline{\sigma}_0^2 \pm 4 \kappa \underline{\sigma}_1\right)}
			{\pm 8 \kappa } \right) |\rangle
	\label{Formula: A|> = ...}
\end{split}\end{equation}
It is convenient to formally assume $\kappa\map-\kappa$ for the reverse 
($\delta,\sigma$)-SLE.

We conclude that $\hat G_t | \rangle$ is a local martingale taking values in
the quotient space 
\eqref{Formula: The quotient rep space}
for the forward and reverse cases if and only if $\nu=0$ (or equivalently, 
\eqref{Formula: delta_-1 - 3 sigma_0 = 0}). 
This property does not depend on the choice of the basis $\ell_n$; however, the
last coefficient in
\eqref{Formula: A|> = ...}
does. We have proved the following theorem.

\begin{theorem}
Assume that for given $\delta$ and $\sigma$ the solution $\hat G_t$ of
\eqref{Formula: dS hat G_t = hat G hat delta dt + hat G hat sigma dS B} 
exists in the highest weight representation space of the Virasoro algebra. Then,  
$\hat G_t | \rangle$ is a local martingale if and only if 
\eqref{Formula; h = h(k)}
is satisfied and
$\nu=0$.
\end{theorem}

\section*{Conclusions and Perspectives}

\begin{enumerate}[1.]
\item We introduced the parameters $x_1$, $x_2$, and $x_3$ to classify
$(\delta,\sigma)$-SLEs that are not essentially equivalent and the corresponding
algebras $\mathcal{U}[\delta,\sigma]$.
Does a better parametrization exist, and what can be said about the
properties of a $(\delta,\sigma)$-SLE, given 
$\mathcal{U}[\delta,\sigma]$, 
and vice versa?

\item It is important to prove Conjecture 
\ref{Conjecture: hat G_t existence}.

\item We showed that the SLE/CFT connection in the sense of the existence of the
related singular highest weight Virasoro algebra representation is only possible
if $\nu=0$.
On the other hand, as we see in Chapter 
\ref{Chapter: Coupling},
there is no difficulty in relating SLEs with $\nu\neq 0$ to the Gaussian free
field. It is essential to find a possible generalization of the above approach
to include the case $\nu\neq0$ as well as to understand the relation between
the Gaussian free field and representations of the Virasoro algebra. 
For example, a more general approach from
\cite{Kytola2009}
may help.
\end{enumerate}

%% file: CFT.tex
\chapter{Gaussian Free Field or \\ probabilistic approach to Euclidean
Conformal Field Theory}
\label{Chapter: CFT}

\section*{Introduction}
\addcontentsline{toc}{section}{Introduction}

This chapter is a kind of preface to the next one. Here, we introduce the 
\emph{Gaussian free field (GFF)}
and explain what we mean by 
\emph{conformal field theory (CFT)}. 
\index{CFT}
The author does not claim to present new
results in this chapter; however, he is not aware of direct analogs of some 
of the calculations and approaches. For example, the generating functional
technique is frequently used in quantum field theory (see, for
instance, \cite{Vasiliev1998})
and in probability theory, but not in mathematically rigorous applications to
CFT. The \emph{Ward identity} in the form 
\eqref{Formula: L_delta E[Phi...] = lim E[T Phi]}
is also rarely used in the literature.

We now give some historical remarks. 
In 1984, Belavin, Polyakov, and Zamolodchikov (BPZ) \cite{Belavin1984} defined a
class of conformal theories, the `minimal models', which describe some discrete
models (Ising, Potts, etc.) at criticality. The central theme is 
universality, i.e.,  the properties of a system close to the critical
point are independent of its microscopic realization. Universality classes are
characterized by a special parameter, the central charge. 
A characteristic property of such models is invariance with respect to
a change of scale. In two dimensions, this usually implies invariance with
respect to arbitrary conformal changes of coordinates. The `group' of these
transformations is infinite-dimensional and is called the conformal group. This
motivates the term conformal field theory. The classical references
are  
\cite{Francesco1997,Blumenhagen2009,Moore1989}, 
see also
\cite{Schottenloher2008} as an introduction for mathematicians.

CFT includes several ingredients, such as correlation functions, operator-valued
distributions, the highest weight representation of the Virasoro algebra,
symmetry with respect to infinitesimal conformal maps, the Ward
identities, the operator product expansion (OPE), and the vertex algebras.
However, there is still no mathematically complete formulation that covers
all of the aspects at the same time. See 
\cite{Summers2012}
for an overview of the current status. The axiomatic approach to CFT grew out of
Hilbert's sixth problem, and the Euclidean axioms were suggested by
Osterwalder and Schrader 
\cite{Osterwalder1975}.
BPZ conjectured that the behavior of the system at
criticality should be described by critical exponents identified as the highest
weights of degenerate representations of infinite-dimensional Lie algebras,
Virasoro in our case. The boundary version of this approach, BCFT, i.e., CFT on
domains with boundary, was developed by Cardy 
\cite{Cardy2004}. 

Our approach to CFT in this chapter is mathematically complete, and
the reader does not need to have any background in this area. However, we consider here
only aspects of CFT that we need to define its relation to $(\delta,\sigma)$-SLE.
We restrict ourselves to Euclidean CFT with one free
scalar (\emph{pre-pre-Schwarzian}) bosonic field on a simply connected
two-dimensional domain (BCFT), because only these models are coupled to SLE
in the next chapter. The generalization is straightforward if one starts to
consider a more general version of the coupling. The only important restriction
is that the theory is supposed to be Euclidean, two-dimensional, and bosonic.

In Section 
\ref{Section: Gaussian free field},
we define the GFF, which is a random law on the
space of distributions over a domain in the complex plane that is invariant
with respect to conformal transformations. The moments of this random variable
can be interpreted as the correlation functions $S_n$ and the analog of the
stress tensor and Ward identities are considered in Section
\ref{Section: Stress tensor and the conformal Ward identity}. 
The operator product expansion can also be derived with this approach. To define
the Hilbert space of states one can apply the so-called Wick rotation and the
field reconstruction theorem,
see \cite{Schottenloher2008},
to the correlation functions. The last step requires an additional element, namely, an anticonformal isomorphism denoted by $*$ in 
\cite{Schottenloher2008}. 

This is what we mean by CFT in this monograph.
A similar approach is used in 
\cite{Simon1979,Nelson1973}
for generic Euclidean field theories and in
\cite{SergioAlbeverio}
for string theory and Liouville theory.

\section{Test functions}
\label{Section: Test function}

We will define the Schwinger functions $S_n$ and the Gaussian free field $\Phi$
in terms of linear functionals over some space of smooth test functions defined
in what follows. 

Let $\Hc_s^{\psi}$ 
\index{$\Hc_s$}
be a linear space of real-valued smooth
functions $f\colon D^{\psi} \map \mathbb{R}$ 
in the domain 
$D^{\psi}:=\psi(\Dc)\subset\C$ 
with compact support, equipped with the topology of uniform
convergence of all derivatives on compact sets; namely, the
topology is generated by the following collection of neighborhoods of the zero
function
\begin{equation}\begin{split}
	&U^{\psi}_{K} := 
	\bigcap\limits_{n,m=0,1,2,\dotso}  
	\{f(z)\in C^{\infty}(D) \colon \supp f \subseteq K \wedge \left| \de^n \bar
	\de^m f^{\psi}(z)\right| <\varepsilon_{n,m},\quad z\in K \},\\
	&\varepsilon_{n,m}>0,\quad n,m = 0,1,2\dotso,
\end{split}\end{equation}
where $K\subset D^{\psi}$ is any compact subset of $D^{\psi}$.

We call $f^{\psi}\in\Hc_s^{\psi}$ the \emph{test functions}
\index{test function} 
and assume that they are $(1,1)$-differentials 
\begin{equation}
 f^{\tilde \psi}
 (\tilde z) = \tau'(\tilde z) \overline{\tau'(\tilde z)} f
 ^{\psi}(\tau(\tilde z)),\quad
 \tau:=\psi \circ \tilde \psi^{-1},
 \label{Formula: f^tilde psi = tau^2 f^psi(tau)}
\end{equation} 
It is straightforward to check that any transition map 
$\tau$ induces a homeomorphism between $\Hc_s^{\psi}$ and 
$\Hc_s^{\tilde \psi}$.
Therefore, we drop the index $\psi$ of $\Hc_s$ henceforth and consider 
the space $\Hc_s$ as a topological space of smooth ($1,1$)-differentials with
compact support.

We will use the space $\Hc_s$ for the forward coupling 
(see Section \ref{Section: Coupling between SLE and GFF}).
To study the reverse coupling we will use a slightly different space
\begin{equation}
	\Hc_s^* :=
	\left\{f \in\Hc_s \colon \int\limits_{D^{\psi}} f^{\psi}(z) l(dz) 
	=	0 \right\}.
\end{equation}
where $l$ is the Lebesgue measure.
\index{$\Hc_s^*$}
This space consists of smooth compactly supported functions with zero mean
value in any chart. In Sections
\ref{Section: Coupling of forward radial SLE and Dirichlet GFF},
\ref{Section: Coupling of reverse radial SLE and Neumann GFF},
and
\ref{Section: Coupling with twisted GFF},
we consider other examples of such spaces. Henceforth, we denote by $\Hc$
any of those nuclear spaces $\Hc_s$, $\Hc_s^*$, $\Hc_{s,b}$, $\Hc_{s,b}^*$, or
$\Hc_{s,b}^{\pm}$ for brevity.

An important property of $\Hc$ is \emph{nuclearity}
(see
\cite{Gelfand1964,Hida2008,Pietsch1972}), 
which is necessary and sufficient for the existence of the uniform Gaussian measure
on the dual space $\Hc'$ (the GFF).

Constructing such a uniform Gaussian measure on a finite-dimensional linear
space is a trivial problem; however, it is impossible on an
infinite-dimensional Hilbert space. On the other hand, if a space $\Hc$ is
nuclear, as $\Hc_s$ and $\Hc_s^*$ are, then the dual space $\Hc'$ admits a uniform
Gaussian measure. A general recipe, which holds not only for Gaussian measures, is
given by the following theorem.

\begin{theorem} 
(\textbf{Bochner-Minlos} \cite{Gelfand1964,SergioAlbeverio,Hida2008}) \\
 \label{Theorem: Bochner-Minols}
 Let $\Hc$ be a nuclear space, and let  
  $\hat \mu\colon \Hc\map \C $ be a (non-linear) functional. 
 Then the following three conditions  
 \begin{enumerate} [1.]
  \item $\hat \mu$ is positive definite
  \begin{equation}
   \forall \{z_1,z_2,\dotso z_n\}\in\C^n,~ 
   \forall \{f_1,f_2,\dotso f_n\}\in \Hc^n~\then
   \sum\limits_{1\leq k,l\leq n} z_k \bar z_l \hat \mu [f_k - f_l]\geq 0;	 	
  \end{equation}
  \item $\hat \mu(0)=1$;
  \item $\hat \mu$ is continuous
 \end{enumerate}  
 are satisfied if and only if there exists a unique probability measure 
 $P$ 
 on
 $(\Omega,\mathcal{F},P)$ 
 for 
 $\Omega=\Hc'$, 
 which has $\hat \mu$ as its characteristic function
 \begin{equation}
  \hat \mu [f] := 
  \int\limits_{\Phi\in \Hc'} e^{(i\Phi[f])} P_{\Phi}(d\Phi), \quad \forall
  f \in \Hc.
  \label{Formula: hat mu = int exp iPhi dPhi 2}
 \end{equation}
 The corresponding $\sigma$-algebra $\mathcal{F}_{\Phi}$ is generated by the
 cylinder sets
 \begin{equation}
  \{ F \in \Hc'\colon \quad F[f]\in B \},\quad 
  \forall f\in \Hc,\quad
  \forall \text{ Borel sets } B \text{ of } \mathbb{R}~. 
 \end{equation} 
\end{theorem}

The random law on $\Hc'$ is called uniform with respect to a bilinear functional
$B:\Hc\times\Hc\map \mathbb{R}$ 
if the characteristic function
$\hat \mu$ 
is of the form
\begin{equation}
	\hat \mu[f] = e^{-\frac12 B[f,f]},\quad f\in \Hc.
\end{equation}  
We consider the class of bilinear functionals we work with in Section
\ref{Section: Fundamental solution to the Laplace-Beltrami equation}.
First, we study the linear and bilinear functionals over $\Hc_s$ and
their transformation properties.

\section{Pre-pre-Schwarzian}
\label{Section: Pre-pre-Schwarzian}

%
%
%
%

In Section
\ref{Section: SLE preliminaries}
we defined a vector field by considering all charts and a certain 
transformation rule. Here we give an analogous construction using the
same method.

A collection of maps $\eta^{\psi}\colon \psi(\Dc)\map \mathbb{C}$, each of
which is given for a global chart map $\psi:\Dc\map \psi(\Dc)\subset
\mathbb{C}$, is called a \emph{pre-pre-Schwarzian} 
\index{pre-pre-Schwarzian}
form of order $\mu,\mu^*
\in\mathbb{C}$ if for any chart map $\tilde \psi$
\begin{equation}
 \eta^{\tilde \psi}(\tilde z) = 
 \eta^{\psi}( \tau (\tilde z) ) 
 + \mu \log \tau'(\tilde z)
 + \mu^* \log \bar \tau'(\tilde z) ,\quad
 \tau = \psi \circ \tilde\psi^{-1},\quad
 \tilde z\in\tilde D,\quad
 \forall \psi,\tilde\psi.
 \label{Formula: tilde phi = phi - chi arg}
\end{equation}
If $\eta$ is defined for one chart map, then it is automatically defined for
all other chart maps. We borrowed the term `pre-pre-Schwarzian' from
\cite{Kang2011}. 
In \cite{Sheffield2010} 
the analogous object is called an `AC
surface', and it is used in CFT as well.

Following Section
\ref{Section: SLE preliminaries}
we define
\begin{equation}
 F_* \eta^{\psi}(z) : = 
 \eta^{\psi \circ F}(z) =
 \eta^{\psi}( \left(F^{\psi}\right)^{-1}(z) ) 
 + \mu \log \left( \left(F^{\psi}\right)^{-1}\right)'(z)
 + \mu^* \log \overline{ \left( \left(F^{\psi}\right)^{-1}\right)'(z)}
 \label{Formula: G eta(z) = eta(G(z)) + mu log G'(z) + ...}
\end{equation}



We are interested in two special cases. The first one is
$\mu=\mu^*=Q/2\in\mathbb{R}$, and
\begin{equation}
 F_* \eta^{\psi}(z) = 
 \eta^{\psi}( \left(F^{\psi}\right)^{-1}(z) ) 
 + Q \log \left| \left(\left(F^{\psi}\right)^{-1}\right)'(z) \right|.
 \label{Formula: tilde eta = eta + Q log}
\end{equation}
The second is $\mu=i\chi/2,~\mu^*=-i\chi/2,~\chi\in\mathbb{R}$, and
\begin{equation}
 F_* \eta^{\psi}(z) = 
 \eta^{\psi}\left( \left(F^{\psi}\right)^{-1}(z) \right) 
 - \chi \arg \left(\left(F^{\psi}\right)^{-1}\right)'(z).
 \label{Formula: tilde eta = eta - chi arg}
\end{equation}
In both cases, $\eta$ can be chosen real in all charts. Moreover, if a
pre-pre-Schwarzian is represented by a real-valued function in all charts, it has
one of the two forms above.
 
A ($\mu,\mu^*$)-pre-pre-Schwarzian can be obtained from a vector field $v$ by
the relation
\begin{equation}
 \eta^{\psi}(z) \equiv -\mu \log v^{\psi}(z) 
 - \mu^* \log \overline{v^{\psi}(z)}. 
 \label{Formula: eta = -mu log v - mu log v}
\end{equation}
For the two special cases above we have
\begin{equation}
 \eta = -Q \log |v|
\end{equation} 
and
\begin{equation}
 \eta = \chi \arg v,
 \label{Formula: eta = chi arg v}
\end{equation} 
where we dropped the upper index $\psi$ and the argument $z$.
 
In Section 
\ref{Section: Linear functionals and change of coordinates},
we obtain the transformation rules 
(\ref{Formula: tilde eta = eta + Q log})
by taking the logarithm of a $(1,1)$-differential.
The second type of real pre-pre-Schwarzian is connected to a sort of
imaginary analog of a metric.

We can define the \emph{Lie derivative} 
\index{Lie derivative}
of $X$, where $X$ can be a
pre-pre-Schwarzian, a vector field, or even an object with a more general
transformation rule (assignment), as
\begin{equation}
 \Lc_{v} X^{\psi} (z):=
 \left. \frac{\de}{\de s} {H_s^{-1}[v]}_* X^{\psi} (z)
 \right|_{s=0},
\end{equation}
see \eqref{Formula: d H = sigma H ds} for the definition of $H_s[v]$. 
If $X$ is a vector field $w$, then 
\begin{equation}
 \Lc_v w^{\psi} = 
 v^{\psi}(z) \de w^{\psi}{}'(z) 
 - v^{\psi}{}'(z) w^{\psi}(z)
 = [v,w]^{\psi}(z),
 \label{Formula: L_v w = [v,w] := v w' - v' w},
\end{equation}
see 
\eqref{Formula: [v,w] := ...} for the definition of $[\cdot,\cdot]$. 

If $X$ is a pre-pre-Schwarzian, then
\begin{equation}
 \Lc_{v} \eta^{\psi} (z) = 
 v^{\psi}(z) \de_z \eta^{\psi}(z) + 
 \overline{v^{\psi}(z)} \de_{\bar z} \eta^{\psi}(z) 
 + \mu \,{v^{\psi}}'(z) + \mu^* \overline{{v^{\psi}}'(z)}. 
 \label{Formula: L eta = v d eta + chi dv}
\end{equation}
Here and below, we use the notation
\begin{equation}\begin{split}
 \de_z := \frac12 \left( \frac{\de}{\de x} - i \frac{\de}{\de y} \right),\quad
 \de_{\bar z} := \frac12 \left( \frac{\de}{\de x} + i \frac{\de}{\de y} \right),
\end{split}\end{equation}
and $f'=\de_z f$ for a holomorphic function $f$.

If $\mu = \mu^* = 0$, then $\eta$ is called a \emph{scalar}. 
\index{scalar}
It is remarkable
that if $\eta$ is a pre-pre-Schwarzian, then $\Lc_v \eta$ is nevertheless a scalar,
as stated in the following lemma.

\begin{lemma} 
Let $\eta$ be a pre-pre-Schwarzian of order $\mu,\mu^*$, let $v$ be a
holomorphic vector field, and let $G$ be a conformal self-map. Then
\begin{equation}
 F^{-1}_* (\Lc_v \eta)^{\psi}(z) = (\Lc_v \eta)^{\psi} \circ F^{\psi}(z)
 \label{Formula: G L phi = L phi circ G},
\end{equation}
or in the infinitesimal form
\begin{equation}
 \Lc_w (\Lc_v \eta)^{\psi}(z) = 
 w^{\psi}(z) \de_z (\Lc_v \eta)^{\psi}(z) +  
 \overline{w^{\psi}(z)} \de_{\bar z} (\Lc_v \eta)^{\psi}(z)
 \label{Formula: L L phi = w de L phi}
\end{equation}
\label{Lemma: L phi is a scalar}
\end{lemma}

We present two short proofs.
\begin{proof} (first)\\
\begin{equation}\begin{split}
 &F^{-1}_* (\Lc_v \eta)^{\psi}(z) = F^{-1}_* \left(
  v^{\psi}(z) \de_z \eta^{\psi}(z) + \overline{v^{\psi}(z)} \de_{\bar z} 
 \eta^{\psi}(z) + \mu {v^{\psi}}'(z) + \mu^* \overline{{v^{\psi}}'(z)} 
 \right)
 =\\=&
 \frac{v^{\psi}\circ F^{\psi}(z) }{{F^{\psi}}'(z)} \de_z 
  \left( \eta^{\psi}\circ F^{\psi}(z) 
  + \mu \log {F^{\psi}}'(z) + \mu^* \log \overline{{F^{\psi}}'(z)} \right) 
 +\\+&
 \overline{\frac{v^{\psi}\circ F^{\psi}(z) }{{F^{\psi}}'(z)}}
  \de_{\bar z} \left( \eta^{\psi}\circ F^{\psi}(z)
  + \mu \log {F^{\psi}}'(z) + \mu^* \log \overline{{F^{\psi}}'(z)} \right) 
 +\\+&
 \mu \de_z \frac{v^{\psi}\circ F^{\psi}(z)}{{F^{\psi}}'(z)} 
 + \mu^* \de_{\bar z} \frac{\overline{v^{\psi}\circ F^{\psi}(z)}}
 {\overline{{F^{\psi}}'(z)}} =\\=&
 v^{\psi}\circ F^{\psi}(z) (\de \eta^{\psi})\circ F^{\psi}(z) 
 + \mu \frac{v^{\psi}\circ F^{\psi}(z) }{{F^{\psi}}'(z)}
 \de_z \log {F^{\psi}}'(z)
 + 0 
 +\\+&
 \overline{v^{\psi}\circ F^{\psi}(z)} (\bar \de \eta^{\psi})\circ F^{\psi}(z) 
 + 0 + \mu^* \overline{ \frac{v^{\psi}\circ F^{\psi}(z) }{{F^{\psi}}'(z)} }
 \de_{\bar z} \log \overline{{F^{\psi}}'(z)} 
 +\\+&
 \mu {v^{\psi}}' \circ F^{\psi}(z)  
 - \mu \frac{v^{\psi}\circ F^{\psi}(z) {F^{\psi}}''(z) }{{F^{\psi}}'(z)^2} 
 + \mu^* \overline{{v^{\psi}}' \circ F^{\psi}(z)}  
 - \mu^* \frac{\overline{v^{\psi}\circ F^{\psi}(z)} \overline{{F^{\psi}}''(z)} }
 {\overline{{F^{\psi}}'(z)^2}} 
 =\\=&
 \left( v^{\psi} \de \eta^{\psi} \right) \circ F^{\psi}(z) 
 + \mu {v^{\psi}}' \circ F^{\psi}(z)
 +\\+&
 \left( \overline{v^{\psi}} \bar \de \eta^{\psi} \right) \circ F^{\psi}(z)
 + \mu^* \overline{{v^{\psi}}' \circ F^{\psi}(z)}  
 =\\=&
 \left( 
  v^{\psi} \de \eta^{\psi} + \overline{ v^{\psi}} \bar \de \eta^{\psi}
  + \mu {v^{\psi}}' + \mu^* \overline{{v^{\psi}}'} 
 \right) \circ F^{\psi}(z)
 =\\=&
 (\Lc_v \eta)^{\psi} \circ F^{\psi}(z) 
\end{split}\end{equation}
\end{proof}

\begin{proof} (second) \\
Let $\alpha$ and $\beta$ be two not necessarily holomorphic vector fields such
that 
\begin{equation}
 \eta^{\psi}(z) = 
 -\mu \log \alpha^{\psi}(z) - \mu^* \log \beta^{\psi}(z)  
\end{equation}
in some fixed chart $\psi$.
For example, we can assume that $\beta^{\psi}\equiv 0$ 
and 
$\alpha^{\psi}=\mu^{-1} e^{-\eta^{\psi}}$ 
in the chart $\psi$. Then, the transformation rules 
\eqref{Formula: tilde phi = phi - chi arg}
are satisfied.
\end{proof}

\section{Linear functionals and change of coordinates}
\label{Section: Linear functionals and change of coordinates}

In this section, we consider linear functionals over $\Hc_s$ and $\Hc$ that
transform as pre-pre-Schwarzians.

Let 
$\eta^{\psi}\in\Hc_s^{\psi}{}'$
be a linear functional over 
$\Hc_s^{\psi}$
for a given chart $\psi$.
The functional is called \emph{regular} if there exists a locally integrable function 
$\eta^{\psi}(z)$ 
such that 
\begin{equation}
 \eta^{\psi}[f] := 
 \int\limits_{\psi(\Dc)} \eta^{\psi}(z) f^{\psi}(z) l(dz),
 \label{Formula: eta[f] = int eta f dlz}
\end{equation}
where $l$ is the Lebesgue measure on $\C$. We use the brackets `$[\cdot]$' for
functionals and the parentheses `$(\cdot)$' for the corresponding functions
(kernels).

We assume that $f$ transforms according to
(\ref{Formula: f^tilde psi = tau^2 f^psi(tau)}).
If $\eta^{\psi}(z)$ is a scalar, then the number $\eta^{\psi}[f]\in\mathbb{R}$
does not depend on the choice of the chart $\psi$. Indeed, for any choice of
another chart $\tilde \psi$ we have
\begin{equation}\begin{split}
 \eta^{\tilde \psi}[f] :=& 
 \int\limits_{\tilde \psi(\Dc)} \eta^{\tilde \psi}(\tilde z) f^{\tilde \psi}(\tilde z) l(d\tilde z)=
 \int\limits_{\tilde \psi(\Dc)} \eta^{\psi}(\tau(\tilde z)) f^{\psi}(\tau(\tilde z))
 |\tau'(\tilde z)|^2  l(d\tilde z)
 =\\=&
 \int\limits_{\psi(\Dc)} \eta^{\psi}(z) f^{ \psi}(z) l(dz)
 =\eta^{\psi}[f] ,
\end{split}\end{equation}

If $\eta^{\psi}(z)$ is a pre-pre-Schwarzian, we have
\begin{equation}\begin{split}
 \eta^{\tilde \psi}[f] =& 
 \int\limits_{\tilde \psi(\Dc)} \eta^{\tilde \psi}(\tilde z) f^{\tilde \psi}(\tilde z) l(d\tilde z)
 =\\=&
 \int\limits_{\tilde \psi(\Dc)} \left(
  \eta^{\psi}(\tau(\tilde z)) 
  + \mu \log \tau'(z) + \mu^* \overline{\log \tau'(z)}
 \right)  
 f^{\psi}(\tau(\tilde z)) |\tau'(\tilde z)|^2  l(d\tilde z)
 =\\=&
 \int\limits_{\psi(\Dc)} \left(
  \eta^{\psi}(z) 
  - \mu \log \tau^{-1}{}'(z) - \mu^* \overline{\log \tau^{-1}{}'(z)}
 \right)  
 f^{\psi}(z)  l(d z)
 =\\=&
 \eta^{\psi}[f] 
 - \int\limits_{\psi(\Dc)} \left(
  \mu \log \tau^{-1}{}'(z) + \mu^* \overline{\log \tau^{-1}{}'(z)}
 \right)  
 f^{\psi}(z) l(d z)
 \label{Formula: eta^tilde psi[f] = eta psi[f] - mu int ...}
\end{split}\end{equation}
according to 
\eqref{Formula: G eta(z) = eta(G(z)) + mu log G'(z) + ...}.

If $\eta^{\psi}$ is not a regular pre-pre-Schwarzian but just a functional from 
$\Hc_s'$, we can consider the last line of
\eqref{Formula: eta^tilde psi[f] = eta psi[f] - mu int ...}
as a definition of the transformation rules for $\eta[f]$ from a chart $\psi$
to a chart $\tilde \psi$.
We denote by $\Hc_s'$ the linear space of pre-pre-Schwarzians as above.

The space $\Hc_s^*{}'$ dual to $\Hc_s^*$ is a quotient space
\begin{equation}
	\Hc_s^*{}' = \Hc_s{}'/ \sim,\quad 
	\eta \sim \eta+C \quad \Leftrightarrow \quad C\in\mathbb{R},\quad 
	\eta\in\Hc_s{}'
\end{equation}
of functionals over $\Hc_s$ defined up to a constant. 

Consider now the pushforward operation $F_*$
on $(1,1)$-differentials $f$ defined by
\begin{equation}
	(F_* f)^{\psi}(z) := 
	\left|\left(F^{\psi}\right)^{-1}{}'(z)\right|^2 
	f^{\psi} \left( \left(F^{\psi}\right)^{-1}(z) \right).
\end{equation}
For a conformal map $F\colon\Dc\setminus\K\map\Dc$ the right-hand side is
a compactly supported function, and $F_*$ is a linear homeomorphism 
$F_*\colon\Hc_s\map\Hc_s$ or $F_*\colon\Hc_s^*\map\Hc_s^*$.
However, if
$F\colon\Dc\map\Dc\setminus\K$, 
then the right-hand side is well-defined only if 
$\supp f\subset \Dc\setminus\K$.
In this work, we define $F_*$ only on the subset of $\Hc_s$ of test functions that
are supported in $\Dc\setminus\K=F^{-1}(\Dc)$.

Define the pushforward operation by
\begin{equation}\begin{split}
	&F_* \eta^{\psi} [f] 
	= \eta^{\psi\circ F} [f] 
	=\\=&
	\eta^{\psi}[F_*^{-1} f] 
	+	\int\limits_{\supp f^{\psi}} 
	\left( 
		\mu \log \left(F^{\psi}\right)^{-1}{'}(z) +
		\mu^{*}	\overline{\log \left(F^{\psi}\right)^{-1}{'}(z)}
	\right)
	f^{\psi}(z)	l(dz),\\
	&	f\in\Hc_s\colon\supp f\subset \Im(F).
	\label{Formula: F eta f= eta F-1 f}
\end{split}\end{equation}
It can be understood as a pushforward $F_*:\Hc_s'\map\Hc_s'$ 
in the dual space if
$F\colon\Dc\setminus\K\map\Dc$.
For
$F\colon\Dc\map\Dc\setminus\K$
the formula
\eqref{Formula: F eta f= eta F-1 f}
is well-defined only for functionals $\eta$ over test functions supported on 
$\Dc\setminus\K$.

Functionals over the space $\Hc_s$ are infinitely differentiable. The
Lie derivative is defined by
\begin{equation}\begin{split}
	\Lc_v \eta[f] =& 
	\left. \frac{\de}{\de s} {H_s^{-1}[v]}_* \eta^{\psi} [f] \right|_{s=0} 
	=\\=&
	-\eta^{\psi}[\Lc_{v} f] 
	+ \int\limits_{\supp f^{\psi}} 
	\left(
		\mu v^{\psi}{}'(z) + \mu^* \overline{v^{\psi}{}'(z)} 
	\right)
	f^{\psi}(z) l(dz),
\end{split}\end{equation}
where
\begin{equation}\begin{split}
	\Lc_v f^{\psi}(z) =& 
	\left. \frac{\de}{\de s} {H_s^{-1}[v]}_* f^{\psi} \right|_{s=0} 
	=\\=&
	v^{\psi}(z) \de_z f^{\psi}(z) + \overline{v^{\psi}(z)} \de_{\bar z} f^{\psi}(z)
	+ v^{\psi}{}'(z) f^{\psi}(z) + \overline{v^{\psi}{}'(z)} f^{\psi}(z).
\end{split}\end{equation}

\section{Fundamental solution to the Laplace-Beltrami equation}
\label{Section: Fundamental solution to the Laplace-Beltrami equation}

In this section, we consider bilinear functionals on $\Hc_s$ that are continuous with respect to
each argument. An important example is the Dirac functional
\begin{equation}
	\delta_{\lambda}[f,g] := 
	\int\limits_{\psi(\Dc)} f^{\psi}(z) g^{\psi}(z) 
	\frac{1}{\lambda^{\psi}(z)} l(dz),\quad 
	f,g\in\Hc_s,
	\label{Formual: delta_lambda[f,g] := ...}
\end{equation}
where $\lambda(z) l(dz)$ is a differential of some measure on $\psi(\Dc)$ that
is mutually absolutely continuous with respect to the Lebesgue measure $l(dz)$. The
Radon–Nikodym derivative $\lambda^{\psi}(z)$ transforms as a
$(1,1)$-differential:
\begin{equation}
	\lambda^{\tilde \psi}
 (\tilde z) 
 = \tau'(\tilde z) \overline{\tau'(\tilde z)} \lambda^{\psi}
 (\tau(\tilde z))
 ,\quad
 \tau:=\psi \circ \tilde \psi^{-1}.
\end{equation}
It is easy to see that the right-hand side of
\eqref{Formual: delta_lambda[f,g] := ...}
does not depend on the choice
of $\psi$.
 
We call a bilinear functional \emph{regular} if there exists a function 
$B^{\psi}(z,w)$ on $\psi^{\Dc}\times\psi^{\Dc}$ such that
\begin{equation}
	B^{\psi}[f,g] := 
	\int\limits_{\psi(\Dc)} \int\limits_{\psi(\Dc)} B^{\psi}(z,w )
	f^{\psi}(z) g^{\psi}(w)
	l(dz) l(dw),\quad 
	f,g\in\Hc_s. 
	\label{Formula: H = int H}
\end{equation}
We use the same convention about the brackets and parentheses as for the linear
functionals. We consider only scalar regular bilinear functionals and
require the transformation rules
\begin{equation}
 B^{\tilde \psi}(\tilde z,\tilde w) = 
 B^{\psi} (\tau(\tilde z),\tau(\tilde w)),\quad 
 \tau=\psi\circ \tilde \psi^{-1},\quad 
 z,w\in\tilde\psi(\Dc).
\end{equation}
Thus, the right-hand side of 
(\ref{Formula: H = int H})
does not depend on the choice 
of the chart $\psi$, and we can drop the index $\psi$ on the left-hand side.

The pushforward is defined by 
\begin{equation}
	F_* B^{\psi} (z,w) 
	= B^{\psi\circ F} (z,w)
	:= B^{\psi}( \left(F^{\psi} \right)^{-1}(z),\left(F^{\psi} \right)^{-1}(w))
	,\quad z,w\in \Im(F^{\psi}).
	\label{Formula: G B(z,w) = B(G(z),G(w))}
\end{equation}
and, for an arbitrary functional,
\begin{equation}
	F_* B^{\psi} [f,g] 
	= B^{\psi\circ F} [f,g]
	:= B^{\psi}[ F^{-1}_* f,F^{-1}_* g ]
	,\quad f,g\in\Hc_s\colon\supp f\subset \Im(F).
\end{equation}
The same remarks as in the previous section for $\eta$ remain true in this case.


Define now the Lie derivative in the same way as above
\begin{equation}\begin{split}
	&\Lc_{v} B^{\psi}(z,w)
	:=\left. \frac{\de}{\de s} 
	H_s[v]^{-1}_* B^{\psi} (z,w) \right|_{s=0} 
	=\\=&
	v^{\psi}(z)\de_z B^{\psi}(z,w) + 
	\overline {v^{\psi}(z)}\de_{\bar z} B^{\psi}(z,w) + 
	v^{\psi}(w)\de_w B^{\psi}(z,w) + 
	\overline {v^{\psi}(w)}\de_{\bar w} B^{\psi}(z,w). 
	\label{Formula: L Gamma = ...}
\end{split}\end{equation}
We remark that $\Lc_v B$ is also a scalar in both variables.
The functionals $\delta_{\lambda}$ and $B$ are both scalar and continuous with
respect to each variable.

Define the Laplace-Beltrami operator $\Delta_{\lambda}$ as
\begin{equation}
	{\Delta_{\lambda}}_1 B^{\psi}(z,w)
	:= -\frac{4}{\lambda^{\psi}(z)} \de_z \de_{\bar z} B^{\psi}(z,w),
\end{equation}   
where the subscript `$1$' means that the operator acts only on
the first argument.

%

Let a regular bilinear functional $\Gamma_{\lambda}$ be a solution to the
equation
\begin{equation}
	{\Delta_{\lambda}}_1 \Gamma_{\lambda}[f,g] 
	= 2 \pi~ \delta_{\lambda}[f,g],\quad \Gamma_{\lambda}[f,g]
	=\Gamma_{\lambda}[g,f],
	\quad f,g\in\Hc_s.
	\label{Formula: Delta Gamma = delta}
\end{equation}
We fix the boundary conditions later.
This equation is conformally invariant in the sense that if 
$\Gamma_{\lambda}^{\psi}(z,w)$ is a solution in a chart 
$\psi$, then 
\begin{equation}
 \Gamma_{\lambda}^{\psi}(\tau(\tilde z),\tau(\tilde w)) = 
 \Gamma^{\tau^{-1}\circ \psi}_{\lambda}(z,w)
\end{equation}
is a solution in the chart 
$\tau^{-1} \circ \psi$.

The solution 
$\Gamma_{\lambda}^{\psi}(z,w)$ 
is a collection of smooth and harmonic functions on 
$\psi(\Dc)\times \psi(\Dc)\setminus\{z\times w\colon z=w\}$ 
of the general form
\begin{equation}
	\Gamma_{\lambda}^{\psi}(z,w)=-\frac12\log (z-w)(\bar z-\bar w)+H^{\psi}(z,w)~,
	\label{Formula: Gamma = Log + H}
\end{equation}
where $H^{\psi}(z,w)$ is a symmetric function, harmonic with respect
to each variable, that is determined by the boundary conditions and will be
specified in what follows.

It is straightforward to verify that the function $\Gamma_{\lambda}^{\psi}(z,w)$
does not depend on the choice of $\lambda$ because the identity
\eqref{Formula: Delta Gamma = delta} 
in the integral form is
\begin{equation}\begin{split}
	&\int\limits_{\psi(\Dc)} \int\limits_{\psi(\Dc)}
		-\frac{4}{\lambda^{\psi}(z)} \de_z \de_{\bar z} 
	\Gamma^{\psi}(z,w)
	f^{\psi}(z) g^{\psi}(w) l(dz)  l(dw) 
	=\\=& 
	\int\limits_{\psi(\Dc)} f^{\psi}(z) g^{\psi}(z) 
		\frac{1}{\lambda^{\psi}(z)} l(dz).
\end{split}\end{equation}
The change $\lambda\map \tilde \lambda$ is equivalent to a change 
$f^{\psi}(z)\map \frac{\lambda^{\psi}(z)}{\tilde\lambda^{\psi}(z)} f^{\psi}(z)$. 
We will drop the subscript $\lambda$ in $\Gamma_{\lambda}$ below.
The fundamental solutions of the Laplace equation are also known as 
\emph{Green's functions} (of the free field).
\index{Green's function}

\begin{example} \textbf{Dirichlet boundary conditions.}
\label{Example: Dirichlet boundary conditions Gamma}
Let us denote by $\Gamma_D$ the solution $\Gamma$ to
(\ref{Formula: Delta Gamma = delta}) 
satisfying the zero boundary conditions, namely,
\begin{equation}
  \left. \Gamma_D^{\HH}(z,w) \right|_{z\in\mathbb{R}}=0,\quad
  \lim_{z\rightarrow \infty}\Gamma_D^{\HH}(z,w) = 0,\quad w\in \HH.
\end{equation}
Then $\Gamma_D$ takes the form
\begin{equation}
  \Gamma_D^{\HH}(z,w):=-\frac12\log
  \frac{(z-w)(\bar z-\bar w)}{(z-\bar w)(\bar z-w)},
  \label{Formula: Gamma_D = Log...} 
\end{equation}
and is symmetric with respect to all M\"obius automorphisms 
$H:\Dc\map\Dc$,
\begin{equation}
	H_* \Gamma_D = \Gamma_D
	\label{Formula: G Gamma_D = G}
\end{equation}
or
\begin{equation}
  \Lc_{\sigma} \Gamma_D(z,w) = 0~,\quad 
  \forall \text{ complete vector field } \sigma~.
  \label{Formula: Lv Gamma_D = 0}
\end{equation}
\end{example}


\begin{example} \textbf{Combined Dirichlet-Neumann boundary conditions.}
\label{Example: Gamma: Combined Dirichlet-Neumann boundary conditions}
 Let $\Gamma_{DN}$ denote the solution to
 (\ref{Formula: Delta Gamma = delta}) 
 satisfying the following boundary conditions in the strip chart 
 \begin{equation}\begin{split}
	&\left. \Gamma_{DN}^{\SSS}(z,w) \right|_{z\in\mathbb{R}}=0, \quad
  \left. \de_y\Gamma_{DN}^{\SSS}(x+i y,w) \right|_{y=\pi }=0, \quad
  x\in\mathbb{R},\\
  &\lim_{z\rightarrow \infty\wedge \Re z>0}\Gamma_{DN}^{\SSS}(z,w) = 0,\quad
  \lim_{z\rightarrow \infty\wedge \Re z<0}\Gamma_{DN}^{\SSS}(z,w) = 0,\quad
  w\in \HH.
\end{split}\end{equation}
We consider this case in Section
\ref{Section: Coupling of forward dipolar SLE and combined Dirichlet-Neumann GFF}
and the exact form of $\Gamma_{DN}$ is given by
\eqref{Formula: G_DN^S = ...}.
It is not invariant with respect to 
all M\"obius automorphisms, but it is invariant if the automorphism
preserves the points where the boundary conditions change, which are
$\pm \infty$ in the strip chart.
%
\end{example}

The set of bilinear functionals $\Hc_s^* \otimes \Hc_s^*\map \mathbb{R}$ over
$\Hc_s^*$ can be obtained from the corresponding set of functionals over
$\Hc_s$ with the equivalence relation
\begin{equation}
	B[f,g] \sim B[f,g] + \beta[f] + \beta[g],\quad \beta\in\Hc_s'.
\end{equation}
In particular, $B^{\psi}(z,w)$ is defined up to a sum 
$\beta^{\psi}(z)+\beta^{\psi}(w)$, for any locally integrable real-valued
function $\beta(z)$.

\begin{example}
\label{Example: Neumann boundary conditions Gamma}
Let $\Hc=\Hc_s^*$, and define the bilinear functional by the regular expression 
\begin{equation}
  \Gamma_N^{\HH}(z,w):=-\frac12\log
  (z-w)(\bar z-\bar w)(z-\bar w)(\bar z-w) + \beta(z) +\beta(w).
  \label{Formula: Gamma_N = Log...} 
\end{equation}
The functional $\Gamma_N$ is invariant with respect to 
M\"obius transformations $H$,
\begin{equation}
	\Gamma[f,g] = H_* \Gamma[f,g],\quad f,g\in \Hc_s^*, 
\end{equation}
just as $\Gamma_D$ is.
Equivalently,
\begin{equation}
	\Gamma(z,w) = \Gamma(H^{-1}(z),H^{-1}(w)) + \varepsilon(z)+\varepsilon(w)
\end{equation}
for some locally integrable function $\varepsilon$.

If we ignore the term $\beta(z) +\beta(w)$, the function 
$\Gamma_N^{\HH}(z,w)$ satisfies the Neumann boundary condition 
\begin{equation}\begin{split}
  &\left. \de_y \Gamma_{N}^{\HH}(x+iy,w) \right|_{y=0 }=0,\quad
  x\in (-\infty,+\infty),\quad w\in \HH.
  \label{Formula: de Gamma_N = 0}
\end{split}\end{equation}
There is a singular point at infinity, but it can be cancelled by a
proper choice of $\beta$. 
\end{example}

We will consider one more example ($\Gamma_{\text{tw},b}$) in Section
\ref{Section: Coupling with twisted GFF}.

\section{Gaussian free field}
\label{Section: Gaussian free field}

\begin{definition}
Let a nuclear space of smooth functions $\Hc$, 
a linear functional $\eta$, and a Green's functional
$\Gamma$ be given.
Set in Theorem
\ref{Theorem: Bochner-Minols}
\begin{equation}
  \hat \mu[f] := \exp{\left(-\frac12 \Gamma[f,f]+ i\eta[f] \right)},
  \quad f\in\Hc.
  \label{Formula: GFF chracteristic function 2}
\end{equation}
Then the $\Hc'$-valued random variable $\Phi$ is called the
\emph{Gaussian free field (GFF)}.
\index{Gaussian free field (GFF)}
We will denote it by 
$\Phi(\Hc,\Gamma,\eta)$.
\index{$\Phi(\Hc,\Gamma,\eta)$}
\end{definition}

For convenience, we change the definition of the characteristic function from
(\ref{Formula: hat mu = int exp iPhi dPhi 2}) to
\begin{equation}
 \hat\phi[f] := 
 \int\limits_{\Phi\in \Hc_s'} e^{\Phi[f]} P_{\Phi}(d\Phi)
 ,\quad \forall f \in \Hc,
 \label{Formula: chi = int exp iPhi dPhi}
\end{equation}
and (\ref{Formula: GFF chracteristic function 2}) changes to
\begin{equation}
	\hat\phi[f] = e^{\left(\frac12 \Gamma[f,f]+\eta[f] \right)},
	\label{Formula: GFF chracteristic function}
\end{equation}
which is possible for Gaussian measures.

The \emph{expectation} of a random variable $X[\Phi]$ ($X:\Hc' \map \C$) is
defined as
\begin{equation}
 \Ev{X}:=\int\limits_{\Phi\in \Hc_s'} X[\Phi] P(d\Phi).
\end{equation}
 
An alternative and equivalent (see, for example, \cite{Hida2008}) definition of
the GFF can be formulated as follows:

\begin{definition}	
The \emph{Gaussian free field} 
$\Phi$ 
is an 
$\Hc'$-valued 
random variable, that is, a map
$\Phi\colon \Hc\times \Omega \map \mathbb{R}$ 
(measurable on 
$\Omega$ 
and continuous linear on the nuclear space
$\Hc$),
or a measurable map 
$\Phi\colon \Omega \map \Hc'$,
such that 
$\mathrm{Law}[\Phi[f]]=N\left(\eta[f],\Gamma[f,f]^{\frac12} \right),
~ f\in \Hc$,
i.e.,
it possesses the properties
\begin{equation}
  \Ev{\Phi[f]}=\eta[f], \quad \forall~f\in \Hc,
 \end{equation}
 \begin{equation}
  \Ev{\Phi[f]\Phi[f]}=\Gamma[f,f]+\eta[f]\eta[f], \quad \forall~f\in \Hc~
\end{equation}
for a positive definite bilinear Green's functional $\Gamma$ and a linear
functional $\eta$.
\end{definition}

The random variable $\Phi$ introduced this way transforms from one chart to
another according to the pre-pre-Schwarzian rule 
\begin{equation}
 \Phi^{\tilde \psi}[f] = \Phi^{\psi}[f] - 
 \int\limits_{\psi(\Dc)} \left( \mu \log \tau^{-1}{}'(z) + \mu^* \overline{ \log \tau^{-1}{}'(z)} \right) 
 f^{\psi}(z) l(dz),\quad \tau:=\psi \circ \tilde \psi^{-1},
 \label{Formula: Phi^tilde psi[f] = Phi^psi [f] - mu in log tau f ldz}
\end{equation}
due to the corresponding property
\eqref{Formula: eta^tilde psi[f] = eta psi[f] - mu int ...}
of $\eta$.

The pushforward can also be defined by
\begin{equation}\begin{split}
	&F_* \Phi^{\psi} [f] = 
	\Phi^{\psi}[F_*^{-1} f] 
	+	\int\limits_{\supp f^{\psi}} 
	\left( 
		\mu \log \left(F^{\psi}\right)^{-1}{'}(z) +
		\mu^{*}	\overline{\log \left(F^{\psi}\right)^{-1}{'}(z)}
	\right)
	f(z)	l(dz),\\
	&	f\in\Hc\colon\supp f\subset \Im(F)
	\label{Formula: F Phi f= Phi F-1 f}
\end{split}\end{equation}
as well as the Lie derivatives
\begin{equation}\begin{split}
	\Lc_v \Phi[f] =& 
	\left. \frac{\de}{\de s} {H_s^{-1}[v]}_* \Phi^{\psi} [f] \right|_{s=0} 
	=\\=&
	-\Phi^{\psi}[\Lc_{v} f] 
	+ \int\limits_{\supp f^{\psi}} 
	\left(
		\mu v^{\psi}{}'(z) + \mu^* \overline{v^{\psi}{}'(z)} 
	\right)
	f^{\psi}(z) l(dz),
\end{split}\end{equation}



\begin{example}
Let $\Hc:=\Hc_s$,
$\Gamma := \Gamma_D$ (as in Example \ref{Example: Dirichlet boundary conditions Gamma}), and let
$\eta^{\psi}(z) := 0$ in all charts $\psi$ ($\mu=\mu^*=0$). 
Then we call $\Phi(\Hc_s,\Gamma_D,0)$ the Gaussian free field with \emph{zero
boundary conditions}.
\end{example}

\begin{example} 
Let us relax the conditions of the previous example. Let $\eta^{\psi}$ be a harmonic function in
$D^{\psi}$ continuously extendable to the boundary $\de D^{\psi}$ if the
chart map $\psi$ can be extended to $\de\Dc$.
Then 
we call $\Phi$ the Gaussian free field with the \emph{Dirichlet boundary condition}.  
\end{example}


We can define the Laplace-Beltrami operator $\Delta_{\lambda}$ over $\Phi$ 
as well as the Lie derivative by
\begin{equation}
	(\Delta_{\lambda} \Phi)[g]:=\Phi[\Delta_{\lambda} g]
	,\quad g\in\Hc,
\end{equation}
where $\Delta_{\lambda}$ on a ($1,1$)-differential is defined by
\begin{equation}
	\Delta_{\lambda} g^{\psi}(z):=-4\de_z \de_{\bar z}
	\frac{g^{\psi}(z)}{\lambda^{\psi}(z)}
\end{equation}
in any chart $\psi$. 
If $\eta$ is harmonic, the identity 
\begin{equation}\begin{split}
	&\Ev{(\Delta_{\lambda} \Phi)[g] \Phi[f_1]\Phi[f_2]\dotso\Phi[f_n]}
	=\\=&
	\sum\limits_{i=1,2,\dotso,n}
	\delta_{\lambda}[g,f_i] \,
	\Ev{\Phi[f_1]\Phi[f_2]\dotso \Phi[f_{i-1}]\Phi[f_{i+1}]...\dotso\Phi[f_n]}
	\label{Formula: E[DD Phi ] = delta}
\end{split}\end{equation}
holds. Thus, one can write heuristically
\begin{equation}
	\Delta_{\lambda} \Phi(z) = 0
	,\quad z\not\in\supp f_1 \cup \supp f_2 \cup\dotso \cup\supp f_n.   
\end{equation}

It turns out that the characteristic functional $\hat\phi$ is also a generating functional for the correlation functions.
Define the variational derivative of a functional $\nu$ as a map
$\frac{\delta}{\delta f}\colon  \nu\mapsto \frac{\delta}{\delta f} \nu$ to the set of functionals by 
\begin{equation}
 \left(\frac{\delta}{\delta f} \nu \right) [g]:= 
  \left.\frac{\de}{\de \alpha} \nu [g+\alpha f] \right|_{\alpha=0},\quad \forall f,g\in \Hc.
\end{equation}

If $\nu$ is such that $\nu[g+\alpha f]$ is an analytic function with respect to $\alpha$ in a neighborhood of $\alpha=0$ for each $f$ and $g$, like $\hat \phi$, 
it is straightforward to see that for each $g,f_1,f_2,\dotso\in\Hc$
\begin{equation}
 \nu[g] = \nu[0],\quad g\in \Hc  
 \quad \Leftrightarrow \quad
 \left( \frac{\delta }{\delta f_1} \frac{\delta }{\delta f_2} \dotso \frac{\delta }{\delta f_n}
  \nu\right)[0] = 0,\quad n=1,2,\dotso.
\end{equation}
Define the Schwinger functionals as
\begin{equation}
  S_n[f_1,f_2,\dotso ,f_n]:=
  \Ev{\Phi[f_1]\Phi[f_2]\dotso\Phi[f_n]}=
  \left( \frac{\delta }{\delta f_1} \frac{\delta }{\delta f_2} \dotso \frac{\delta }{\delta f_n}
  \hat\phi\right)[0]
 \label{Formula: S[f f...f] = d d...d mu}
\end{equation}
where $\hat\phi\colon \Hc\map \mathbb{R}$ is defined in 
(\ref{Formula: GFF chracteristic function}).

The identity 
\eqref{Formula: E[DD Phi ] = delta}
can be reformulated as 
\begin{equation}
	\Ev{(\Delta_{\lambda}\Phi)[g] e^{\Phi[f]}} =
	\left(
		\delta_{\lambda}[g,f] + \eta [\Delta_{\lambda} f]
	\right) \hat \phi[f]
	,\quad f,g\in\Hc.
\end{equation}

\section{The Schwinger functionals}
\label{Formula: The Schwinger functions}

In this section, we consider the Schwinger functionals defined by
(\ref{Formula: S[f f...f] = d d...d mu}) 
and their generating functional 
$\hat \phi$ in more detail.

For any finite collection $\{f_1,f_2,\dotso,f_n\}$ of functions from $\Hc_s$ or
$H_{\Gamma}$, the collection of random variables 
$\{ \Phi[f_1],\Phi[f_2],\dotso,\Phi[f_n] \}$ 
has a multivariate normal distribution. Thus, we have
\begin{equation}
 \Ev{\Phi[f_1]\Phi[f_2]\dotso\Phi[f_n]}=\sum\limits_{\text{partitions}} 
  \prod \limits_k \Gamma[f_{i_k},f_{j_k}]~,
\end{equation}
for 
$\eta(z)\equiv 0$, 
where the sum is taken over all partitions of the set 
$\{1,2\dotso,n\}$ 
into disjoint pairs 
$\{i_k,j_k\}$.
In particular, the expectation of the product of an odd number of fields is
identically zero.
In the general case ($\eta\not\equiv 0$), the Schwinger functionals are   
\begin{equation}
 S[f_1,f_2,\dotso,f_n]:=
 \Ev{\Phi[f_1]\Phi[f_2]\dotso\Phi[f_n]}=
 \sum\limits_{\text{partitions}} 
   \prod \limits_k \Gamma[f_{i_k},f_{j_k}] \prod \limits_l \eta[f_{i_l}],
 \label{Formula: general form of correlation functions}
\end{equation}
where the sum is taken over all partitions of the set 
$\{1,2\dotso,n\}$ 
into
disjoint non-ordered pairs 
$\{i_k,j_k\}$, 
and non-ordered single elements
$\{i_l\}$.
In particular,
\begin{equation}\begin{split}
 S_1[f_1]=&\eta[f_1],\\
 S_2[f_1,f_2]=&\Gamma[f_1,f_2] + \eta[f_1]\eta[f_2],\\
 S_3[f_1,f_2,f_3]=&\Gamma[f_1,f_2]\eta[f_3] + 
  \Gamma[f_3,f_1] \eta[f_2] + 
  \Gamma[f_2,f_3] \eta[f_1]+
  \eta[f_1] \eta[z_2] \eta[f_3],\\
 S_4[f_1,f_2,f_3,f_4]=&
 \Gamma[f_1,f_2]\Gamma[f_3,f_4]+\Gamma[f_1,f_3]\Gamma[f_2,f_4]+\Gamma[f_1,f_4]\Gamma[f_2,f_3]
 +\\+&
 \Gamma[f_1,f_2]\eta[f_3]\eta[f_4]+\Gamma[f_1,f_3]\eta[f_2]\eta[f_4]+\Gamma[f_1,f_4]\eta[f_2]\eta[f_3]
 +\\+&
 \eta[f_1]\eta[f_2]\eta[f_3]\eta[f_4].
 \label{Formula: S_1=... S_2=... S_3=... S_4=...}
\end{split}\end{equation}
Such correlation functionals 
are called the \emph{Schwinger functionals}. Their kernels 
\begin{equation}
	S_n(z_1,z_2,\dotso,s_n)
\end{equation}
are known as Schwinger functions or $n$-point
functions.
For regular functionals $\Gamma$ and $\eta$, the Schwinger functions are also
regular, but it is still reasonable to understand $S_n$ as a functional because
the derivatives are not regular. For example,
\begin{equation}
	{\Delta_{\lambda}}_1 S_2^{\psi}(z,w) = 2\pi \delta_{\lambda}(z-w).
\end{equation}

The transformation rules for $S_n$ (the behavior under the action of $G_*$) 
are rather complicated. We present here only the infinitesimal ones:
\begin{equation}\begin{split}
 &\Lc_v S_n^{\psi}[f_1,f_2,\dotso] =
 -\sum\limits_{1\leq k\leq n} S_n^{\psi}[f_1,f_2,\dotso \Lc_v f_k,\dotso, f_n] 
 -\\-&
 \sum\limits_{1\leq k\leq n}
 S_{n-1}^{\psi}[f_1,f_2,\dotso f_{k-1},f_{k+1},\dotso f_{n-1}] 
 \int\limits_{\psi(\Dc)} \left( \mu v^{\psi}{}'(z) + \mu^* \overline{ v^{\psi}{}'(z)} \right) 
 f_k^{\psi}(z) l(dz) 
\end{split}\end{equation}

We prefer to work with the characteristic functional $\hat\phi$, rather
than with $S_n$. For instance, for any endomorphism 
$F\colon \tilde \Dc \map \Dc$, 
we can define the pushforward operation 
$F_*\colon \hat\phi (\Gamma, \eta) \mapsto \hat\phi (F_*\Gamma,F_*\eta)$ 
which maps functionals 
on $\Dc$ to functionals on $\tilde \Dc$. Equivalently,
\begin{equation}
	\left( F_* \hat\phi (\Gamma, \eta) \right) [\tilde f] :=
	\hat\phi(F_*\Gamma, F_* \eta)[\tilde f],
	\quad \tilde f\in \Hc_s[\tilde \Dc].
 \label{Formula: F phi(Gamma, eta) = phi(F^-1 Gamma, F^-1 eta)}
\end{equation}
(here we need to indicate the dependence on the functionals $\Gamma$ and $\eta$).

The Lie derivative $\Lc_v$ of an arbitrary nonlinear functional 
$\rho\colon \Hc_s\map \mathbb{C}$ 
can also be defined as
\begin{equation}
 \Lc_v \rho[f]:=(\Lc_v \rho)[f]=
 \frac{\de}{\de \alpha}\left. (H_{\alpha}[v]_*^{-1} \rho)[f] \right|_{\alpha=0}
\end{equation} 
(if the partial derivative with respect to $\alpha$ 
is well-defined).

For example,
\begin{equation}
 \Lc_v \exp{\left( \rho[f] \right)} = 
 (\Lc_v \rho[f]) \exp{\left( \rho[f] \right)},
 \label{Formula: L exp nu = L nu exp nu}
\end{equation}
\begin{equation}
 \Lc_v^2 \exp{\left(\rho[f] \right)} = 
 \left(\Lc_v^2 \rho[f] + (\Lc_v \rho[f])^2 \right) \exp{\left( \rho[f] \right)}.
 \label{Formula: L2 exp nu = (L2 nu + L nu 2) exp nu}
\end{equation}
In our case $\rho[f] = \hat \phi[f]= \frac12 \Gamma[f,f] + \eta[f]$.
Recall that the Lie derivatives of $\eta$ and $\Gamma$ are defined in 
(\ref{Formula: L eta = v d eta + chi dv})
and
(\ref{Formula: L Gamma = ...})
respectively. 

The operations $G_*^{-1}$ and $\frac{\delta}{\delta f}$ or 
$\Lc$ and $\frac{\delta}{\delta f}$ 
commute. Thus, for example, we have
\begin{equation}
 \Lc_v S_n[f_1,f_2,\dotso,f_n] = 
 \left( \frac{\delta }{\delta f_1} \frac{\delta }{\delta f_2} \dotso \frac{\delta }{\delta f_n}
  \Lc_v \hat \phi \right)[0].
 \label{Formula: LS[f f...f] = d d...d L mu}
\end{equation}
We use this to deduce the martingale properties of $G_t^{-1}{}_*S_n$ and of all
their variational derivatives from the martingale property of
$G_t^{-1}{}_*\hat\phi$, which will be discussed in the next section.

\section{Gaussian Hilbert space}
\label{Section: Gaussian Hilbert space}

Define an inner product on $\Hc$ by
\begin{equation}
	(f,g)_{\Gamma} := \Gamma[f,g],\quad f,g\in \Hc.
	\label{Formula: (f,g)_Gamma = Gamma[f,g]}
\end{equation}
If the bilinear functional $\Gamma$ is positive definite, consider the
real Hilbert space
\begin{equation}
	H_{\Gamma} := \overline{\Hc}^{(\cdot,\cdot)_{\Gamma}}
\end{equation}
obtained by the completion of 
$\Hc$ 
with respect to the inner product
$(\cdot,\cdot)_{\Gamma}$. 

\begin{proposition}
Let $\eta$ be continuous with respect to
the topology of $H_{\Gamma}$, so that
$\eta^{\psi}\in H_{\Gamma}'$.
For any $f\in H_{\Gamma}$ there exists a Gaussian random variable
$\Phi[f]\in L_2(\Omega)$
with expectation $\eta[f]$ and standard deviation 
$(f,f)_{\Gamma}^{\frac12}$. Moreover, for any collection
$f_1,f_2,\dotso,f_N\in H_{\Gamma}$, the set 
$\{\Phi[f_1],\Phi[f_2],\dotso,\Phi[f_N]\}$ is a collection of jointly Gaussian
variables and
\begin{equation}
	\Ev{(\Phi[f_n]-\eta[f_n])(\Phi[f_m]-\eta[f_m])} = \Gamma[f_n,f_m]
	,\quad n,m=1,2,\dotso,N.
	\label{Formula: E[Phi Phi] = G[f,f]}
\end{equation}    
\end{proposition}

\begin{proof}
Consider a Cauchy sequence $\{g_n\}_{n=1,2,\dotso}$ such that
\begin{equation}
	g_k \xrightarrow{H_{\Gamma}} f,\quad k\map +\infty.
\end{equation}
Then, the sequence $\Phi[g_k]$, $k=1,2,\dotso$ converges in $L_2(\Omega)$
because
\begin{equation}
	\Ev{\left(\Phi[g_n] - \Phi[g_m]\right)^2} 
	= \Gamma[g_k - g_l,g_k - g_l] = (g_k-g_l,g_k-g_l)_{\Gamma}
\end{equation}
tends to zero according to the definition of the
sequence $\{g_k\}_{n=1,2,\dotso}$ we started from. 

To show
\eqref{Formula: E[Phi Phi] = G[f,f]}
we can consider the limit of characteristic functions
\begin{equation}\begin{split}
	&\Ev{e^{\Phi[g_{1,k_1}] + \Phi[g_{2,k_2}] + \dotso + \Phi[g_{N,k_N}] }}
	=\\=&
	e^{\frac12 \Gamma[g_{1,k_1} + g_{2,k_2} + \dotso + 
	 g_{N,k_N}, g_{1,k_1} + g_{2,k_2} + \dotso
	+ g_{N,k_N}] + \eta[g_{1,k_1} +
	g_{2,k_2} + \dotso + g_{N,k_N}]} \map\\
	&\map  e^{\frac12 \Gamma[f_{1}+f_{2}+\dotso+f_{N}, f_{1}+f_{2}+\dotso+f_{N}] + 
	\eta[f_{1}+f_{2}+\dotso+f_{N}]},\\
	&g_{n,k_n} \xrightarrow{H_{\Gamma}} f_n
	,\quad k_n \map + \infty
	,\quad n=1,2,\dotso N.
\end{split}\end{equation}
\end{proof}

We remark that 
$\Phi$ 
cannot be understood as a random variable that takes values in 
$H_{\Gamma}'$, 
even though 
$\Phi[f]$ 
is a random variable for any 
$f\in H_{\Gamma}$. 
However, the functionals 
$\Gamma$, $\eta$, $S_n$,
and 
$\hat\phi$ 
are defined for 
$f\in H_{\Gamma}$
as well, and the machinery of the previous section can be extended from
$f\in\Hc$ to $f\in H_{\Gamma}$.

\begin{definition}
The \emph{Gaussian Hilbert space} 
\index{Gaussian Hilbert space}
is the subspace $\Phi[H_{\Gamma}]\subset L_2(\Omega)$ consisting of the vectors 
$\Phi[f]$ for $f\in H_{\Gamma}$. 
\end{definition}

Define the convolution
\begin{equation}
	(\Gamma * f)^{\psi}(z) :=
	\int\limits_{\psi(\Dc)} \Gamma(z,w) f^{\psi}(w) l(dz)
	,\quad f\in\Hc
\end{equation}
for any chart $\psi$. Let
\begin{equation}
	\Hc^1_{\Gamma}:=\Gamma*\Hc
\end{equation}
be the image space, which is also nuclear with respect to the image topology.
It consists of functions
\begin{equation}
	h = \Gamma * f\in\Hc_{\Gamma}^1,\quad f\in\Hc
	\label{Formula: h = Gamma f}
\end{equation}
that transform as scalars
\begin{equation}
	F_* h^{\psi} = h^{\psi} \circ \left(F^{\psi}\right)^{-1} .
\end{equation}
We remark that $h$ satisfies the same boundary
conditions as $\Gamma$. In particular, if $\Gamma=\Gamma_{DN}$ is as in Example
\ref{Example: Gamma: Combined Dirichlet-Neumann boundary conditions}
then $h$ satisfies Dirichlet and Neumann boundary conditions on the corresponding
intervals of $\de\Dc$.

Let
\begin{equation}
	(h,f)
	:= \int\limits_{\psi(\Dc)} f^{\psi}(z) h^{\psi}(z) l(dz)
	,\quad f\in \Hc,\quad h\in \Hc_{\Gamma}^1. 
	\label{Formula: (,) := ...}
\end{equation}
and
\begin{equation}
	(h_1,h_2)_{\nabla}
	:= 4\int\limits_{\psi(\Dc)} \overline{\de h_1^{\psi}(z)} \de h_2^{\psi}(z)
	l(dz)
	\label{Formula: (,)_nabla := ...}
\end{equation}
be the Dirichlet inner product. We note that the form
\eqref{Formula: (,)_nabla := ...}
is chart-independent. This property can be called 
\emph{conformal invariance of the Dirichlet inner form}.
\index{conformal invariance of the Dirichlet inner form}
It is straightforward to check that
\begin{equation}
	(f_1,f_2)_{\Gamma}
	=(f_1,h_2)=(h_1,f_2)
	= \frac{1}{2\pi}\left(h_1, h_2 \right)_{\nabla}
	,\quad h_i=\Gamma * f_i
	,\quad i=1,2,
	\label{Formula: ()_G = () = ()_nabla}
\end{equation}
if $\Gamma$ satisfies either Dirichlet or Neumann boundary conditions on each
interval of the boundary. This condition is satisfied for all choices of
$\Gamma$ in this text.
From the last relation and 
\eqref{Formula: (,)_nabla := ...},
we can conclude that 
$(\cdot,\cdot)_{\Gamma}$ 
is positive definite.

The separable Hilbert space
\begin{equation}
	H^1_{\Gamma} := \overline{\Hc_s^1}^{(\cdot,\cdot)_{\nabla}}
\end{equation}
is naturally isomorphic to $H_{\Gamma}$ via the convolution map 
\eqref{Formula: h = Gamma f}. 
Like $\Hc^1_{\Gamma}$, it consists of functions satisfying the same
boundary conditions as $\Gamma$.

Due to 
\eqref{Formula: ()_G = () = ()_nabla},
it is natural to associate the Gaussian free field $\Phi$ with the space 
$H_{\Gamma}^{1}$. 
Let
$\{\xi_n\}_{n=1,2\dotso}$
be a countable collection of independent standard normally distributed random
variables, let $\{e_n\}_{n=1,2,\dotso}$ be an orthonormal basis of
$H_{\Gamma}$, let $\{e_n^1\}_{n=1,2,\dotso}$ be an orthonormal basis of
$H^1_{\Gamma}$ such that 
$(e_i^1,e_j)=\delta_{i,j}$, $i,j=1,2,\dotso$, 
and let
\begin{equation}
	f= \sum\limits_{n=1,2,\dotso} f_i e_i.
\end{equation}
Then, the GFF $\Phi(\Hc,\Gamma,0)$ can be equivalently defined by
\begin{equation}
	\Phi[f] := \sum\limits_{n=1,2,\dotso} \xi_i f_i,\quad f\in H_{\Gamma}
\end{equation}
and can be understood as the formal series
\begin{equation}
	\Phi = \sum\limits_{n=1,2,\dotso} \xi_i e_i^1.
	\label{Formula: Phi = sum xi_i e_i}
\end{equation}
We note that 
$(\Phi,\Phi)_{\nabla}$ diverges a.s. 

\begin{figure}
\centering
\includegraphics[width=\textwidth,keepaspectratio]{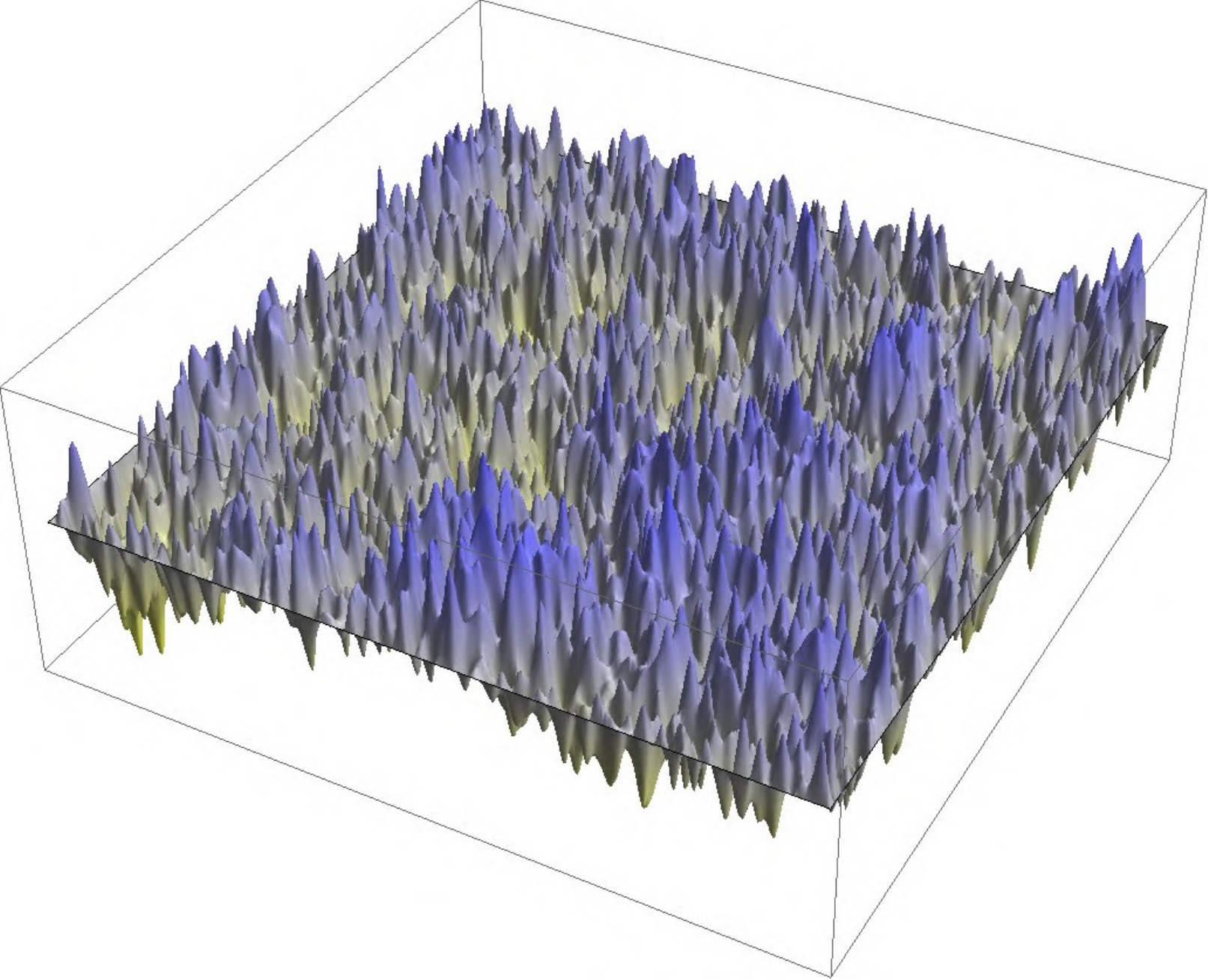}
\caption{Visualization of a sample of the Gaussian free field with
$\Gamma=\Gamma_D$ on a square $[0,\pi]\times[0,\pi]$. The picture is obtained as
the sum of the first $10^4$ terms in 
\eqref{Formula: Phi = sum xi_i e_i}. The basis is 
$e^1_{i,j}:= \frac{1}{\sqrt{i^2+j^2}}\sin (x i)\sin (x j)$, $i,j=1,2,\dotso$.
We remark that the series 
\eqref{Formula: Phi = sum xi_i e_i}
converges neither uniformly for a.a. $z\in[0,\pi]\times[0,\pi]$, nor
pointwise a.s. Thus, this picture provides only a heuristic visualization.
\label{Figure: GFF_3D}}
\end{figure}

%
%
%
%
% ...
%
%
%
%
%
%
%

\section{Stress tensor and the conformal Ward identity}
\label{Section: Stress tensor and the conformal Ward identity}

This section is not essential for understanding the next chapter. We
explain what can be meant by the
\emph{stress tensor} 
from the mathematical point of view within the framework of the probabilistic approach
to CFT. We also derive the \emph{Ward identities}. This provides a connection
between the relations 
(\ref{Formula: L eta + 1/2 L^2 eta}-\ref{Formula: L_sigma Gamma = 0})
of the coupling in the next chapter and the structure of CFT. 


We start by studying the properties of the Dirichlet form 
$(\cdot,\cdot)_{\nabla}$ 
introduced in the previous section. It is conformally invariant in the sense
that the expression on the right-hand side of
\eqref{Formula: (,)_nabla := ...}
is identical in all charts. It is also invariant with respect to M\"obius
transformations $M$:
\begin{equation}
	(M_* h,M_* k)_{\nabla} = (h,k)_{\nabla}
	,\quad h,k\in	\Hc_{\Gamma}^1.
\end{equation}

Our purpose is to study how the Dirichlet inner product 
$(\cdot,\cdot)_{\nabla}$
transforms under infinitesimal non-conformal automorphisms that are
close to conformal in the sense explained below. Such automorphisms
can be induced by a vector field $v$ with the aid of
\eqref{Formula: d H = sigma H ds}.
In order to induce an automorphism, the vector field should be tangent to the
boundary, and the automorphisms are conformal if the field is holomorphic. If
both conditions are satisfied, then $v$ is complete, see Section 
\ref{Section: SLE preliminaries}.
The vector field $\delta$ defined in
\eqref{Formula: delta with 1 simple pole}
satisfies both conditions except at a single point $a\in\Dc$, where there is a
simple pole.

We now consider the Lie derivative with respect to a not necessarily
holomorphic, but continuously differentiable vector field $v$ tangent to the boundary. Define
\begin{equation}
	\Lc_{v} (h,k)_{\nabla}: = 
	(\Lc_v h,k) + (h,\Lc_v k),\quad h,k \in\Hc_{\Gamma}^1,
\end{equation}
and note that
\begin{equation}\begin{split}
	&(f_1,\Lc_v f_2)_{\Gamma}
	=(h_1,\Lc_v f_2)
	= -(\Lc_v h_1,f_2)
	= -\frac{1}{2\pi}\left(\Lc_v h_1, h_2 \right)_{\nabla},\\
	&f_i\in \Hc
	,\quad h_i=\Gamma * f_i
	,\quad i=1,2,
	\label{Formula: (f,L f)_G = (L h,h)_nubla}
\end{split}\end{equation}
due to 
\eqref{Formula: ()_G = () = ()_nabla} and integration by parts.
We use the relation
\begin{equation}
	\de_{\bar z} (\Lc_{v} h^{\psi}(z))  
	= \Lc_{v} (\de_{\bar z} h^{\psi}(z))
	+ (\de_{\bar z} v^{\psi}(z)) \de_z h^{\psi}(z)
	+ \left(\de_{\bar z} \overline{ v^{\psi}(z)}\right) \de_{\bar z} h^{\psi}(z)
\end{equation}
to conclude that
\begin{equation}\begin{split}
	&\Lc_{v} (h, k)_{\nabla} =
	(\Lc_{v} h, k)_{\nabla} +
	( h,\Lc_{v} k)_{\nabla} 
	=\\=&
	4\int\limits_{\psi(\Dc)} \left(
		\de_{\bar z}(\Lc_{v} h^{\psi}(z)) \de_z k^{\psi}(z) + 
		\de_{\bar z} h^{\psi}(z) \de_z ( \Lc_{v}  k^{\psi}(z)) 
	\right)	
	=\\=&
	4\int\limits_{\psi(\Dc)} \left[
 	\de_z 
 	\left( v^{\psi}(z) \de_{\bar z} h^{\psi}(z) \de_z k^{\psi}(z) \right)+
 	\de_{\bar z} 
 	\left(
 		\overline{v^{\psi}(z)} \de_{\bar z} h^{\psi}(z) \de_z k^{\psi}(z)
 	\right)
 	\right.
 	+\\+& 
 	\left.
 	\de_{\bar z} v^{\psi}(z) \de_z h^{\psi}(z) \de_z k^{\psi}(z) 
 	+ \de_{z} \overline{v^{\psi}(z)} \de_{\bar z} h^{\psi}(z) \de_{\bar z}
 	k^{\psi}(z) \right] l(dz) =\\=&
	4\int\limits_{\psi(\Dc)} \left[
 	\left(\de_{\bar z} v^{\psi}(z)\right) \de_z h^{\psi}(z) \de_z k^{\psi}(z) 
 	+ \left(\de_{z} \overline{v^{\psi}(z)} \right) \de_{\bar z} h^{\psi}(z)
 	\de_{\bar z} k^{\psi}(z) \right] l(dz).
 	\label{Formula: formula 6}
\end{split}\end{equation}
The last expression equals zero if $v$ is a holomorphic vector field, which
is possible only if it is complete.
Define now a regularized version of the vector
field $\delta$ from
\eqref{Formula: delta with 1 simple pole in H}
by
\begin{equation}
	\delta^{\HH}_{\varepsilon} (z) := 
	\begin{cases} 
		\delta^{\HH}(z)  
		&\mbox{if } |z|>\varepsilon \\
		\delta_{-2} 
		\left(-\frac{2}{\varepsilon^3} |z|^2 + \frac{3}{\varepsilon^2} |z| \right)
		e^{-i\arg z} 
		& \mbox{if } |z|\leq\varepsilon
	\end{cases}
	,\quad z\in\HH,\quad \varepsilon>0.
	\label{Formula: delta_e := ...}
\end{equation}
This vector field is continuously differentiable and tangent to the entire
boundary, and 
$\delta_{\varepsilon }\xrightarrow{\varepsilon\map 0} \delta$ 	
pointwise. Besides,
\begin{equation}
	\de_{\bar z} \delta_{\varepsilon}^{\HH}(z) =
	\begin{cases} 
		0	&\mbox{if } |z|>\varepsilon \\
		3\delta_{-2} \frac{\varepsilon-|z|}{\varepsilon^3}
		& \mbox{if } |z|\leq\varepsilon
	\end{cases}	
	,\quad z\in\HH.
\end{equation} 
We substitute now $v:=\delta_{\varepsilon}$ in 
\eqref{Formula: formula 6}
to calculate
\begin{equation}\begin{split}
	\lim\limits_{\varepsilon \map 0} \Lc_{\delta_{\varepsilon}} (h,k)_{\nabla} 
	=& 2 \pi \delta_{-2}
	\left( 
		\left.
			\de_z h(z) \de_z k(z) 
		\right|_{z=0}
		+
		\left.
			\de_{\bar z} h(z) \de_{\bar z} k(z)
		\right|_{z=0} 
	\right)
	=\\=&	
	\pi \delta_{-2} 
	\left( 
		\left. \de_x h(x) \de_x k(x) \right|_{x=0}
		- \left. \de_{y} h(iy) \de_{ y} k(iy) \right|_{y=0}
	\right).
\end{split}\end{equation} 
Geometrically, this relation means that 
all variation of the conformally invariant inner product $(h,k)_{\nabla}$ is
concentrated at the point $a$, where $\delta$ has a pole.
 
Let now 
$h:=\Gamma * f$ and $k:=\Gamma * g$ for some $f,g\in\Hc$. Using
\eqref{Formula: Delta Gamma = delta}
and
\eqref{Formula: (f,L f)_G = (L h,h)_nubla},
we obtain
\begin{equation}\begin{split}
	&\lim\limits_{\varepsilon \map 0} (\Lc_{\delta_{\varepsilon}} \Gamma)[f,g]
	=-\lim\limits_{\varepsilon \map 0} 
	\left( 
		\Gamma[\Lc_{\delta_{\varepsilon}}f,g] + \Gamma[f,\Lc_{\delta_{\varepsilon}}g]
	\right) 
	=\\=&
	\frac{1}{2\pi}
	\lim\limits_{\varepsilon \map 0}
	\Lc_{\delta_{\varepsilon}} \left(\Gamma * f,\Gamma * g\right)_{\nabla} 
	=\\=&
	\frac{\delta_{-2}}{2}
	\left(
		\left. 
			\de_x (\Gamma * f)^{\HH}(x) \de_x (\Gamma * g)^{\HH}(x) 
		\right|_{x=0}
		-\left. 
			\de_y (\Gamma * f)^{\HH}(iy) \de_y (\Gamma * g)^{\HH}(iy)
		\right|_{y=0} 
	\right)
	\label{Formula: formula 7} 
\end{split}\end{equation}
This is a version of Hadamard's variational formula, which describes the
variations of the Green's function under small perturbations of the domain, for
the case when the perturbations are concentrated at the point $\psi^{\HH}(a)=0$.

We consider first the case $\Gamma=\Gamma_D$, and the case $\Gamma=\Gamma_N$
is considered afterwards. Other cases can be studied analogously; the most essential point is the boundary
conditions of $\Gamma$ at the point $a$, which are identical for $\Gamma_D$,
$\Gamma_{DN}$, and $\Gamma_{\text{tw},b}$.

In the Dirichlet case,
\begin{equation}
	\left. \de_x (\Gamma_D * f)^{\HH}(x) \right|_{x=0} = 0
	\label{Formula: d_x Gamma_D (x) = 0}
\end{equation}
and the first term in 
\eqref{Formula: formula 7}
cancels. In order to have covariant relations and a convenient connection with
the next chapter, we introduce a complete vector field $\sigma$ such that
$\sigma(a)\neq 0$ and the Lie derivative in a direction orthogonal to a
vector field $v$ by
\begin{equation}
	\Lc^{\perp}_{v} := \Lc_{i v}
\end{equation}
In particular, for a scalar $s$ we have
\begin{equation}
	\Lc^{\perp}_{v} s = \left(i \Lc^+_{v} - i \Lc^-_{v}\right) s 
\end{equation}
and, in a fixed chart $\psi$, the same formula is
\begin{equation}
	\Lc^{\perp}_{v} s^{\psi}(z) 
	= \left( i v^{\psi}(z)\de_z - i \overline{v^{\psi}(z)}\de_{\bar z} \right) 
	s^{\psi}(z). 
\end{equation}
Thus, 
\begin{equation}
	\left. \de_y (\Gamma_D * f)^{\HH}(iy) \right|_{y=0}
	= \left.
		\frac{1}{\sigma^{\HH}(0)}	\Lc_{\sigma}^{\perp} (\Gamma_D * f)^{\HH}(z) 
	\right|_{z=0} 
	= \frac{1}{\sigma_{-1}} \Lc_{\sigma}^{\perp} (\Gamma_D * f)(a)	
\end{equation}
and
\begin{equation}
	\Lc_{\delta} \Gamma[f,g] = 
	-\frac{\delta_{-2}}{2\sigma_{-1}^2} 
	\Lc_{\sigma}^{\perp} (\Gamma_D * f)(a) \Lc_{\sigma}^{\perp} (\Gamma_D * g)(a).
	\label{Formula: L_delta G = Hadamard Dirichlet}
\end{equation}

Thus, the Lie derivative 
$\Lc_{\delta} \Gamma[f,g]$ 
is a product of two linear functionals, of $f$ and of $g$. This relation plays an
important role in the coupling considered in the next chapter, see 
\eqref{Formula: Hadamard's formula}.

Let $f_{\varepsilon,b}\in\Hc$, $\varepsilon>0$ be a delta sequence
for the point $b\in\bar\Dc$, namely,
\begin{equation}\begin{split}
	&\int\limits_{\psi(\Dc)} f_{\varepsilon,b}^{\psi}(z) l(dz) = 1
	, \quad \varepsilon>0,\\
 	&\lim\limits_{\varepsilon \map 0} \supp f_{\varepsilon,b} 
 	= b \in \bar\Dc
 	\label{Formula: delta sequence for b}
\end{split}\end{equation}
We now define the 
\emph{stress tensor} 
for the Dirichlet-type GFF (or CFT) by
\index{stress tensor for Dirichlet type CFT} 
\begin{equation}\begin{split}
	T_{\varepsilon}(b) 
	:=& -\frac{c_1}{2}
	\left( 
		\Phi[\Lc_{\sigma}^{\perp} f_{\varepsilon,b}] 
		- \eta[\Lc_{\sigma}^{\perp}	f_{\varepsilon,b}] \right)^2
	+\frac{c_1}{2} \Gamma_D[\Lc_{\sigma}^{\perp} f_{\varepsilon,b},
		\Lc_{\sigma}^{\perp}f_{\varepsilon,b}] 
	+\\+& 
	c_2 \left( 
			\Phi[\Lc_{\sigma}\Lc^{\perp}_{\sigma} f_{\varepsilon,b}] -
			\eta[\Lc_{\sigma}\Lc^{\perp}_{\sigma} f_{\varepsilon,b}] 
	\right),
\end{split}\end{equation}
where $c_1$ and $c_2$ are some real coefficients that we fix below.
For each $\varepsilon$, we have the following properties:
\begin{equation}
	\Ev{T_{\varepsilon}(b)} = 0,
\end{equation}  
\begin{equation}
	\Ev{T_{\varepsilon}(b) \Phi[g]} 
	= c_2 \, \Gamma_D[\Lc_{\sigma}\Lc^{\perp}_{\sigma}f_{\varepsilon,b},g], 
\end{equation} 
\begin{equation}\begin{split}
	\Ev{T_{\varepsilon}(b) \Phi[g]\Phi[k]} 
	=& - c_1 \Gamma_D[\Lc_{\sigma} f_{\varepsilon,b},g]
	\Gamma_D[\Lc_{\sigma} f_{\varepsilon,b},k]
	+\\+& 
	c_2 \, \Gamma_D[\Lc_{\sigma}\Lc^{\perp}_{\sigma}f_{\varepsilon,b},g] \eta[k]
	+ c_2 \, \Gamma_D[\Lc_{\sigma}\Lc^{\perp}_{\sigma}f_{\varepsilon,b},k]
	\eta[g],
	\label{Formula: E[T Phi Phi] = ...}
\end{split}\end{equation} 
and more generally
\begin{equation}\begin{split}
	&\Ev{T_{\varepsilon}(b) e^{\Phi[g]}} 
	=\\=&
	\left(
		- \frac{c_1}{2} \Gamma_D[\Lc_{\sigma}^{\perp} f_{\varepsilon,b},g]
		\Gamma_D[\Lc_{\sigma}^{\perp} f_{\varepsilon,b},g]  
		+ c_2 \, \Gamma_D[\Lc_{\sigma}\Lc^{\perp}_{\sigma}f_{\varepsilon,b},g] 
	\right)
	e^{\frac12 \Gamma_D[g,g]+\eta[g]}
\end{split}\end{equation} 

Assume now that $b:=a$; then, according to 
\eqref{Formula: delta sequence for b}
we have
\begin{equation}\begin{split}
	\lim\limits_{\varepsilon\map 0} 
		\Gamma_D[\Lc_{\sigma}^{\perp} f_{\varepsilon,a},g]
	=& -\Lc_{\sigma}^{\perp} (\Gamma_D * g)(a),\\
	\lim\limits_{\varepsilon\map 0}
	\Gamma_D[\Lc_{\sigma}\Lc^{\perp}_{\sigma}f_{\varepsilon,a},g] 
	=& \Lc_{\sigma} \Lc_{\sigma}^{\perp} (\Gamma_D * g)(a).
\end{split}\end{equation}

Let us choose $\eta$ such that
\begin{equation}
	\Lc_{\delta} \eta[g] 
	= c_2 \Lc_{\sigma} \Lc_{\sigma}^{\perp} (\Gamma_D * g)(a)
	\label{Formula: L_delta eta = c L L Gamma_D}
\end{equation}
and assume
\begin{equation}
	c_1 = \frac{\delta_{-2}}{2 \sigma_{-1}^2},
\end{equation} 
then, due to 
\eqref{Formula: L_delta G = Hadamard Dirichlet},
\begin{equation}\begin{split}
	\lim\limits_{\varepsilon\map 0} \Ev{T_{\varepsilon}(a) e^{\Phi[g]}} 
	=& \left( \Lc_{\delta} \frac12 \Gamma_D[g,g] + \Lc_{\delta} \eta[g] \right)
	e^{\frac12 \Gamma_D[g,g]+\eta[g]} 
	=\\=& 
	\Lc_{\delta} e^{\frac12 \Gamma[g,g]+\eta[g]}
	,\quad f\in\Hc
	\label{Formula: lim E[T e...] = L_delta e...}
\end{split}\end{equation} 
In particular, 
\begin{equation}\begin{split}
	&\Lc_{\delta} \Ev{\Phi[f_1] \Phi[f_2]\dotso \Phi[f_n]}
	=\lim\limits_{\varepsilon\map 0} \Ev{T_{\varepsilon}(a) 
	\Phi[f_1] \Phi[f_2]\dotso \Phi[f_n]},\\ 
	&f_i\in\Hc,\quad i=1,2,\dotso,n.
	\label{Formula: L_delta E[Phi...] = lim E[T Phi]}
\end{split}\end{equation} 
This relation is known as the \emph{conformal Ward identity}.
We note that the limit $\lim_{\varepsilon\map 0} T_{\varepsilon}(a)$ 
is not well-defined. This is why it cannot be interchanged with the
expectation `$\Ev{\dotso}$'. This is not just a technical difficulty because, in
order to have a nonzero commutator of the Lie derivatives
\begin{equation}\begin{split}
	& [\Lc_{\delta}, \Lc_{\delta'}]\,
	\Ev{\Phi[f_1] \Phi[f_2]\dotso \Phi[f_n]}
	=\\=&
	\left(
		\lim_{\varepsilon\map 0} \lim_{\varepsilon'\map 0} -
		\lim_{\varepsilon'\map 0} \lim_{\varepsilon\map 0} 
	\right)
	\Ev{T_{\varepsilon}(a)T_{\varepsilon'}(a')\Phi(z_1)
	\Phi(z_2)\dotso\Phi(z_n)}
	\neq 0,
\end{split}\end{equation}
where $\delta'$ is some vector field analogous to $\delta$, but with the pole at
a different point $a'\in\de\Dc$.

In order to study the Neumann case, we consider the Green's function
$\Gamma=\Gamma_N$
(see \eqref{Formula: Gamma_N = Log...})
for the space of test functions $\Hc=\Hc_s$. The property
\eqref{Formula: ()_G = () = ()_nabla}
is satisfied because the singularity of $\Gamma_N$ at the point $z=\infty$ in
the half-plane chart is logarithmic.

Instead of 
\eqref{Formula: d_x Gamma_D (x) = 0}
we have the condition 
\begin{equation}
	\left. \de_y (\Gamma_N * f)^{\HH}(iy) \right|_{y=0} = 0
\end{equation}
and some other slightly different relations: 
\begin{equation}
	\left. \de_x (\Gamma_D * f)^{\HH}(x) \right|_{x=0}
	= \left.
		\frac{1}{\sigma^{\HH}(0)}	\Lc_{\sigma} (\Gamma_N * f)^{\HH}(z) 
	\right|_{z=0} 
	= \frac{1}{\sigma_{-1}} \Lc_{\sigma} (\Gamma_N * f)(a),	
\end{equation}
\begin{equation}
	\Lc_{\delta} \Gamma_N[f,g] = 
	\frac{\delta_{-2}}{2\sigma_{-1}^2} 
	\Lc_{\sigma} (\Gamma_N * f)(a) \Lc_{\sigma} (\Gamma_N * g)(a).
\end{equation}
We can define
\begin{equation}\begin{split}
	T_{\varepsilon}(b) 
	:=& -\frac{c_1}{2} 
	\left( 
		\Phi[\Lc_{\sigma} f_{\varepsilon,b}] 
		- \eta[\Lc_{\sigma}	f_{\varepsilon,b}] \right)^2
	+\frac{c_1}{2}
	\Gamma_N[\Lc_{\sigma}	f_{\varepsilon,b},
	\Lc_{\sigma} f_{\varepsilon,b}] 
	+\\+& 
	c_2	\left( 
			\Phi[\Lc_{\sigma}^2 f_{\varepsilon,b}] -
			\eta[\Lc_{\sigma}^2 f_{\varepsilon,b}] 
	\right),
	\label{Formula: T_e(b) = ...}
\end{split}\end{equation}
and assume
\begin{equation}
	\Lc_{\delta} \eta[g] 
	= c_2 \Lc_{\sigma}^2 (\Gamma_N * g)(a).
	\label{Formula: L_delta eta = c L L Gamma_N}
\end{equation} 
in order to obtain 
\eqref{Formula: lim E[T e...] = L_delta e...}.

The expression
\eqref{Formula: T_e(b) = ...}
is in agreement with the heuristic expression
\begin{equation}
	T(z) = - \frac12 \de_z \Phi(z) \de_z \Phi(z) + i \alpha \de_z^2 \Phi(z) 
\end{equation}
for the stress tensor; see, for example,
\cite[(1.15)]{Ginsparg1993a}. 
We note that the singular part 
$\frac{c_1}{2} \Gamma_N[\Lc_{\sigma} f_{\varepsilon,b},
	\Lc_{\sigma} f_{\varepsilon,b}]$
in
\eqref{Formula: T_e(b) = ...}
is added in order to have a finite expression in 
\eqref{Formula: E[T Phi Phi] = ...} 
after taking the limit $\varepsilon \map 0$.
It is invariant with respect to the choice of chart. However, it can also be
chosen to be just $\frac{C}{\varepsilon^{2}}$ in any chart with some constant $C$, as
in 
\cite{Kang2011} 
for example.
In this case, $T_{\varepsilon}(b)$ transforms from one chart to another in a
complicated manner (as a Schwarzian), rather than as a scalar, as it does in our
framework.
 
It is straightforward to check that the condition 
\eqref{Formula: L_delta eta = c L L Gamma_D}
or
\eqref{Formula: L_delta eta = c L L Gamma_N}
is satisfied for all cases considered in this monograph. For example, for the
case considered in Section
\ref{Section: Coupling of chordal SLE and Dirichlet GFF},
from 
\eqref{Formula: eta - chordal with drift in H}
and
\eqref{Formula: Gamma_D = Log...}
we have
\begin{equation}
	c_1 = 1,\quad c_2=- \kappa^{-\frac32}.
\end{equation}

\section{Vertex operators}

In this section, we consider how the
so-called vertex operator in CFT can be interpreted from the probabilistic point
of view.
We specify all the limits that are usually omitted in the classical literature
on CFT and avoid some heuristic notation such as 
`$\Phi(z)$' 
and 
`$\mathcal{V}(z)$`.
This differs from
\cite{Kang2011},
where a similar problem is considered. Like the previous section, this one
is not necessary for understanding what follows.

We notice first that 
\begin{equation}
	\Ev{e^{\Phi[g] } e^{\Phi[f]} } 
	= e^{\frac12 \Gamma[f,f] + \Gamma[f,g] + \frac12 \Gamma[g,g] 
	+ \eta[f] +	\eta[g]},\quad f,g\in H_{\Gamma}
\end{equation}
due to
\eqref{Formula: chi = int exp iPhi dPhi}.
Define now the random variable
\begin{equation}
	\mathcal{V}_1[g] := e^{\Phi[g] - \eta[g] - \frac12 \Gamma[g,g]}
	,\quad g\in H_{\Gamma},
\end{equation} 
and note that inserting it into an expectation is equivalent to
the change
\begin{equation}
	\eta[\cdot]\map\eta[\cdot]+\Gamma[g,\cdot].
\end{equation}
In other words,
\begin{equation}
	\Evv{\Gamma,\eta}{\mathcal{V}_1[g]\, X} =
	\Evv{\Gamma,\eta[\cdot]+\Gamma[g,\cdot]}{X}
	\label{Formula: E^eta X = E^eta+Gamma X}
\end{equation}
for any random variable $X$.

For example, let $\Gamma=\Gamma_N$, $\eta=0$, $\mu=\mu^*=0$, and assume 
$g:=-\frac12 f_{\varepsilon,b}$, 
see
\eqref{Formula: delta sequence for b}.
Consider the random variable
\begin{equation}
	\mathcal{V}_{\varepsilon}(b):=e^{-\frac12 \Phi[f_{\varepsilon,b}] 
	+\frac12 \eta[f_{\varepsilon,b}] 
	- \frac18 \Gamma[f_{\varepsilon,b},f_{\varepsilon,b}]}
	,\quad \varepsilon>0, \quad b\in\bar\Dc, 
\end{equation} 
and the limit $\varepsilon \map 0$
\begin{equation}
	\lim\limits_{\varepsilon \map 0}
	\Evv{\Gamma_N,0}{\mathcal{V}_{\varepsilon}(a)\, X} =
	\Evv{\Gamma_N,\eta}{X},\quad
	\eta^{\HH}(z)=  \log|z|.
	\label{Formula: E^0 V X = E^eta X}
\end{equation}
This combination of $\Gamma$ and $\eta$ is similar to 
the combination considered in Section
\ref{Section: Coupling of reverse chordal SLE and Neumann GFF}
for $\nu=0$.
The difference is that $\mu=\mu^*\neq 0$ for all values of $\kappa$.
Thus, it may be possible to avoid introducing a 
nonzero $\eta$ in the definition of the GFF and work with the insertion of 
$\mathcal{V}_{\varepsilon}(a)$
instead. However, we did not use this approach for the following reasons:
\begin{enumerate}
\item The term $\Gamma[g,\cdot]$ in
\eqref{Formula: E^eta X = E^eta+Gamma X} 
is a scalar, whereas we need $\eta$ to behave as a pre-pre-Schwarzian;
\item The limit in 
\eqref{Formula: E^0 V X = E^eta X}
cannot be interchanged with the expectation because
$\lim\limits_{\varepsilon \map 0} \mathcal{V}_{\varepsilon}(b)$ 
does not exist as an element of a reasonable
normed space such as $H_{\Gamma}$.
\end{enumerate}   

The random variable $\mathcal{V}_{\varepsilon}(b)$ can be thought of as an
analog of the heuristic vertex operator $e^{\alpha \Phi(z)}$ frequently used
in the CFT literature. For the singular part
$-\frac18 \Gamma[f_{\varepsilon,b},f_{\varepsilon,b}]$ 
in
\eqref{Formula: E^0 V X = E^eta X}
the same remark applies as for the singular part of $T_{\varepsilon}(b)$ at the
end of the previous section. If one uses $C\log \varepsilon$ for some
constant $C$ instead of 
$-\frac18 \Gamma[f_{\varepsilon,b},f_{\varepsilon,b}]$
in any chart, this leads to a more complicated transformation rule for 
$\mathcal{V}_{\varepsilon}$.

It is also possible to introduce a similar random variable, the insertion of
which changes the covariance $\Gamma$. Calculate first the expectation
\begin{equation}\begin{split}
	&\Ev{e^{\frac12 \Phi[g]^2} e^{\Phi[f]} } 
	= \frac{1}{\sqrt{1-\Gamma[g,g]}}e^{\frac12 \Gamma[f,f] + \eta[f] 
	+ \frac12 \frac{\left( \Gamma[f,g]+\eta[g]\right)^2 }{1-\Gamma[g,g]}},\\
	&f,g\in H_{\Gamma},\quad \Gamma[g,g]<1.
\end{split}\end{equation}
This can be done using the standard finite-dimensional machinery for Gaussian
integrals. Define the variable
\begin{equation}
	\mathcal{V}_2[g] := \sqrt{1-\Gamma[g,g]} e^{\frac12 \left(\Phi[g] -
	\eta[g]\right)^2}.
\end{equation} 
It possesses the property
\begin{equation}
	\Ev{\mathcal{V}_2[g] e^{\phi[f]}}
	= e^{\frac12 \Gamma[f,f]  
	+ \frac12 \frac{\Gamma[f,g]^2}{1-\Gamma[g,g]} + \eta[f]}.  
\end{equation}
Consequently, its insertion leads to
\begin{equation}
	\Gamma[\cdot,\cdot] \map
	\Gamma[\cdot,\cdot]+\frac{\Gamma[\cdot,g] \Gamma[\cdot,g] }{1-\Gamma[g,g]},
\end{equation}
or equivalently,
\begin{equation}
	\Evv{\Gamma,\eta}{\mathcal{V}_2[g]\, X}
	= \Evv{\Gamma[\cdot,\cdot]+\frac12
	\frac{\Gamma[\cdot,g] \Gamma[\cdot,g]}{1-\Gamma[g,g]},\eta}{X}.
\end{equation}
%
%
% ... 

%% file: Coupling.tex
\chapter{Coupling of ($\delta,\sigma$)-SLE and the GFF}
\label{Chapter: Coupling}

\section*{Introduction}
\label{Section: Coupling Introduction}

In this introduction, we slightly simplify the notation from the previous
chapters, drop the chart indices, and use the notation $\Phi(z)$ instead of
$\Phi[f]$ for the GFF.
 
The relationships between CFT and both forward and reverse forms of SLE
have several aspects. The most important are:
\begin{enumerate} [1.]
  \item For some choice of the covariance $\Gamma$ and the expectation $\eta$ of
  the GFF $\Phi$, the random laws of 
  $\Phi(z)$
  and
  $G^{-1}_t{}_*\Phi(z)$	
  (in the simplest case $\Phi(G_t(z))$)
  are identical if the GFF $\Phi(z)$ and the SLE map $G_t(z)$ are sampled independently,
  \cite{Sheffield2010}. 
  \item The CFT correlation functions $S_n(z_1,z_2,\dotso, z_n)$ induce SLE
  (local) martingales 
  $${G_t^{-1}}_* S_n(z_1,z_2,\dotso z_n),$$ 
  \cite{Bauer2004b}.
  \item Two Riemann surfaces equipped with independent random metrics 
  according to Liouville theory can be glued together along boundary
  segments in a boundary-length-preserving way (conformal welding). The
  resulting law of the interface between the two surfaces is the SLE, 
  \cite{Duplantier2011}.
  \item Some CFT correlation functions are related to the probabilities that
  the SLE slit touches the boundary, \cite{Bauer2004b,Friedrich2003a}.
  \item The heuristic Lie semigroup of conformal endomorphisms of $\Dc$, in which
  the SLE map takes values, has a highest weight representation in the CFT space
  of states. The diffusion operator $\mathcal{A}$ of the SLE differential
  equation corresponds to the null vector of this representation, which is
  singular, see Chapter 
  \ref{Chapter: Representation theory approach}
  and 
  \cite{Bauer2004b}.
\end{enumerate}

\begin{figure}[h]
\centering
	\includegraphics[keepaspectratio=true]
    	{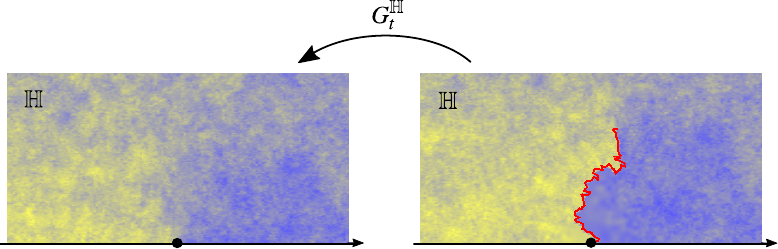}
\caption{This figure illustrates the first aspect of the coupling discussed in
the introduction to Chapter
\ref{Chapter: Coupling}. 
On the left-hand side, we present a sample of the GFF
$\Phi(\Hc_s,\Gamma_D,\chi \left( -\arg{z}+\pi/2 \right))$
(see the caption to Fig. 
\ref{Figure: GFF_3D}) in the half-plane chart. Blue corresponds to
positive values of $\Phi(z)$, and yellow corresponds to negative
values of $\Phi(z)$.
We take the coefficient $\chi$ larger than it
should be according to formula 
\eqref{Formula: chi = 2/k - k/2 intro}. This is done  
in order to make the predominance of blue near 
$\mathbb{R}^+$ 
more visible, and likewise that of yellow near 
$\mathbb{R}^-$. 
The red line on the
right-hand side is an independent sample of the chordal SLE slit for some $t>0$
and $\kappa=2$. Blue and yellow on the right-hand side correspond to
the pushforward $G_t^{-1}{}_* \Phi(z)$ of the GFF $\Phi(z)$ on the
left-hand side.
The coupling proposition states that the expectation with respect to the SLE
random law gives a sample of the GFF 
$\Phi(\Hc_s,\Gamma_D,\chi \left( -\arg{z}+\pi/2 \right))$
on the right-hand side.
\label{Figure. Coupling illustration}}
\end{figure}

The coupling proposition (the first item above) was first formulated in a talk
by Scott Sheffield in 2005. A detailed proof for the chordal SLE case is
presented in 
\cite{Sheffield2010}
for $\kappa<4$.
The SLE/CFT relation is extended from the
chordal equation to the radial one in 
\cite{Bauer2004,Kang2012a} 
and to the dipolar one in
\cite{Bauer2004a,Kang,Kanga}, 
see also
\cite{Izyurov2010} 
for the SLE-($\kappa,\rho$) case.
The fifth aspect is considered in Chapter 
\ref{Chapter: Representation theory approach}.

In this chapter, we discuss only the first two aspects. 
Thus, in our framework, we consider only the SLE/GFF coupling and reserve the
term `SLE/CFT coupling' for a broader context.
We consider the general case of $(\delta,\sigma)$-SLE and a GFF that
transforms as a pre-pre-Schwarzian
\eqref{Formula: Phi^tilde psi[f] = Phi^psi [f] - mu in log tau f ldz}
and has the covariance $\Gamma$ as in
\eqref{Formula: Gamma = Log + H}.
To motivate this assumption, let us discuss a geometric interpretation of the
coupling. 

It can be heuristically explained as follows. A properly defined zero-level line
of the GFF $\Phi(z)$ on the half-plane 
$\HH$ starting from the origin
has the same law as the chordal SLE slit if $\kappa=4$. 
This proposition can be generalized
to other values of 
$\kappa<4$, 
see \cite{Sheffield2010}.
In order to explain its geometric meaning, let us associate with $\Phi$ another
field $J(z)$ of unit vectors $|J(z)|= 1$ by 
\begin{equation}
	J(z): = e^{i \Phi(z)/\chi}
\end{equation}
for some 
$\chi\in\mathbb{R}\setminus\{0\}$.
We assume that the unit vector field transforms from the domain $\HH$ to 
$\HH\setminus \gamma$ for some $\gamma\subset\HH$ according to the rule
\begin{equation}
	J(z)\map \arg (G'(z))	J(G(z)),
\end{equation}  
where $G$ is a conformal map $G:\HH\setminus \gamma \map \HH$ (which can be
understood as a change of coordinates).
Thus, $\Phi(z)$ transforms according to the rule
\begin{equation}
 \Phi(z) \map G^{-1}_* \Phi(z) := \Phi(G(z)) 
  - \chi \, \mathrm{arg}\, {G}'(z).
 \label{Formula: Phi ->Phi - chi arg}
\end{equation}
The coupling statement is the agreement in law of   
$\Phi(G_t(z))-\chi \arg G'(z)$
and
$\Phi(z)$, 
where
\begin{equation}
	\chi = 
	\frac{2}{\sqrt{\kappa}}-\frac{\sqrt{\kappa}}{2},
	\label{Formula: chi = 2/k - k/2 intro}
\end{equation}
if
$G_t$ 
is a forward chordal SLE map. Besides, a flow line of 
$J(z)$, 
starting at the origin, has the same law as the
SLE${}_{\kappa}$ curve.

For the geometric interpretation of the coupling of the reverse SLE, assume that
$e^{\Phi(z)/Q} l(dz)$
is a random measure on
$\HH$
for some
$\kappa$-dependent constant $Q\in\mathbb{R}\setminus\{0\}$ and for the Lebesgue
measure $l$.
A Riemann surface $\Dc$ equipped with such a random measure is what is
called the Liouville model. Its relation to SLE is explained in item 3 above
if
\begin{equation}
	Q = 
	\frac{2}{\sqrt{\kappa}}+\frac{\sqrt{\kappa}}{2}.
	\label{Formula: chi = 2/k + k/2 intro}
\end{equation}
In 
\cite{Duplantier2011},
it was shown that the expectation of the measure of a subset $A\subset\HH$
\begin{equation}
	\Ev{\int\limits_{A} e^{\Phi(z)/Q} l(dz) }
\end{equation}
is invariant with respect to the conformal change of coordinates 
$G$ 
if the
GFF $\Phi$ transforms according to the rule 
\begin{equation}
	\Phi(z)\map G^{-1}_* \Phi(z):=\Phi(G(z))+Q\,\log|G{}'(z)|.
	\label{Formula: Phi -> Phi + gamma log}
\end{equation}

Both of these interpretations of the coupling are related to the fact that the
stochastic process
\begin{equation}
	G_t^{-1}{}_* \Ev{\Phi(z_1)\Phi(z_2)\dotso\Phi(z_n)}
\end{equation}
is a (local) martingale if we assume the rules
\eqref{Formula: Phi ->Phi - chi arg}
and
\eqref{Formula: Phi -> Phi + gamma log}
for the forward and reverse SLEs, respectively.
We show in Chapter
\ref{Chapter: Slit Loewner equation and its stochastic version}
that the slit of  
$(\delta,\sigma)$-SLE
has the same local behavior as the chordal SLE. Thus, it is reasonable 
to expect that the general form of 
$(\delta,\sigma)$-SLE
can be coupled with a GFF with the same local covariance. Namely, we assume
for $\Gamma^{\psi}(z,w)$ 
the form
\eqref{Formula: Gamma = Log + H}
which corresponds to a logarithmic singularity at $z=w$ plus a general
harmonic part.

Before describing the organization of the chapter, we discuss a technical question regarding the
definition of the pushforward operation in the forward case. As discussed in
Section
\ref{Section: Linear functionals and change of coordinates},
the pushforward 
$G_t^{-1}{}_*\Phi[f]$
of 
$\Phi[f]$
is well-defined if
$\supp f \in \Im[G^{-1}_t]$. 
This is not a restriction for the reverse SLE 
$\{G_t\}_{t\in[0,+\infty)}$.
To handle this for the forward SLE, we introduce a stopping time 
$T[f]$, for which the hull $\K_t$ touches the support of $f$ for the first time
(see \eqref{Formula: T[f] = ...} below) and
we consider a stopped process 
$\{G_{t\wedge T[f]}\}_{t\in[0,+\infty)}$. This approach is also used in 
\cite{Izyurov2010}.
The most important property of the process for us
$\{{G_{t\wedge T[f]}^{-1}}_*\Phi[f]\}_{t\in[0,+\infty)}$ 
is that it is a local martingale.
A stopped local martingale is also a local martingale. This is why
stopping 
$\{G_t\}_{t\in[0,+\infty)}$
does not change our results. However, we lose some information,
which makes the coupling statement weaker than one might
expect. An approach to avoid the stopping for the
case when the hull $\K_t$ is a simple curve ($\kappa\leq4$)
is considered in
\cite{Sheffield2010,Schramm2008}.
For the case $\kappa>4$, it is a more complicated and interesting problem.
Once the pushforward $G_t^{-1}{}_*\Phi[f]$ is defined for a larger stopping
time than
\eqref{Formula: T[f] = ...}
such that the propositions in Section 
\ref{Section: Some technical propositions}
are satisfied, the coupling given in the key Theorem 
\ref{Theorem: The coupling theorem}
below can be extended. In the last section, Section 
\ref{Section: Alternative definition og G_* eta},
we consider, without proof, a possible alternative to the approach of
\cite{Sheffield2010,Schramm2008}
for defining 
$G_t^{-1}{}_*\Phi[f]$ 
for $t\in[0,+\infty)$.

We remark that the stopping is not essential if $\kappa\leq 4$ in the
following sense. If the Lebesgue measure of the support of the test function
tends to zero, then 
$P(\{T[f]<+\infty\})$ 
also tends to zero. Together with the regularity of the Schwinger
functions $S_n$, we can conclude that 
\begin{equation}
	\{G_{t}^{-1}{}_* S_n(z_1,z_2,\dotso,z_n)\}_{t\in[0,+\infty)}
	,\quad z_1,z_2,\dots,z_n\in\Dc,
\end{equation}
is a local martingale a.s. for $\kappa\leq 4$.

This chapter is organized as follows. After some technical remarks in Section
\ref{Section: Some technical propositions}, 
we prove Theorem 
\ref{Theorem: The coupling theorem} 
in Section 
\ref{Section: Coupling between SLE and GFF}. 
It states that
for $(\delta,\sigma$)-SLE, the pushforward 
${G_t^{-1}}_* S_n$ 
is a local martingale if and only if the system of partial differential
equations (\ref{Formula: L eta + 1/2 L^2 eta})-(\ref{Formula: L_sigma Gamma = 0})
for $\delta$, $\sigma$, $\Gamma$, and $\eta$  
is satisfied. The second equation (\ref{Formula: Hadamard's formula}) is known
as Hadamard's formula, and the third (\ref{Formula: L_sigma Gamma = 0}) states
that the covariance $\Gamma$ must be invariant with respect to the M\"obius
automorphisms generated by the vector field $\sigma$. The same theorem also
states that both the local martingale property of
$\{G_t\}_{t\in[0,+\infty)}$
and the system of equations are equivalent to a \emph{local coupling}, which is a
weaker version of the coupling from \cite{Sheffield2010} discussed above. We
expect that the local coupling leads to the
same property, namely, that the flow lines of $e^{i\Phi(z)/\chi}$ agree in law with the
($\delta,\sigma$)-SLE curves, as well as to the connection with the welding of
Liouville surfaces.

A general solution to the system 
(\ref{Formula: L eta + 1/2 L^2 eta}--\ref{Formula: L_sigma Gamma = 0})
gives all possible ways to couple $(\delta,\sigma)$-SLEs with the GFF at least
within the framework of our assumptions on the pre-pre-Schwarzian behavior of $\eta$ and the scalar behavior of $\Gamma$.
Theorem
\ref{Theorem: eta structure}
in Section 
\ref{Section: The coupling in case of Dirichlet Neumann boundary conditions}
explains why it is natural to couple the forward ($\delta,\sigma$)-SLE to the GFF
with the transformation rule
\eqref{Formula: Phi ->Phi - chi arg}
and the reverse SLE to the GFF with the transformation rule 
\eqref{Formula: Phi -> Phi + gamma log}.
It also shows that the conditions
\eqref{Formula: chi = 2/k - k/2 intro}
and
\eqref{Formula: chi = 2/k + k/2 intro}
are necessary. 

In Section
\ref{Section: The coupling in case of Dirichlet Neumann boundary conditions}, 
we assume the simplest choices of the Dirichlet ($\Gamma=\Gamma_D$) and Neumann
($\Gamma=\Gamma_N$)
boundary conditions for the covariance $\Gamma$ and study all ($\delta,\sigma$)-SLEs that
can be coupled.
It turns out that only the following ($\delta,\sigma$)-SLEs are allowed: 
\begin{enumerate}
\item Classical SLEs with the drift $\nu\in\mathbb{R}$, $\kappa>0$;
\item Two one-parameter cases that are reparametrizations of the
driftless ($\nu=0$) chordal SLE, $\kappa>0$;
\item All ($\delta,\sigma$)-SLEs with $\kappa=6$ and $\nu=0$ for
$\Gamma=\Gamma_D$.
\end{enumerate}
The last item makes a bridge to the result of Theorem
\ref{Theorem: Absolute continuity of SLE},
which states that all such ($\delta,\sigma$)-SLEs define the same random law on
unparametrized curves. 

We consider all these cases of the coupling in detail, one by one, in Chapter
\ref{Chapter: Classical cases}. Two more examples of the coupling can
be obtained if we assume less trivial choices of $\Gamma$. We also consider them in  
Chapter 
\ref{Chapter: Classical cases},
Sections
\ref{Section: Coupling of forward dipolar SLE and combined Dirichlet-Neumann
GFF}
and
\ref{Section: Coupling with twisted GFF}.

\section{Some technical propositions}
\label{Section: Some technical propositions}

Here we formulate and prove some technical statements that will be
used in the proof of Theorem \ref{Theorem: The coupling theorem} below.
This section can be skipped if the reader is not interested in the details of
the proof.

As we discussed in Section 
\ref{Section: Linear functionals and change of coordinates},
we can define
${G_t^{-1}}_*\eta[f]$ 
and
${G_t^{-1}}_*\Gamma[f,f]$
only if
$\supp f\subset \Im(G^{-1})$. 
This is why we introduce a stopping time. For $f\in\Hc$,
let $T[f]$ be the stopping
time at which the hull $\K_t$ of the forward ($\delta,\sigma$)-SLE first touches
some small neighborhood $U(\supp f)$ of the support of $f$:
\index{$T[f]$}
\begin{equation}
	T[f]:=\sup \{t>0\colon \K_t\cap U(\supp f) 
	= \emptyset\},\quad f\in\Hc.
	\label{Formula: T[f] = ...}
\end{equation}
The neighborhood $U(\supp f)$ can be defined, for example, as the set of points
with Poincar\'e distance less than some $\varepsilon>0$ from $\supp f$. Thus,
$T[f]>0$ a.s. We also define the stopping time $T(x)$, $x\in\bar\Dc$ analogously
using the neighborhood $U(x)$ of a single point $x\in\Dc$.  
In the reverse case, we set $T[f]:=+\infty$. 

Consider now an It\^o process $\{X_t\}_{t\in[0,+\infty)}$ such that
\begin{equation}
	\dI X_t = a_t dt + b_t \dI B_t,\quad t\in[0,+\infty),
\end{equation}
for some continuous processes 
$\{a_t\}_{t\in[0,+\infty)}$
and
$\{b_t\}_{t\in[0,+\infty)}$.
We denote by 
$\{X_{t\wedge T}\}_{t\in[0,+\infty)}$ 
the process stopped at a stopping time $T$. It satisfies
\begin{equation}
	\dI X_{t\wedge T} = \theta(T-t) a_t dt + \theta(T-t) b_t \dI B_t,\quad
	t\in[0,+\infty),
\end{equation}
where
\begin{equation}
	\theta(t) :=
	\begin{cases} 
		0 	&\mbox{if } t\leq 0 \\
		1		&\mbox{if } t>0.
	\end{cases}
\end{equation}
If $\{X_t\}_{t\in[0,+\infty)}$ is a local martingale ($a_t=0$,
$t\in[0,+\infty)$), then 
$\{X_{t\wedge T}\}_{t\in[0,+\infty)}$ 
is also a local martingale. 

We consider below the stopped processes 
$\{Y(G_{t\wedge T}) \}_{t\in[0,+\infty)}$ 
instead of 
$\{Y(G_t)\}_{t\in[0,+\infty)}$ 
for some functions 
$Y:\mathscr{G}\map \mathbb{R}$
and the It\^o differential equations for them. In order to make the relations
less cluttered, we usually drop the terms `$...\wedge T[f]$' and $\theta(T-t)$.
However, in places where they are essential, such as the
proof of Theorem 
\ref{Theorem: The coupling theorem},
we specify the stopping times explicitly.

Define the \emph{diffusion operator}
\index{diffusion operator} 
\index{$\Ac$}
\begin{equation}
  \Ac: = \Lc_{\delta} + \frac12 \Lc_{\sigma}^2.
  \label{Formula: A = L + 1/2 L^2}
\end{equation}
and consider how a regular pre-pre-Schwarzian $\eta$ changes under the random
evolution $G_t$. 
The functions  
${G^{-1}_t}_* \eta^{\psi}(z)$ 
and 
${G^{-1}_t}_* \Gamma^{\psi}(z,w)$ 
are defined by 
\eqref{Formula: G eta(z) = eta(G(z)) + mu log G'(z) + ...}
and
\eqref{Formula: G B(z,w) = B(G(z),G(w))} 
until the stopping times $T(z)$ and $\min(T(z),T(w))$, respectively.

\begin{proposition} 
Let  $\{G_t\}_{t\in[0,+\infty)}$ be a ($\delta,\sigma$)-SLE.
\begin{enumerate}[1.]
\item 
Let $\eta$ be a regular pre-pre-Schwarzian such that the Lie derivatives
$\Lc_{\sigma} \eta$, $\Lc_{\delta} \eta$,
and
$\Lc_{\sigma}^2 \eta$ are well-defined. Then
\begin{equation}
  \dI {G^{-1}_{t}}_* \eta^{\psi}(z) = 
  {G_{t}^{-1}}_* \left( \Ac \eta^{\psi}(z) ~dt + \Lc_{\sigma}
  \eta^{\psi}(z) \dI B_t \right).
  \label{Formula: d G^-1 eta  = G A^+ eta dt + G L eta dB}
\end{equation}
\item
Let $\Gamma$ be a scalar 
(\eqref{Formula: G B(z,w) = B(G(z),G(w))} is satisfied) 
bilinear functional 
such that the Lie derivatives
$\Lc_{\sigma} \Gamma$, 
$\Lc_{\delta} \Gamma$,
and
$\Lc_{\sigma}^2 \Gamma$ 
are well-defined.
Then
\begin{equation}
  \dI {G_{t}^{-1}}_* \Gamma^{\psi}(z,w) = 
  {G_{t}^{-1}}_* \left( \Ac \Gamma^{\psi}(z,w) ~dt 
  	+ \Lc_{\sigma} \Gamma^{\psi}(z,w) \dI B_t \right).
  \label{Formula: d G Gamma  = G A Gamma dt + G L Gamma dB}
\end{equation}
\end{enumerate}
\label{Proposition: dI G eta = G A eta dt + G L eta dB}
\end{proposition}

This can be proved by direct calculation, but we present a
preferable approach, which is valid not only for pre-pre-Schwarzians but also, for
instance, for vector fields and, even more generally, for assignments
whose transformation rules contain an arbitrary finite number of derivatives at
a finite number of points. To this end, let us prove the following lemma.
\begin{lemma}
\label{lemma2}
Let $X^i(t)$ ($i=1,2,\dotso,n$) be a finite collection of stochastic processes
defined by the following system of equations in the Stratonovich form
\begin{equation}\begin{split}
	d^{\mathrm{S}} X_t^i = \alpha^i(X_t) dt + \beta^i(X_t) d^{\mathrm{S}}B_t,
	\label{Formula: dX = a X dt + b X dB}
\end{split}\end{equation}
for some fixed functions $\alpha,\beta\colon \mathbb{R}^n\map\mathbb{R}^n$.
Let us define $Y_s^i$, $Z_s^i$ as the solutions to the initial value problems
\begin{equation}\begin{split}
&\dot Y_s^i = \alpha^i(Y_s), \quad Y_0^i=0, \\
&\dot Z_s^i = \beta^i(Z_s), \quad Z_0^i=0,
\label{Formula: Y_s = ..., Z_s = ...}
\end{split}\end{equation}
defined in some neighborhood of $s=0$. Let also
$F:\mathbb{R}^n\map\C$
be a twice-differentiable function.
Then the It\^{o} differential of $F(X_t)$ can be written in the following form:
\begin{equation}
d^{\mathrm{Ito}}F(X_t) =
\left.
\frac{\de}{\de s} F(X_t+Y_s) dt +
\frac{\de}{\de s} F(X_t+Z_s) d^{\mathrm{Ito}}B_t +
\frac12\frac{\de^2}{\de s^2} F(X_t+Z_s) dt
\right|_{s=0}.
\label{Formula: Ito derivative lemma}
\end{equation}
\label{Lemma: Ito derivative lemma}
\end{lemma}

\begin{proof}
The direct calculation of the right-hand side of
(\ref{Formula: Ito derivative lemma})
gives
\begin{equation}\begin{split}
  &
  F'_i(X_t) \left(\alpha^i(X_t) + \frac12 {\beta'_j}^i(X_t)\beta^j(X_t) \right) dt +
  F'_i(X_t) \beta^i(X_t) d^{\mathrm{Ito}}B_t +
  \frac12 F''_{ij}(X_t) \beta^i(X_t)\beta^j(X_t) dt,
\end{split}\end{equation}
which is indeed the It\^{o} differential of $F(X_t)$. We use summation
over repeated indices and used the notation 
$F'_i(X):=\frac{\de}{\de X^i}F(X)$.
\end{proof}

\medskip

\noindent
{\it Proof of  Proposition \ref{Proposition: dI G eta = G A eta dt + G L eta dB}.}
We use the lemma above.
Let $n=4$, and let us define a vector-valued linear map
$\{ \cdot \}$ for an analytic function $x(z)$ as
\begin{equation}
 \{ x(z) \} : = \{ \Re x(z),\Im x(z),\Re x'(z),\Im x'(z) \}.
 \label{Formula: [x] :=  [Re,Im,Re',Im']}
\end{equation}
For example,
\begin{equation}
 X_t := \{ G_t^{\psi}(z) \} =
  \{ \Re G_t^{\psi}(z), \Im G_t^{\psi}(z), \Re {G_t^{\psi}}'(z), \Im {G_t^{\psi}}'(z) \}.
\end{equation}
From
\eqref{Formula: Slit hol stoch flow Strat}
we have
\begin{equation}
 \alpha(X_t) = \{ \delta^{\psi}(G_t^{\psi}(z)) \},\quad
 \beta(X_t) = \{ \sigma^{\psi}(G_t^{\psi}(z)) \}.
\end{equation}
Let also
\begin{equation}\begin{split}
	F(X_t) =& F(\{{G_{t}^{\psi}}(z) \}):=
	{G_{t}^{-1}}_* \eta^{\psi}(z) 
	=\\=& 
	\eta^{\psi}(G_{t}^{\psi}(z)) 
	+ \mu \log {G_{t}^{\psi}}'(z) 
	+ \mu^* \log \overline{{G_{t}^{\psi}}'(z)}.
\end{split}\end{equation}
We have
\begin{equation}
 Y_s = \{ H_s[\delta]^{\psi}(z)-z \},\quad
 Z_s = \{ H_s[\sigma]^{\psi}(z)-z \}
\end{equation}
due to
(\ref{Formula: Slit hol stoch flow Strat}),
(\ref{Formula: dX = a X dt + b X dB}),
(\ref{Formula: Y_s = ..., Z_s = ...}),
and
(\ref{Formula: d H[sigma](z) = sigma( H[sigma] ) (z)}).

Now we can use Lemma~\ref{lemma2} in order to obtain (\ref{Formula: d G^-1 eta  = G A^+ eta dt + G L eta dB}) for $t=0$:
\begin{equation}\begin{split}
 &\left. \dI {G_{t}^{-1}}_* \eta^{\psi}(z) \right|_{t=0} =
 \left. \dI F[X_t] \right|_{t=0} 
 =\\=&
 \text{ (right-hand side of (\ref{Formula: Ito derivative lemma}) with $t=0$
 ) }.
 \label{Formula: d G eta |_t=0 = ...}
\end{split}\end{equation}
But
\begin{equation}\begin{split}
 &\left. \frac{\de}{\de s} F(X_t+Y_s) \right|_{s=0,t=0} =
 \left. \frac{\de}{\de s} F( \{z+ H_s[\delta]^{\psi}(z) - z \}) \right|_{s=0}
 =\\=&
 \left. \frac{\de}{\de s} F( \{H_s[\delta]^{\psi}(z)\}) \right|_{s=0}
 =\left. \frac{\de}{\de s} \{H_s[\delta]^{-1}_* \eta^{\psi}(z) \right|_{s=0}
 = \Lc_{\delta} \eta^{\psi}(z).
\end{split}\end{equation}
A similar observation for other terms in
(\ref{Formula: d G eta |_t=0 = ...})
yields that
\begin{equation}\begin{split}
 \left. \dI {G_t^{-1}}_* \eta^{\psi}(z) \right|_{t=0} =&
 \Lc_{\delta} \eta^{\psi}(z) dt
 + \Lc_{\sigma} \eta^{\psi}(z) \dI B_t
 + \frac12 \Lc_{\sigma}^2 \eta^{\psi}(z) dt
 =\\=&
 \Ac \eta^{\psi}(z) dt + \Lc_{\sigma} \eta^{\psi}(z) \dI B_t.
\end{split}\end{equation}
For $t>0$, we use the results of Section
\ref{Section: Definition and basic properties}
to conclude
\begin{equation}\begin{split}
 &\dI {G_t^{-1}}_* \eta^{\psi}(z) =
 \dI \left( \tilde G_{t-t_0} \circ G_{t_0} \right)^{-1}_*  \eta^{\psi}(z) =
 \dI \left. {G_{t_0}^{-1}}_* { \tilde G_{t-t_0}^{-1}\,  }_* \eta^{\psi}(z) \right|_{t_0=t}
 =\\=&
 \left. {G_{t}^{-1}}_* \dI  { \tilde G_{s}^{-1}\,  }_* \eta^{\psi}(z) \right|_{s=0} =
 {G_{t}^{-1}}_* \left( \Ac \eta^{\psi}(z) dt + \Lc_{\sigma} \eta^{\psi}(z) \dI B_t \right).
\end{split}\end{equation}

The proof of 
\eqref{Formula: d G Gamma  = G A Gamma dt + G L Gamma dB}
is analogous. The only difference is that we do not have
the pre-pre-Schwarzian terms with the derivatives but there are two points $z$
and $w$. We can assume
\begin{equation}
  \{ x \} : = \{ \Re x(z),\Im x(z),\Re x(w),\Im x(w) \}
\end{equation}
instead of
(\ref{Formula: [x] :=  [Re,Im,Re',Im']})
and the remaining part of the proof is the same.
\quad\qed
\medskip

Below, we obtain the It\^o differentials of 
${G_t^{-1}}_*\eta[f]$ 
and 
${G_t^{-1}}_* \Gamma [f,g]$ 
for 
($\delta,\sigma$)-SLE
$\{G_t\}_{t\in[0,+\infty)}$
and
$f,g\in\Hc$.  
To this end, we need the It\^o formula for nonlinear functionals over $\Hc$.
For linear functionals on the Schwartz space this has been shown in
\cite{Krylov2009}.
However, the author is not aware of such results for nonlinear
functionals.
The following propositions are the special cases we will need in the next
section. They are consequences of the proposition above, the classical It\^o 
formula, and the stochastic Fubini theorem.

\begin{proposition} ~\\
Let the conditions of Proposition
\ref{Proposition: dI G eta = G A eta dt + G L eta dB}
be satisfied. Then:
\begin{enumerate}[1.]
\item
The It\^{o} differential is
interchangeable with the integration over $\Dc$. Namely,
\begin{equation}\begin{split}
  &\dI \int\limits_{\psi(\supp f)} {G^{-1}_{t}}_*
  \eta^{\psi}(z) f^{\psi}(z) l(dz) 
  =\\=&
  \int\limits_{\psi(\supp f)} {G^{-1}_{t}}_* \Ac
  \eta^{\psi}(z) f^{\psi}(z) l(dz)\, dt 
  +\\+& 
  \int\limits_{\psi(\supp f)} 
  	{G^{-1}_{t}}_* \Lc_{\sigma} \eta^{\psi}(z) f^{\psi}(z)
  	 l(dz) \,\dI B_t.
  \label{Formula: dI int G^-1 eta = int dI G^-1 eta}
\end{split}\end{equation}
An equivalent shorter formulation is
\begin{equation}\begin{split}
  \dI {G_{t }^{-1}}_* \eta[f]
  = {G^{-1}_t}_* \Ac \eta[f] dt +
  {G^{-1}_t}_* \Lc_{\sigma} \eta[f] \dI B_t
  \label{Formula: dI G^-1 eta[f] = G^-1 A eta[f] dt + G^-1 L eta[f] dB}.
\end{split}\end{equation}
\item 
The It\^{o} differential is interchangeable with the double integration
over $\Dc$, namely,
\begin{equation}\begin{split}
  &\dI \int\limits_{\psi(\supp f)} 
  \int\limits_{\psi(\supp f)} {G^{-1}_{t}}_* \Gamma (x,y)
  f^{\psi}(x)f^{\psi}(y)  l(dx) l(dy) =\\=&
  \int\limits_{\psi(\supp f)} \int\limits_{\psi(\supp f)}
  {G^{-1}_{t}}_* \Ac \Gamma (x,y)
  f^{\psi}(x) f^{\psi}(y) l(dx) l(dy)\, dt
  +\\+&
  \int\limits_{\psi(\supp f)} \int\limits_{\psi(\supp f)}
  {G^{-1}_{t}}_* \Lc_{\sigma} \Gamma(x,y)
  f^{\psi}(x) f^{\psi}(y) l(dx) l(dy)\, \dI B_t.
  \label{Formula: dI int G^-1 Gamma = int dI G^-1 Gamma}
\end{split}\end{equation}
 An equivalent shorter formulation is
 \begin{equation}\begin{split}
  \dI {G_t^{-1}}_* \Gamma[f,g]
  =
  {G^{-1}_t}_* \Ac \Gamma[f,g]\, dt +
  {G^{-1}_t}_* \Lc_{\sigma} \Gamma[f,g]\, \dI B_t.
  \label{Formula: dI G^-1 Gamma[f,g] = G^-1 A Gamma[f,g] dt + G^-1 L Gamma[f,g] dB}
 \end{split}\end{equation}
\end{enumerate}
\label{Proposition: dI int eta = int dI eta}
\end{proposition}

\begin{proof}

The relation
(\ref{Formula: dI int G^-1 eta = int dI G^-1 eta})
in integral form is
%
%
%
%
\begin{equation}\begin{split}
  &\int\limits_{\psi(\supp f)} 
  	{G^{-1}_t}_* \eta^{\psi} (z) f^{\psi}(z) l(dz)
  =\eta[f]  
  +\\+&
  \int\limits_0^{t}
   \int\limits_{\psi(\supp f)} 
   {G^{-1}_{\tau}}_* \Ac \eta^{\psi} (z)  f^{\psi}(z) l(dz) d{\tau} +
  \int\limits_0^{t}
  \int\limits_{\psi(\supp f)} 
  	{G^{-1}_{\tau}}_* \Lc_{\sigma} \eta^{\psi}(z) f^{\psi}(z) l(dz) \dI
  B_{\tau}
	\label{Formula: 1}
\end{split}\end{equation}
The order of the It\^{o} and the Lebesgue integrals can be changed using the
stochastic Fubini theorem; see, for example, \cite{Protter2004a}.
It is now enough to use
(\ref{Formula: d G^-1 eta  = G A^+ eta dt + G L eta dB}) to obtain
(\ref{Formula: dI int G^-1 eta = int dI G^-1 eta}).

The proof of 
\eqref{Formula: dI G^-1 Gamma[f,g] = G^-1 A Gamma[f,g] dt + G^-1 L Gamma[f,g] dB}
is analogous.
\end{proof}

\begin{proposition}
\label{Proposition: G chi[f] is an Ito process}
Let
\begin{equation}
 \hat\phi[f] = \exp \left( W[f] \right),\quad
 W[f] := \frac12 \Gamma[f,f] + \eta[f].
\end{equation}
Then
${G_t^{-1}}_* \hat\phi[f]$  is an It\^{o} process defined by the integral
\begin{equation}\begin{split}
 &{G_t^{-1}}_* \hat\phi[f] 
 =\\=&
 \int\limits_0^t
 \exp \left( {G_{\tau}^{-1}}_* W[f] \right)
 \left(
  {G_{\tau}^{-1}}_* \Ac W[f] d\tau +
  {G_{\tau}^{-1}}_* \Lc_{\sigma} W[f] \dI B_{\tau} +
  \frac12 \left( {G_{\tau}^{-1}}_* \Lc_{\sigma} W[f] \right)^2 d\tau
 \right).
 \label{Formula: G chi = chi + int ...}
\end{split}\end{equation}
\end{proposition}

\begin{proof}
The stochastic process
${G_t^{-1}}_* W^{\psi}[f]$
has the integral form
\begin{equation}\begin{split}
 &{G_t^{-1}}_* W^{\psi}[f] = 
 \frac12 {G_t^{-1}}_* \Gamma^{\psi}[f,f] + {G_t^{-1}}_* \eta^{\psi}[f]
 =\\=&
 \int\limits_{\psi(\supp f)} \int\limits_{\psi(\supp f)}
  {G_t^{-1}}_* \Gamma^{\psi}(z,w) f^{\psi}(z) f^{\psi}(w) l(dz) l(dw)
 +
 \int\limits_{\psi(\supp f)} {G_t^{-1}}_* \eta^{\psi}(z) f^{\psi}(z) l(dz).
\end{split}\end{equation}
due to Proposition 
\ref{Proposition: dI int eta = int dI eta}. 
In terms of It\^{o} differentials, it is
\begin{equation}
 \dI {G_t^{-1}}_* W^{\psi}[f] =
 {G_{t}^{-1}}_* \Ac W^{\psi}[f] dt +
 {G_{t}^{-1}}_* \Lc_{\sigma} W^{\psi}[f] \dI B_{t}.
\end{equation}
In order to obtain the exponential function, we can simply use It\^{o}'s lemma:
\begin{equation}\begin{split}
 &\dI {G_t^{-1}}_* \exp \left( W^{\psi}[f] \right) =
 \dI \exp \left( {G_t^{-1}}_* W^{\psi}[f] \right)
 =\\=&
 \exp \left( {G_{t}^{-1}}_* W^{\psi}[f] \right)
 \left(
  {G_{t}^{-1}}_* \Ac W^{\psi}[f] dt
  + {G_{t}^{-1}}_* \Lc_{\sigma} W^{\psi}[f] \dI B_t
  + \frac12 \left( {G_{t}^{-1}}_* \Lc_{\sigma} W^{\psi}[f] \right)^2 dt
 \right)
\end{split}\end{equation}
\end{proof}

\section{Coupling between SLE and GFF}
\label{Section: Coupling between SLE and GFF}

Let 
$(\Omega^{\Phi},\mathcal{F}^{\Phi},P^{\Phi})$ 
be the probability space
for the GFF 
$\Phi$
and let
$(\Omega^{B},\mathcal{F}^{B},P^{B})$ 
be the independent probability space for
the Brownian motion 
$\{B_t\}_{t\in[0,+\infty)}$, 
which governs some
$(\delta,\sigma)$-SLE 
$\{G_t\}_{t\in[0,+\infty)}$. 
In this section, we consider a coupling between
these random laws. Let $T[f]$ be the stopping time as defined above.


\begin{definition}
A GFF 
$\Phi(\Hc,\Gamma,\eta)$
is called \emph{coupled}
to the forward or reverse
$(\delta,\sigma)$-SLE, driven by $\{B_t\}_{t\in[0,+\infty)}$, if
the random variable ${G^{-1}_{t\wedge T[f]}}_*\Phi^{\psi}[f]$
obtained by independent samplings of $\Phi$ and $G_t$ has the same law as
$\Phi^{\psi}[f]$ for any test function $f\in\Hc$, chart map $\psi$, and
$t \in[0,+\infty)$.
 \label{Definition: Coupling}
\end{definition}

If the coupling holds for a fixed chart map $\psi$ and for any $f\in\Hc$,
then it also holds for any other chart map $\tilde \psi$,
due to
(\ref{Formula: Phi^tilde psi[f] = Phi^psi [f] - mu in log tau f ldz}).
We also give a weaker version of the coupling statement that we plan to use
here. To this end, we have to consider stopped versions of the stochastic
process $\{G_{t\wedge T[f]}\}_{t\in[0,+\infty)}$.

A collection of stopping
times $\{T_n\}_{n=1,2,\dotso}$ is called a
\emph{fundamental sequence}
\index{fundamental sequence}
if $0\leq T_n\leq T_{n+1}\leq \infty$, $n=1,2,\dotso$ a.s., and
$\lim\limits_{n\map \infty} T_n=\infty$ a.s.
A stochastic process $\{x_t\}_{t\in[0,+\infty)}$ is called a 
\emph{local martingale}
\index{local martingale}
if there exists a fundamental sequence of stopping times
$\{T_n\}_{n=1,2,\dotso}$, such that the stopped process
$\{x_{t\wedge T_n}\}_{t\in[0,+\infty)}$ 
is a martingale for each
$n=1,2,\dotso$.
%

Suppose now that the coupling statement above is valid only for the process stopped at
$T_n$, $\{G_{t\wedge {T[f]} \wedge T_n}\}_{t\in[0,+\infty)}$, 
for each $n=1,2,\dotso$.
Namely,
${G^{-1}_{t\wedge {T[f]} \wedge T_n}}_*\Phi^{\psi}[f]$
has the same law as
$\Phi^{\psi}[f]$ for each $n=1,2,\dotso$.
We are now ready to define local coupling.

\begin{definition}
A GFF 
$\Phi(\Hc,\Gamma,\eta)$
is called \emph{locally coupled}
\index{local coupling}
to the forward or reverse
$(\delta,\sigma)$-SLE, driven by $\{B_t\}_{t\in[0,+\infty)}$, if
there exists a fundamental sequence 
$\{T_n[f,\psi]\}_{n=1,2,\dotso}$ such that
the random variable ${G^{-1}_{t\wedge T[f]}}_*\Phi^{\psi}[f]$
obtained by independent samplings of $\Phi$ and $G_t$ has the same law as
$\Phi^{\psi}[f]$ until the 
stopping time $T_n[f,\psi]$ for each $n=1,2,\dotso$, for any test function
$f\in\Hc$, and chart map $\psi$.

\label{Definition: Local coupling}
\end{definition}

\begin{remark}
If 
$T_n[f,\psi]=+\infty$ a.s. 
for each $n$ and each
$f\in\Hc$,
then the coupling is not local.
\end{remark}

%
%

In this subsection, we prove the following theorem.
\begin{theorem}
\label{Theorem: The coupling theorem}
The following three statements are equivalent:
\begin{enumerate} [1.]
\item
The GFF
$\Phi(\Hc,\Gamma,\eta)$
is locally coupled to the
$(\delta,\sigma)$-SLE;
\item
${G_{t \wedge T[f]}^{-1}}_* \hat\phi^{\psi}[f]$ is a local martingale for
$f\in\Hc$
in any chart $\psi$;
\item
The system of equations
 \begin{equation}\begin{split}
  \Lc_{\delta} \eta[f] + \frac12\Lc_{\sigma}^2 \eta[f] = 0
  ,\quad f\in\Hc,
  \label{Formula: L eta + 1/2 L^2 eta}
 \end{split}\end{equation}
 \begin{equation}
  \Lc_{\delta} \Gamma[f,g] + \Lc_{\sigma} \eta[f] \Lc_{\sigma} \eta[g] = 0
  ,\quad f,g\in\Hc,
  \label{Formula: Hadamard's formula}
 \end{equation}
 and
 \begin{equation}
  \Lc_{\sigma} \Gamma[f,g] = 0
  ,\quad f,g\in\Hc.
  \label{Formula: L_sigma Gamma = 0}
 \end{equation}
is satisfied.
 \end{enumerate}
\end{theorem}

We start the proof after some remarks.
Just for clarity (not for applications), we reformulate the system
(\ref{Formula: L eta + 1/2 L^2 eta}--\ref{Formula: L_sigma Gamma = 0})
directly in terms of partial derivatives using
(\ref{Formula: L eta = v d eta + chi dv}),
(\ref{Formula: L L phi = w de L phi}),
(\ref{Formula: L Gamma = ...}),
and
(\ref{Formula: A = L + 1/2 L^2})
as
\begin{equation}\begin{split}
 &\delta(z) \de_z \eta(z) + \overline{\delta(z)} \de_{\bar z} \eta(z)
 + \mu \delta'(z) + \mu^* \overline{\delta'(z)}
 +\\+&
 \frac12 \sigma^2(z) \de_z^2 \eta(z) 
 + \frac12 \overline{\sigma^2(z)} \de_{\bar z}^2 \eta(z)
 + \sigma(z) \overline{\sigma(z)} \de_z \de_{\bar z} \eta
 +\\+&
 \frac12 \sigma(z) \sigma'(z) \de_z \eta(z)
 + \frac12 \overline{\sigma(z)} \overline{\sigma'(z)} \de_{\bar z} \eta(z)
 + \mu \sigma(z) \sigma''(z) + \mu^* \overline{ \sigma(z) \sigma''(z) } =
 \beta;\\
 &\delta(z) \de_z \Gamma(z,w) + \delta(w) \de_w \Gamma(z,w)
 +\overline{\delta(z)} \de_{\bar z} \Gamma(z,w) 
 + \overline{\delta(w)} \de_{\bar w} \Gamma(z,w)
 +\\+&
 \left(
  \sigma(z) \de_z \eta(z) + \overline{\sigma(z)} \de_{\bar z} \eta(z)
  + \mu \sigma'(z) + \mu^* \overline{ \sigma'(z) }
 \right)
 \times \\ \times &
 \left(
  \sigma(w) \de_w \eta(w) + \overline{\sigma(w)} \de_{\bar w} \eta(w)
  + \mu \sigma'(w) + \mu^* \overline{ \sigma'(w) }
 \right) = \beta_1(z)+\beta_1(w);\\
 &\sigma(z) \de_z \Gamma(z,w) + \sigma(w) \de_w \Gamma(z,w)
 +\overline{\sigma(z)} \de_{\bar z} \Gamma(z,w) 
 + \overline{\sigma(w)} \de_{\bar w} \Gamma(z,w) 
 =\\=& 
 \beta_2(z)+\beta_2(w),
	\label{Formula: master equations in ditails}
\end{split}\end{equation}
for $\Hc=\Hc_s^*$. We drop the superscript $\psi$ above for brevity.
The right-hand side is not zero, but an arbitrary real
constant $\beta$ for the first equation and a sum of two real-valued functions
$\beta_i(z)+\beta_i(w)$, $i=1,2$, for the other two equations, because linear
functionals over $\Hc_s^*$ are defined up to a constant and bilinear ones are
defined up to a sum $\beta(z)+\beta(w)$, see Sections
\ref{Section: Test function}
and
\ref{Section: Fundamental solution to the Laplace-Beltrami equation}.
For the space $\Hc=\Hc_s$, the right-hand sides in 
\eqref{Formula: master equations in ditails}
are simply zero.

The first equation
(\ref{Formula: L eta + 1/2 L^2 eta})
is simply a local martingale condition for $\eta$. The second one 
\eqref{Formula: Hadamard's formula} 
is known as Hadamard's formula, see also Section 
\ref{Section: Stress tensor and the conformal Ward identity}. 
The third 
means that $\Gamma$ should be invariant under the one-parameter family of M\"obius automorphisms generated by $\sigma$.


\medskip
\noindent
{\it Proof of  Theorem \ref{Theorem: The coupling theorem}}.
Let us start by showing how statement 1 about the coupling implies
statement 2 about the local martingale property. 

\medskip
\noindent
\textbf{1.$\Leftrightarrow$2.}
Let 
$G_{t\wedge T[f] \wedge T_n[f,\psi]}$ 
be the process $G_{t\wedge T[f]}$ stopped at the stopping times
$T_n[f,\psi]$ 
forming some fundamental sequence. 
The coupling statement can be reformulated as an equality of characteristic
functions for the random variables 
${G^{-1}_{t\wedge T[f]\wedge T_n[f,\psi]}}_*\Phi^{\psi}[f]$
and
$\Phi^{\psi}[f]$
for all test functions $f$.
Namely, the following expectations must be equal:
\begin{equation}
 \Evv{B}{ \Evv{\Phi} 
 {e^{ {G^{-1}_{t\wedge T[f]\wedge T_n[f,\psi]}}_*\Phi^{\psi}[f] } } }
 = \Evv{\Phi}{e^{\Phi^{\psi}[f]}}
 ,\quad f\in \Hc,\quad t\in[0,+\infty),\quad n=1,2,\dotso~,
\end{equation}
which, in particular, means the integrability of
$e^{{G^{-1}_{t\wedge T[f]\wedge T_n[f,\psi]}}{}_*\Phi^{\psi}[f]}$ 
with respect to 
$\Omega^{B}$ 
and 
$\Omega^{\Phi}$.
We used 
$\Evv{B}{\cdot}$ 
for the expectation with respect to the random law of
$\{B_t\}_{t\in[0,+\infty)}$ 
(or $\{G_t\}_{t\in[0,+\infty)}$) 
and
$\Evv{\Phi}{\cdot}$ 
for the expectation with respect to $\Phi$.
Let us use Definitions
(\ref{Formula: GFF chracteristic function})
and
(\ref{Formula: F phi(Gamma, eta) = phi(F^-1 Gamma, F^-1 eta)})
to simplify this identity to
\begin{equation}
	\Evv{B}{{G^{-1}_{t\wedge T[f]\wedge T_n[f,\psi]}}_* \hat\phi^{\psi}[f]} =
	\hat\phi^{\psi}[f],\quad f\in \Hc
	,\quad t\in[0,+\infty),\quad n=1,2,\dotso~.
	\label{Formula: E[G chi] = chi}
\end{equation}
After the change 
$f\map \tilde G_{s \wedge T[f]}{}_*f$ 
for some independently
sampled 
$\tilde G_s$ 
and
$s\in[0,+\infty)$, 
we obtain
\begin{equation}\begin{split}
	&\Evv{B}{G^{-1}_{t\wedge T[\tilde G_{s \wedge T[f]}{}_* f]
		\wedge T_n[\tilde
	G_{s \wedge T[f]}{}_*f,\psi]}{}_* 
	\hat\phi^{\psi}[\tilde G_{s \wedge T[f]}{}_*f]} 
	= \hat\phi^{\psi}[\tilde G_{s \wedge T[f]}{}_*f],\\
	&f\in \Hc
	,\quad t\in[0,+\infty),\quad n=1,2,\dotso.
\end{split}\end{equation}
Multiplying both sides by 
\begin{equation}
	e^{\int\limits_{~\supp f^{\psi}} 
	\left( 
		\mu \log \left(\tilde G_{s \wedge T[f]}^{\psi}\right){'}(z) +
		\mu^{*}	\overline{\log \left(\tilde G_{s \wedge T[f]}^{\psi}\right){'}(z)}
	\right)
	f(z)	l(dz)}
\end{equation}
and using 
\eqref{Formula: F eta f= eta F-1 f}
and
\eqref{Formula: GFF chracteristic function},
we conclude that
\begin{equation}\begin{split}
	&\Evv{B}{
	\tilde G_{s \wedge T[f]}^{-1}{}_*
	G^{-1}_{t\wedge T[\tilde G_{s \wedge T[f]}{}_* f]
		\wedge T_n[\tilde
	G_{s \wedge T[f]}^{-1}{}_*f,\psi]}{}_* 
	\hat\phi^{\psi}[f]} 
	= \tilde G_{s \wedge T[f]}{}_*^{-1} \hat\phi^{\psi}[f],\\
	&f\in \Hc
	,\quad t\in[0,+\infty),\quad n=1,2,\dotso.
	\label{Formula: 5}
\end{split}\end{equation}
Define now the process
\begin{equation}
	\tilde {\tilde G}_{t+s} 
	:= G_{t} \circ \tilde G_{s}
	,\quad s,t\in[0,+\infty),
\end{equation}
which has the law of ($\delta,\sigma$)-SLE.
Its stopped version satisfies
\begin{equation}
	\tilde{{\tilde G}}_{t+s\wedge T[f]} 
	= G_{t\wedge T[\tilde G_{s \wedge T[f]}{}_* f]} \circ \tilde G_{s\wedge T[f]}
	,\quad s,t\in[0,+\infty),\quad f\in\Hc.
\end{equation}
The left-hand side of
\eqref{Formula: 5}
equals
\begin{equation}\begin{split}
	&\Evv{B}{
	\left(
		{G_{t\wedge T[\tilde G_{s \wedge T[f]}{}_* f]
		\wedge T_n[\tilde
		G_{s \wedge T[f]}^{-1}{}_*f,\psi]}} 
		\circ \tilde G_{s \wedge T[f]}
	\right)_*^{-1} 
	\hat\phi^{\psi}[f]}
	=\\=&
	\Evv{B}{
	\left(
		\tilde {\tilde G}_{t+s\wedge T[f]
		\wedge T_n[\tilde	G_{s \wedge T[f]}^{-1}{}_*f,\psi]+s } 
	\right)_*^{-1} 
	\hat\phi^{\psi}[f]~|~
	\mathcal{F}^B_{s\wedge T[f]}}.
\end{split}\end{equation}
We now use the Markov property of ($\delta,\sigma$)-SLE and conclude that
$T_n'[f,\psi]:= T_n[\tilde	G_{s \wedge T[f]}^{-1}{}_*f,\psi]+s$
is a fundamental sequence for the pair of $f$ and 
$\psi$. 
Thus, 
\eqref{Formula: 5}
simplifies to
\begin{equation}
	\Evv{B}{
		G_{t+s\wedge T[f]
		\wedge T_n'[f,\psi] }^{-1} {}_* \hat\phi^{\psi}[f]~|~
	\mathcal{F}^B_{s\wedge T[f]\wedge T_n'[f,\psi]}}
	=G_{t+s\wedge T[f]\wedge T_n'[f,\psi] }^{-1}{}_* \hat \phi[f],
\end{equation}
hence, 
$\{G_{t\wedge T[f] }^{-1}{}_* \hat \phi[f]\}_{t\in[0,+\infty)}$ 
is a local martingale. 

The converse statement can be obtained by the same method in the reverse
order.

\medskip
\noindent
\textbf{2.$\Leftrightarrow$3.}
According to Proposition \ref{Proposition: G chi[f] is an Ito process},
the drift term, i.e., the coefficient of $dt$, vanishes identically when
\begin{equation}
 \Ac W[f] +  \frac12 \left( \Lc_{\sigma} W[f] \right)^2 = 0,\quad f\in\Hc.
 \label{Formula: A W + 1/2 L W L W = 0}
\end{equation}
The left-hand side is a functional polynomial of degree four.
We use the fact that a regular symmetric functional 
$P[f] := \sum\limits_{k=1,2,\dotso n}
p_k[f,f,\dotso,f]$
of degree $n$ over such spaces as $\Hc_s$,
$\Hc_s^*$, $\Hc_{s,b}$ (see Sections 
\ref{Section: Coupling of forward radial SLE and Dirichlet GFF}),
$\Hc_{s,b}^{*}$ (see Section
\ref{Section: Coupling of reverse radial SLE and Neumann GFF}),
or
$\Hc^{\pm}_{s,b}{}^{*}$ (see Section
\ref{Section: Coupling with twisted GFF})
is identically zero if and only if 
\begin{equation}
	p_k[f_1,f_2,\dotso,f_n] = 0,\quad k=1,2,\dotso n, \quad f\in \Hc.
\end{equation}
Thus, each of the following functions must be identically zero:
\begin{equation}\begin{split}
	&\Ac \eta[f] = 0,\quad
	\frac12 \Ac \Gamma[f,g]
	+ \frac12 \Lc_{\sigma} \eta[f] \Lc_{\sigma} \eta[g] = 0,\\
	& \Lc_{\sigma} \eta[f] \Lc_{\sigma} \Gamma[g,h] + \text{symmetric terms} = 0,\\
	&\Lc_{\sigma} \Gamma[f,g] \Lc_{\sigma} \Gamma[h,l]  + \text{symmetric terms} =0,\\
	& f,g,h,l\in \Hc.
	\label{Formula: 6}
\end{split}\end{equation}
We can conclude that $\Lc_{\sigma} \Gamma[f,g]= 0$,
$\Ac \Gamma[f,g] = \Lc_{\delta}\Gamma[f,g]$
for any $f,g\in\Hc$,
and this system is equivalent to the system
(\ref{Formula: L eta + 1/2 L^2 eta}--\ref{Formula: L_sigma Gamma = 0}).
For the case $\Hc=\Hc_s$, we can write \eqref{Formula: 6} in terms of functions
on $\psi(\Dc)$:
\begin{equation}\begin{split}
	&\Ac \eta(z) = 0,\quad
	\frac12 \Ac \Gamma(z,w)
	+ \frac12 \Lc_{\sigma} \eta(z) \Lc_{\sigma} \eta(w) = 0,\\
	& \Lc_{\sigma} \eta(z) \Lc_{\sigma} \Gamma(w,u) 
	+ \text{symmetric terms} = 0,\\
	& \Lc_{\sigma} \Gamma(z,w) \Lc_{\sigma} \Gamma(u,v) 
	+ \text{symmetric terms} =0,\\
	& z,w,u,v\in \psi(\Dc)
	,\quad z \neq w,~u \neq v, \dotso~.
\end{split}\end{equation}



\begin{remark}
Fix a chart $\psi$.
The coupling and the martingales are not local if, in addition to
statement 3 in Theorem
\ref{Theorem: The coupling theorem}
the relation
\begin{equation}
 \Evv{B}{ \left| \int\limits_0^t \exp
 	\left( {G_{\tau \wedge T[f]}^{-1}}_* W^{\psi}[f] \right) 
 	{G_{\tau \wedge T[f]}^{-1}}_*
 	\Lc_{\sigma} W^{\psi}[f] \dI B_{\tau} \right| } <\infty,\quad
 t\geq 0,
 \label{Formula: E[ |int GLW dB| ]}
\end{equation}
holds. This is the condition that the diffusion term with
$\dI B_t$ in
(\ref{Formula: G chi = chi + int ...})
is in $L_1(\Omega^{B})$.
However, this may not be true, in general, in another chart $\tilde \psi$.
Meanwhile, if the local martingale property of
${G_{t \wedge T[f]}^{-1}}_* \hat\phi^{\psi}[f]$
is satisfied in one chart $\psi$ for any $f\in\Hc$, then it is also true in
any other chart due to the invariance of the condition 
\eqref{Formula: A W + 1/2 L W L W = 0} in the proof.
\end{remark}

The study of the general solution of 
(\ref{Formula: L eta + 1/2 L^2 eta}-\ref{Formula: L_sigma Gamma = 0})
is an interesting and complicated problem. Take the Lie derivative
$\Lc_{\sigma}$ 
of the second equation and the Lie derivative 
$\Lc_{\delta}$
of the third equation, and consider the difference of the resulting equations.
This gives an algebraically independent equation
\begin{equation}
	\Lc_{[\delta,\sigma]} \Gamma[f,g] = 
	- \Lc_{\sigma}^2 \eta[f] \Lc_{\sigma} \eta[g]
	\Lc_{\sigma}^2 \eta[f] \Lc_{\sigma}^2 \eta[g].
\end{equation}
Continuing by induction, we obtain an infinite system of a priori algebraically
independent equations because the Lie algebra $\mathcal{U}[\delta,\sigma]$,
induced by $\delta$ and $\sigma$ and 
introduced in Section
\ref{Section: Classification and normalization from algebraic point of view},
is infinite-dimensional. Thus, the existence of a solution of the system
(\ref{Formula: L eta + 1/2 L^2 eta}-\ref{Formula: L_sigma Gamma = 0})
is a special situation that is strongly related to the properties of 
$\mathcal{U}[\delta,\sigma]$. A geometric interpretation of the second equation 
\eqref{Formula: Hadamard's formula}
and a hint to solve this equation are discussed in Section 
\ref{Section: Stress tensor and the conformal Ward identity}.

Before studying special solutions to the system
(\ref{Formula: L eta + 1/2 L^2 eta}--\ref{Formula: L_sigma Gamma = 0}),
let us consider some of its general properties. We also reformulate it in terms
of the analytic functions $\eta^+$, $\Gamma^{++}$ and $\Gamma^{+-}$, which is
technically more convenient.

\begin{theorem}
\label{Theorem: eta structure}
Let $\delta$, $\sigma$, $\eta$, and $\Gamma$ be such that the system
(\ref{Formula: L eta + 1/2 L^2 eta}--\ref{Formula: L_sigma Gamma = 0})
is satisfied, let
$\Gamma$ be a fundamental solution to the Laplace equation
(see (\ref{Formula: Gamma = Log + H}))
that transforms as a scalar,
see (\ref{Formula: G B(z,w) = B(G(z),G(w))}), 
and let $\eta$ be a pre-pre-Schwarzian. Then,
\begin{itemize}
\item For the forward case and $\Hc=\Hc_s$: 
\begin{enumerate}[1.]
\item
$\eta$ is a $(i\chi/2,-i\chi/2)$-pre-pre-Schwarzian
(\ref{Formula: tilde eta = eta - chi arg})
given by a harmonic function in any chart with $\chi$ given by
\begin{equation}
	\chi = \frac{2}{\sqrt{\kappa}} - \frac{\sqrt{\kappa}}{2}.
	\label{Formula: chi = 2/k - k/2}
\end{equation}
\item
The boundary value of $\eta$ undergoes a jump $2\pi/\sqrt{\kappa}$ at the
source point $a$, namely, its local behavior in the half-plane chart is given
by (\ref{Formula: eta = -2/k (arg z - pi/2) + hol}) up to a sign;
\item
The system
(\ref{Formula: L eta + 1/2 L^2 eta}--\ref{Formula: L_sigma Gamma = 0})
is equivalent to the system
(\ref{Formula: eta = eta^+ + eta^-}),
(\ref{Formula: delta/sigma j + mu [sigma,delta]/sigma + 1/2 L j = ibeta}),
(\ref{Formula: Gamma = Gamma^++ + bar Gamma^++ + Gamma^+- + bar Gamma^+-}),
(\ref{Formula: L Gamma^++ + L sigma^+ L sigma^+ = e + e}),
and
(\ref{Formula: L Gamma^++ = 0, L Gamma^+- = 0}).
\end{enumerate}
\item For the reverse case and $\Hc=\Hc_s^*$:
\begin{enumerate}[1.]
\item
$\eta$ is a $(Q/2,Q/2)$-pre-pre-Schwarzian
(\ref{Formula: tilde eta = eta + Q log})
given by a harmonic function in any chart with $Q$ given by
\begin{equation}
	Q = \frac{2}{\sqrt{\kappa}} + \frac{\sqrt{\kappa}}{2}.
	\label{Formula: Q = 2/k + k/2}
\end{equation}
\item
The value of $\eta$ has a logarithmic singularity
at the source point $a$, namely, its local behavior in the half-plane chart is
given by 
\eqref{Formula: Q = 2/k log |z| + hol}
up to a sign;
\item
The system
(\ref{Formula: L eta + 1/2 L^2 eta}--\ref{Formula: L_sigma Gamma = 0})
is equivalent to the system
(\ref{Formula: eta = eta^+ + eta^-}),
(\ref{Formula: delta/sigma j + mu [sigma,delta]/sigma + 1/2 L j = ibeta}),
(\ref{Formula: Gamma = Gamma^++ + bar Gamma^++ + Gamma^+- + bar Gamma^+-}),
(\ref{Formula: L Gamma^++ + L sigma^+ L sigma^+ = e + e}),
and
(\ref{Formula: L Gamma^++ = 0, L Gamma^+- = 0}).
\end{enumerate}
\end{itemize}
\end{theorem}

\begin{proof}
The system
(\ref{Formula: L eta + 1/2 L^2 eta}--\ref{Formula: L_sigma Gamma = 0})
defines $\eta$ only up to an additive constant $C$ that we keep writing in
the formulas for $\eta$ below. The condition that the pre-pre-Schwarzian $\eta$
be real leaves only two possibilities:
\begin{enumerate} [1.]
 \item $\mu=-\mu^*$ is purely imaginary, as in
 (\ref{Formula: tilde eta = eta - chi arg});
 \item $\mu=\mu^*$ is real, as in
 (\ref{Formula: tilde eta = eta + Q log}).
\end{enumerate}

The equation
\eqref{Formula: Hadamard's formula} 
shows that the functional $\Lc_{\sigma} \eta$, as well as $\Lc_{\sigma}^2 \eta$, has to be given by a harmonic
function in any chart. On the other hand,
\eqref{Formula: L eta + 1/2 L^2 eta} 
implies that $\Lc_{\delta}\eta$ is also harmonic. The vector fields $\delta$
and $\sigma$ are transversal almost everywhere. We conclude that $\eta$ is
harmonic. We have also used the fact that the additional $\mu$-terms in 
\eqref{Formula: L eta = v d eta + chi dv} 
are harmonic.

The harmonic function
$\eta^{\psi}(z)$ 
can be represented as a sum of an analytic function 
${\eta^+}^{\psi}(z)$ 
and its complex conjugate in any chart $\psi$:
\begin{equation}
 \eta^{\psi}(z) = {\eta^+}^{\psi}(z) + \overline{{\eta^+}^{\psi}(z)}.
 \label{Formula: eta = eta^+ + eta^-}
\end{equation}
Below in this proof, we drop the chart index $\psi$, which can
be chosen arbitrarily.

We can define $\eta^+$ and $\overline {\eta^+}$ to be pre-pre-Schwarzians of
orders $(\mu,0)$ and $(0,\mu^*)$, respectively, due to
(\ref{Formula: L eta = v d eta + chi dv}). Thus, $\eta^+$ is defined up to a
complex constant $C^+$.
We denote
\begin{equation}
  j^+ := \Lc_{\sigma} \eta^+.
  \label{Formula: j^+ = L eta^+}
\end{equation}
and
\begin{equation}
  j := \Lc_{\sigma} \eta = \Lc_{\sigma} \eta^+ + \overline{\Lc_{\sigma} \eta^+}.
  \label{Formula: j = L eta}
\end{equation}
The reciprocal formula is
 \begin{equation}\begin{split}
  \eta^+(z) 
  := \int { \frac{j^+(z) - \mu \sigma'(z)}{\sigma(z)} }dz.
	\label{Folmula: phi+ = int ...}
\end{split}\end{equation}
This integral can be a multivalued function if $\sigma(z)$ has a zero inside
$\Dc$ (the elliptic case). We consider how to handle this technical difficulty
in Sections
\ref{Section: Coupling of forward radial SLE and Dirichlet GFF}
and
\ref{Section: Coupling of reverse radial SLE and Neumann GFF}.

Let us now reformulate
(\ref{Formula: L eta + 1/2 L^2 eta}) in terms of $j^+$.
Using the fact that
\begin{equation}
 \Lc_{v}^2 (\eta^+ + \overline{\eta^+}) =
 \Lc_{v}^2 \eta^+ + \Lc_{v}^2 \overline{\eta^+},
\end{equation}
we conclude that
\begin{equation}
 \Lc_{\delta} \eta^+ + \frac12 {\Lc_{\sigma}}^2 \eta^+ = C^+.
 \label{Formula: Lc eta^+ + 1/2 L^2 eta^+ = C}
\end{equation}
Here $C^+=i\beta$ for some $\beta\in\mathbb{R}$ for the forward case. For the
reverse case, 
$C^+=-\beta + i \beta'$
for some 
$\beta,\beta'\in\mathbb{R}$
because 
\eqref{Formula: Lc eta^+ + 1/2 L^2 eta^+ = C}
is an identity in the sense of functionals over $\Hc_s^*$. 

The relation 
\eqref{Formula: Lc eta^+ + 1/2 L^2 eta^+ = C}
is equivalent to
\begin{equation}\begin{split}
 &\frac{\delta}{\sigma} \Lc_{\sigma} \eta^+ +
 \frac{\sigma \Lc_\delta \eta^+ - \delta \Lc_\sigma \eta^+}{\sigma}
 + \frac12 {\Lc_{\sigma}}^2 \eta^+ = C^+ \quad
  \Leftrightarrow \\
 &\frac{\delta}{\sigma} j^+
 + \frac{\sigma \delta \de \eta^+ + \mu \sigma \delta'
  - \delta \sigma \de \eta^+ - \mu \delta \sigma'}{\sigma}
 + \frac12 \Lc_\sigma j^+ = C^+ \quad \Leftrightarrow
\end{split}\end{equation}
\begin{equation}\begin{split}
 &\frac{\delta}{\sigma} j^+
 + \mu \frac{[\sigma, \delta]}{\sigma}
 + \frac12 \Lc_\sigma j^+ = C^+. 
 \label{Formula: delta/sigma j + mu [sigma,delta]/sigma + 1/2 L j = ibeta}
\end{split}\end{equation}
We used
\eqref{Formula: L eta = v d eta + chi dv}
and
\eqref{Formula: L_v w = [v,w] := v w' - v' w}.

Consider now the function $\Gamma^{\HH}(z,w)$. It is harmonic with respect to
both variables, with only a logarithmic singularity. Hence, it can be split into a sum of four terms
\begin{equation}
 \Gamma^{\HH}(z,w) :=
 {\Gamma^{++}}^{\HH}(z,w) + \overline{{\Gamma^{++}}^{\HH}(z,w)}
 \mp {\Gamma^{+-}}^{\HH}(z,\bar w) \mp \overline{{\Gamma^{+-}}^{\HH}(z,\bar w)},
 \label{Formula: Gamma = Gamma^++ + bar Gamma^++ + Gamma^+- + bar Gamma^+-}
\end{equation}
where
${\Gamma^{++}}^{\HH}(z,w)$
and
${\Gamma^{+-}}^{\HH}(z,w)$
are analytic with respect to both variables, except on the diagonal $z=w$ for
${\Gamma^{++}}^{\HH}(z,w)$.
We use the upper sign in the pairs $\mp$ for the forward case and the lower sign
for the reverse case.

So, e.g.,
$\overline{{\Gamma^{+-}}^{\HH}(z,\bar w)}$
is anti-analytic with respect to $z$ and analytic with respect to $w$.
We can assume that both
${\Gamma^{++}}(z,w)$ and ${\Gamma^{+-}}(z,w)$ transform as scalars represented
by analytic functions in all charts and are symmetric with respect to $z
\leftrightarrow w $. Observe that these functions are defined at least up to 
the transform
\begin{equation}\begin{split}
 &{\Gamma^{++}}^{\HH}(z,w) \rightarrow
  {\Gamma^{++}}^{\HH}(z,w) + \epsilon^{\HH}(z) + \epsilon^{\HH}(w),\\
 &{\Gamma^{+-}}^{\HH}(z,w) \rightarrow
  {\Gamma^{+-}}^{\HH}(z,w) + \epsilon^{\HH}(z) + \epsilon^{\HH}(w)
\end{split}\end{equation}
for any analytic function $\epsilon^{\HH}(z)$ such that
\begin{equation}
 \overline{\epsilon^{\HH}(z)} = \epsilon^{\HH}(\bar z).
\end{equation}
In the forward case, these additional terms cancel due to the choice
of minus in the pairs `$\mp$' in 
\eqref{Formula: Gamma = Gamma^++ + bar Gamma^++ + Gamma^+- + bar Gamma^+-}. 
In the reverse case, the contribution of these functions is equivalent to the
zero bilinear functional over $\Hc_s^*$.

Consider the equation
(\ref{Formula: L_sigma Gamma = 0}). It leads to
\begin{equation}\begin{split}
	&\Lc_{\sigma} {\Gamma^{++}}^{\HH}(z,w) = \beta_2^{\HH}(z) +
	\beta_2^{\HH}(w),\quad \Lc_{\sigma} {\Gamma^{+-}}^{\HH}(z,w) = \beta_2^{\HH}(z)
	+ \beta_2^{\HH}(w)
	\label{Formula: L Gamma^++ = beta + beta, L Gamma^+- = beta + beta}
\end{split}\end{equation}
for any analytic function $\beta_2^{\HH}(z)$ such that
$\overline{\beta_2^{\HH}(z)} = \beta_2^{\HH}(\bar z)$.
One can fix this freedom (the function $\beta_2^{\HH}$) by
the conditions
\begin{equation}\begin{split}
 &\Lc_{\sigma} {\Gamma^{++}}^{\HH}(z,w) = 0,\quad
 \Lc_{\sigma} {\Gamma^{+-}}^{\HH}(z,w) = 0.
 \label{Formula: L Gamma^++ = 0, L Gamma^+- = 0}
\end{split}\end{equation}
Thus,
${\Gamma^{++}}^{\HH}(z,w)$ and ${\Gamma^{+-}}^{\HH}(z,w)$
are fixed up to an inessential constant.

The second equation
(\ref{Formula: Hadamard's formula})
can be reformulated now as
\begin{equation}\begin{split}
 &\Lc_{\delta} {\Gamma^{++}}^{\HH}(z,w)
 + \Lc_{\sigma} {\eta^+}^{\HH}(z) \Lc_{\sigma} {\eta^+}^{\HH}(w)
 = \beta_1^{\HH}(z) + \beta_1^{\HH}(w),\\
 &\Lc_{\delta} {\Gamma^{+-}}^{\HH}(z,\bar w)
 + \Lc_{\sigma} {\eta^+}^{\HH}(z) \overline{\Lc_{\sigma} {\eta^+}^{\HH}(w)}
 = \beta_1^{\HH}(z) + \beta_1^{\HH}(\bar w)
 \label{Formula: L Gamma^++ + L sigma^+ L sigma^+ = e + e}
\end{split}\end{equation}
for any analytic function $\beta_1^{\HH}(z)$ such that
$\overline{\beta_1^{\HH}(z)} = \beta_1^{\HH}(\bar z)$ analogous to 
\eqref{Formula: L Gamma^++ = beta + beta, L Gamma^+- = beta + beta}.
We can conclude now that the system
(\ref{Formula: L eta + 1/2 L^2 eta}--\ref{Formula: L_sigma Gamma = 0})
is equivalent to the system
(\ref{Formula: eta = eta^+ + eta^-}),
(\ref{Formula: delta/sigma j + mu [sigma,delta]/sigma + 1/2 L j = ibeta}),
(\ref{Formula: Gamma = Gamma^++ + bar Gamma^++ + Gamma^+- + bar Gamma^+-}),
(\ref{Formula: L Gamma^++ + L sigma^+ L sigma^+ = e + e}),
and
(\ref{Formula: L Gamma^++ = 0, L Gamma^+- = 0}).

Use now the fact
\begin{equation}
	{\Gamma^{++}}^{\HH}(z,w) = -\frac12 \log(z-w) + \text{analytic terms}
	\label{Formula: Gamma^++ = - 1 log z-w + hol}
\end{equation}
to obtain the singularity of
${j^+}^{\HH}$
at the origin in the half-plane chart.
Relation
(\ref{Formula: delta_N = ..., sigma_N = ... without drift})
yields
\begin{equation}
 \pm \frac{2}{z} \de_z \left( -\frac12 \log(z-w) \right) 
 \pm \frac{2}{w} \de_w \left( -\frac12 \log(z-w) \right)
 = \pm \frac{1}{zw},
 \label{Formula: 2/z de_z - 12 log(z-w) + 2/z de_w - 12 log(z-w) = ...}
\end{equation}
hence,
\begin{equation}\begin{split}
	{j^+}^{\HH}(z) &= \frac{ -i  }{ z} + \text{holomorphic part},
	\quad \text{for the forward case;}\\
	{j^+}^{\HH}(z) &= \frac{ -1  }{ z} + \text{holomorphic part},
	\quad \text{for the reverse case.}	
	\label{Formula: j^+ = mp/z + hol}
\end{split}\end{equation}
The choice of the sign of
${j^+}^{\HH}(z)$
is irrelevant. We made this choice only to be consistent with
\cite{Sheffield2010}.
The analytic terms in
(\ref{Formula: Gamma^++ = - 1 log z-w + hol})
can give a term with a sum of simple poles at $z$ and $w$, but not in the form of
the product $1/zw$.

From
(\ref{Folmula: phi+ = int ...})
we conclude that the singular part of
${\eta^+}^{\HH}$
is proportional to the logarithm of~$z$:
\begin{equation}\begin{split}
	&{\eta^+}^{\HH}(z) = \frac{i}{\sqrt{\kappa}} \log z + \text{holomorphic part}, 
	\quad \text{for the forward case;}\\
	&{\eta^+}^{\HH}(z) = \frac{1}{\sqrt{\kappa}} \log z + \text{holomorphic part}, 
	\quad \text{for the reverse case.}	 
\end{split}\end{equation}
Thus, we have
\begin{equation}\begin{split}
	\eta^{\HH}(z) =& 
	\frac{-2}{\sqrt{\kappa}} \arg z + \text{non-singular harmonic part},
	\quad \text{for the forward case;}
	\label{Formula: eta = -2/k (arg z - pi/2) + hol}
\end{split}\end{equation}
\begin{equation}\begin{split}
	\eta^{\HH}(z) =& 
	\frac{2}{\sqrt{\kappa}} \log|z| + \text{non-singular harmonic part},
	\quad \text{for the reverse case.}	
	\label{Formula: Q = 2/k log |z| + hol}
\end{split}\end{equation}
We can choose the additive constant such that, in the half-plane chart, we have
\begin{equation}
  \eta^{\HH}(+0)=-\eta^{\HH}(-0) = \frac{\pi}{\sqrt{\kappa}}
  \label{Formula: eta(+0) = eta(-0) = ...}
\end{equation}
in the forward case. This provides the jump $2 \pi/\sqrt{\kappa}$ of the value
of $\eta$ at the boundary near the origin, which is exactly the behavior
of $\eta$ needed for the flow line construction in \cite{Miller2012} and
\cite{Sheffield2010}.
However, the form
\eqref{Formula: eta(+0) = eta(-0) = ...}
is not chart independent, and only the jump
$2 \pi/\sqrt{\kappa} = \eta^{\psi}(+0) - \eta^{\psi}(-0)$
does not change its value if the boundary of $\psi(\Dc)$ is not singular in the
neighborhood of the source $\psi(a)$. 
In the reverse case, such a constant has no meaning either, because $\eta$ is
a functional over $\Hc_s^*$.

Substitute now
\eqref{Formula: j^+ = mp/z + hol}
in 
\eqref{Formula: delta/sigma j + mu [sigma,delta]/sigma + 1/2 L j = ibeta} 
in the half-plane chart, use 
\eqref{Formula: delta_N = ..., sigma_N = ... without drift},
and consider the corresponding Laurent series.
We are interested in the coefficient of the leading term
$\frac{1}{z^2}$:
\begin{equation}
	\frac{2}{z} \frac{1}{-\sqrt{\kappa}} \frac{- i }{ z} +
	\mu \frac{-2}{z^2} + \frac12 (-\sqrt{\kappa})\frac{ i }{ z^2}
	+ o\left(\frac{1}{z^2}\right) = C^+
\end{equation}
for the forward case, and we have the same expression without $i$ in the
reverse case. We can conclude that
\begin{equation}\begin{split}
	&\mu = i\frac{4 - \kappa}{4\sqrt{\kappa}},
	\quad \text{for the forward case;}\\
	&\mu = \frac{4 + \kappa}{4\sqrt{\kappa}},
	\quad \text{for the reverse case.}
	\label{Formula: mu = i pm (-4 pm k)/4/k}
\end{split}\end{equation}
Thus, the pre-pre-Schwarzians 
\eqref{Formula: tilde eta = eta - chi arg}
with $\chi$ given by
(\ref{Formula: chi = 2/k - k/2}) 
is the only one that can be realized in the forward case. The same is true for the
pre-pre-Schwarzians
\eqref{Formula: tilde eta = eta + Q log},
$Q$ given by 
\eqref{Formula: Q = 2/k + k/2}, 
and the reverse case.
\end{proof}


\section{Coupling in the case of the Dirichlet and Neumann boundary conditions}
\label{Section: The coupling in case of Dirichlet Neumann boundary conditions}

In this section, we consider some special solutions to the system
(\ref{Formula: L eta + 1/2 L^2 eta}--\ref{Formula: L_sigma Gamma = 0})
with the help of Theorem
\ref{Theorem: eta structure}.
We assume Dirichlet and Neumann boundary conditions for $\Gamma$ 
(see Examples 
\ref{Example: Dirichlet boundary conditions Gamma}
and
\ref{Example: Neumann boundary conditions Gamma})
and find the general solution in these cases. In other words, we systematically
study which $(\delta,\sigma)$-SLEs can be coupled to the GFF if
$\Gamma=\Gamma_{D}$ and $\Gamma=\Gamma_{N}$.



	\begin{theorem}
	Let a forward ($\delta$,$\sigma$)-SLE be coupled to the GFF with
	$\Hc=\Hc_s$, $\Gamma=\Gamma_D$, and
	let $\eta$ be a pre-pre-Schwarzian
	\eqref{Formula: tilde eta = eta - chi arg} 
	of order $\chi$. Then only the special combinations of
	$\delta$ and $\sigma$ summarized in Table \ref{Table: some simple cases}
	and all combinations with $\kappa=6$ and $\nu=0$ are possible.
	\label{Theorem: ppS -> simple cases}
	\end{theorem}

\begin{center}
\begin{table}
\begin{tabular}{| c | p{3cm}  | c | c | c | c | }
\hline
   & Name &   
  $\delta=$ 
  & $\sigma= $ & $\alpha=$ & $\beta=$ \\ 
 \hline
  1 & Chordal with drift &
   $ \pm 2 \ell_{-2} - \nu \ell_{-1} $ &
   $ -\sqrt{\kappa} \ell_{-1} $ &
   $ -\frac{\nu}{2} $ &
   $ \frac{-\nu^{2}}{2\sqrt{\kappa}} $ \\ \hline
  2 & Chordal with fixed time change &
   $ \pm 2  \ell_{-2} + 2 \xi \ell_0 $ &
   $ -\sqrt{\kappa} \ell_{-1} $ &
   $ 0 $ &
   $ -\frac{\xi(\pm \kappa-8)}{2\sqrt{\kappa}} $ \\ \hline
  3 & Dipolar with drift &
   $\begin{aligned}
   	\pm &2 \left(\ell_{-2} - \ell_0 \right) -\\
   	- &\nu (\ell_{-1} - \ell_{1})
   \end{aligned}$ &
   $ -\sqrt{\kappa}(\ell_{-1} - \ell_{1}) $ &
   $ -\frac{\nu}{2} $ &
   $ \frac{4-\nu^{2}}{2\sqrt{\kappa}} $ \\ \hline
  4 & One right fixed boundary point &
	$\begin{aligned}
		&\pm 2 \ell_{-2} + \\
		&+(\kappa \mp 6) \ell_{-1} + \\
    &+2(\pm 3 - \kappa \pm \xi) \ell_{0} + \\
    &+(\mp 2 + \kappa \mp 2 \xi) \ell_{1}
	\end{aligned}$ &
   $ -\sqrt{\kappa}(\ell_{-1}-\ell_{1}) $ &
   $ +\frac12(\kappa \mp 6)$ &
   $ \frac{\xi(\mp \kappa+8)}{2\sqrt{\kappa}} $ \\ \hline
  5 & One left fixed boundary point &
   $\begin{aligned}
    &\pm 2 \ell_{-2} - \\
    &-(\kappa \mp 6) \ell_{-1} + \\
    &+2(\pm 3 - \kappa \pm \xi) \ell_{0} + \\
    &-(\mp 2 + \kappa \mp 2 \xi) \ell_1
   \end{aligned}$ &
   $ -\sqrt{\kappa}(\ell_{-1}-\ell_{1}) $ &
   $ -\frac12(\kappa \mp 6)$ &
   $ \frac{\xi(\mp \kappa + 8)}{2\sqrt{\kappa}} $ \\ \hline
  6 & Radial with drift &
		$\begin{aligned} 
   		\pm &2 \left( \ell_{-2} + \ell_0 \right) - \\
   		-&\nu( \ell_{-1} + \ell_1 ) 
		\end{aligned}$ & 
   $ -\sqrt{\kappa}( \ell_{-1} + \ell_1 ) $ &
   $ -\frac{\nu}{2} $ &
   $ \frac{4-\nu^{2}}{2\sqrt{\kappa}} $ \\ \hline
\end{tabular}
 \caption{The list of ($\delta$,$\sigma$)-SLE types that can be coupled with
 the GFF with the Dirichlet boundary conditions $\Gamma=\Gamma_D$ 
 (the upper sign in the `$\pm$' and `$\mp$' pairs)
 and with the Neumann boundary conditions $\Gamma=\Gamma_N$ 
 (the lower sign).
}
\label{Table: some simple cases}
\end{table}
\end{center}

The table consists of six cases, three of which are classical. Each case is a
one-parameter family of ($\delta$,$\sigma$)-SLEs parametrized by the drift
$\nu\in\mathbb{R}$, or by the parameter $\xi\in\mathbb{R}$. The cases
may overlap for vanishing values of $\nu$ or $\xi$.
Different combinations of $\delta$ and $\sigma$ can correspond
to essentially the same process in $\Dc$ but written in different coordinates
as discussed in Section 
\ref{Section: Equivalence and normalization SLE}.
In the columns of $\delta$ and $\sigma$, we present only one example of such
choices in each case.

Some special cases of the forward coupling presented here have been considered
before.
The chordal SLE without drift (case $1$ from the table with $\nu=0$) was
considered in 
\cite{Kang2011}, 
the radial SLE without drift (case $6$ from the
table with $\nu=0$) in 
\cite{Kang2012a}, 
and the dipolar SLE without drift
(case $3$ from the table with $\nu=0$) appeared in 
\cite{Kang}.
Case $2$ actually corresponds to the same measure as the chordal SLE but
stopped at time $t=\frac14\xi$ if $\xi>0$
(see Section 
\ref{Section: Chordal case with fix time change}).
Cases $4$ and $5$ are mirror images of each other. They are discussed in
Section 
\ref{Section: Case with one fixed point}. 
The reverse coupling of the chordal SLE with zero drift is considered in
\cite{Sheffield2010}.



\medskip
\noindent
{\it Proof of Theorem \ref{Theorem: ppS -> simple cases}}.
Let us use Theorem \ref{Theorem: The coupling theorem} and assume the
Dirichlet boundary conditions for $\Gamma=\Gamma_D$, i.e., 
\begin{equation}
 {\Gamma^{++}}^{\HH}(z,w) = -\frac12 \log(z-w),\quad
 {\Gamma^{+-}}^{\HH}(z,\bar w) = -\frac12\log (z- \bar w)
 \label{Formula: Gamma_D^++ = Log...}
\end{equation}
in Theorem
\ref{Theorem: eta structure}.
The condition
(\ref{Formula: L Gamma^++ = beta + beta, L Gamma^+- = beta + beta})
is satisfied for any complete vector field $\sigma$ and some $\sigma$-dependent
$\beta_2$, which is irrelevant.

In order to obtain $j^+$, we first remark that due to the M\"obius invariance
\eqref{Formula: Lv Gamma_D = 0}
we can ignore the polynomial part of $\delta^{\HH}(z)$:
\begin{equation}
 \Lc_{\delta} \Gamma^{\HH}(z,w) =
 \left(
  \frac{2}{z} \de_z + \frac{2}{\bar z} \de_{\bar z} +
  \frac{2}{w} \de_w + \frac{2}{\bar w} \de_{\bar w}
 \right)
 \Gamma^{\HH}(z,w).
\end{equation}
Using
(\ref{Formula: L Gamma^++ + L sigma^+ L sigma^+ = e + e}),
(\ref{Formula: Gamma_D^++ = Log...}),
and
(\ref{Formula: 2/z de_z - 12 log(z-w) + 2/z de_w - 12 log(z-w) = ...}),
we obtain
\begin{equation}\begin{split}
	&{j^+}^{\HH}(z) = \frac{-i}{z} + i\alpha,\quad \alpha\in\C, 
	\label{Formula: j+ = i/z + ia}
\end{split}\end{equation}
with
\begin{equation}
 \beta_1(z) = \frac{\alpha}{z} - \frac{{\alpha}^2}{2}.
\end{equation}

In order to satisfy all the conditions formulated in Theorem
\ref{Theorem: eta structure}
we need to check
\eqref{Formula: delta/sigma j + mu [sigma,delta]/sigma + 1/2 L j = ibeta}.
Substituting
\eqref{Formula: j+ = i/z + ia} into
\eqref{Formula: delta/sigma j + mu [sigma,delta]/sigma + 1/2 L j = ibeta}
gives
\begin{equation}\begin{split}
 & \frac{\delta}{\sigma} j^+ + \mu \frac{[\sigma,\delta]}{\sigma}
 + \frac12 \Lc_\sigma^+ j^+ = i \beta ~ \Leftrightarrow \\
 & \delta j^+ + \mu\, [\sigma,\delta] + \frac12 \sigma \Lc_{\sigma}^+ j^+
  - i \beta \sigma  = 0
 ~ \Leftrightarrow \\
 & \delta^{\HH}(z) \left( \frac{-i }{z} + i \alpha \right)
 + \mu\, [\sigma,\delta]^{\HH}(z)
 + \frac12 \left(\sigma^{\HH}(z)\right)^2 \de \left( \frac{-i }{z} 
 + i \alpha \right) - i \beta \sigma^{\HH}(z) = 0. 
 \label{Formula: A phi = ib then ...}
\end{split}\end{equation}
In what follows, we will use the half-plane chart in the proof.
With the help of
(\ref{Formula: S - transform})
and
(\ref{Formula: R - transform})
we can assume without loss of generality that $\sigma^{\HH}$ has one of three
possible forms:
\begin{enumerate} [1.]
  \item $\sigma^{\HH}(z)=-\sqrt{\kappa}$,
  \item $\sigma^{\HH}(z)=-\sqrt{\kappa}(1-z^2)$,
  \item $\sigma^{\HH}(z)=-\sqrt{\kappa}(1+z^2)$.
\end{enumerate}
Let us consider these cases one by one.

\medskip
\noindent
{\bf 1. $\sigma(z) =-\sqrt{\kappa}$.} \\
 Inserting
 \eqref{Formula: delta_N = ..., sigma_N = ... without drift}
 the relation
 \eqref{Formula: A phi = ib then ...}
 reduces to
 \begin{equation}\begin{split}
  &\frac{-2 + \frac{\kappa}{2} -2 i \sqrt{\kappa}\mu }{z^2} +
  \frac{ 2 \alpha - \delta_{-1} }{z} +
  \left(\beta \sqrt{\kappa} + \alpha \delta_{-1} - \delta_{0}
   + i \sqrt{\kappa} \mu \delta_{0} \right)
   +\\+&
   z \left( \alpha \delta_0 - \delta_1 + 2 i \sqrt{\kappa} \mu \delta_1 \right)
   + z^2 \alpha \delta_1 \equiv 0 \quad \Leftrightarrow \\
  &(\ref{Formula: mu = i pm (-4 pm k)/4/k}),\quad
  2 \alpha - \delta_{-1}=0,\quad
  \beta \sqrt{\kappa} + \alpha \delta_{-1} - \delta_{0}
   + i \sqrt{\kappa} \mu \delta_{0} = 0,\quad \\
  &\alpha \delta_0 - \delta_1 + 2 i \sqrt{\kappa} \mu \delta_1 = 0,\quad
  \alpha \delta_1 =0.
\end{split}\end{equation}
There are three possible cases:
\begin{enumerate}
 	\item 
  \begin{equation}
  	\delta_{-1} = 2\alpha,\quad
  	\delta_0 = 0,\quad
  	\delta_{1} = 0,\quad
  	\kappa>0,\quad
  	\beta = \frac{-2 \alpha^2}{\sqrt{\kappa}}.
	\end{equation}
	It is convenient to use the drift parameter
	\eqref{Formula: nu := ...}. Thus,
	\begin{equation}
		\nu=-2\alpha,
		\label{Formula: nu := -2 alpha}
	\end{equation}
	which is related to the drift in the chordal equation.
	This case is presented in the first line of Table
	\ref{Table: some simple cases}.	  
	\item 
	\begin{equation}
  	\delta_{-1} = 0,\quad
  	\delta_0 = -\frac{4 \beta \sqrt{\kappa}}{\kappa - 8},\quad
   	\delta_{1} = 0,\quad
   	\kappa>0,\quad
   	\alpha = 0.
	\end{equation}
	This case is presented in the second line of Table ($\xi\in\mathbb{R}$)
	and discussed in detail in Section
	\ref{Section: Chordal case with fix time change}.
	\item 
	\begin{equation}
		\delta_{-1}=0,\quad
		\delta_0=2\sqrt{5}\beta,\quad
		\delta_1\in\mathbb{R},\quad
		\kappa=6,\quad
		\alpha=0.
	\end{equation}	
	This is a general case of $\delta$ with $\kappa=6$ and $\nu=0$.	
\end{enumerate}

\medskip
\noindent
{\bf 2. $\sigma^{\HH}(z)=-\sqrt{\kappa}(1-z^2)$.} \\
Relation (\ref{Formula: A phi = ib then ...}) reduces to
\begin{equation}\begin{split}
  &\frac{-2 + \frac{\kappa}{2} + 2 i \mu \sqrt{\kappa} }{z^2} +
   \frac{ 2 \alpha - \delta_{-1} }{z}
   +
   \left(\beta \sqrt{\kappa} - \kappa + 6 i \sqrt{\kappa} \mu + \alpha
   \delta_{-1} - \delta_0 + i \sqrt{\kappa} \mu \delta_0 \right)
   +\\+&
   z \left( 2 i \sqrt{\kappa} \mu \delta_{-1} + \alpha \delta_0 - \delta_1
    + 2 i \sqrt{\kappa} \mu \delta_1 \right)
   +
   z^2 \left( - \beta \sqrt{\kappa} + \frac{\kappa}{2} + i \sqrt{\kappa} \mu
   \delta_0 + \alpha \delta_1 \right) = 0 \quad \Leftrightarrow \\
  &(\ref{Formula: mu = i pm (-4 pm k)/4/k}),\quad
  2 \alpha - \delta_{-1} = 0,\quad
   \beta \sqrt{\kappa} - \kappa + 6 i \sqrt{\kappa} \mu + \alpha \delta_{-1} -
   \delta_0 + i \sqrt{\kappa} \mu \delta_0 = 0,\\
  & 2 i \sqrt{\kappa} \mu \delta_{-1} + \alpha \delta_0 - \delta_1
    + 2 i \sqrt{\kappa} \mu \delta_1 = 0,\quad
  - \beta \sqrt{\kappa} + \frac{\kappa}{2} + i \sqrt{\kappa} \mu \delta_0
    + \alpha \delta_1 = 0.
\end{split}\end{equation}
There are four solutions, each of which is a two-parameter family.
The first one corresponds to the dipolar SLE with the drift $\nu$,
line 3 in Table \ref{Table: some simple cases}.
The second and third solutions are `mirror images' of each other,
as can be seen from lines 4 and 5 in the table. They are parametrized
by $\xi:=\frac{2\beta\sqrt{\kappa}}{8-\kappa}$ and discussed in detail in
Section 
\ref{Section: Case with one fixed point}. 
The fourth case is when
\begin{equation}
	\delta_{-1}=0,\quad
	\delta_0=2(\sqrt{6}\beta-3),\quad
	\delta_1\in\mathbb{R},\quad
	\kappa=6,\quad
	\alpha=0. 
\end{equation}
This is a general form of $\delta$ with $\kappa=6$ and $\nu=0$.

\medskip
\noindent
{\bf 3. $\sigma^{\HH}(z)=-\sqrt{\kappa}(1+z^2)$.} \\
Relation (\ref{Formula: A phi = ib then ...}) reduces to
\begin{equation}\begin{split}
  &\frac{-2 + \frac{\kappa}{2} -2 i \sqrt{\kappa} \mu }{z^2} +
   \frac{ 2 \alpha - \delta_{-1} }{z}
   +\\+&
   \left( \beta \sqrt{\kappa} - \kappa + 6 i \sqrt{\kappa} \mu + \alpha
   \delta_{-1} - \delta_0 + i \sqrt{\kappa} \mu \delta_0 \right)
   +\\+&
   z \left( 2 i \kappa \mu \delta_{-1} + \alpha \delta_0 - \delta_1
    + 2 i \sqrt{\kappa} \mu \delta_1 \right)
   +\\+&
   z^2 \left( -\beta \sqrt{\kappa} + \frac{\kappa}{2} + i \sqrt{\kappa} \mu
   \delta_0 + \alpha \delta_1 \right) = 0 \quad \Leftrightarrow \\
  &(\ref{Formula: mu = i pm (-4 pm k)/4/k}),\quad
  2 \alpha - \delta_{-1} = 0,\quad
   \beta \sqrt{\kappa} - \kappa + 6 i \sqrt{\kappa} \mu + \alpha
   \delta_{-1} - \delta_0 + i \sqrt{\kappa} \mu \delta_0 = 0,\\
  &2 i \kappa \mu \delta_{-1} + \alpha \delta_0 - \delta_1
    + 2 i \sqrt{\kappa} \mu \delta_1 = 0,\quad
  -\beta \sqrt{\kappa} + \frac{\kappa}{2} + i \sqrt{\kappa} \mu
  \delta_0 + \alpha \delta_1 =0.
\end{split}\end{equation}
The first solution is presented in line 6 of Table
\ref{Table: some simple cases},
where it is again convenient to introduce the parameter $\nu$
related to the drift in the radial equation. The second solution is
\begin{equation}
	\delta_{-1}=0,\quad
	\delta_0=2(\sqrt{6}\beta-3),\quad
	\delta_1\in\mathbb{R},\quad
	\kappa=6,\quad
	\alpha=0. 
\end{equation}
This is a general form of $\delta$ with $\kappa=6$ and $\nu=0$.
\quad\qed

The following theorem is an analog of the previous one. 
\begin{theorem}
Let a reverse ($\delta$,$\sigma$)-SLE be coupled to the GFF with
$\Gamma=\Gamma_N$ and $\Hc=\Hc_s^*$, and let $\eta$ be a pre-pre-Schwarzian
\eqref{Formula: tilde eta = eta + Q log}
of order $Q$.
Then only the combinations of
$\delta$ and $\sigma$ 
summarized in Table \ref{Table: some simple cases}
are possible.
\label{Theorem: ppS -> simple cases 2}
\end{theorem}


\begin{proof}
We use the same method as in the proof of the previous theorem. The difference
is that we use the space $\Hc_s^*$ instead of $\Hc_s$. We still have
\eqref{Formula: Gamma_D^++ = Log...}, 
but instead of 
\eqref{Formula: j+ = i/z + ia} we have
\begin{equation}
 {j^+}^{\HH}(z) = \frac{-1}{z} - \alpha,\quad \alpha\in\C. 
 \label{Formula: j+ = -1/z - a}
\end{equation}
This definition of $\alpha$ makes the formulas more similar to those in the previous case.
We also have
\begin{equation}\begin{split}
	&\delta^{\HH}(z) \left( \frac{-1 }{z} - \alpha \right)
	+ \mu\, [\sigma,\delta]^{\HH}(z)
	+ \frac12 \left(\sigma^{\HH}(z)\right)^2 \de 
	\left( \frac{-1 }{z} - \alpha \right) 
	+ (\beta-i\beta') \sigma^{\HH}(z) \equiv 0. 
\end{split}\end{equation}
instead of 
\eqref{Formula: A phi = ib then ...}. 
We obtain that $\beta'=0$, the same relations for $\beta$, and the same set of
solutions from Table
\ref{Table: some simple cases}
with slightly different signs (the lower ones in all pairs `$\pm$' and `$\mp$').
We do not have the special case $\kappa=6$ as before, because the analogous
solutions give the negative value $\kappa=-6$.
\end{proof}

\section{Alternative definition of $G_t^{-1}{}_* \eta[f]$}
\label{Section: Alternative definition og G_* eta}

This section actually begins the conclusions and perspectives part.
We briefly consider how one can resolve the difficulties in the definition of 
$G_t^{-1}{}_* \eta[f]$, $G_t^{-1}{}_* \Gamma[f,f]$,
and other such quantities, which appear when 
$\Im G_t^{-1}\cap \supp f\neq \emptyset$. 

Consider the initial value problem 
\begin{equation}
	\dS G_{t,\varepsilon} 
	= \delta_{\varepsilon} \circ G_{t,\varepsilon} dt 
	+ \sigma \circ G_{t,\varepsilon} \dS B_t
	,\quad t\in[0,+\infty),\quad \varepsilon>0,
\end{equation}
where the vector field $\delta_{\varepsilon}$ is defined by 
\eqref{Formula: delta_e := ...}.
The solution 
$\{G_{t,\varepsilon}\}_{t\in[0,+\infty)}$
is a family of automorphisms of $\Dc$ that are not conformal but continuously
differentiable. It can be understood as a regularized version of 
$\{G_t\}_{t\in[0,+\infty)}$
because
\begin{equation}
	G_{t,\varepsilon}\xrightarrow{\varepsilon \map 0} G_t,\quad t\in[0,+\infty),
\end{equation}
pointwise. 

\begin{center}
\begin{figure}[h]
\centering
	\includegraphics[keepaspectratio=true]
    	{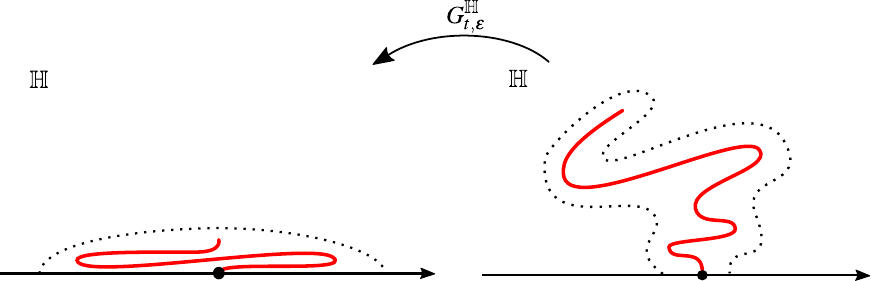}
   \caption{Schematic illustration of how the map 
   $G_{t,\varepsilon}^{\HH}$
   acts in the case when the slit $\gamma_t$ of 
   $G_{t}^{\HH}$ is a simple curve. It is denoted by the red line on the
   right-hand side. On the left-hand side, the red line is the corresponding 
   image $G_{t,\varepsilon}^{\HH}(\gamma_t)$. 
   The map $G_{t,\varepsilon}^{\HH}$ is conformal outside the dashed line.
   \label{Figure: smooth map - simple curve }}
\end{figure}
\end{center}

We remark that for a non-conformal map $F$ the definition of the
pre-pre-Schwarzian 
\eqref{Formula: tilde eta = eta + Q log}
should be generalized to 
\begin{equation}
	F_* \eta^{\psi}(z) = 
	\eta^{\psi}( \tau(z) ) 
	+ Q \, \log \left( 
		\de_z \tau(z) \de_{\bar z} \overline{\tau(z)} -  	
		\de_{\bar z} \tau(z) \de_{z} \overline{\tau(z)}  	
	\right),\quad \tau(z) = \left(F^{\psi}\right)^{-1}(z).
\end{equation} 
Random variables such as 
$G_{t,\varepsilon}^{-1}{}_* \eta[f]$ 
and 
$G_{t,\varepsilon}^{-1}{}_* \Gamma[f,f]$
are defined for $t\in[0,+\infty)$ because 
$G_{t,\varepsilon}$ is an automorphism.
Thus, we can define 
\begin{equation}
	G_{t}^{-1}{}_* \eta[f] 
	:= \lim\limits_{\varepsilon\map 0} G_{t,\varepsilon}^{-1}{}_* \eta[f]
	,\quad t\in[0,+\infty). 
	\label{Formula: G_* eta := limit G_e_* eta}	
\end{equation}
If $t\leq T[f]$ this definition coincides with the standard one.  
We expect it to be equivalent to the approach described in 
\cite{Schramm2010}
if the hull $\K_t$ is a simple curve ($\kappa\leq4$). 
If $\K_t$ is generated by a curve with self-touchings, then the pushforward 
$G_{t}^{-1}{}_* \eta[f]$ 
depends also on the chain 
$\{G_s\}_{s\in[0,t]}$,
or equivalently, on the driving function 
$\{u_s\}_{s\in[0,t]}$, 
and not only on the final map 
$G_t$.

Let us now consider how the limit in
\eqref{Formula: G_* eta := limit G_e_* eta}
works in this case; see Fig. 
\ref{Figure: closed domains}.
Let the hull $\K_t$ be generated by a curve $\gamma_t$ with self-touchings. It
can be represented as the union of the curve $\gamma_t$ and the connected open
components $U^+_n$ and $U^-_n$, $n=1,2,\dotso$, bounded by $\gamma_t$:
\begin{equation}
	\K_t = \gamma_t \cup \bigcup\limits_{n=1,2,\dotso} U^+_n 
	\cup \bigcup\limits_{n=1,2,\dotso} U^-_n.  
\end{equation}

\begin{center}
\begin{figure}[h]
\centering
	\includegraphics[keepaspectratio=true]
    	{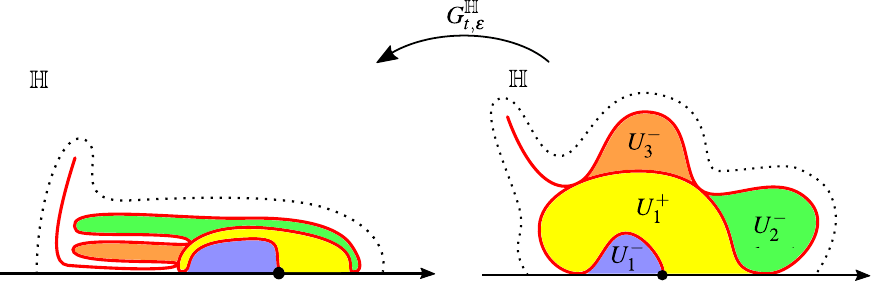}
  \caption{This is a schematic illustration of how the map
  $G_{t,\varepsilon}^{\HH}$ acts in the case of a self-touching curve $\gamma_t$.
  The red line on the right-hand side is the slit $\gamma$ of $G_{t}^{\HH}$, on
  the left-hand side, the red line is the image
  $G_{t,\varepsilon}^{\HH}(\gamma_t)$. The map $G_{t,\varepsilon}^{\HH}$ is
  conformal beyond the dashed line.
  \label{Figure: closed domains}}
\end{figure}
\end{center}

We denote by the superscripts `$+$' and `$-$' the components bounded by
$\gamma_t$ traversed clockwise and counterclockwise, respectively.
We state without proof that
\begin{equation}
	\lim\limits_{\varepsilon\map 0} G_{t,\varepsilon}^{\HH}(U_n^{+}) 
	\subset [0,+\infty),\quad 
	\lim\limits_{\varepsilon\map 0} G_{t,\varepsilon}^{\HH}(U_n^{-}) 
	\subset (-\infty,0],\quad n=1,2,\dotso.
\end{equation}
In other words, the sets $U_n^{+}$, $n=1,2,\dotso$ are `pressed' onto the positive
real semi-axis and $U_n^{-}$ onto the negative one. 
Consider, for instance, the case from
Section 
\ref{Section: Coupling of chordal SLE and Dirichlet GFF}
with $\kappa=4$, $\nu=0$,
when $\eta$ is a scalar and
\begin{equation}
	\eta^{\HH}(z)=-\arg z + \pi/2,\quad z\in\HH.
\end{equation}
We use the results of Section
\ref{Section: Linear functionals and change of coordinates}
to conclude that 
\begin{equation}\begin{split}
	&G_t^{-1}{}_* \eta[f] 
	=	\int\limits_{\psi^{\HH}(\Dc \setminus \K_t)} \eta^{\HH}(G^{\HH}_t(z))
	f^{\HH}(z) l(dz) +\\+& \frac{\pi}{2}\sum\limits_{n=1,2,\dotso} \left[ 
		\int\limits_{\psi^{\HH}(U_n^+)} f^{\HH}(z) l(dz)
	  -\int\limits_{\psi^{\HH}(U_n^-)} f^{\HH}(z) l(dz).
	 \right]
	 ,\quad f\in\Hc_s.
\end{split}\end{equation}
Equivalently, we can define 
$\eta^{\HH}(G^{\HH}_t(z))$ 
for $z\in\K_t$ to be equal to 
$+\pi/2$ and $-\pi/2$ on $U_n^+$ and $U_n^-$, respectively. 
 
In order to see that the proposed definition of the pushforward gives the
extension of Theorem
\ref{Theorem: The coupling theorem}
for $t>T[f]$,
one can apply the following method.
Let $\eta_{\varepsilon}$ and $\Gamma_{\varepsilon}$ be solutions of 
(\ref{Formula: L eta + 1/2 L^2 eta}--\ref{Formula: L_sigma Gamma = 0})
with $\delta_{\varepsilon}$ instead of $\delta$. Then 
\begin{equation}
	\hat \phi_{\varepsilon}[f] 
	= e^{\frac12 \Gamma_{\varepsilon}[f,f]+\eta_{\varepsilon}[f]}
\end{equation}
induces local martingales, just as $\hat \phi$ does. The last part of the
method is to show that 
$\eta_{\varepsilon}\xrightarrow{\varepsilon \map 0}\eta$,
$\Gamma_{\varepsilon}\xrightarrow{\varepsilon \map 0}\Gamma$,
and consequently,
$\hat \phi_{\varepsilon}\xrightarrow{\varepsilon \map 0}\hat \phi$.

\section*{Conclusions and perspectives}

\begin{enumerate}[1.]

\item
We did not prove it in this monograph, but our experience suggests that we have listed
all possible ways of coupling the GFF with $(\delta,\sigma)$-SLE if we assume that
$\Gamma$ transforms as a scalar, see (\ref{Formula: G B(z,w) = B(G(z),G(w))}),
and that $\eta$ is a pre-pre-Schwarzian. It would be useful to prove this as
well as to understand what should be relaxed in our framework in order to couple all 
$(\delta,\sigma)$-SLEs.

\item 
The pre-pre-Schwarzian rule 
(\ref{Formula: G eta(z) = eta(G(z)) + mu log G'(z) + ...})
is motivated by the local geometry of the SLE curve 
\cite{Sheffield2010}. 
In principle, one can consider alternative rules. Moreover, the scalar behavior
of $\Gamma$
can also be relaxed because the harmonic part 
$H^{\psi}(z,w)$ 
in 
\eqref{Formula: Gamma = Log + H} 
can transform in many ways. Such a more general coupling is natural and can be
thought of as a generalization of the coupling in Sections 
\ref{Section: Coupling of forward dipolar SLE and combined Dirichlet-Neumann GFF}
and
\ref{Section: Coupling with twisted GFF}
for arbitrary $\kappa$.  

\item
We considered only the simplest case of one GFF. It would be
interesting to examine tuples of 
$\{\Phi_i\}_{i=1,2,\dotso n}$. 

\item
The Bochner-Minlos Theorem
\ref{Theorem: Bochner-Minols}
suggests considering not only free fields but also, for example, some polynomial
combinations in the exponential of
(\ref{Formula: GFF chracteristic function 2}).
In particular, the quartic functional 
\begin{equation}
	\hat \phi[f] = e^{
		\eta[f] + \frac12 \Gamma[f,f] 
		+ \frac16 W_3[f,f,f] + \frac{1}{24} W_4[f,f,f,f] 
	},
\end{equation}
for some symmetric cubic $W_3$ and quartic $W_4$ functionals with some
transformation properties (not necessarily scalar),
corresponds to conformal field theories
related to 2-to-2 scattering of particles in dimension two.
Then, the generalization of the system 
(\ref{Formula: L eta + 1/2 L^2 eta}--\ref{Formula: L_sigma Gamma = 0})
is
\begin{equation}\begin{split}
	& \Lc_{\delta} \eta[f] + \frac12 \Lc_{\sigma}^2 \eta[f] = 0,\\
	& \Lc_{\delta} \Gamma[f,g] + \Lc_{\sigma} \eta[f] \Lc_{\sigma} \eta[g] 
	+ \frac12 \Lc_{\sigma}^2 \Gamma[f,g]= 0,\\
	& \Lc_{\delta} W_3[f,g,h] 
	+ \Lc_{\sigma} \Gamma[f,g] \Lc_{\sigma} \eta[h] 
	+ \text{symmetric terms} =0,\\
	& \Lc_{\delta} W_4[f,g,h,l] 
	+ \Lc_{\sigma} \Gamma[f,g] \Lc_{\sigma} \Gamma[h,l] 
	+ \text{symmetric terms} =0,\\
	& \Lc_{\sigma} W_3[f,g,h]=0,\quad 
	\Lc_{\sigma} W_4[f,g,h,l] =0,\\	
	&f,g,h,l\in\Hc.	
\end{split}\end{equation}
Any nontrivial solution of this system (which corresponds to positive
definite $\hat \phi$) gives a new type of coupling, which is particularly
interesting because it is related to a field theory with a self-interaction. 

For the general form of the Schwinger functionals $S_n$, the necessary and
sufficient condition for the coupling is
\begin{equation}
	\left( \Lc_{\delta} + \frac12\Lc_{\sigma}^2 \right)
	S_n(z_1,z_2,\dotso z_n)=0,\quad n=1,2,\dotso,
\end{equation}
which can be understood as a version of the BPZ equation from \cite{Belavin1984}.

\end{enumerate}

%% file: Classical_SLE.tex
\chapter{Important special cases}
\label{Chapter: Classical cases}

In this chapter, we collect the most important special cases of
$(\delta,\sigma)$-L\"owner equations and couplings. 
Each section starts with a special choice of $\delta$ and $\sigma$ and the
corresponding motivation. We continue with a discussion of the
properties of the deterministic and stochastic versions that are specific to the
given $\delta$ and $\sigma$. The last part of each section is dedicated to
the coupling with the GFF. We consider all known types of couplings related to the
given $\delta$ and $\sigma$.

To compare different cases, we have to consider them in the same chart.
We use two charts for this purpose: the half-plane $\HH$ chart and the unit
disk $\D$ chart. The advantage of the first one is the simplicity of formulas.
The unit disk is sometimes more convenient because it is a bounded domain with
smooth boundary. It is the chart in which we compare the global behavior of
the slits. Numerical simulations of all $(\delta,\sigma)$-SLEs are also
presented in the unit disk chart.

We discuss in detail only the chordal case. The analogous comments for
the other cases are omitted in order to avoid repetition.

\section{Chordal case}
\label{Section: chordal case}

\subsection{Chordal L\"owner equation}
\label{Section: Chordal Loewner equation}

The chordal version of the L\"owner equation is the best known and most studied. 
The deterministic equation was introduced by Charles L\"owner in 1923. Its stochastic
version was introduced by Schramm \cite{Schramm2000} in 2000.

\begin{figure}[h]
\centering
  \begin{subfigure}[t]{0.33\textwidth}
		\centering
    \includegraphics[width=5cm,keepaspectratio=true]
    	{Pictures/Vector_fileds_v=_-2_.pdf}
    \caption{Flow lines of the vector field $\delta_c$.}
  \end{subfigure}%
  ~  
  \begin{subfigure}[t]{0.33\textwidth}
		\centering
    \includegraphics[width=5cm,keepaspectratio=true]
    	{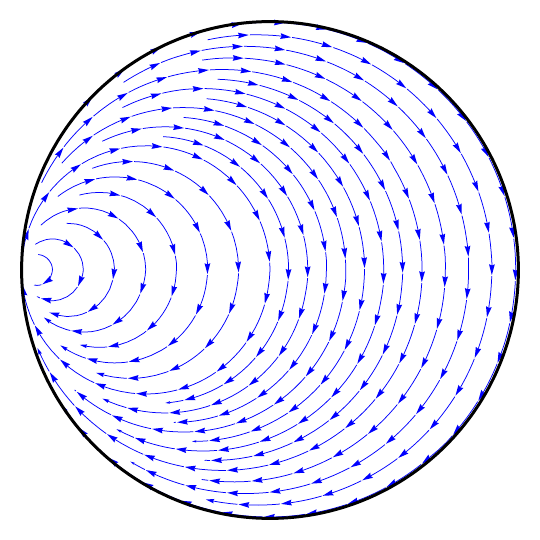} 
  	\caption{Flow lines of the vector field $\sigma_c$.}
	\end{subfigure}%
	~
	\begin{subfigure}[t]{0.33\textwidth}
	\centering
		\includegraphics[width=5cm,keepaspectratio=true]
			{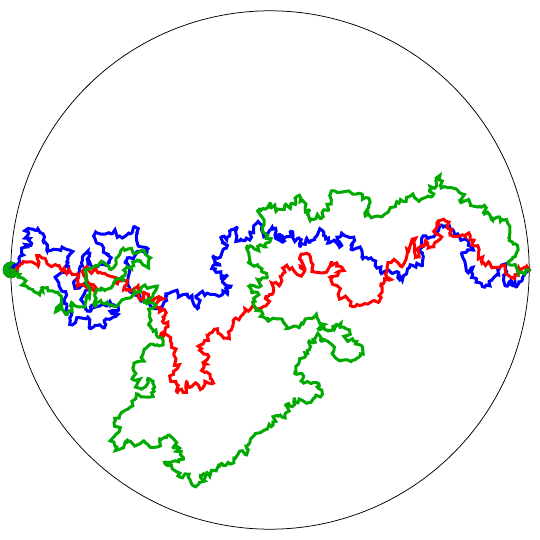}
		\caption{Three samples of chordal SLE slits for $\kappa=2$. All slits tend to
		the point $\phi^{\D}(b)=-1$.}
	\end{subfigure}
 	\caption{Flow lines of the chordal vector fields $\delta_c$ and $\sigma_c$, and
	corresponding slit samples in the unit disk chart.
 	\label{Figure: Chordal SLE illustration}}
\end{figure}

Assume that both holomorphic vector fields $\delta$ and $\sigma$ have a
second-order zero at a boundary point $b\in\Dc$. This point can always be mapped
to 
$\psi^{\D}(b)=1$ 
in the unit disk chart, or equivalently, to 
$\psi^{\HH}(b)=\infty$ 
in the half-plane chart. Thus, without loss of
generality, we can assume that 
\begin{equation}\begin{split}
	&\sigma = \sigma_{-1} \ell_{-1}, \quad \sigma_{-1}\in\mathbb{R}, \\
	&\delta = \delta_{-2} \ell_{-2}+\delta_{-1} \ell_{-1}, \quad \delta_{-2}>0,
	\quad \delta_{-1}\in\mathbb{R},
\end{split}\end{equation}
where we restrict ourselves to the forward case.
Using the transforms $\mathscr{V}$, $\mathscr{T}$, and $\mathscr{D}$, we can
normalize
\begin{equation}\begin{split}
	\delta_c = 2 \ell_{-2},\quad
	\sigma_c = -\ell_{-1}.
 \label{Formula: Chordal delta and sigma}
\end{split}\end{equation}
This normalization is chosen for historical reasons, because the chordal
L\"owner equation in the half-plane chart usually has the form  
\begin{equation}
	\dot G_t^{\HH}(z) = \frac{2}{G_t^{\HH}(z)} - \dot u_t.
	\label{Formula: dot G = 2/G - dot u}
\end{equation}
This corresponds to the vector fields
\begin{equation}\begin{split}
	\delta_c^{\HH}(z) = \frac{2}{z},\quad
	\sigma_c^{\HH}(z) = -1.
 \label{Formula: Chordal delta and sigma in H det}
\end{split}\end{equation}
which are the relations
\eqref{Formula: Chordal delta and sigma}
in the half-plane chart.


To obtain a more traditional form of the chordal L\"owner equation
which is also valid for non-differentiable
driving functions, we can apply the procedure from the beginning of Section
\ref{Section: Definition and basic properties}.
The one-parameter group of M\"obius automorphisms 
$H[\sigma]_s$ 
can be found in the half-plane chart from the equation 
(\ref{Formula: d H = sigma H ds}), 
which in our case takes the form
\begin{equation}\begin{split}
 \dot H_s[\sigma_c]^{\HH}(z) = -1 ~,\quad
 H_0[\sigma_c]^{\HH} (z) = z,\quad z\in \HH,
\end{split}\end{equation}
has the solution 
\begin{equation}\begin{split}
 H_s[\sigma_c]^{\HH}(z) = z - s,\quad z\in \HH.
\end{split}\end{equation}
We define
\begin{equation}
	g_t^{\HH}(z) := 
	\left( H_{u_t}[\sigma_c]^{-1} \circ G_t \right)^{\HH}(z)=
	G_t^{\HH}(z) + u_t,
\end{equation}
which is a form of
\eqref{Formula: g_t = H[sigma] circ G_t}.
Taking into account that 
\begin{equation}
	\left( H_s[\sigma_c]_*^{-1} \delta \right)^{\HH}(z) = 
	\frac{2}{z-s}
\end{equation}
we derive the equation for $g_t^{\HH}(z)$:
\begin{equation}\begin{split}
 \dot g_t^{\HH}(z) = \frac{2}{g_t^{\HH}(z)-u_t},
 \label{Formula: dot g = 2/(g-u)}
\end{split}\end{equation}
which is known as the chordal L\"owner equation.

We remark that the maps $G_t^{\HH}$ and $g_t^{\HH}$ satisfy the
normalization conditions
\begin{equation}
	G_t^{\HH}(z) = z - u_t + O(z^{-1})
	\label{Formula: chordal normalization for G in H}
\end{equation}
and
\begin{equation}
	g_t^{\HH}(z) = z + O(z^{-1}).
\end{equation}
These conditions uniquely determine the maps $G_t$ and $g_t$ for a given $\K_t$ if 
$\psi^{\HH}(\K_t)$ is compact.
\index{chordal normalization}


In the unit disk chart, the relations above have the form 
\begin{equation}
	\delta^{\D}_c(z) = \frac{(1+z)^3}{1-z},\quad
	\sigma^{\D}_c(z) = -\frac{i}{2}(1+z)^2.
 \label{Formula: Chordal delta and sigma in D det}
\end{equation}

The analog of
\eqref{Formula: dot G = 2/G - dot u}
is
\begin{equation}
	\dot G_t^{\D}(z) = 
	\frac{\left( 1+G_t^{\D}(z) \right)^3}
		{1-G_t^{\D}(z)}
	- \frac{i}{2}(1+G_t^{\D}(z))^2 \dot u_t.
\end{equation}
The equation 
\eqref{Formula: dot g = 2/(g-u)}
in the unit disk chart is
\begin{equation}
	\dot g_t^{\D}(z) = 
	\frac{-4i(1+g_t^{\D}(z))^3}
		{\left(2i+u_t+ u_t g_t^{\D}(z)\right)
			(i+u_t-i g_t^{\D}(z) + u_t g_t^{\D}(z))}.
\end{equation}


To define the stochastic version of the chordal L\"owner equation (chordal SLE),
we can substitute $u_t=\sqrt{\kappa}B_t+\nu t$
in
\eqref{Formula: dot G = 2/G - dot u}
or
\eqref{Formula: dot g = 2/(g-u)}.
In the first case, we thus obtain the chordal SLE in the Stratonovich form
\begin{equation}\begin{split}
 d^{\mathrm{S}} G_t^{\HH}(z) = 
 \frac{2}{G_t^{\HH}(s)}dt - \sqrt{\kappa} d^{\mathrm{S}} B_t - \nu dt
 \label{Formula: chordal SLE in H in Strat}
\end{split}\end{equation}
The It\^o form of the same equation in the half-plane chart is
\begin{equation}\begin{split}
	d^{\mathrm{Ito}} G_t^{\HH}(z) = 
	\frac{2}{G_t^{\HH}(s)}dt - \sqrt{\kappa} d^{\mathrm{Ito}} B_t - \nu dt.
\end{split}\end{equation}
The last two equations are identical, because
${\sigma^{\HH}}'(z)\equiv 0$.
In other charts, ${\sigma^{\psi}}'(z)\not \equiv 0$, and the Stratonovich and
It\^o forms differ. 

Thus, we define $\delta$ and $\sigma$ for the stochastic case by
\begin{equation}
  \delta = \pm\delta_c + \nu \sigma_c = \pm 2 \ell_{-2} - \ell_{-1} \nu,\quad
  \sigma = \sqrt{\kappa} \sigma_c = -\sqrt{\kappa} \ell_{-1},\quad
  \kappa>0,\quad \nu \in \mathbb{R},
  \label{Formula: delta and sigma chordal}
 \end{equation}
where we add the parameters $\kappa$ and $\nu$ (see the comments in Section
\ref{Section: Equivalence and normalization SLE}), and
`$\pm$' 
corresponds to the forward and reverse cases.

The chordal SLE without drift has a specific property that holds only in this
case. The random laws of 
$\{G_t\}_{t\in[0,+\infty)}$ 
and
$\{\K_t\}_{t\in[0,+\infty)}$ 
are invariant with respect to the transform $\mathscr{P}$, see Section
\ref{Section: Equivalence and normalization SLE}. In other words, this is the
only ($\delta,\sigma$)-SLE whose law is scale invariant. This property makes the
chordal case special, see also the end of Section 
\ref{Section: Classification and normalization from algebraic point of view}.
This is why it is natural to use the chordal equation to study local properties of
the slit.

%
%

\subsection{Coupling of forward chordal SLE and Dirichlet GFF}
\label{Section: Coupling of chordal SLE and Dirichlet GFF}

This is the simplest type of coupling considered in
\cite{Sheffield2010}.
Another aspect is discussed in 
\cite{Kang2011}.
We consider here a generalization with a nonzero drift $\nu\in\mathbb{R}$.
It is a special case of
Theorem \ref{Theorem: The coupling theorem}
and
Theorem \ref{Theorem: ppS -> simple cases}
(the first row of Table \ref{Table: some simple cases} and the sign `$+$' in the
`$\pm$' pair).

We assume
\begin{equation}
	\delta=\delta_c+\nu \sigma_c, \quad \sigma=\sqrt{\kappa}\sigma_c
	,\quad \Hc=\Hc_s, \quad	\Gamma=\Gamma_D.
\end{equation}
The pre-pre-Schwarzian $\eta$
can be obtained directly from the Hadamard formula
\eqref{Formula: Hadamard's formula}.

We show here how to calculate $\eta$ from the relations
\eqref{Formula: j^+ = L eta^+}
by substituting
\eqref{Formula: j+ = i/z + ia}
and
\eqref{Formula: nu := -2 alpha}.
In the half-plane chart, we have 
\begin{equation}
  -\kappa^{\frac12} \de_z {\eta^+}^{\HH}(z) + \mu \cdot 0 =
  \frac{-i}{z} + i \alpha.
\end{equation}
Then
\begin{equation}
  {\eta^+}^{\HH}(z) =
  \frac{i}{\sqrt{\kappa}} \log z - \frac{i \alpha z}{\sqrt{\kappa}} + C^+ ,
  \quad C^+ \in \mathbb{C}.
  \label{Formula: phi+ - chordal with drift}
\end{equation}
and taking into account that $\alpha=-\frac{\nu}{2}$
(see \eqref{Formula: nu := -2 alpha})
and
\eqref{Formula: eta = eta^+ + eta^-},
we finally obtain
\begin{equation}
  {\eta}^{\HH}(z) =
  \frac{-2}{\sqrt{\kappa}} \arg z
  - \frac{ \nu}{\sqrt{\kappa}} \Im z + C,\quad C\in\mathbb{R}.
  \label{Formula: eta - chordal with drift in H}
\end{equation}

We present here an explicit form of the time evolution
${M_1}_t$ of the one-point function $S_1(z)=\eta(z)$
 \begin{equation}\begin{split}
  {M_1}_t^{\HH}(z) =& ({G_t^{-1}}_* \eta)^{\HH}(z) =
  \frac{-2}{\sqrt{\kappa}} \arg G_t^{\HH}(z)
  - \frac{ \nu}{\sqrt{\kappa}} \Im G_t^{\HH}(z)
  +\frac{\kappa-4}{2\sqrt{\kappa}} \arg {G_t^{\HH}}'(z) + C.
  \label{Formula: M1 - chordal with drift}
 \end{split}\end{equation}
 When $\nu=0$, this expression coincides (up to a constant)
 with the analogous one from
 \cite[Section 8.5]{Kang2011}.

\subsection{Coupling of reverse chordal SLE and Neumann GFF}
\label{Section: Coupling of reverse chordal SLE and Neumann GFF}

We assume 
\begin{equation}
	\delta=-\delta_c+\nu \sigma_c, \quad \sigma=\sqrt{\kappa}\sigma_c
	,\quad \Hc=\Hc_s^*,\quad \Gamma=\Gamma_N.
\end{equation}
This type of coupling is also considered in 
\cite{Sheffield2010}. 
It corresponds to the first row of Table
\ref{Table: some simple cases} 
and the sign `$-$' in the `$\pm$' pairs. We use the space $\Hc_s^*{}'$ for the
GFF, which means that $\Phi$ is a linear functional defined up to a constant 
over the space of smooth functions with compact support. The covariance
$\Gamma=\Gamma_N$ is defined in 
\eqref{Formula: Gamma_N = Log...}.

Proceeding in the same way
as in the previous section we obtain
\begin{equation}
	{j}^{\HH}(z) = -\Re \frac{1}{z} + \nu,\quad \nu=\mathbb{R},
\end{equation}
\begin{equation}
	{\eta^{+}}^{\HH}(z)
	= \frac{1}{\sqrt{\kappa}} \log z + \frac{\alpha}{\sqrt{\kappa}} z,
\end{equation}
\begin{equation}
	{\eta}^{\HH}(z)
	= \frac{2}{\sqrt{\kappa}} \log |z| - \frac{\nu}{\sqrt{\kappa}} \Re z,
\end{equation}
\begin{equation}
 \eta
 = \frac{2}{\sqrt{\kappa}} \log 
	\left|
  	\frac{\sigma_c^{\frac{\kappa}{4}}}{\delta_c}
	\right|
  + 
  \frac{2\nu}{\sqrt{\kappa}} \Re \frac{\sigma_c} {\delta_c}+ C
 = \frac{2}{\sqrt{\kappa}} \log
	\left| 
		\frac{\sigma^{\frac{\kappa}{4}}}{\delta  -
  	\frac{\nu}{\sqrt{\kappa}}\sigma } 
	\right| -
	\frac{2\nu}{\kappa} \Re \frac{\sigma}
 		{\delta-\frac{\nu}{\sqrt{\kappa}}\sigma}+ C.
\end{equation}
The one-point local martingale in the half-plane chart is
\begin{equation}\begin{split}
	&{M_1}_t^{\HH}(z)
	= {G_t^{-1}}_* {\eta^*}^{\HH} (z)
	=\\=&
	\frac{2}{\sqrt{\kappa}} \log \left|{G_t}^{\HH}(z)\right|
	- \frac{\nu}{\sqrt{\kappa}} \Re \left( {G_t}^{\HH}(z) \right)
	+ \frac{\kappa+4}{2\sqrt{\kappa}} \log \left| \left(G_t^{\HH}\right)'(z)
	\right|+C
\end{split}\end{equation}

\section{Dipolar Case}
\label{Section: Dipolar Case}

\subsection{Dipolar L\"owner equation}
\label{Section: Dipolar Loewner equation}

The dipolar version of the L\"owner equation was first considered in
\cite{Bauer2003b}.

Assume that both vector fields $\delta$ and $\sigma$ vanish at the
same two distinct points $b_1$ and $b_2$ on the boundary $\de\Dc$. Namely,
\begin{equation}
	\delta(b_1)=\sigma(b_1)=\delta(b_2)=\sigma(b_2)=0,\quad b_1\neq b_2\in\de\Dc.
\end{equation}
With the help of the transforms 
$\mathscr{R}$ and $\mathscr{P}$
we can always normalize them so that
\begin{equation}\begin{split}
	&\sigma=\sigma_{-1} \left( \ell_{-1}-\ell_{1} \right),
	\quad \sigma_{-1}\in\mathbb{R}\setminus \{0\},	\\
	&\delta=\delta_{-2} \left( \ell_{-2}-\ell_{0} \right)
		+	\delta_{-1} \left( \ell_{-1}-\ell_{1} \right),
	\quad \delta_{-2}>0, \quad \delta_{-1}\in\mathbb{R}.
\end{split}\end{equation}
This corresponds to
$\psi^{\D}(b_1)=i$
and
$\psi^{\D}(b_2)=-i$
in the unit disk chart and
$\psi^{\HH}(b_1)=1$
and
$\psi^{\HH}(b_2)=-1$
in the half-plane chart.
Using the transforms $\mathscr{V}$, $\mathscr{D}$, and $\mathscr{T}$, we can
normalize further and assume
\begin{equation}\begin{split}
	\sigma = \ell_{-1}-\ell_{1},\quad
	\delta = 2 \left( \ell_{-2} - \ell_{0} \right).
\end{split}\end{equation}

\begin{figure}[t]
\centering
  \begin{subfigure}[t]{0.33\textwidth}
		\centering
    \includegraphics[width=5cm,keepaspectratio=true]
    	{Pictures/Vector_fileds_v=_-2_-_0_.pdf}
    \caption{Flow lines of the vector field $\delta_d$.}
  \end{subfigure}%
  ~
  \begin{subfigure}[t]{0.33\textwidth}
		\centering
    \includegraphics[width=5cm,keepaspectratio=true]
    	{Pictures/Vector_fileds_v=-_-1_+_1_.pdf}
  	\caption{Flow lines of the vector field $\sigma_d$.}
	\end{subfigure}%
	~
	\begin{subfigure}[t]{0.33\textwidth}
		\centering
		\includegraphics[width=5cm,keepaspectratio=true]
			{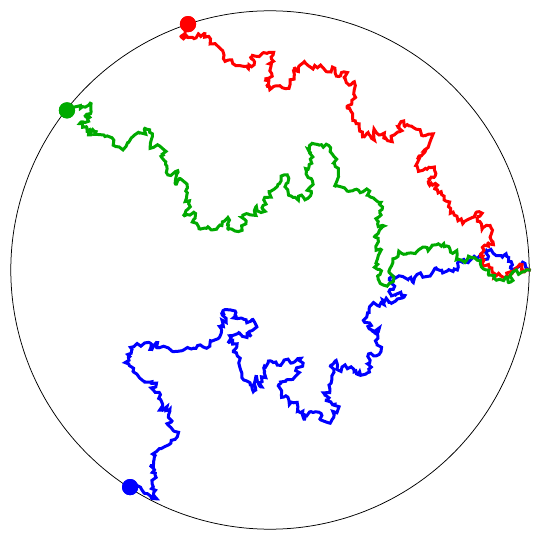}
		\caption{Three samples of dipolar SLE slits for $\kappa=2$. The slits tend to
		random points (big colored dots) on the left arc between $\psi^{\D}(b_1)=i$ and
		$\psi^{\D}(b_2)=-i$.}
	\end{subfigure}
	\caption{Flow lines of the dipolar fields $\delta_d$ and $\sigma_d$, and
	corresponding slit samples in the unit disk chart.}
\end{figure}

It is natural to consider the dipolar L\"owner equation in the chart defined by
the following transition map:
\begin{equation}
  \tau_{\HH,\mathbb{S}}(z)=\mathrm{th}\frac{z}{2}
  \colon \mathbb{S} \map\HH. 
  \label{Formula: tau_H S}
\end{equation}
This domain $\SSS$ is called the \emph{strip}
\index{strip chart}
\begin{equation}
	\SSS := \psi^{\SSS}(\Dc)
	= \{ z\in\C\colon 0 < \Im z < \pi \}
\end{equation}
and we call the chart $\psi^{\SSS}$ the 
\emph{strip chart}. 
We denote vector
fields, conformal maps, and pre-pre-Schwarzians in these coordinates by the superscript
$\SSS$, analogously to $\D$ and $\HH$.
In these coordinates
$\psi^{\SSS}(b_1)=$`right infinity',
$\psi^{\SSS}(b_2)=$`left infinity',
and
$\psi^{\SSS}(a)=0$.

The vector fields have the form
\begin{equation}\begin{split}
	\delta^{\SSS}(z) = 4\cth \frac{z}{2}, \quad
	\sigma^{\SSS}(z) = - 2.
 \label{Formula: delta and sigma dipolar in S deterministic}
\end{split}\end{equation}
The L\"owner equation can be written as
\begin{equation}
  \dot G_t^{\SSS}(z) = 4 \cth \frac{G_t^{\SSS}(z)}{2} - 2 \dot u_t.
  \label{Formula: dot G^S = 4 cth G/2 - 2 dot u}
\end{equation}
A more frequently used normalization is
\begin{equation}
  \dot G_t^{\SSS}(z) = \cth \frac{G_t^{\SSS}(z)}{2} - \dot u_t,
\end{equation}
which can be obtained from
\eqref{Formula: dot G^S = 4 cth G/2 - 2 dot u}
by the transform $\mathscr{P}_c$ with $c=2$.

We prefer this normalization only for uniformity (the invariant coefficients
$\delta_{-2}=2$ and $\sigma_{-1}=-1$ are the same in all cases).

We obtain the following form of the L\"owner equation:
\begin{equation}
	\dot g_t^{\SSS}(z) = 4 \cth \left( \frac{g_t^{\SSS}(z)}{2} - u_t \right)
\end{equation}
with
\begin{equation}
	g_t^{\SSS}(z) = G_t^{\SSS}(z) + 2 u_t
\end{equation}
in the same way as in the previous section.

We define $\delta_d$ and $\sigma_d$ for the stochastic case by
\begin{equation}
  \delta_d = 2 (\ell_{-2} - \ell_{0}) - \nu (\ell_{-1} - \ell_{1}),\quad
  \sigma_d = -\left( \ell_{-1} - \ell_{1} \right)
  ,\quad \nu \in \mathbb{R},
  \label{Formula: delta and sigma dipolar stoch}
 \end{equation}
where we add the parameters $\kappa$ and $\nu$; see the comments in Section
\ref{Section: Equivalence and normalization SLE}.
The dipolar SLE can be expressed as
\begin{equation}
  \dS G_t^{\SSS}(z)
  = 4 \cth \frac{G_t^{\SSS}(z)}{2} dt
  -2 \sqrt{\kappa} d^{\mathrm{S}} B_t
  -\frac{\nu}{2} dt
\end{equation}
in the Stratonovich form and as
\begin{equation}
  \dI G_t^{\SSS}(z)
  = 4 \cth \frac{G_t^{\SSS}(z)}{2} dt
  -2 \sqrt{\kappa} d^{\mathrm{It\^{o} }} B_t
  -\frac{\nu}{2} dt
\end{equation}
in terms of the It\^{o} differentials, because $\sigma^{\SSS}(z)$ is a constant.

Consider now the same relations for the stochastic case in the half-plane chart.
They are
\begin{equation}
  \delta_d^{\HH}(z)= 2 \left( \frac{1}{z} - z \right) - \nu(1-z^2),\quad
  \sigma_d^{\HH}(z)=-(1-z^2),\quad
  \kappa>0,\quad \nu \in \mathbb{R},
  \label{Formula: delta and sigma dipolar in H stoch}
\end{equation}
\begin{equation}
	\dS G_t^{\SSS}(z) =
	2\left( \frac{1}{G_t^{\SSS}(z)} - G_t^{\SSS}(z) \right)dt
	- \nu \left( 1 - G_t^{\SSS}(z)^2 \right)dt
	- \sqrt{\kappa} \left( 1 - G_t^{\SSS}(z)^2 \right) \dS B_t
\end{equation}
\begin{equation}\begin{split}
	\dI G_t^{\SSS}(z) =&
	2\left(
		\frac{1}{G_t^{\SSS}(z)} - G_t^{\SSS}(z)
		- \kappa G_t^{\SSS}(z)
		\left( 1- G_t^{\SSS}(z)^2 \right)
	\right)dt
	- \nu \left( 1 - G_t^{\SSS}(z)^2 \right)dt
	-\\-&
	\sqrt{\kappa} \left( 1 - G_t^{\SSS}(z)^2 \right) \dI B_t,
\end{split}\end{equation}
where we used
\eqref{Formula: Slit hol stoch flow Ito}.

We remark that the maps $G_t^{\SSS}$ and $g_t^{\SSS}$ satisfy the
normalization conditions
\begin{equation}\begin{split}
	&G_t^{\SSS}(z) = z - 2u_t + O(z^{-1}),\quad \Re z >0,\\
	&G_t^{\SSS}(z) = z - 2 u_t + O(z^{-1}),\quad \Re z <0,
\end{split}\end{equation}
and
\begin{equation}\begin{split}
	&g_t^{\SSS}(z) = z + O(z^{-1}),\quad \Re z >0,\\
	&g_t^{\SSS}(z) = z + O(z^{-1}),\quad \Re z <0.
	\label{Formula: dipolar normalization for G in S}
\end{split}\end{equation}
Besides, 
$G_t(z_{\text{tip}})=a$ 
and 
$g_t^{\SSS}(z_{\text{tip}})=u_t$ 
in the forward case and 
$G_t(a)=z_{\text{tip}}$ 
in the reverse case.
These conditions uniquely determine the maps $G_t$ and $g_t$ for a given $\K_t$ if 
$\psi^{\SSS}(\K_t)$ is compact and does not touch the upper boundary of
the strip $\SSS$.
\index{dipolar normalization}

It is also reasonable to consider the unit disk chart:
\begin{equation}
  \delta_d^{\D}(z)= 2\frac{(1+z)(1+z^2)}{1-z} - \nu(1-z^2),\quad
  \sigma_d^{\D}(z)=-i (1+z^2),\quad
  \kappa>0,\quad \nu \in \mathbb{R},
  \label{Formula: delta and sigma dipolar in D stoch}
\end{equation}
\begin{equation}
	\dS G_t^{\D}(z) =
	2\frac{(1+G_t^{\D}(z))(1+G_t^{\D}(z)^2)}{1-G_t^{\D}(z)} dt
	-i \nu (1+G_t^{\D}(z)^2) dt
	-i \sqrt{\kappa} (1+G_t^{\D}(z)^2) \dS B_t
\end{equation}
\begin{equation}\begin{split}
	\dI G_t^{\D}(z) =&
	2 \left(
		\frac{(1+G_t^{\D}(z))(1+G_t^{\D}(z)^2)}{1-G_t^{\D}(z)}
		- \kappa G_t^{\D}(z) (1+G_t^{\D}(z)^2)
	\right)	dt
	-\\-&
	i \nu (1+G_t^{\D}(z)^2) dt
	-i \sqrt{\kappa} (1+G_t^{\D}(z)^2) \dI B_t.
\end{split}\end{equation}

We define $\delta$ and $\sigma$ for the stochastic case by 
\begin{equation}\begin{split}
  &\delta = \pm\delta_d + \nu \sigma_d 
  = \pm 2 (\ell_{-2}-\ell_{0}) - \nu(\ell_{-1}-\ell_{1}),\\
  &\sigma = \sqrt{\kappa} \sigma_d 
  = -\sqrt{\kappa} (\ell_{-1}-\ell_1),\quad
  \kappa>0,\quad \nu \in \mathbb{R},
  \label{Formula: delta and sigma dipolar}
\end{split}\end{equation}
where 
`$\pm$' corresponds to the forward and reverse cases.



\subsection{Coupling of forward dipolar SLE and Dirichlet GFF}
\label{Section: Coupling fo dipolar SLE}

We assume 
\begin{equation}
	\delta=\delta_d+\nu \sigma_d, \quad \sigma=\sqrt{\kappa}\sigma_d
	,\quad \Hc=\Hc_s,\quad \Gamma=\Gamma_D.
\end{equation}
This case corresponds to the third row of Table 
\ref{Table: some simple cases}
and the sign `$+$' in the `$\pm$' pairs. For $\nu=0$, it is also considered in 
\cite{Kang}. 

We first obtain $\eta^+$, $\eta$ and ${M_1}_t$ in the half-plane chart.
In the same way as in the previous subsection, we calculate
 \begin{equation}
  -\sqrt{\kappa}(1-z^2) \de_z {\eta^+}^{\HH}(z)
  + \mu \left( -\sqrt{\kappa} (1-z^2) \right)' =
  \frac{-i}{z} + i \alpha.
 \end{equation}
 Taking into account
 (\ref{Formula: mu = i pm (-4 pm k)/4/k})
 and $\alpha=-\nu/2$ we obtain
 \begin{equation}
  {\eta^+}^{\HH}(z) =
  \frac{i}{\sqrt{\kappa}} \log z
  + \frac{i(\kappa-6)}{4 \sqrt{\kappa}} \log (1-z^2)
  + \frac{i\nu}{2\sqrt{\kappa}} \mathrm{arcth} z + C^+.
 \end{equation}
 \begin{equation}
  {\eta}^{\HH}(z) =
  \frac{-2}{\sqrt{\kappa}} \arg z
  - \frac{(\kappa-6)}{2 \sqrt{\kappa}} \arg (1-z^2)
  - \frac{\nu}{\sqrt{\kappa}} \Im \mathrm{arcth} z + C.
  \label{Formula: phi - dipolar with drift in H}
 \end{equation}
 \begin{equation}\begin{split}
  &{M_c}_t^{\HH}(z) = ({G_t^{-1}}_* \eta)^{\HH}(z)
  =\\=&
  \frac{-2}{\sqrt{\kappa}} \arg G_t^{\HH}(z)
  - \frac{(\kappa-6)}{2 \sqrt{\kappa}} \arg (1-G_t^{\HH}(z)^2)
  -\\-&
  \frac{\nu}{\sqrt{\kappa}} \Im \mathrm{arcth}\, G_t^{\HH}(z)
  +\frac{\kappa-4}{2\sqrt{\kappa}} \arg {G_t^{\HH}}'(z) + C.
 \end{split}\end{equation}

The corresponding relations in the strip chart are
\begin{equation}\begin{split}
 \delta_d^{\SSS}(z)
 = 4 \cth \frac{z}{2} - 2 \nu,
 = - 2,\quad
 \nu \in \mathbb{R},
\end{split}\end{equation}
\begin{equation}
	j^{\SSS}(z) = -i \cth \frac{z}{2} + i \alpha,
\end{equation}
The function $\eta^{\SSS}(z)$ can be found as the solution to
(\ref{Formula: j^+ = L eta^+})
in the strip chart
\begin{equation}
	\eta^{\SSS}(z)
	= \frac{-2 }{\sqrt{\kappa}} \arg \mathrm{sh} \frac{z}{2}
	- \frac{\nu}{2\sqrt{\kappa}} \Im z + C,
\end{equation}
\begin{equation}\begin{split}
	&{M_1}_t^{\SSS}(z) = ({G_t^{-1}}_* \eta)^{\SSS}(z)
	=\\=&
	\frac{-2}{\sqrt{\kappa}} \arg \mathrm{sh} \frac{G_t^{\SSS}(z)}{2}
	- \frac{\nu}{2\sqrt{\kappa}} \Im G_t^{\SSS}(z)
	+\frac{\kappa-4}{2\sqrt{\kappa}} \arg \left(G_t^{\SSS}\right)'(z) + C.
\end{split}\end{equation}
The expression for $\Gamma_D$ in the strip chart is
\begin{equation}
	\Gamma_D^{\SSS}(z,w)
	= \Gamma_D^{\HH}\left( \tau_{\HH,\SSS}(z),\tau_{\HH,\SSS}(w) \right)
	= -\frac12 \log \frac{ \mathrm{sh}(\frac{z-w}{2}) 
	\mathrm{sh}(\frac{\bar z-\bar w}{2})}
  { \mathrm{sh}(\frac{\bar z-w}{2}) \mathrm{sh}(\frac{z-\bar w}{2})}.
\end{equation}

\subsection{Coupling of reverse dipolar SLE and Neumann GFF}
\label{Section: Coupling of reverse dipolar SLE and Neumann GFF}

We assume 
\begin{equation}
	\delta=-\delta_d+\nu \sigma_d, \quad \sigma=\sqrt{\kappa}\sigma_d
	,\quad \Hc=\Hc_s^*,\quad \Gamma=\Gamma_N.
\end{equation}
This case corresponds to the third row of Table 
\ref{Table: some simple cases}
and the sign `$-$' in the `$\pm$' pairs.
We have
\begin{equation}
  {\eta^+}^{\HH}(z) =
  \frac{1}{\sqrt{\kappa}} \log z
  - \frac{(\kappa+6)}{4 \sqrt{\kappa}} \log (1-z^2)
  - \frac{\nu}{2\sqrt{\kappa}} \arcth z + C^+,
\end{equation}
\begin{equation}
	{\eta}^{\HH}(z)
	= \frac{2}{\sqrt{\kappa}} \log |z|
	- \frac{(\kappa+6)}{2 \sqrt{\kappa}} \log |1-z^2|
	- \frac{\nu}{\sqrt{\kappa}} \Re \arcth z + C,
\end{equation}
and 
\begin{equation}\begin{split}
	&{M_1}_t^{\HH}(z)
	= {G_t^{-1}}_* {\eta}^{\HH} (z)
	=\\=&
	\frac{2}{\sqrt{\kappa}} \log |{G_t}^{\HH}(z)|
	- \frac{(\kappa+6)}{2 \sqrt{\kappa}}
		\log \left|1-\left({G_t}^{\HH}(z)\right)^2\right| -
	\frac{\nu}{\sqrt{\kappa}} \Re \arcth {G_t}^{\HH}(z)
	+\\+&
	\frac{\kappa+4}{2\sqrt{\kappa}} \log \left| \left({G_t}^{\HH}\right)'(z)
	\right| + C,
\end{split}\end{equation}
in the half-plane chart. In the strip chart, the covariance $\Gamma_N$ and $\eta$
have the form
\begin{equation}\begin{split}
	&\Gamma_{N}^{\SSS}(z,w)
	=\\=& 
	-\frac{1}{2} \log
		\left(
		\left(\th \frac{z}{2}-\th \frac{w}{2}\right)
		\left(\th \frac{z}{2}-\th	\frac{\bar w}{2}\right) 
		\left(\th \frac{\bar z}{2}-\th \frac{w}{2}\right) 
		\left(\th \frac{\bar z}{2}-\th \frac{\bar w}{2}\right)
		\right)
	+\\+& 
	\beta(z) + \beta(w),
\end{split}\end{equation}
and
\begin{equation}
	\eta^{\SSS}(z)
	= \frac{2 }{\sqrt{\kappa}} \log 
	\left| \mathrm{sh} \frac{z}{2} \right|
	+ \frac{\nu}{2\sqrt{\kappa}} \Re z + C.
\end{equation}


\subsection{Coupling of forward dipolar SLE and combined Dirichlet-Neumann GFF}
\label{Section: Coupling of forward dipolar SLE and combined Dirichlet-Neumann
GFF}

We assume 
\begin{equation}
	\delta=\delta_d+\nu \sigma_d,\quad
	\sigma=\sqrt{\kappa}\sigma_d,\quad \Hc=\Hc_s, \quad \Gamma=\Gamma_{DN}
	\label{Formula: DN coupling}
\end{equation}
(see Example 
\ref{Example: Gamma: Combined Dirichlet-Neumann boundary conditions}).
Below, we show that the coupling is possible if and only if $\kappa=4$ and
$\nu=0$.
This case is considered in 
\cite{Kanga} 
and it is also a special case in
\cite{Izyurov2010}.

In the strip chart, the covariance $\Gamma_{DN}$ is of the form
\begin{equation}
 \Gamma_{DN}^{\mathbb{S}}(z,w)= -\frac12\log \frac{
  \mathrm{th}\frac{z-w}{4} \mathrm{th}\frac{\bar z-\bar w}{4}}
  {\mathrm{th}\frac{\bar z-w}{4} \mathrm{th}\frac{z -\bar w}{4}},\quad
  z,w\in \SSS:=\{z~:~ 0<\Im z<\pi \}.
 \label{Formula: G_DN^S = ...}
\end{equation}

The function $\Gamma_{DN}^{\SSS}(z,w)$ satisfies the boundary conditions
\begin{equation}
 \left.\Gamma_{DN}^{\SSS}(x,w)\right|_{x\in \mathbb{R}}=0,\quad
 \left. \de_y \Gamma_{DN}^{\SSS}(x+i y,w)\right|_{x \in \mathbb{R},y=\pi}=0,
\end{equation}
and the symmetry property
\begin{equation}
 \Lc_{\sigma} \Gamma_{DN}(z,w)=0.
 \label{Formula: L_s G = 0}
\end{equation}

The coupling of the GFF with this $\Gamma$ and the dipolar SLE is geometrically
motivated. The zeros of both $\delta$ and $\sigma$ are located at the
same boundary points where $\Gamma_{DN}$ changes the boundary conditions from
Dirichlet to Neumann. In the strip chart, these points are $-\infty$ and
$+\infty$.

\begin{proposition}
Let 
\eqref{Formula: DN coupling} 
be satisfied, and let $\eta$ be a pre-pre-Schwarzian
\eqref{Formula: tilde eta = eta - chi arg}. 
Then the coupling is possible only for $\kappa=4$ and $\nu=0$.
\label{Proposition: D-N coupling}
\end{proposition}

\begin{proof}
We use Theorem
\ref{Theorem: eta structure}
and the strip chart.
From
\eqref{Formula: G_DN^S = ...},
we obtain
\begin{equation}
 {\Gamma^{++}_{DN}}^{\SSS}(z,w) = -\frac12 \log \th \frac{z-w}{4},\quad
 {\Gamma^{+-}_{DN}}^{\SSS}(z,\bar w) = -\frac12 \log \th \frac{z-\bar w}{4}.
\end{equation}
The relations
(\ref{Formula: L Gamma^++ = beta + beta, L Gamma^+- = beta + beta})
are satisfied. From
(\ref{Formula: L Gamma^++ + L sigma^+ L sigma^+ = e + e})
we find that
\begin{equation}
	(\Lc_{\sigma} \eta^+)^{\SSS}(z)
	= {j^+}^{\SSS}(z) = \frac{-i}{\sh\frac{z}{2}} + i
	\alpha,\quad \alpha\in\mathbb{R}.
\end{equation}
Substituting this into
(\ref{Formula: delta/sigma j + mu [sigma,delta]/sigma + 1/2 L j = ibeta})
gives
\begin{equation}
 -i \frac{
  \left( \beta \sqrt{\kappa} - \alpha \nu \right) \sh^2\frac{z}{2}
  + \nu \sh\frac{z}{2}
  - 2 i \sqrt{\kappa} \mu
  \ch\frac{z}{2} \left( 4 \sh\frac{z}{2} \alpha + (\kappa -4) \right)
  }{2 \sqrt{\kappa} \sh^2\frac{z}{2} }
 \equiv 0,
\end{equation}
which is possible only if $\kappa=4$, $\nu=0$, $\beta=0$, and $\mu=0$
(the latter agrees with (\ref{Formula: mu = i pm (-4 pm k)/4/k})).
\end{proof}

From
(\ref{Formula: j^+ = L eta^+})
we obtain that
\begin{equation}
 {\eta^+}^{\SSS}(z) = \frac{i}{2} \log \th \frac{z}{4}+C^+,
\end{equation}
and
\begin{equation}
 \eta^{\SSS}(z) = - \arg \,\mathrm{cth}\, \frac{z}{4} + C.
\end{equation}

We also present here the relations in the half-plane chart:
\begin{equation}
 \eta^{\HH}(z)
 = - \arg \frac{z}{1+\sqrt{1-z^2}} + C,
\end{equation}
\begin{equation}\begin{split}
	&\Gamma_{DN}^{\HH}(z)
	=\\=& 
 -\frac12 \log \frac{(z-w)(\bar z-\bar w)
  (1-\bar z w+\sqrt{1-\bar z^2}\sqrt{1-w^2})(1- z \bar w+\sqrt{1- z^2}\sqrt{1-\bar w^2})}
  {(\bar z-w)(z-\bar w)
  (1-zw+\sqrt{1-z^2}\sqrt{1-w^2})(1-\bar z \bar w+\sqrt{1-\bar z^2}\sqrt{1-\bar w^2})}.
  \label{Formula: Gamma_DN^H := ...}
\end{split}\end{equation}


\section{Radial Case}
\label{Section: Radial case}

\subsection{Radial L\"owner equation}
\label{Section: Radial Loewner equation}

The radial L\"owner equation has been studied along with the chordal
version.

Assume that both vector fields $\delta$ and $\sigma$ vanish at the
same point $b$ in the interior of $\Dc$:
\begin{equation}
	\delta(b)=\sigma(b)=0, \quad b\in\Dc.
\end{equation}
With the aid of the transforms
$\mathscr{R}$ and $\mathscr{S}$ it is always possible to arrange
$\psi^{\D}(b)=0$ in the unit disk chart or, equivalently, $\psi^{\HH}(b)=i$ in
the half-plane chart. Thus, in terms of the basis vectors $\ell_n$, we have
\begin{equation}\begin{split}
	&\sigma=\sigma_{-1} \left( \ell_{-1} + \ell_{1} \right),
	\quad \sigma_{-1}\in\mathbb{R}\setminus \{0\},	\\
	&\delta=\delta_{-2} \left( \ell_{-2} + \ell_{0} \right)
		+	\delta_{-1} \left( \ell_{-1} + \ell_{1} \right),
	\quad \delta_{-2}\neq0, \quad \delta_{-1}\in\mathbb{R},
\end{split}\end{equation}
Using the transforms $\mathscr{V}$, $\mathscr{D}$, and $\mathscr{T}$, we can
normalize further and assume
\begin{equation}\begin{split}
	\delta_r = 2 \left( \ell_{-2} + \ell_{0} \right),\quad
	\sigma_r = -\ell_{-1} - \ell_{1}
\end{split}\end{equation}
for the forward case.
\begin{figure}[h]
\centering
  \begin{subfigure}[t]{0.35\textwidth}
		\centering
    \includegraphics[width=5cm,keepaspectratio=true]
    	{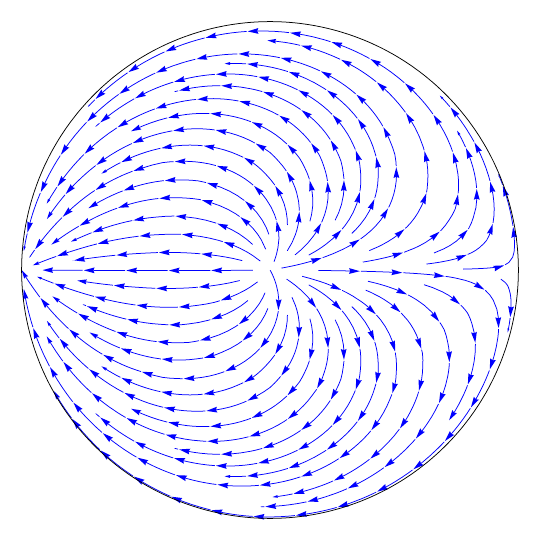}
    \caption{Flow lines of the vector field $\delta_r$.}
  \end{subfigure}%
  ~
  \begin{subfigure}[t]{0.35\textwidth}
		\centering
    \includegraphics[width=5cm,keepaspectratio=true]
    	{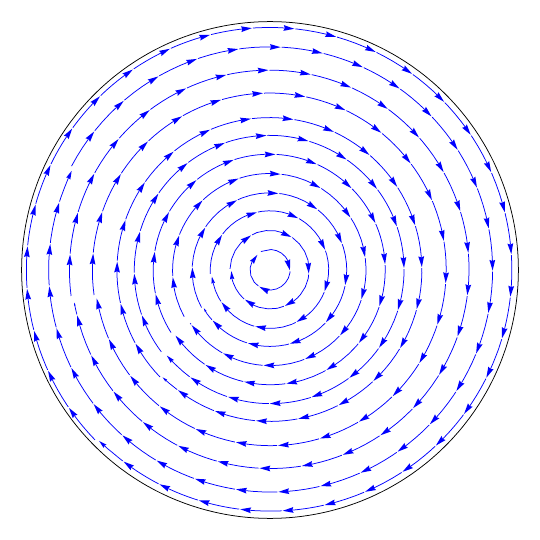}
    \caption{Flow lines of the vector field $\sigma_r$.}
	\end{subfigure}%
	~
	\begin{subfigure}[t]{0.33\textwidth}
		\centering
		\includegraphics[width=5cm,keepaspectratio=true]
			{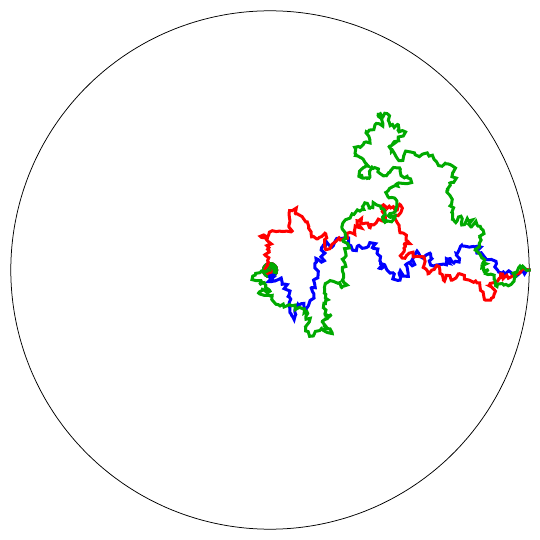}
		\caption{Three samples of radial SLE slits for $\kappa=2$ in the unit disk
		chart. The slits tend to the fixed point $\psi^{\D}(b)=0$.}
	\end{subfigure}
	\caption{Flow lines of the radial vector fields $\delta_r$ and $\sigma_r$, and
	corresponding slit samples in the unit disk chart.}
\end{figure}

It is natural to consider the radial L\"owner equation in the unit disk and
logarithmic charts. We will define the logarithmic chart later and first present 
the relations in the unit disk chart. The vector fields have the form
\begin{equation}\begin{split}
	\delta_r^{\D}(z) = 4 z \frac{1+z}{1-z}, \quad
	\sigma_r^{\D}(z) = -2 i z.
 \label{Formula: delta and sigma radial in D deterministic}
\end{split}\end{equation}
The L\"owner equation in terms of $G_t$ can be written as
\begin{equation}
  \dot G_t^{\D}(z) = 4 G_t^{\D}(z) \frac{1+G_t^{\D}(z)}{1-G_t^{\D}(z)}
  - 4 i G_t^{\D}(z) \dot u_t.
  \label{Formula: dot G^D = 4G^D(1+G^D)/(1-G^D) - 2iG^D dot u}
\end{equation}

The equation for $H_s[\sigma]$ in the unit disk chart is
\begin{equation}\begin{split}
 &\dot H_s[\sigma]^{\D}(z) = -2 i H_s[\sigma]^{\D}(z),\quad
  H_0[\sigma]^{\D}(z) = z,\quad
 z\in \D.
\end{split}\end{equation}
The solution is
\begin{equation}\begin{split}
 H_s[\sigma]^{\D}(z) = e^{-2 i s}z.
 \label{Formula: h_t for radial case in D}
\end{split}\end{equation}
Thus, the equation for $g_t$ is
\begin{equation}\begin{split}
 \dot g_t^{\D}(z) =&
 \left(
  \frac{1}{{H_{u_t}[\sigma]^{\D}}'}
   \left( 4 H_{u_t}[\sigma]^{\D}
   \frac{1+H_{u_t}[\sigma]^{\D}}{1-H_{u_t}[\sigma]^{\D}} \right)
 \right)
 \circ g_t^{\D}(z)
 =\\=&
 \frac{1}{ e^{-2i u_t} }
 \left( e^{-2i u_t} g_t^{\D}(z)
  \frac{ 1 + e^{-2i u_t} g_t^{\D}(z) }{ 1 - e^{-2i u_t} g_t^{\D}(z) }
 \right)
\end{split}\end{equation}
or
\begin{equation}\begin{split}
	\dot g_t^{\D}(z) =
	4g_t^{\D}(z) \frac{e^{2i u_t} +g_t^{\D}(z)}{ e^{2iu_t} - g_t^{\D}(z)} .
	\label{Formula: g radial S
	LE in D}
\end{split}\end{equation}

We remark that the maps $G_t^{\D}$ and $g_t^{\D}$ satisfy the
normalization conditions
\begin{equation}\begin{split}
	&G_t^{\D}(z) = z\, e^{4t-2i u_t} + O(z^2),
	\label{Formula: radila normalization for G in D}
\end{split}\end{equation}
and
\begin{equation}
	g_t^{\D}(z) = z\, e^{4t} + O(z^2).
\end{equation}
These conditions uniquely determine the maps $G_t$ and $g_t$ for a given $\K_t$ if 
$\psi^{\D}(\K_t)$ does not touch the origin $\psi^{\D}(b)=0$.
\index{radial normalization}

The stochastic version is given in the Stratonovich form by
\begin{equation}
	d^{\mathrm{S}} G_t^{\D}(z) =
	4 G_t^{\D}(z) \frac{ 1 + G_t^{\D}(z) }{1 - G_t^{\D}(z)} dt
	-2 i \sqrt{\kappa} G_t^{\D}(z) d^{\mathrm{S}} B_t.
\end{equation}
The same equation in the It\^o form is
\begin{equation}
	\dI G_t^{\D}(z) =
	4 \left(
		G_t^{\D}(z) \frac{ 1 + G_t^{\D}(z) }{1 - G_t^{\D}(z)}
		- \frac12 \kappa G_t^{\D}(z)
	\right) dt
	-2 i \sqrt{\kappa} G_t^{\D}(z) \dI B_t.
\end{equation}

As in the dipolar case, in a more frequently used normalization, the
same relations are
\begin{equation}\begin{split}
	\delta^{\D}(z) = z \frac{1+z}{1-z} - i \nu z, \quad
	\sigma^{\D}(z) = -i z,
\end{split}\end{equation}
\begin{equation}\begin{split}
	\dot g_t^{\D}(z) =
	g_t^{\D}(z) \frac{e^{i u_t} +g_t^{\D}(z)}{ e^{iu_t} - g_t^{\D}(z)},
\end{split}\end{equation}
\begin{equation}
	\dS G_t^{\D}(z) =
	G_t^{\D}(z) \frac{ 1 + G_t^{\D}(z) }{ 1 - G_t^{\D}(z) } dt
	-i \sqrt{\kappa}\, G_t^{\D}(z) \dS B_t.
\end{equation}
They can be obtained from
\eqref{Formula: dot G^D = 4G^D(1+G^D)/(1-G^D) - 2iG^D dot u}
by the transform $\mathscr{P}_c$ with $c=2$.

We prefer this normalization only for uniformity (the invariant coefficients
$\delta_{-2}=2$ and $\sigma_{-1}=-1$ are the same in all cases in our
framework).

In the half-plane chart, we have
\begin{equation}\begin{split}
	\delta^{\HH}(z) = 2 \left(\frac{1}{z}+z\right) - \nu (1 + z^2), \quad
	\sigma^{\HH}(z) = - (1 + z^2),
\end{split}\end{equation}
\begin{equation}
	H_s[\sigma]^{\HH}(z) = \frac{z-\tg s}{1+ z \tg s},
\end{equation}
\begin{equation}\begin{split}
	\dot g_t^{\HH}(z) =
	\left(1+g_t^{\HH}(z)^2\right)
	\frac{ 2+ 2 (\tg u_t) g_t^{\HH}(z) + \nu(\tg u_t - g_t^{\HH}(z)) }
		{g_t^{\HH}(z)-\tg u_t}
\end{split}\end{equation}
\begin{equation}
	\dS G_t^{\HH}(z) =
	\left(\frac{1}{G_t^{\HH}(z)}+G_t^{\HH}(z)\right) dt
	- \nu (1 + G_t^{\HH}(z)^2) dt
	-\sqrt{\kappa} (1 + G_t^{\HH}(z)^2) \dS B_t.
\end{equation}

In a chart where the complete vector field $\sigma$ is a constant, the relations
are always especially simple. We have seen this for the chordal case (the
half-plane chart) and the dipolar case (the strip chart) in the previous two sections.
The chart in which the radial field $\sigma_r$ is constant is
called \emph{logarithmic}.
\index{logarithmic chart}
We denote it by $\LL$ and define by the transition map
\begin{equation}
 \tau_{\D,\LL}(z):=e^{iz}:\HH\map\D,\quad
 \tau_{\LL,\D}(z)=\tau_{\D,\LL}^{-1}(z) =-i\log \, z.
 \label{Formula: tau_D LL = ...}
\end{equation}
Thus,
\begin{equation}
 \tau_{\HH,\LL}(z)
 = \tau_{\HH,\D} \circ \tau_{\D,\LL}(z)
 = \mathrm{tg}\frac{z}{2}:\HH\map\HH,\quad
 \tau_{\LL,\HH}(z) = \tau_{\HH,\LL}^{-1}(z) = 2 \,\mathrm{arctg}\, z,
\end{equation}
The chart maps $\psi^{\HH}$, $\psi^{\D}$, and $\psi^{\SSS}$ considered
above are global, which means that they are defined on the entire $\Dc$ and are
single-valued. The map $\psi^{\LL}$ is not defined at one point of $\Dc$, which is
$b$ in our normalization of the radial equation. Besides, the function $\log$
is multivalued, and the upper half-plane (the image of
$\tau_{\HH,\LL}$) contains an infinite number of identical copies of the radial
slit because
\begin{equation}
	\tau_{\LL,\HH}(z+2\pi)=\tau_{\LL,\HH}(z),\quad z\in\HH.
\end{equation}
The advantage of this chart is
that the automorphisms $H_t[\sigma_r]^{\LL}$ induced by $\sigma_r$ 
(see \eqref{Formula: d H[sigma](z) = sigma( H[sigma] ) (z)})
are horizontal translations because $\sigma_r^{\LL}(z)$ is a real constant. In
the logarithmic chart, we also have
\begin{equation}\begin{split}
  \delta_r^{\LL}(z)
  = 4\ctg \frac{z}{2},\quad
  \sigma_r^{\LL}(z)
  = - 2.
  \label{Formula: delta and sigma radial in L chart}
\end{split}\end{equation}
From this and the relations below, one can see that if one formally replaces
all hyperbolic functions in the dipolar case in the strip chart $\SSS$ by the
corresponding trigonometric functions, then one obtains the relations for
the radial equation in the logarithmic chart $\LL$.
\begin{equation}
 	\dot G_t^{\LL}(z) = 4 \ctg \frac{G_t^{\LL}(z)}{2} - 2 \dot u_t,
	\label{Formula: dot G^L = 4 ctg G_t^L (z)/2 - 2 dot u}
\end{equation}
\begin{equation}
	\dot g_t^{\LL}(z) = 4 \ctg \left( \frac{g_t^{\LL}(z)}{2} - u_t \right)
\end{equation}


\subsection{Coupling of forward radial SLE and Dirichlet GFF}
\label{Section: Coupling of forward radial SLE and Dirichlet GFF}

We assume 
\begin{equation}
	\delta=\delta_r+\nu \sigma_r, \quad \sigma=\sqrt{\kappa}\sigma_r
	,\quad \Hc=\Hc_{s,b}
	,\quad	\Gamma=\Gamma_D.
\end{equation}
We define $\Hc_{s,b}$ below.
This case of the coupling with $\nu=0$ is considered in 
\cite{Kang2012a}. 

A naive calculation in this case gives a multivalued function for 
$\eta^{\psi}(z)$ 
due to the zero of $\sigma_r(z)$ in the denominator of 
\eqref{Folmula: phi+ = int ...} at a point inside $\Dc$. 
To handle this, we have to depart from working with test
functions from the space $\Hc_s$.
We define the space $\Hc_{s,b}$ as the image of $\Hc_s$ with respect
to the Lie derivative $\Lc_{\sigma_r}$:
\begin{equation}
	\Hc_{s,b}:=\Lc_{\sigma_r}(\Hc_s).
\end{equation}
\index{$\Hc_{s,b}$ and $\Hc_{s,b}'$}
Equivalently, it consists of all functions from $\Hc_s$ whose integrals
along all concentric circles centered at the origin in the unit disk chart are equal to zero.
Thus, the dual space is larger than $\Hc_{s}'$ 
($\Hc_{s}'\subset \Hc_{s,b}'$) 
and in addition contains, for example,
\begin{equation}
	\arg z \in \Hc^{\D}_{s,b}{}'.
\end{equation}
In the logarithmic chart, the same functional can be represented as 
\begin{equation}
	\Re z \in \Hc^{\LL}_{s,b}{}'.
\end{equation}
Thus, the multivalued expressions for $\eta^{\psi}(z)$ below can be
understood in the sense of elements of $\Hc_{s,b}'$.


% ...
%
%
% ...

We present here the expressions for
$\delta$, $\sigma$, $\Gamma_D$, $\eta$, and ${M_1}_t$ in three different charts
(the half-plane, the unit disk, and the logarithmic chart), using
the same method as before. The calculations are similar to the
dipolar case. In fact, it is enough to change some signs and
replace the hyperbolic functions by trigonometric ones.
In contrast to the dipolar case, $\eta$ is multivalued
in the half-plane and unit disk charts. 
From the heuristic point of view, this is not an essential problem. In any chart,
$\eta^{\psi}(z)$ changes its value only by an irrelevant constant
after harmonic continuation around the point $b$.

In the half-plane chart, we have
\begin{equation}
 -\sqrt{\kappa}(1+z^2) \de_z {\eta^+}^{\HH}(z)
 + \mu \left( -\sqrt{\kappa} (1+z^2) \right)' =
 \frac{-i}{z} + i \alpha,
\end{equation}
\begin{equation}
 {\eta}^{\HH}(z) =
 \frac{-2}{\sqrt{\kappa}} \arg z
 - \frac{(\kappa-6)}{2 \sqrt{\kappa}} \arg (1+z^2)
 - \frac{\nu}{\sqrt{\kappa}} \Im \mathrm{arctg}\, z + C,
 \label{Formula: phi - radial with drift in H}
\end{equation}
\begin{equation}\begin{split}
 &{M_1}_t^{\HH}(z) = ({G_t^{-1}}_* \eta)^{\HH}(z)
 =\\=&
 \frac{-2}{\sqrt{\kappa}} \arg G_t^{\HH}(z)
 - \frac{(\kappa-6)}{2 \sqrt{\kappa}} \arg (1 + G_t^{\HH}(z)^2)
 -\\-&
 \frac{\nu}{\sqrt{\kappa}} \Im \mathrm{arctg}\, z
 +\frac{\kappa-4}{2\sqrt{\kappa}} \arg {G_t^{\HH}}'(z) + C
\end{split}\end{equation}
 analogously to
\eqref{Formula: phi+ - chordal with drift}.

The unit disk chart is defined in
(\ref{Formula: tau_H D}), and
\begin{equation}
 \eta^{\D}(z)
 = \frac{-2}{\sqrt{\kappa}} \arg (1-z)
 - \frac{\kappa-6}{2\sqrt{\kappa}} \arg z
 + \frac{\nu}{2\sqrt{\kappa}} \log |z| + C,
\end{equation}
\begin{equation}\begin{split}
 &{M_1}_t^{\D}(z) = ({G_t^{-1}}_* \eta)^{\D}(z)
 =\\=&
 \frac{-2}{\sqrt{\kappa}} \arg (1-G_t^{\D}(z))
 - \frac{\kappa-6}{2\sqrt{\kappa}} \arg G_t^{\D}(z)
 + \frac{\nu}{2\sqrt{\kappa}} \log |G_t^{\D}(z)|
 + \frac{\kappa-4}{2\sqrt{\kappa}} \arg {G_t^{\D}}'(z) + C,
\end{split}\end{equation}
\begin{equation}
 \Gamma_D^{\D}(z,w)
 = \Gamma_D^{\HH}\left( \tau_{\HH,\D}(z),\tau_{\HH,\D}(w) \right)
 = -\frac12 \log \frac{ (z-w)(\bar z - \bar w) }
  { (\bar z-w)(z - \bar w) }.
\end{equation}

The relations for the radial SLE in the logarithmic chart can be
easily obtained from those for the dipolar SLE in the strip chart just by replacing the
hyperbolic functions by their trigonometric analogs:
\begin{equation}
	{\Gamma_{D}^{++}}^{\LL}(z,w)
	= -\frac12 \log \sin \frac{z-w}{2}, \quad
	{\Gamma_{D}^{+-}}^{\LL}(z,\bar w)
	= -\frac12 \log \sin \frac{z-\bar w}{2},
\end{equation}
\begin{equation}
	\Gamma_D^{\LL}(z,w)
	= -\frac12 \log \frac{\sin(\frac{z-w}{2}) \sin(\frac{\bar z-\bar w}{2})}
   {\sin(\frac{\bar z-w}{2}) \sin(\frac{z-\bar w}{2})},
\end{equation}
\begin{equation}
 \eta^{\LL}(z)
 = \frac{-2 }{\sqrt{\kappa}} \arg \mathrm{sin} \frac{z}{2}
 - \frac{\nu}{2\sqrt{\kappa}} \Im z + C,
\end{equation}
\begin{equation}\begin{split}
 &{M_1}_t^{\LL}(z) = ({G_t^{-1}}_* \eta)^{\LL}(z)
 =\\=&
 \frac{-2}{\sqrt{\kappa}} \arg \sin \frac{G_t^{\LL}(z)}{2}
 - \frac{\nu}{2\sqrt{\kappa}} \Im G_t^{\LL}(z)
 +\frac{\kappa-4}{2\sqrt{\kappa}} \arg {G_t^{\LL}}'(z) + C.
\end{split}\end{equation}
The relations above coincide up to a constant with the analogous ones
established in 
\cite{Kang2012a}.


We remark that for $\kappa=6$ we can use the usual space $\Hc_s$ instead of
$\Hc_{s,b}$.

\subsection{Coupling of reverse radial SLE and Neumann GFF}
\label{Section: Coupling of reverse radial SLE and Neumann GFF}

We assume 
\begin{equation}
	\delta=-\delta_r+\nu \sigma_r, \quad \sigma=\sqrt{\kappa}\sigma_r
	,\quad \Hc=\Hc_{s,b}^{*}	
	,\quad \Gamma=\Gamma_N,
\end{equation}
where 
\begin{equation}
	\Hc_{s,b}^* :=
	\left\{f \in\Hc_{s,b} \colon \int\limits_{D^{\psi}} f^{\psi}(z) l(dz) 
	=	0 \right\}.
\end{equation}
We use the space $\Hc_{s,b}^*$ rather than $\Hc_{s}^*$ for the same reasons
as in the previous subsection.

Analogously to the previous cases, we have in the half-plane chart
\begin{equation}
  {\eta^+}^{\HH}(z) =
  \frac{1}{\sqrt{\kappa}} \log z
  - \frac{(\kappa+6)}{4 \sqrt{\kappa}} \log (1 + z^2)
  - \frac{\nu}{2\sqrt{\kappa}} \arctg z + C^+,
\end{equation}
\begin{equation}
	{\eta}^{\HH}(z)
	= \frac{2}{\sqrt{\kappa}} \log |z|
	- \frac{(\kappa+6)}{2 \sqrt{\kappa}} \log |1+z^2|
	- \frac{\nu}{\sqrt{\kappa}} \Re \arctg z + C,
\end{equation}
\begin{equation}\begin{split}
	&{M_1}_t^{\HH}(z)
	= {G_t^{-1}}_* {\eta}^{\HH} (z)
	=\\=&
	\frac{2}{\sqrt{\kappa}} \log |{G_t}^{\HH}(z)|
	- \frac{(\kappa+6)}{2 \sqrt{\kappa}}
		\log \left|1+\left({G_t}^{\HH}(z)\right)^2\right| -
	\frac{\nu}{\sqrt{\kappa}} \Re \arctg {G_t}^{\HH}(z)
	+\\+&
	\frac{\kappa+4}{2\sqrt{\kappa}} \log \left| {G_t}^{\HH}{}'(z) \right| +
	C.
\end{split}\end{equation}

In the logarithmic chart, the same expressions are
\begin{equation}
	\Gamma_{N}^{\LL}(z,w)
	= -\frac12 \log 
	\left(
		\sin \frac{z-w}{2} \sin \frac{z-\bar w}{2}
		\sin \frac{\bar z-w}{2} \sin \frac{\bar z- \bar w}{2}
	\right)
	+\beta(z)+\beta(w),
\end{equation}
\begin{equation}
	{\eta^+}^{\LL}(z) = 
	\frac{1}{\sqrt{\kappa}} \log \sin \frac{z}{2} 
	- \frac{\nu}{4\sqrt{\kappa} } z,
\end{equation}\begin{equation}
	\eta^{\LL}(z) = 
	\frac{2}{\sqrt{\kappa}} \log \left| \sin \frac{z}{2} \right|
	- \frac{\nu}{2\sqrt{\kappa} } \Re z.
\end{equation}
The function $\eta^{\LL}(z)$ satisfies the Neumann boundary condition along the whole
boundary:
\begin{equation}\begin{split}
  &\left. \de_y \eta^{\LL}(x+iy,w) \right|_{y=0 } = -\frac{\nu}{2\sqrt{k}},\quad
  x\in \mathbb{R}
  ,\quad w\in \C,
  \label{Formula: Gamma_DN := ...}
\end{split}\end{equation}

In the unit disk chart, $\eta$ takes the form
\begin{equation}
 \eta^{\D}(z)
 = \frac{2}{\sqrt{\kappa}} \log |1-z|
 + \frac{\kappa + 	6}{2\sqrt{\kappa}} \log |z|
 - \frac{\nu}{2\sqrt{\kappa}} \arg z + C,
\end{equation}

If $\nu=0$, we do not need to introduce the space $\Hc_{s,b}$ and
can use $\Hc_s^*$ because there is no branch point at $b$.

\subsection{Coupling with twisted GFF}
\label{Section: Coupling with twisted GFF}

We assume
\begin{equation}
	\delta=\delta_r+\nu \sigma_r,\quad
	\sigma=\sqrt{\kappa}\sigma_r,\quad \Hc=\Hc_{s,b}^{\pm}
	,\quad \Gamma=\Gamma_{\text{tw},b}.
	\label{Formula: twisted assumptions}
\end{equation}
We define $\Hc_{s,b}^{\pm}$ and $\Gamma_{\text{tw}}$ below.

This model is similar to the one from Section
\ref{Section: Coupling of forward dipolar SLE and combined Dirichlet-Neumann
GFF}. At the algebraic level, as will
be shown below, it is enough to formally replace all hyperbolic functions in
the dipolar case in the strip chart $\SSS$ by the corresponding trigonometric
functions in order to obtain the relations for the radial SLE in the
logarithmic chart $\LL$. But at the analytic level, we have to consider the
correlation functions that are double-valued on $\Dc$ and change sign
after going once around the interior point $b$. 
This construction has been
considered before under the name `twisted CFT', as we were informed by Nam-Gyu Kang.

Let 
$\Dc_b^{\pm}$ 
be the double cover of 
$\Dc\setminus \{b\}$ 
and let
$\Hc_{s,b}^{\pm}$
be the space of test functions 
$f\colon \Dc_b^{\pm}\map\mathbb{R}$
as in Section 
\ref{Section: Test function}
with the extra condition
$f(z_1)=-f(z_2)$, 
where $z_1$ and $z_2$ are two points of 
$\Dc^{\pm}_b$
corresponding to the same point of 
$\Dc\setminus\{b\}$. 
Thus, in the logarithmic chart, we have
\begin{equation}
 f^{\LL}(z)=f^{\LL}(z+4\pi k) = - f^{\LL}(z+ 2\pi k),\quad k\in \mathbb{Z},\quad z\in\HH.
 \label{Formula: f = -f  = f}
\end{equation}
Such functions are $4\pi$-periodic and $2\pi$-antiperiodic. In particular,
$f^{\LL}$ is not of compact support in the logarithmic chart, but we require
that the support is compact in $\Dc_b^{\pm}$. The dual space
$\Hc_{s,b}^{\pm}{}'$
also consists of $4\pi$-periodic and $2\pi$-antiperiodic functionals.

We define the covariance functional by
\index{$\Gamma_{\text{tw},b}$}
\begin{equation}
 \Gamma_{\text{tw},b}^{\LL}(z,w)= -\frac12\log \frac{
  \mathrm{tg}\frac{z-w}{4} \mathrm{tg}\frac{\bar z-\bar w}{4}}
  {\mathrm{tg}\frac{\bar z-w}{4} \mathrm{tg}\frac{z -\bar w}{4}},\quad
  z,w\in \HH, \quad z\neq w+2\pi k,\quad k\in\mathbb{Z}.
 \label{Formula: G_d^S = ...}
\end{equation}
in the logarithmic chart. Observe that
\begin{equation}
	\Gamma_{\text{tw},b}^{\LL}(z,w) 
	= \Gamma_{\text{tw},b}^{\LL}(z+4\pi k,w) 
	= -\Gamma_{\text{tw},b}^{\LL}(z+2\pi k,w),\quad k\in \mathbb{Z}
  ,\quad z,w\in\HH.
\end{equation}

In the unit disk chart, the covariance $\Gamma_{\text{tw}}^{\D}$ has the form
\begin{equation}
	\Gamma^{\D}_{\text{tw},b}(z,w) = -\frac12\log \frac
  {(\sqrt{z}-\sqrt{w})(\sqrt{\bar z}-\sqrt{\bar w})(\sqrt{z}+\sqrt{\bar w})(\sqrt{\bar z}-\sqrt{w})}
  {(\sqrt{z}+\sqrt{w})(\sqrt{\bar z}+\sqrt{\bar w})(\sqrt{z}-\sqrt{\bar w})(\sqrt{\bar z}+\sqrt{w})},
\end{equation}
and in the half-plane chart it has the form
\begin{equation}
	\Gamma_{\text{tw},b}^{\HH}(z)
	= -\frac12 \log \frac{(z-w)(\bar z-\bar w)
  (1+\bar z w+\sqrt{1+\bar z^2}\sqrt{1+w^2})(1+ z \bar w+\sqrt{1+ z^2}\sqrt{1+\bar w^2})}
  {(\bar z-w)(z-\bar w)
  (1+zw+\sqrt{1+z^2}\sqrt{1+w^2})(1+\bar z \bar w+\sqrt{1+\bar z^2}\sqrt{1+\bar w^2})}.
\end{equation}
It is double-valued because of the square root, and analytic continuation
around the center changes its sign.

The covariance $\Gamma_{\text{tw,b}}^{\LL}$ satisfies the Dirichlet boundary
conditions and tends to zero as one of the variables tends to the center point
$b$ ($\infty$ in the $\LL$ chart):
\begin{equation}
	\left.\Gamma_{\text{tw},b}^{\LL}(x,w)\right|_{x\in \mathbb{R}}=0,\quad
	\lim\limits_{y\map +\infty} \Gamma_{\text{tw}}^{\LL}(x+iy,w) = 0,\quad x\in
	\mathbb{R}, \quad w\in \HH;
\end{equation}
\begin{equation}
	\left.\Gamma_{\text{tw},b}^{\D}(z,w)\right|_{|z| = 1 }=0,\quad
	\lim\limits_{z\map 0} \Gamma_{\text{tw}}^{\D}(z,w) = 0, \quad w\in \D.
\end{equation}

The $\sigma$-symmetry property 
\begin{equation}
	\Lc_{\sigma_r} \Gamma_{\text{tw}}(z,w)=0
	\label{Formula: L_s_r G = 0}
\end{equation}
holds.

We call the GFF 
$\Phi(\Hc_{s,b}^{\pm},\Gamma_{\text{tw},b},\eta)$
the \emph{twisted Gaussian free field}
$\Phi_{\text{tw}}$ 
\index{twisted Gaussian free field}.
%
Similarly to the dipolar case in Section 
\ref{Section: Coupling of forward dipolar SLE and combined Dirichlet-Neumann GFF},
the following proposition
can be proved.
\begin{proposition}
With the assumptions as in 
\eqref{Formula: twisted assumptions}
the twisted GFF $\Phi_{\text{tw}}$
can be coupled if and only if $\kappa=4$ and $\nu=0$.
\label{Proposition: Twisted coupling}
\end{proposition}

The {\it proof} in the logarithmic chart essentially repeats the proof of
Proposition \ref{Proposition: D-N coupling}.

%

We give here the expressions for $\eta$ in the logarithmic, unit disk, and half-plane charts:
\begin{equation}
 \eta^{\LL}(z) = - 2\arg \tg \frac{z}{4} + C.
\end{equation}
\begin{equation}
 {\eta}^{\D}(z)
 = -2 \arg \frac{1-\sqrt{z}}{1+\sqrt{z}} + C
 = 4\Im \,\mathrm{arctgh}\sqrt{z} + C.
\end{equation}
\begin{equation}
 \eta^{\HH}(z)
 = - 2\arg \frac{z}{1+\sqrt{1+z^2}} + C.
\end{equation}
From this relation, it is clear that $\eta$ is antiperiodic.



The function $\eta^{\LL}$ obtained above also has the property
\eqref{Formula: f = -f  = f}. 
Thus, the construction of the level (flow) lines, which we discussed in the
introduction to Chapter 
\ref{Chapter: Coupling},
can be performed for both layers simultaneously, and the lines will be
identical. In particular, this means that the line can turn around the central
point and appear in the second layer but cannot intersect itself. This agrees
with the property of the SLE slit of avoiding self-intersections.

\section{Chordal case with a fixed time change}
\label{Section: Chordal case with fix time change}

\begin{figure}[ht]
\centering
  \begin{subfigure}[t]{0.33\textwidth}
		\centering
    \includegraphics[width=5cm,keepaspectratio=true]
    	{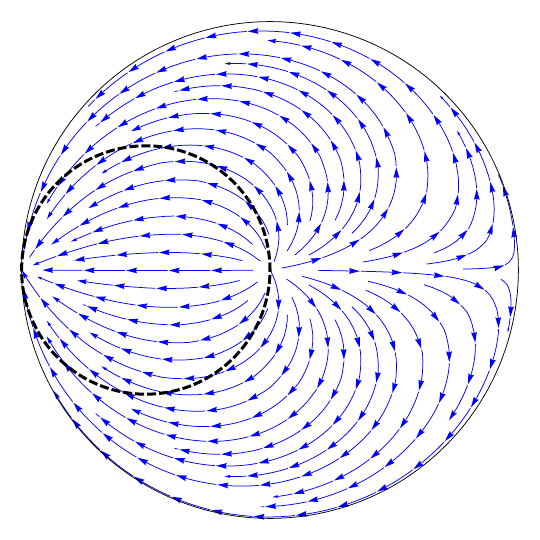}
    \caption{Flow lines of the vector field $\delta$.}
  \end{subfigure}%
  ~
  \begin{subfigure}[t]{0.33\textwidth}
		\centering
    \includegraphics[width=5cm,keepaspectratio=true]
    	{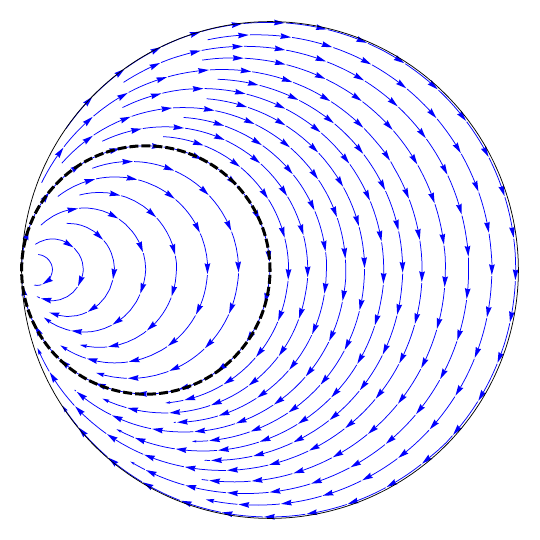}
  	\caption{Flow lines of the vector field $\sigma$.}
	\end{subfigure}%
	~
  \begin{subfigure}[t]{0.33\textwidth}
		\centering
    \includegraphics[width=5cm,keepaspectratio=true]
    	{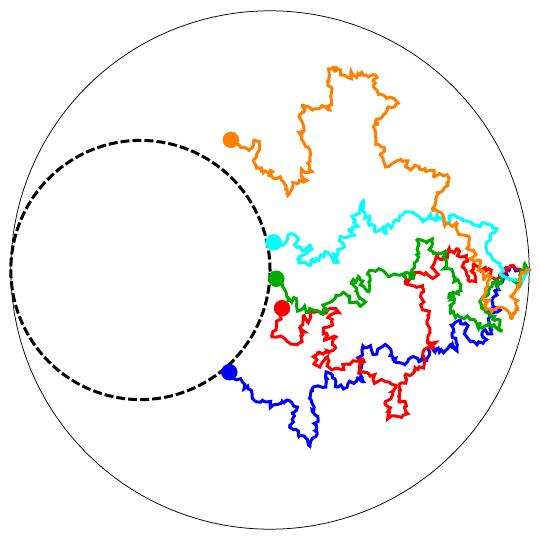}
    \caption{Five samples of the SLE slits for $\kappa=2$ in the unit disk
		chart. The slits tend to a random point near the boundary of the
		invariant domain $\mathcal{I}$.}
  \end{subfigure}
 	\caption{Flow lines of the vector fields $\delta$ and $\sigma$ from
 	\eqref{Formula: delta and sigma for fix time change},
 	and corresponding ($\delta,\sigma$)-SLE slit samples for $\xi=1$ in the unit
 	disk chart. The dashed line is the boundary of the invariant domain
 	$\mathcal{I}$.} 
\end{figure} 	

\begin{figure}[h]
  \begin{subfigure}[t]{0.33\textwidth}
		\centering
    \includegraphics[width=5cm,keepaspectratio=true]
    	{Pictures/Vector_fileds_v=_-2_-_0_.pdf}
    \caption{Flow lines of the vector field $\delta$.}
  \end{subfigure}%
  ~
  \begin{subfigure}[t]{0.33\textwidth}
		\centering
    \includegraphics[width=5cm,keepaspectratio=true]
    	{Pictures/Vector_fileds_v=-_-1_.pdf}
  	\caption{Flow lines of the vector field $\sigma$.}
	\end{subfigure}%
	~
  \begin{subfigure}[t]{0.33\textwidth}
		\centering
    \includegraphics[width=5cm,keepaspectratio=true]
    	{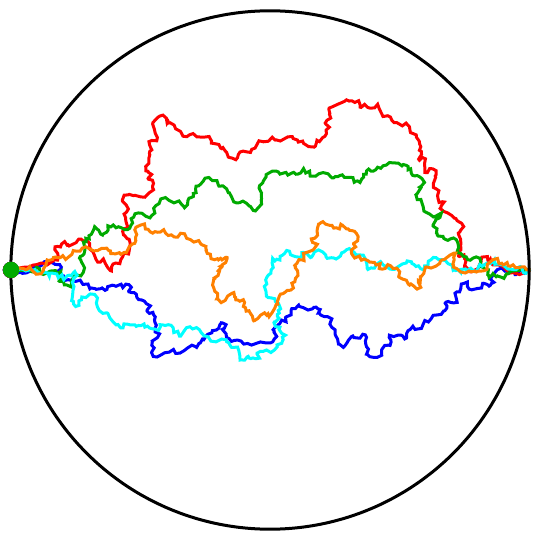}
    \caption{Five samples of the SLE slits for $\kappa=2$. The slits have the
    same law as the chordal ones in Fig.
 	\ref{Figure: Chordal SLE illustration}.}
  \end{subfigure}
 	\caption{Vector fields 
 	\eqref{Formula: delta and sigma for fix time change}
 	and corresponding ($\delta,\sigma$)-SLE slit samples for $\xi=-1$ in the unit
 	disk chart.}
\end{figure}

We assume
\begin{equation}
	\delta = \pm 2 \ell_{-2} + 2 \xi \ell_0,\quad
	\sigma = -\sqrt{\kappa}\ell_{-1}
	,\quad \xi\in\mathbb{R}\setminus\{0\}. 
	\label{Formula: delta and sigma for fix time change}
\end{equation}
We consider the forward case (`$+$' in `$\pm$' pair) first.
The stochastic equation with such parameters appeared in Theorem
\ref{Theorem: cordal -> fix time change and fixed boundary point}
and Theorem
\ref{Theorem: ppS -> simple cases}. 
We do not consider the deterministic L\"owner
chain, because only the stochastic version is motivated by these theorems. According
to Theorem
\ref{Theorem: cordal -> fix time change and fixed boundary point}
this choice of $\delta$ and $\sigma$ is
equivalent to the chordal one for any $\kappa>0$ if we perform a fixed time
reparametrization and a scale transformation. However, this case and the chordal case are not essentially equivalent in the sense
described in Section
\ref{Section: Equivalence and normalization of slit L chains}. 
The stochastic equation in the half-plane chart is
\begin{equation}
	\dS G_t^{\HH}(z)
	= \frac{2}{G_t^{\HH}(z)} dt + 2 \xi G_t^{\HH}(z) dt 
	- \sqrt{\kappa} 1 \dS B_t.
\end{equation}

%
We explain now the relation to the chordal case described in Section
\ref{Section: Chordal Loewner equation}.
Let 
$\xi \in\mathbb{R}\setminus \{0\}$ 
and let
$\{\tilde G_{\tilde t}\}_{t\in[0,+\infty)}$
be the chordal stochastic flow
(chordal SLE (\ref{Formula: delta and sigma chordal})).
Define
\begin{equation}
 \tilde G_{\tilde t}^{\HH}(z) = e^{2\xi \tilde t} G_{\lambda(\tilde t)}^{\HH}(z)
\end{equation}
in the half-plane chart and
assume that
\begin{equation}\begin{split}
 &\lambda(\tilde t) := \frac{1 - e^{-4\xi \tilde t}}{4\xi};\\
 &\lambda:~ [0,+\infty)\map[0,(4\xi)^{-1}), \quad \xi > 0;\\
 &\lambda:~ [0,+\infty)\map[0,+\infty), \quad \xi < 0.
\end{split}\end{equation}
This choice of $\lambda$ corresponds to
$c_{\tilde t}=-2\xi \tilde t$
in Theorem \ref{Theorem: cordal -> fix time change and fixed boundary point}.
We note that the time reparametrization is not random.

The flow $\{\tilde G_{\tilde t}\}_{t\in[0,+\infty)}$ satisfies the
autonomous equation
\eqref{Formula: Slit hol stoch flow Strat}
with $\delta$ and $\sigma$ from
\eqref{Formula: delta and sigma for fix time change},
which are the vector fields from the second row of Table
\ref{Table: some simple cases}
and a special case of
(\ref{Formula: 3})
with $x_{\tilde t} = - 4 \xi e^{2c_{\tilde t}}$ and $y_{\tilde t}=0$.

There is a common zero of $\delta$ and $\sigma$ at infinity in the half-plane
chart and at $z=-1$ in the unit disk chart, so infinity is a fixed point:
$\tilde G^{\HH}_{\tilde t}: \infty \map \infty$. However, in contrast to the chordal
case, the coefficient of $z^1$ in the Laurent series is not $1$ but
$e^{2 \xi \tilde t}$.

The vector field $\delta$ is of radial type if $\xi> 0$, and of dipolar type if
$\xi< 0$.
In the second case, the equation induces exactly the same
measure as the chordal stochastic flow but with a different time
parametrization.
If $\xi>0$, the measures also coincide when the chordal stochastic flow is
stopped at the time $t=(4\xi)^{-1}$.



Due to this observation, it is reasonable to expect that the Dirichlet GFF
($\Hc=\Hc_s$ and $\Gamma=\Gamma_D$)
coupled with this kind of ($\delta,\sigma$)-SLE is the same as in the chordal
case, because it is supposed to induce the same random law of the flow lines.
Indeed, $\sigma$ from
\eqref{Formula: delta and sigma for fix time change}
coincides with that from the chordal case; hence, $\eta$, defined by
(\ref{Formula: j = L eta}),
with
$\alpha=0$ (see the table) also coincides with
(\ref{Formula: eta - chordal with drift in H})
with $\nu=0$. Thus, the martingales are the same as in the chordal case. 
The same observation can be made in the reverse case (`$-$' in `$\pm$' pair)
and Neumann GFF ($\Hc=\Hc_s^*$ and $\Gamma=\Gamma_N$).

\section{The case with one fixed point}
\label{Section: Case with one fixed point}

\begin{figure}[ht]
\centering
  \begin{subfigure}[t]{0.33\textwidth}
		\centering
    \includegraphics[width=5cm,keepaspectratio=true]
    	{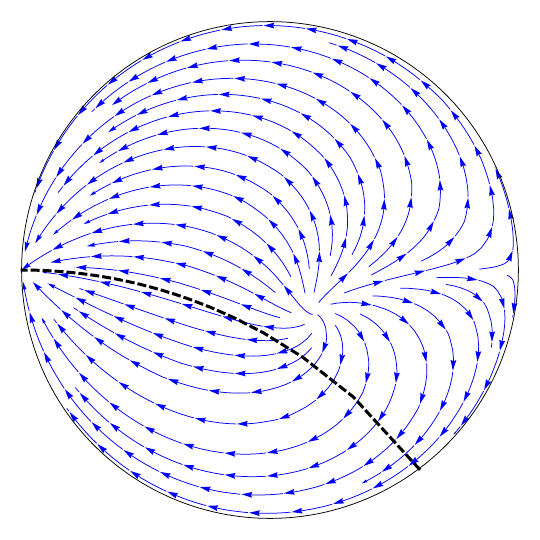}
    \caption{Flow lines of the vector field $\delta$.}
  \end{subfigure}%
  ~
  \begin{subfigure}[t]{0.33\textwidth}
		\centering
    \includegraphics[width=5cm,keepaspectratio=true]
    	{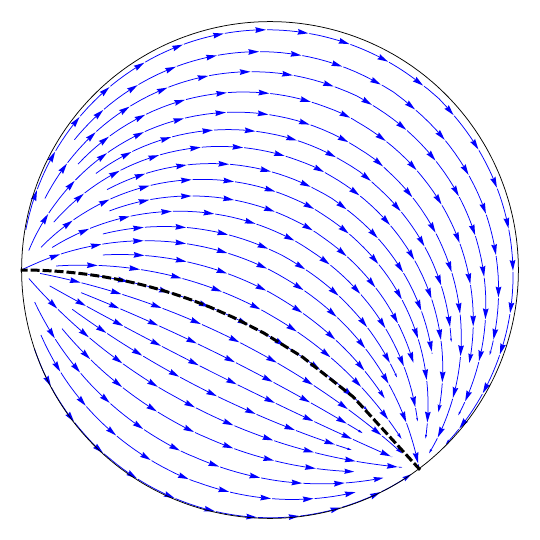}
  	\caption{Flow lines of the vector field $\sigma$.}
	\end{subfigure}%
  ~
  \begin{subfigure}[t]{0.33\textwidth}
		\centering
    \includegraphics[width=5cm,keepaspectratio=true]
    	{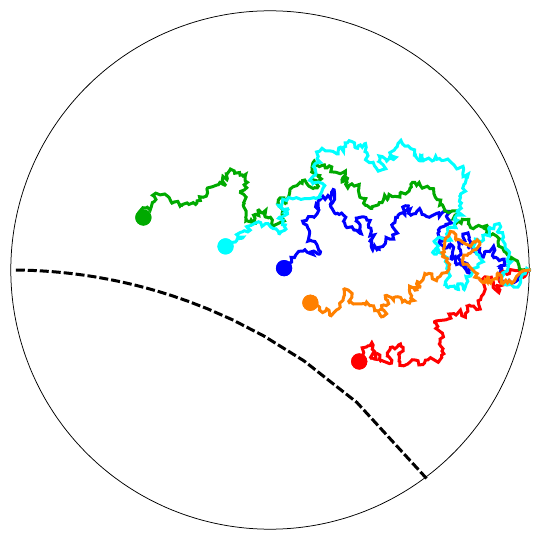}
    \caption{Five slit samples for $(\delta,\sigma)$-SLE. The slits tend to a
    random point inside the disk near the boundary of
    $\tilde {\mathcal{I}}$.}
   \end{subfigure}%
	\caption{Flow lines of the vector fields $\delta$ and $\sigma$ from
	\eqref{Formula: assumtions 1FP}
	and corresponding slit samples in the unit disk
	chart. The dashed line is the boundary of the invariant domain
	$\tilde {\mathcal{I}}$.}
\end{figure}

We assume
\begin{equation}
	\delta = \pm 2\ell_{-2} + \kappa \ell_{-1} \pm 2\xi \ell_{0},\quad
	\sigma = -\sqrt{\kappa}(\ell_{-1} + 2 \ell_{0})
	,\quad \xi\in\mathbb{R}
	,\quad \kappa>0.
	\label{Formula: assumtions 1FP}
\end{equation}
These vector fields have a common zero $b$ at infinity,
$\psi^{\HH}(b)=\infty$,
in the half-plane chart.
Thus, $b$ is a fixed point ($G_t(b)=b$) of the map $G_t$ for
$t\in[0,+\infty]$.

As in the previous section, we consider the forward case 
(`$+$' in `$\pm$' pair) first. There is also a mirror case 
\begin{equation}
	\delta = \pm 2\ell_{-2} - \kappa \ell_{-1} \pm 2\xi \ell_{0},\quad
	\sigma = -\sqrt{\kappa}(\ell_{-1} - 2 \ell_{0})
	,\quad \xi\in\mathbb{R}
	,\quad \kappa>0.
\end{equation}
that can be considered analogously.

The stochastic equation with such parameters appeared as the third case in
Theorem
\ref{Theorem: cordal -> fix time change and fixed boundary point}.
The corresponding stochastic equation 
\begin{equation}\begin{split}
 &d^S G_t^{\HH}(z) =
  \left(
   \frac{2}{G_t^{\HH}(z)} + \kappa + 2 \xi G_t^{\HH}(z)
  \right)dt -
  \sqrt{\kappa} \left(1 + 2 G_t^{\HH}(z) \right) d^S B_t, 
 \label{Formula: moved SLE in half-plane}
\end{split}\end{equation}
can be obtained from the chordal equation after a time change
$\lambda_{\tilde t}= e^{4\xi \tilde t - 4\sqrt{\kappa} \tilde B_{\tilde t}}$
and multiplication by
$e^{2 \xi \tilde t - 2\sqrt{\kappa} \tilde B_{\tilde t}}$.

On the other hand, the case from the fourth row of Table
\ref{Table: some simple cases}
is essentially equivalent to this one. Indeed, let us apply the transformation
$\mathscr{R}_{c}$ with $c=-1$ corresponding to the M\"obius automorphism
\begin{equation}
	H_{-1}[\ell_1]\colon \Dc\map\Dc
\end{equation}
having the form
\begin{equation}
 H_{-1}[\ell_1]^{\HH}(z)=\frac{z}{1+z}.
\end{equation}
in the half-plane chart. It maps the fixed point $z=\infty$ to $z=1$ in the
half-plane chart and keeps the origin and the normalization 
\eqref{Formula: chordal normalization for G in H} 
unchanged. It results in
\begin{equation}\begin{split}
	&\tilde G_t := H_{-1}[\ell_1] \circ G_t \circ H_{-1}[\ell_1]^{-1}, \\
	&\tilde \delta^{\HH}(z)
	:= H_{-1}[\ell_1]_* \delta^{\HH}(z)
	= \frac{2}{z} + \kappa-6 + 2(3-\kappa+\xi)z + (- 2 + \kappa - 2\xi)z^2, \\
	&\tilde \sigma^{\HH}(z)
	:= H_{-1}[\ell_1]_* \sigma^{\HH}(z)
	= -\sqrt{\kappa}(1-z^2).
\end{split}\end{equation}
That corresponds to
\begin{equation}\begin{split}
	&\tilde \delta = 2\ell_{-2} + (\kappa-6) \ell_{-1} + 2 (3 - \kappa + \xi)
	\ell_{0} +  (- 2 + \kappa - 2\xi)\ell_{1}, \quad
	\xi\in\mathbb{R},\\
	&\tilde \sigma = -\sqrt{\kappa}(\ell_{-1} - \ell_{1}),
	\label{Formula: formula 3}
\end{split}\end{equation}
which are exactly the vector fields $\delta$ and $\sigma$ from the fourth row of
Table \ref{Table: some simple cases}. The fifth row corresponds to the mirror
case, for which the point $z=1$ is fixed.

We do not consider the deterministic L\"owner
chain for
\eqref{Formula: assumtions 1FP}, 
because only the stochastic version is motivated by Theorems
\ref{Theorem: ppS -> simple cases}
and
\ref{Theorem: cordal -> fix time change and fixed boundary point}.
We note that the subsurface
$\tilde {\mathcal{I}}\subset \Dc$
defined in the half-plane chart by
\begin{equation}
	\psi^{\HH} (\tilde {\mathcal{I}}) 
	= \{ z\in\HH:\Re(z) > -\frac12 \}
\end{equation}
is invariant
($G_t^{-1}(\tilde {\mathcal{I}}) \subset \tilde{\mathcal{I}}$)
if and only if
$\xi\geq \kappa-4$.
In order to see this, it is enough to calculate the real parts of
\begin{equation}\begin{split}
 \delta^{\HH}(z) = \frac{2}{z} + \kappa + 2\xi z, \quad
 \sigma^{\HH}(z) = -\sqrt{\kappa}(1+2z),
 \label{Formula: delta and sigma 1FP in H}
\end{split}\end{equation}
which are the horizontal components of the vector fields
on the boundary of
$\psi^{\HH}(\tilde{\mathcal{I}})$ in $\HH$,
$\{z\in\HH:~\mathrm{Re}(z)=-\frac12\}$,
\begin{equation}\begin{split}
 \Re\left(\tilde \delta^{\HH}\left(-\frac12 + i h\right)\right) =&
 \Re\left(\frac{2}{-\frac12 + i h} + \kappa
 + 2\xi\left(-\frac12 + i h\right)\right)
 =\\=&
 -\frac{1}{h^2+\frac14} + \kappa - \xi,\\
 \Re \left(\tilde \sigma^{\HH}\left(-\frac12 + i h\right) \right)=&
 \Re \left(\sqrt{\kappa}\left(1+2\left(-\frac12
 + i h\right)\right)\right)=0,\quad
 h>0.
\end{split}\end{equation}
The first number is negative for all values of $h$ if and only if
$\xi\geq \kappa - 4$.

We remark that the $H_{-1}[\ell_1]$-transform has the invariant subsurface
$\tilde{\mathcal{I}}:=H_{-1}[\ell_1](\tilde {\mathcal{I}}) \subset \Dc$
for the $(\delta,\sigma)$-SLE above, which is the upper half of the unit disk:
\begin{equation}
	\psi^{\HH}(\tilde{\mathcal{I}}) =\{ z\in\HH\colon |z|<1 \}
\end{equation}

Similarly to the previous section, it is reasonable to expect that $\eta$ for the
coupling with the Dirichlet GFF is the same as in the chordal case, because it
is supposed to induce the same random law of the flow lines. Indeed, the
solution to (\ref{Formula: j^+ = L eta^+}),
with $\sigma$ and $\alpha$ as in the fifth row of the table, is
\begin{equation}
 {\eta^+}^{\HH}(z)
 = \frac{i}{\sqrt{\kappa}} \log z
 + i\frac{\kappa-6}{2\sqrt{\kappa}} \arg (1-z) + C^+.
\end{equation}
Thus,
\begin{equation}
 \eta^{\HH}(z)
 = \frac{-2}{\sqrt{\kappa}} \arg z
 - \frac{\kappa-6}{\sqrt{\kappa}} \arg (1-z) + C.
\end{equation}
After the $\mathscr{R}_{-1}$-transform for $\tilde \delta$ and $\tilde \sigma$, we have
\begin{equation}
 \tilde \eta^{\HH}(z)
 = \frac{-2}{\sqrt{\kappa}} \arg z + C.
\end{equation}
The last relation coincides with
(\ref{Formula: eta - chordal with drift in H})
with $\nu=0$. Recall that $\Gamma_D$ is invariant under M\"obius transforms,
in particular, under $H_{-1}[\ell_1]$. 

Analogous statements hold for the reverse case and the Neumann GFF.

\section{Degenerate case}
\label{Section: Degenerate case}

\begin{figure}[ht]
\centering
  \begin{subfigure}[t]{0.33\textwidth}
		\centering
    \includegraphics[width=5cm,keepaspectratio=true]
    	{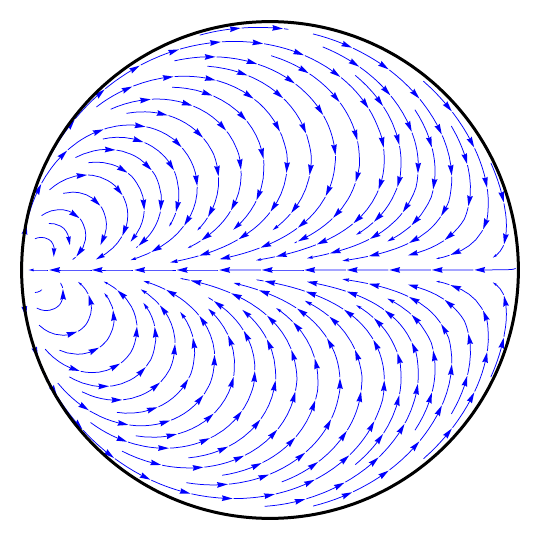}
    \caption{Flow lines of the vector field $\delta$.} 
  \end{subfigure}%
  ~
  \begin{subfigure}[t]{0.33\textwidth}
		\centering
    \includegraphics[width=5cm,keepaspectratio=true]
    	{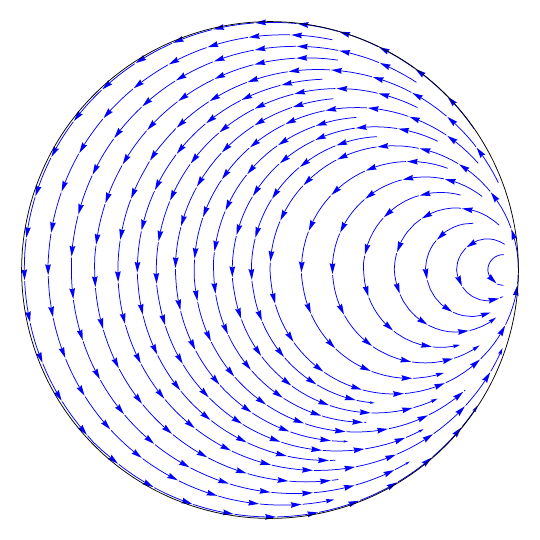}
  	\caption{Flow lines of the vector field $\sigma$.}
	\end{subfigure}%
	~
  \begin{subfigure}[t]{0.33\textwidth}
		\centering
    \includegraphics[width=5cm,keepaspectratio=true]
    	{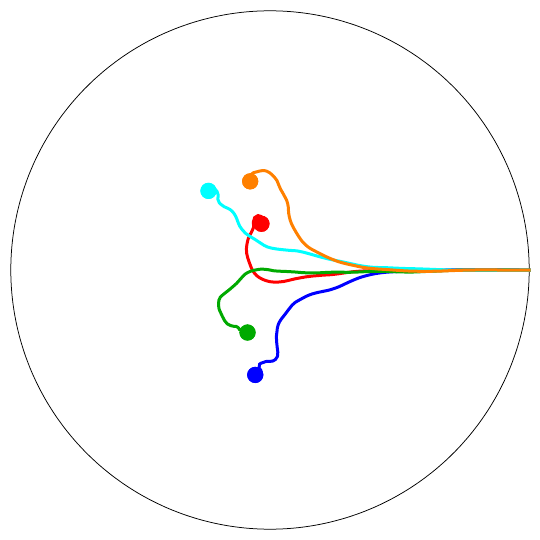}
    \caption{Five samples of the degenerate SLE slits. The slits look smooth and
    tend to a random point inside the disk.
		All slits are close to the trivial straight line for small $t$.}
  \end{subfigure}
	\caption{Vector fields 
	\eqref{Formula: degeneate SLE delta and sigma}
	and corresponding ($\delta,\sigma$)-SLE slit samples in the unit
	disk chart.
	\label{Figure: degenerate SLE}}
\end{figure}

In this section, we consider the special case $\sigma(a)=0$:
\begin{equation}\begin{split}
	\delta =&
	\delta_{-2} \ell_{-2} + \delta_{-1} \ell_{-1} + \delta_0 \ell_{0} 
	+ \delta_1 \ell_{1},\quad
	\delta_{-2},\delta_{-1}, \delta_0, \delta_1 \in \mathbb{R},\quad
	\delta_{-2}\neq 0\\
	\sigma =&
	\sigma_0 \ell_{0} + \sigma_1 \ell_{1},\quad
	\sigma_0, \sigma_1 \in \mathbb{R},\quad \sigma_{-1}= 0,\\
	\label{Formula: delta and sigma for degenerate Lowner chain}
\end{split}\end{equation}
We call this choice the \emph{degenerate L\"owner equation}.
\index{degenerate L\"owner equation} 
This case was excluded above because most of the propositions, such as
Theorems
\ref{Theorem: The master theorem}
and
\ref{Theorem: Absolute continuity of SLE},
fail, and we cannot construct a bridge from well-studied classical L\"owner
equations to this one as above.
In particular, it is not possible to show that the hulls 
$\{\K_t\}_{t\in[0,+\infty)}$ 
are curve-generated using the same methods.
However, we expect that the hulls are smooth curves for continuous driving
functions. We do not present here any systematic investigation of the degenerate
case. Instead, we consider some special cases and apply the methods of
numerical simulation considered in Chapter
\ref{Chapter: Numerical Simulation}
for the stochastic version.

If we assume, for example,
\begin{equation}
	\delta=-2 \ell_{-2},\quad \sigma=\ell_0
\end{equation}
then the chain is a composition of $H_t[-2 \ell_{-2}]$ and some M\"obius
automorphism. Thus, the hull is trivial: it is just a vertical straight line
segment from the origin in the half-plane chart.

A more interesting case is 
\begin{equation}
	\delta=-2 \ell_{-2},\quad \sigma=\ell_1.
	\label{Formula: degeneate SLE delta and sigma}
\end{equation}
A numerical simulation for $u_t=B_t$ is given in Fig. 
\ref{Figure: degenerate SLE}.

\begin{remark}
The ($\delta,\sigma$)-SLE with
\eqref{Formula: degeneate SLE delta and sigma}
induces a random law on \emph{smooth} planar curves that
possesses the property of conformal invariance and a version of the Markov property.
There is no contradiction with the known fact that conformal invariance and
the Markov property necessarily lead to fractal SLE curves, because we have relaxed the
condition that fixes the endpoint of the curve.
\end{remark}

Other choices of degenerate $\delta$ and $\sigma$ give similar pictures or
trivial (non-random) slits.

%% file: Stochastic_calculus_stuff.tex
\chapter{Some relations from stochastic calculus}
\label{Appendix: Some relations from stochastic calculus}

We refer to 
\cite{Oksendal2003},
\cite{Protter2004a},
\cite{Gardiner1982}
for the definitions of the It\^{o} and Stratonovich integrals
and other details.

A stochastic differential equation in the Stratonovich form 
\begin{equation}
 \dS X_t = a_t dt + b_t \dS B_t
\end{equation}
is, by definition, a shorthand for the Stratonovich integral equation 
\begin{equation}
 \int\limits_0^T \dS X_t = \int\limits_0^T a_t dt + \int\limits_0^T  b_t \dS
 B_t.
\end{equation}
The same holds for the It\^{o} differential and integral. We use the notation $\dS$ and $\dI$ for these cases, respectively. The advantage of the Stratonovich form is that the usual rules of differentiation hold. The advantage of the It\^{o} form becomes apparent when working with martingales.

In this text, we do not use differential equations that contain both It\^{o} and Stratonovich differentials. Instead, we use the following relation between the integrals:
\begin{equation}
 \int\limits_{0}^T F(X_t,t) \dS B_t =
 \int\limits_{0}^T F(X_t,t) \dI B_t 
 + \frac12 \int\limits_{0}^T b_t \de_1 F(X_t,t) dt.
 \label{Formula: Strat to Ito}
\end{equation}
The last term can also be expressed in terms of the covariation
\begin{equation}
 \int\limits_{0}^T B_t \de_1 F(X_t,t) dt = \langle F(X_T) ,B_t \rangle.
\end{equation}

For example, to obtain 
(\ref{Formula: Slit hol stoch flow Ito})
from
(\ref{Formula: Slit hol stoch flow Strat}),
we set 
\begin{equation}
 X_t:=G_t(z),\quad
 b_t:=\sigma(G_t(z)),\quad
 F(x_t,t):=\sigma(X_t)=\sigma(G_t(z)).
\end{equation}
Then
\begin{equation}
 \de_1F(X_t,t) = \sigma'(G_t(z)),
\end{equation}
and
\begin{equation}
 \int\limits_0^T \sigma(G_t(z)) \dS B_t =
 \int\limits_0^T \sigma(G_t(z)) \dI B_t +
 \frac12 \int\limits_0^T \sigma(G_t(z))\sigma'(G_t(z)) dt.
\end{equation}
It is enough now to add $\int\limits_0^T \delta(G_t(z)) dt$ to both sides to obtain the right-hand sides of the integral forms of 
(\ref{Formula: Slit hol stoch flow Ito})
and
(\ref{Formula: Slit hol stoch flow Strat}).

We have also used the fact that for any monotonic continuously differentiable function 
$\lambda:[0,\tilde T]\map[0,T]$
\begin{equation}
 \tilde B_{\tilde T} := 
 \int\limits_0^{\tilde T} \dot \lambda^{\frac12}_{\tilde t} \dI B_{\lambda_{\tilde t}}
 = \int\limits_0^{\lambda_{\tilde T}} \dot \lambda^{\frac12}_{\lambda^{-1}_t} \dI B_{t}
\end{equation}
has the same law as $B_{\tilde T}$.
In differential form, this relation is
\begin{equation}
 \dI \tilde B_{\tilde t} = \dot \lambda^{\frac12}_{\tilde t} \dI B_{\lambda_{\tilde t}}.
 \label{Formula: dB = lambda dB}
\end{equation}

We need to reformulate the relation 
(\ref{Formula: dB = lambda dB})
in the Stratonovich form.
Now let $\lambda$ satisfy 
\begin{equation}
 \dS \dot \lambda_{\tilde t} 
 = x_{\tilde t} d \tilde t + y_{\tilde t} \dS \tilde B_{\tilde t}.
 \label{Formula: d lambda = a dt + b dB}
\end{equation}
\begin{equation}\begin{split}
 \int\limits_0^{\tilde T}  \dS \tilde B_{\tilde t}
 =&\int\limits_0^{\tilde T}  \dI \tilde B_{\tilde t}
 = \int\limits_0^{\lambda_{\tilde T}} \dot \lambda^{\frac12}_{\lambda^{-1}_t} \dI B_{t}
 = \int\limits_0^{\lambda_{\tilde T}} \dot \lambda^{\frac12}_{\lambda^{-1}_t} \dS B_{t}
 - \frac12 
 \langle \dot \lambda^{\frac12}_{\tilde T}, B_{\lambda_{\tilde T}}  \rangle
 =\\=& 
 \int\limits_0^{\tilde T} \dot \lambda^{\frac12}_{\tilde t} \dS B_{\lambda_{\tilde t}}
 - \frac12 
 \langle 
  \dot \lambda^{\frac12}_{\tilde T}, 
  \int\limits_0^{\tilde T} \lambda^{-\frac12}_{\tilde t} d\tilde B_{\tilde t}
 \rangle
 = \int\limits_0^{\tilde T} \dot \lambda^{\frac12}_{\tilde t} \dS B_{\lambda_{\tilde t}}
 - \frac12 \int\limits_0^{\tilde T} 
 \frac12 \dot \lambda^{-\frac12}_{\tilde t} y_{\tilde t} 
 \dot \lambda^{-\frac12}_{\tilde t} d \tilde t
 =\\=&
  \int\limits_0^{\tilde T} \dot \lambda^{\frac12}_{\tilde t} \dS B_{\lambda_{\tilde t}}
 - \frac14 \int\limits_0^{\tilde T} 
 \frac{y_{\tilde t}}{\dot \lambda_{\tilde t} } d \tilde t.
\end{split}\end{equation}
Thus, we conclude that
\begin{equation}
 \dS \tilde B_{\tilde t} = 
 \dot \lambda_{\tilde t}^{\frac12} \dS B_{\lambda_{\tilde t}} - 
 \frac14 \frac{y_{\tilde t}}{\dot \lambda_{\tilde t} } d \tilde t.
 \label{Folrmula: dB = lambda dB - 1/4 b/lambda dt}
\end{equation}